\documentclass[11pt]{article}
\usepackage{enumerate}
\usepackage{url}
\usepackage{pdflscape} 
\usepackage{rotating} 
\usepackage{booktabs}
\usepackage{threeparttable} 

\usepackage{geometry}

\usepackage{perpage} 
\usepackage{xr}

\MakePerPage[1]{footnote} 

\usepackage{scalefnt}

\usepackage{amsthm, amsmath, amsmath, amssymb, bbm, bm}
\usepackage{graphicx}
\usepackage{mathrsfs}
\usepackage{threeparttable}
\usepackage[title]{appendix}
\usepackage{url}
\usepackage[round]{natbib}
\usepackage[colorlinks, 
linkcolor=blue,
anchorcolor=blue,
urlcolor=blue,
citecolor=blue]{hyperref}
\usepackage{bookmark}

\usepackage{caption}
\usepackage{tikz}
\usepackage{subfigure}
\usepackage{mathrsfs}

\usepackage{algorithm}
\usepackage{algpseudocode}

\usepackage{booktabs}
\usepackage{chngpage}
\usepackage{xcolor}

\usepackage{pifont}
\usepackage{enumerate}
\usepackage{enumitem}

\pgfdeclarelayer{background}
\pgfsetlayers{background,main}

\newcommand{\mc}[1]{\ensuremath{\mathcal{#1}}}

\newcommand{\p}[1]{\ensuremath{\left(#1\right)}}
\newcommand{\br}[1]{\ensuremath{\left\{#1\right\}}}

\newcommand{\norm}[1]{\ensuremath{\left|\left|#1\right|\right|}}
\newcommand{\sbr}[1]{\ensuremath{\left[#1\right]}}
\newcommand{\abs}[1]{\ensuremath{\left|#1\right|}}

\def\P{\mathbb{P}}
\def\R{\mathbb{R}}
\def\E{\mathbb{E}}

\def\var{\mbox{var}}
\def\Var{\mbox{Var}}

\def\argmin{\mbox{argmin}}

\newtheorem{theorem}{Theorem}
\newtheorem{assumption}{Assumption}
\newtheorem{proposition}{Proposition}[section]
\newtheorem{lemma}{Lemma}[section]

\theoremstyle{remark}
\newtheorem{remark}{Remark}

\def\T{{\mathrm{\scriptscriptstyle T}}}
\newcommand{\tbl}[2]{\caption{#1}\begin{center}#2\end{center}}
\newenvironment{tabnote}{\begin{flushleft}\footnotesize\textit{Note.}\ }{\end{flushleft}}

\begin{document}

\def\spacingset#1{\renewcommand{\baselinestretch}%
{#1}\small\normalsize} \spacingset{1}

%%%%%%%%%%%%%%%%%%%%%%%%%%%%%%%%%%%%%%%%%%%%%%%%%%%%%%%%%%%%%%%%%%%%%%%%%%%%%%

\title{\bf Semi-supervised inference using unlabeled summary statistics}
\author{Facheng Yu\hspace{.2cm}\\
Department of Statistics, University of Washington\\
and \\
Zhen Qi\hspace{.2cm}\\
Institute of Statistics and Big Data, Renmin University of China\\
and \\
Yuqian Zhang\thanks{Corresponding author (email: yuqianzhang@ruc.edu.cn)} \\
Institute of Statistics and Big Data, Renmin University of China}
\date{}
\maketitle

\bigskip
\begin{abstract}
Semi-supervised inference assumes access to a labeled dataset together with a large unlabeled dataset in which the outcome variable is missing, and it is widely used to improve statistical efficiency and support generalizability across populations. In many modern applications, however, individual-level unlabeled data may not be directly accessible due to privacy restrictions, data-sharing limits, or storage constraints, while summary statistics such as sample means and covariances from the unlabeled population are often available. In this work, we study this constrained semi-supervised setting where, in addition to labeled data with individualized observations, auxiliary information from the unlabeled population is available only through summary statistics. We propose new semi-supervised inference methods for mean estimation under both covariate-independent and covariate-dependent labeling and show that unlabeled summaries can still improve efficiency and help correct selection bias. The proposed methods apply in high dimensions and are robust to model misspecification. Valid inference is obtained under sparsity conditions comparable to those required by semi-supervised methods that assume access to individual-level unlabeled samples. Our approach relies on a specialized cross-fitting procedure, where sample splitting is applied only to the labeled data, which removes the need for individualized unlabeled covariates. We further extend this framework to average treatment effect estimation, enabling generalizability and transportability of causal conclusions in this constrained semi-supervised setting.
\end{abstract}

\noindent%
\vfill

\newpage
\spacingset{1.9} % DON'T change the spacing!

\section{Introduction}

\subsection{Background and motivation}
Semi-supervised inference studies how to improve efficiency by combining a small labeled sample with a much larger unlabeled sample. This setting arises when collecting covariates is relatively easy but obtaining outcomes is costly or time-consuming. By leveraging additional unlabeled samples, semi-supervised methods can achieve substantial gains in statistical efficiency. Such gains have been established across a range of inference tasks, including mean estimation \citep{zhang2019semi, zhang2022high, zhang2023double}, linear regression \citep{chakrabortty2018efficient, azriel2022semi, deng2024optimal, chen2025semi}, treatment effect estimation \citep{cheng2021robust, chakrabortty2022general, zhang2023semi, kallus2025role}, M-estimation \citep{song2024general}, and prediction-powered inference \citep{angelopoulos2023prediction, zrnic2024cross}, among others.

Most existing approaches assume full access to individual-level covariates from the unlabeled data. In many modern applications, however, auxiliary datasets cannot be shared at the individual level due to privacy, storage, or cost constraints, and only summary statistics such as sample means and covariances are available. This restriction limits the applicability of current semi-supervised procedures. We study estimation and inference in this constrained setting, where auxiliary information from the unlabeled sample is available only through summary statistics. Such situations arise, for example, when privacy regulations prevent hospitals from sharing patient-level data across sites, while aggregate demographic or health summaries remain accessible.

In this work, we study the estimation of the population mean outcome for the entire population by combining individual-level data from the labeled sample with covariate information from the unlabeled sample, a problem known as \emph{generalizability}. We also consider the case where the target population is the unlabeled sample itself, with the goal of estimating its mean outcome, known as \emph{transportability}. In addition, we extend our framework to causal inference problems, aiming to estimate the average treatment effect for both the overall population and the unlabeled population. Our objective is to develop semi-supervised methods that rely only on summarized covariate information from the unlabeled data.

The main difficulty is that standard semi-supervised procedures rely on individual-level unlabeled data to construct regression-based imputations or inverse probability weights. When only summary information is available, these procedures cannot be applied directly. In addition, when the unlabeled data are represented only through summaries, fully nonparametric approaches are not feasible, since such methods require access to individual-level observations or an increasing number of moments. In high-dimensional settings, nonparametric methods also suffer from the ``curse of dimensionality,'' which often leads to inaccurate estimation. Together, these issues push semi-supervised inference toward the use of parametric nuisance models. This reliance increases the risk of model misspecification, which may result in biased estimates and invalid inference. In this work, we develop robust inference methods that remain valid under possible model misspecification, improving reliability even in high-dimensional settings.

\subsection{Related literature}

Semi-supervised inference is closely connected to data integration and to problems of generalizability and transportability in causal inference, where labeled data come from a primary source, such as a randomized controlled trial, and additional covariate information is obtained from external observational data. A large body of work studies how to combine these sources to generalize or transport causal conclusions; see, for example, \citep{buchanan2018generalizing, dahabreh2021study, ung2024combining} and the reviews \citep{degtiar2023review, shi2023data, colnet2024causal}. These methods typically require access to individual-level covariates from the external source.

The use of summary statistics has also been studied in settings such as Mendelian randomization in genome-wide association studies \citep{burgess2013mendelian, morrison2020mendelian, zhao2020bayesian} and data fusion in meta-analysis \citep{lin2010relative, liu2015multivariate, zhu2015meta}. In these works, the summaries typically involve outcome-related quantities, such as regression estimates. In contrast, we study a constrained semi-supervised setting in which only summary statistics of unlabeled covariates are available and the outcomes are entirely missing in the unlabeled population.

Several recent works study data integration using both individualized data and summary information. For example, \citet{chu2023targeted} study the transportation of optimal individualized treatment rules using calibrated estimators. Their estimator is consistent for a target parameter defined by the unlabeled population only when a density ratio model is correctly specified. In contrast, we develop doubly robust inference procedures that do not require correct specification of any single nuisance model. \citet{hu2026semiparametric} develop a semiparametric framework for incorporating auxiliary summaries to improve estimation for parameters defined by the labeled data. Our focus is different: we aim to generalize or transport inference from the labeled sample to a larger or different population. Moreover, while these existing works apply to low-dimensional settings, we consider the more challenging high-dimensional scenario where nuisance functions cannot be estimated at parametric convergence rates.

The data integration problem is also related to the calibration weighting framework \citep{deville1992calibration, hainmueller2012entropy, zubizarreta2015stable, zhao2017entropy}, where covariate summaries are used to construct balancing weights. In low-dimensional settings, these methods can be extended to the scenarios considered here. In high-dimensional settings, however, the use of regularized estimators such as the lasso introduces additional bias in nuisance estimation, leading to slow convergence rates and requiring strong ultra-sparsity conditions for valid inference \citep{athey2018approximate, ning2020robust, tan2020model}. To reduce such bias, cross-fitting is widely used \citep{chernozhukov2018double, zhang2022high, qian2024changepoint}. Conventional cross-fitting, however, cannot be applied directly to the constrained semi-supervised problem studied here, since it requires access to individual-level data. This limitation motivates the development of new cross-fitting methods that rely only on covariate summaries.

\subsection{Our contributions}

We first consider the standard semi-supervised setting in which labeled and unlabeled observations follow the same distribution and label availability is independent of the covariates. We propose an estimator that achieves higher efficiency than the supervised sample mean, even when the nuisance models are misspecified. The proposed estimator attains the same efficiency as the semi-supervised estimators studied by \citet{zhang2019semi} and \citet{zhang2022high}, while not requiring access to individual-level unlabeled covariates and remaining applicable in high dimensions. In addition, when partial covariate information is available from auxiliary unlabeled data, we further develop a new mean estimator based on partially linear regression to improve efficiency. The corresponding methods and theoretical results are presented in Section~\ref{sec: mcar}. 

We also study a more general semi-supervised setting in which the labeling mechanism depends on covariates; see Section~\ref{sec: mar}. Under the missing at random condition, we propose estimators for the mean response in both the overall population and the unlabeled population. The proposed estimators are \emph{model doubly robust} in high-dimensional settings, in the sense that \emph{root-$n$ consistency and asymptotic normality} hold as long as at least one of the two working models is correctly specified. This property provides greater robustness to model misspecification than existing data integration methods \citep{zhao2017entropy,athey2018approximate, chu2023targeted}. When both nuisance models are correctly specified, the proposed estimators attain the same efficiency as existing methods, such as \citet{bang2005doubly}, \citet{chernozhukov2017double}, and \citet{kallus2025role}, while not requiring full access to individual-level unlabeled covariates. We further extend the framework to causal inference by developing estimators of the average treatment effect for both the overall population, which addresses generalizability, and the unlabeled population, which addresses transportability.

The proposed methods are built on a new cross-fitting strategy that performs sample splitting only within the labeled data and avoids splitting the unlabeled data. This design eliminates the need for individual-level covariate information from the unlabeled sample in both nuisance estimation and final averaging, and is applicable to both covariate-independent and covariate-dependent labeling mechanisms. Our theory shows that this strategy plays a role analogous to conventional cross-fitting: it reduces the bias induced by regularized nuisance estimation and enables valid inference under weak sparsity conditions; see Remarks~\ref{remark:sample-splitting} and \ref{remark:challenge of cross-fitting}. The result also provides insight into the role of cross-fitting and may help reduce implementation complexity, storage costs, and communication burden in other problems where cross-fitting is used to control overfitting bias in high-dimensional or nonparametric estimation.

Finally, unlike existing semi-supervised methods that require access to and storage of the full unlabeled matrix of size \(n_U \times d\), our approach substantially reduces this burden. Under covariate-independent labeling, only a vector of size \(O_p(s)\) is required. Under covariate-dependent labeling, point estimation requires a \(d\)-dimensional vector, while statistical inference further involves a Gram matrix with \(O_p(s^2)\) entries. Detailed discussions are provided in Remarks~\ref{remark:data-MCAR} and \ref{remark: mar gen use less info}. Here, \(n_U\) denotes the size of the unlabeled sample, which may be much larger than both the labeled sample size \(n_L\) and the dimension \(d\), and \(s \ll d \wedge n_L\) denotes the sparsity level of the outcome regression model.

\subsection{Notation}

For a random vector $X \in \R^d$, we denote its $\psi_\alpha$-Orlicz norm for any $\alpha>0$ by $\|X\|_{\psi_\alpha}$ and its $\ell_p$-norm for any $p> 0$ by $\|X\|_p := (\sum_{j=1}^d X[j]^p)^{1/p}$, where $X[j]$ denotes the $j$th coordinate of a vector $X$. For simplicity, we define $\|X\|_0 := \max \{1, \sum_{j=1}^d \mathrm{I}_{(X[j] \neq 0)}\}$ to avoid degenerate cases with zero sparsity. For any integer $n>0$, let $[n]:=\{1,\ldots,n\}$, and for any set $S$, denote its cardinality by $|S|$.

\section{\label{sec: mcar} Mean estimation under covariate-independent labeling}

\subsection{Problem setup}

Suppose there exist underlying independent and identically distributed samples \((\Gamma_i, X_i, Y_i)_{i=1}^n\), and let \((\Gamma, X, Y)\) be an independent copy. Here, \(X_i \in \mathbb{R}^d\) denotes a covariate vector with its first coordinate equal to one, corresponding to the intercept in the regression model. Let \(Y_i \in \mathbb{R}\) be the outcome and \(\Gamma_i \in \{0,1\}\) be a binary labeling indicator. For interpretability, we treat \(\Gamma_i\) as random while keeping the total sample size \(n\) fixed. Due to access constraints, individualized data \((X_i, Y_i)\) are observed only when \(\Gamma_i = 1\). For observations with \(\Gamma_i = 0\), we assume access only to summary covariate statistics, namely the sample mean \(\bar X_0\) and the sample Gram matrix \(\bar \Xi_0\):
\begin{equation}\label{def:summary}
\bar X_0 = \frac{\sum_{i=1}^n (1 - \Gamma_i)X_i}{\sum_{i=1}^n (1 - \Gamma_i)}, \quad 
\bar \Xi_0 = \frac{\sum_{i=1}^n (1 - \Gamma_i)X_i X_i^\T}{\sum_{i=1}^n (1 - \Gamma_i)}.
\end{equation}
Given the sample mean \(\bar X_0\), observing the sample Gram matrix is equivalent to observing the sample covariance matrix \(\widehat\Sigma_0 = \sum_{i=1}^n (1 - \Gamma_i)(X_i - \bar X_0)(X_i - \bar X_0)^\T / \sum_{i=1}^n (1 - \Gamma_i)\) since they are in one-to-one correspondence through \(\bar \Xi_0 = \widehat\Sigma_0 + \bar X_0 \bar X_0^\T\). We also assume access to the unlabeled sample size, \(n_U=\sum_{i=1}^n (1 - \Gamma_i)\), so that the total sample size \(n\) is known. Because only summary statistics from the unlabeled data are required, the amount of information needed is much smaller than in standard semi-supervised learning, which assumes access to individual-level covariates.

Our procedure extends naturally to settings with multiple unlabeled datasets by aggregating them. The only requirement is access to the overall sample mean and Gram matrix across all unlabeled data, which can be computed from the sample statistics provided by each dataset. Since unlabeled summaries often come from large samples, especially when combining multiple datasets, we allow a decaying labeling probability \(\gamma_n := \P(\Gamma=1) \to 0\) as $n\to\infty$.

In this section, we focus on the standard semi-supervised setting in which labeling is independent of covariates. This scenario is often referred to as the missing completely at random (MCAR) setting in the missing data literature and is commonly assumed in semi-supervised learning; see, for example, \citep{chakrabortty2018efficient,zhang2019semi,tony2020semisupervised,azriel2022semi,zhang2022high,angelopoulos2023prediction,zrnic2024cross}.

\begin{assumption}[Missing completely at random]\label{assumption mcar} $\Gamma \perp (X, Y)$.
\end{assumption}
Our parameter of interest is the mean outcome, \(\theta = \E(Y)\). Under Assumption~\ref{assumption mcar}, the mean outcome is the same across the overall, labeled, and unlabeled populations: \(\theta = \E(Y) = \E(Y \mid \Gamma = 1) = \E(Y \mid \Gamma = 0)\). A direct approach is to use the supervised sample mean, \(\bar{Y}_1 = \sum_{i=1}^n \Gamma_i Y_i / \sum_{i=1}^n \Gamma_i\). However, this estimator ignores information contained in the unlabeled covariates and can be inefficient. We develop methods that use the available summary information to improve efficiency.

\subsection{\label{sec: mcar lasso method} Methodology}

In the following, we propose an estimator for the mean outcome \(\theta = \E(Y)\) that uses only the individualized information from the labeled data \(\mc{D}_n = \{(\Gamma_i, \Gamma_i X_i, \Gamma_i Y_i)\}_{i=1}^n\) together with the unlabeled sample mean \(\bar{X}_0\).

Define the population slope \(\beta^* = \argmin_{\beta \in \mathbb{R}^{d}}\E \{(Y - X^\T \beta)^2 \}\), without requiring a correctly specified linear model. By construction of \(\beta^*\) and the Karush--Kuhn--Tucker conditions, we have \(\theta=\E(Y)=\E(X)^\T\beta^*\), even when \(\E(Y \mid X) \neq X^\T \beta^*\). Based on this representation, a natural approach is to estimate the linear slope \(\beta^*\) with some estimate \(\widehat{\beta}\) using labeled samples, while estimating the mean covariate \(\E(X)\) using the sample mean over the entire population, which can be represented as a function of \(\mathcal{D}_n\) and \(\bar{X}_0\): \(\bar{X}_{\rm all} = (1 - n^{-1} \sum_{i=1}^n \Gamma_i) \bar{X}_0 + n^{-1} \sum_{i=1}^n \Gamma_i X_i\). A plug-in estimate of \(\theta\) can then be constructed as \(\widehat{\theta}_{\rm PI} = \bar{X}_{\rm all}^\T \widehat{\beta}\). When the population slope is estimated via least squares, this approach aligns with the semi-supervised least squares estimator proposed by \citet{zhang2019semi}. This construction therefore requires only summary statistics from the unlabeled data, although this aspect was not emphasized in the original work. In low-dimensional settings, \citet{zhang2019semi} showed that the resulting estimator is asymptotically normal and at least as efficient as the sample mean \(\bar{Y}_1\), provided that the feature dimension satisfies \(d=o((n_L)^{1/2})\), where \(n_L=\sum_{i=1}^n \Gamma_i\).

As further shown by \citet{zhang2019semi}, the efficiency gain of semi-supervised methods over supervised methods depends on the variance explained by the linear model, \(\var(X_i^\T\beta^*)\). To increase this quantity and improve estimation efficiency, it is desirable to collect more outcome-related covariates, or to include richer basis expansions to increase the variation explained by the linear model. However, when the feature dimension exceeds the number of labeled samples, the least squares solution is ill-defined, making regularized estimators such as the lasso more suitable. A drawback of regularization is the bias it introduces, which often leads to slower convergence rates for plug-in estimators. To address this issue, \citet{zhang2022high} propose a debiasing procedure based on cross-fitting. Their method, however, relies on access to the sample means of covariates within each fold, which requires individualized unlabeled samples.

To overcome this limitation, we propose a refined approach that imputes all unlabeled covariates by their sample mean, \(\bar{X}_0\), so that the method remains feasible without access to individualized covariates. We introduce the proposed procedure below.

Step 1: Divide the index set $[n]$ into $K$ disjoint subsets $\mathcal{I}_1, \ldots, \mathcal{I}_K$ with equal sizes such that $n_k := |\mathcal{I}_k| = n/K$ for $k \in [K]$. Let $\widehat \gamma_{k} = n_k^{-1}\sum_{i\in \mathcal{I}_k} \Gamma_i$ and $\mathcal{I}_{-k} = [n] \setminus \mathcal{I}_k$. 

Step 2: For each $k \in [K]$, compute the lasso estimator $\widehat \beta^{(-k)}$: with some $\lambda_n\geq0$,
\[
\widehat \beta^{(-k)} = \underset{\beta \in \R^{d}}{\argmin}
\left\{
\frac{\sum_{i \in \mathcal{I}_{-k}} \Gamma_i\p{Y_i - X_i^\T\beta}^2}
{\sum_{i \in \mathcal{I}_{-k}} \Gamma_i}
+ \lambda_n \|\beta\|_1
\right\}.
\]

Step 3: Denote $\widehat{\varepsilon}_i^{(-k)}=Y_i - X_i^\T \widehat \beta^{(-k)}$. The mean estimator is proposed as
\begin{align}
\widehat \theta = n^{-1}\sum_{k=1}^K\sum_{i\in \mathcal{I}_k} \{\Gamma_iX_i + (1-\Gamma_i)\bar{X}_0\}^\T \widehat \beta^{(-k)}  + n^{-1}\sum_{k=1}^K\sum_{i\in \mathcal{I}_k}\frac{ \Gamma_i}{\widehat \gamma_{k}}\widehat{\varepsilon}_i^{(-k)}. \label{definition k-th fold ss outcome estimator MCAR}
\end{align}
When the sample Gram matrix $\bar \Xi_0$ is also observable, we define:
\begin{align}
\widehat \sigma^2 =& n^{-1}\sum_{k=1}^K\sum_{i\in \mathcal{I}_k}(1-\Gamma_i)\widehat \beta^{(-k),\T}\bar \Xi_0 \widehat \beta^{(-k)}+n^{-1}\sum_{k=1}^K\sum_{i\in \mathcal{I}_k} \left(\Gamma_iX_i^\T\widehat\beta^{(-k)}  + \frac{\Gamma_i}{\widehat \gamma_k}\widehat{\varepsilon}_i^{(-k)}\right)^2 - \widehat \theta^2. \label{variance estimator under MCAR}
\end{align}

\begin{remark}[Required summary statistics under covariate-independent labeling]\label{remark:data-MCAR}
To construct a point estimator for \(\theta\), only first-moment information from the unlabeled data is required, whereas statistical inference and asymptotic variance estimation additionally require second-moment information. Importantly, the additional covariate summaries from the unlabeled data can be collected after obtaining the lasso estimators in Step 2. Since the estimator \eqref{definition k-th fold ss outcome estimator MCAR} depends only on the inner product \(\bar{X}_0^\T \widehat \beta^{(-k)}\), it is sufficient to collect \(\bar{X}_{0,\widehat S} = n_U^{-1}\sum_{i=1}^n(1 - \Gamma_i)X_{i,\widehat S}\), where \(\widehat S = \cup_{k=1}^K\mathrm{supp}(\widehat \beta^{(-k)})\) satisfies \( |\widehat S| = O_p(\|\beta^*\|_0) \) by arguments as in Lemma~\ref{lemma sparsity betatilde -> beta}. The estimator \eqref{definition k-th fold ss outcome estimator MCAR} remains unchanged if \(\bar{X}_0\) is replaced by \(\tilde{X}_0\), where \(\tilde{X}_{0,\widehat S} = \bar{X}_{0,\widehat S}\) and \(\tilde{X}_{0,\widehat S^c} = 0\).

To construct the asymptotic variance estimator in \eqref{variance estimator under MCAR}, it suffices to compute \(\bar{\Xi}_{0,\widehat S} = n_U^{-1}\sum_{i=1}^n(1 - \Gamma_i)X_{i,\widehat S}X_{i,\widehat S}^\T\). Furthermore, in scenarios where the sub-matrix \(\bar{\Xi}_{0,\widehat S}\) is unavailable, the asymptotic variance can still be estimated using only the labeled data under covariate-independent labeling: $\widehat \sigma_{L}^2 = n^{-1}\sum_{k=1}^K\sum_{i\in \mathcal{I}_k} \{\Gamma_i(X_i^\T\widehat\beta^{(-k)})^2/\widehat \gamma_k  + \Gamma_i(\widehat{\varepsilon}_i^{(-k)})^2/\widehat \gamma_k^2\} - \widehat \theta^2.$ The asymptotic normality results in Theorem~\ref{theorem for the asymptotics under MCAR} remain valid when \(\widehat \sigma^2\) is replaced by \(\widehat \sigma_{L}^2\). However, when second-moment information from unlabeled data is available, inference based on \(\widehat \sigma^2\) is recommended, as incorporating additional information generally leads to improved finite-sample performance.
\end{remark}

\begin{remark}[Comparison with existing cross-fitting schemes]
In the presence of high-dimensional covariates, existing semi-supervised procedures typically apply cross-fitting to the entire dataset and therefore require access to both labeled and unlabeled observations at the individual level \citep{zhang2022high, zhang2023double, zrnic2024cross, kallus2025role}. In contrast, the proposed strategy performs cross-fitting only within the labeled data and incorporates unlabeled information solely through the summary statistic \(\bar X_0\). Figure~\ref{fig:crossfit} illustrates the difference between the proposed cross-fitting strategy and existing approaches. By avoiding cross-fitting on the unlabeled data, our method removes the need for access to individualized unlabeled covariates. At the same time, this simplification does not weaken the role of cross-fitting in reducing overfitting bias; see Theorem~\ref{theorem for the asymptotics under MCAR} and Remark~\ref{remark:sample-splitting}.
\end{remark}

\begin{figure}[tbp]
\centering
\resizebox{0.92\textwidth}{!}{%
\begin{tikzpicture}[
    x=1cm,
    y=1cm,
    every node/.style={font=\small}
]

% ===== parameters =====
\def\w{3.2}
\def\hup{0.55}
\def\hdown{1.45}
\def\gap{0.22}
\def\cutw{0.64}

\pgfmathsetmacro{\upMid}{\hup/2}
\pgfmathsetmacro{\lowMid}{-\gap-\hdown/2}

% =========================
% (a) Existing procedures
% =========================
\begin{scope}[shift={(0,0)}]

    % labeled data
    \draw (0,0) rectangle (\w,\hup);

    % unlabeled data
    \draw (0,-\gap-\hdown) rectangle (\w,-\gap);

    % partitions
    \foreach \x in {\cutw,1.28,1.92,2.56}
    {
    \draw (\x,0) -- (\x,\hup);
    \draw (\x,-\gap-\hdown) -- (\x,-\gap);
    }

    % labels
    \node[anchor=east, align=center]
    at (-0.45,\upMid)
    {Labeled\\[-1mm]Data};

    \node[anchor=east, align=center]
    at (-0.45,\lowMid)
    {Unlabeled\\[-1mm]Data};

    % panel label
    \node at (\w/2,-\gap-\hdown-0.42) {(a)};

\end{scope}

% =========================
% (d) Proposed procedure
% =========================
\begin{scope}[shift={(5.3,0)}]

    % labeled data
    \draw (0,0) rectangle (\w,\hup);

    % partitions only for labeled data
    \foreach \x in {\cutw,1.28,1.92,2.56}
    {
        \draw (\x,0) -- (\x,\hup);
    }

    % unavailable unlabeled sample
    \draw[dashed]
    (0,-\gap-\hdown) rectangle (\w,-\gap);

    % arrow
    \draw[->, thick]
    (\w+0.12,-\gap-\hdown/2)
    -- (\w+0.75,-\gap-\hdown/2);

    % summary statistics block (auto-fit)
    \node[draw, align=center]
    at (\w+1.775,-\gap-0.725)
    {Summary\\[-2pt]Statistics};

    % panel label
    \node at (\w/2,-\gap-\hdown-0.42) {(b)};

\end{scope}
\path (11,0);
\end{tikzpicture}%
}
\caption{
Schematic illustration of cross-fitting strategies.
(a) Existing semi-supervised procedures \citep{zhang2022high, zhang2023double, zrnic2024cross, kallus2025role}, which apply cross-fitting to both labeled and unlabeled samples.
(b) The proposed procedure, in which cross-fitting is applied only to the labeled sample. The unlabeled sample (dashed rectangle) is not available at the individual level; only its summary statistics are used.
}
\label{fig:crossfit}
\end{figure}

\subsection{\label{mcar lasso theory}Asymptotic theory}

We introduce the following regularity assumptions.

\begin{assumption}\label{assumption mcar lasso (a)}
Let the following conditions hold with constants $\kappa_l,\sigma,\sigma_w,\delta_w>0$: (a) $X$ is a sub-Gaussian random vector with $\|X^\T v\|_{\psi_2} \leq \sigma \|v\|_2$ for all $v \in \R^{d}$. In addition, $\|X^\T \beta^*\|_{\psi_2} \leq \sigma$ and
\[
 \inf_{v \in \R^d, \|v\|_2=1}\E\{(X^\T v)^2\} \geq \kappa_l .
\]
(b) $w = Y - X^\T \beta^*$ is a sub-Gaussian random variable with $\|w\|_{\psi_2} \leq \sigma_w$ and $\E(w^2)\geq \delta_w$.
\end{assumption}

The following theorem characterizes asymptotic properties of the mean estimator $\widehat \theta$.

\begin{theorem}\label{theorem for the asymptotics under MCAR}
Suppose that Assumptions~\ref{assumption mcar} and~\ref{assumption mcar lasso (a)} hold. Choose $\lambda_n \asymp \{\log d/(n\gamma_n)\}^{1/2}$. If $n\gamma_n \gg (\log n)^2\log d$ and the sparsity level satisfies $s=\|\beta^*\|_0=o(n\gamma_n/ \log d)$, then as $n, d \rightarrow \infty$, $\widehat \theta -\theta = O_p(n^{-1/2}\gamma_n^{-1/2})$, $\widehat \sigma^2 = \sigma_n^2\{1 + o_p(1)\}$, and $\widehat\sigma^{-1}n^{1/2}(\widehat \theta - \theta)\rightarrow_d \mc{N}(0,1)$, where $\sigma_n^2 = \var(X^\T \beta^* + \Gamma(Y - X^\T \beta^*)/\gamma_n) = \gamma_n^{-1} \var(Y) - (\gamma_n^{-1}-1) \var(X^\T \beta^*)$.
\end{theorem}

As observed in prior works \citep{zhang2019semi,zhang2022high,zhang2023double,kallus2025role}, the ``effective sample size'' of this problem is \(n\gamma_n = \E(n_L)\), corresponding to the expected labeled sample size. The convergence rate established above depends on \(n\gamma_n\), rather than the full sample size \(n\). 

In addition, Theorem~\ref{theorem for the asymptotics under MCAR} establishes that \(\widehat{\theta}\) remains consistent and asymptotically normal even when the linear working model is misspecified. Compared with the supervised sample mean \(\bar Y_1\), which has asymptotic variance \(\gamma_n^{-1}\var(Y)\), the asymptotic variance of \(\widehat{\theta}\), namely \(\sigma_n^2\), is smaller by $(\gamma_n^{-1}-1)\var(X^\T\beta^*)\geq0$. The asymptotic variance \(\sigma_n^2\) also aligns with that of the linear regression-based methods from \citet{zhang2019semi, zhang2022high} when \(\Gamma_i\) is treated as random. When individualized unlabeled covariates are observable, \citet{zhang2019semi} show that the oracle lower bound for the asymptotic variance is \(\sigma_{\rm oracle}^2 = \gamma_n^{-1} \var(Y) + (1 - \gamma_n^{-1}) \var(\mu(X))\), where \(\mu(X) = \E(Y \mid X)\) represents the true conditional mean function. If the linear model is correct, the asymptotic variance \(\sigma_n^2\) from Theorem~\ref{theorem for the asymptotics under MCAR} matches this semi-supervised oracle lower bound, but notably, we achieve this without requiring individualized unlabeled covariates. In the presence of model misspecification, efficiency gains are possible when individualized covariates are partially observable; see Section~\ref{sec: mcar plm}. 
\begin{remark}[Cross-fitting without access to individualized unlabeled covariates]\label{remark:sample-splitting}
The cross-fitting technique is well known for reducing bias from nuisance estimation, especially when nuisance estimators do not achieve a parametric convergence rate, as is common in high-dimensional or nonparametric settings. Methods that use cross-fitting often attain faster convergence rates and allow valid inference under weaker sparsity conditions; see, for example, \citep{chernozhukov2017double,tony2020semisupervised,zhang2022high,qian2024changepoint}. In the context of mean estimation under covariate-independent labeling, \citet{zhang2022high} note that statistical inference without cross-fitting typically requires an ultra-sparse condition \(s = o((n_L)^{1/2}/\log d)\), assuming \(\Gamma_i\) and \(n_L\) are non-random. To relax this requirement, they apply cross-fitting across the entire dataset, which improves the sparsity condition to \(s = o(n_L/\log d)\). 

Our findings show that cross-fitting is needed only within the labeled data and does not need to be applied to the unlabeled data. This is because the outcome regression estimator \(\widehat\beta^{(-k)}\) is constructed using only labeled observations, which makes it conditionally independent of the unlabeled summary \(\bar X_0\) given all \(\Gamma_i\), even without cross-fitting on the unlabeled data. This property allows careful control of how nuisance estimation error affects the final mean estimator. In Theorem~\ref{theorem for the asymptotics under MCAR}, we impose a sparsity condition \(s = o(n\gamma_n/\log d)\), which matches the requirement in \citet{zhang2022high} when \(\Gamma_i\) and \(n_L\) are treated as random, since \(n\gamma_n=\E(n_L)\). This refinement removes the need to access individualized unlabeled covariates and makes the method more practical for large-scale semi-supervised problems where such access can be difficult.
\end{remark}

\subsection{Improving efficiency with additional individualized features}\label{sec: mcar plm}

The proposed estimator above uses only unlabeled summary statistics, which is advantageous when individual-level access is difficult or impossible. When additional individualized covariates \(Z_i \in \mathbb{R}^r\) from the unlabeled sample are available, with \(Z\) denoting an independent copy, we can further improve efficiency through a partially linear model. We introduce the procedure below.

Step 1: Divide the index set $[n]$ into $K$ disjoint subsets $\mathcal{I}_1, \ldots, \mathcal{I}_K$ with equal sizes such that $n_k := |\mathcal{I}_k| = n/K$ for $k \in [K]$. Let $\widehat \gamma_{k} = n_k^{-1}\sum_{i\in \mathcal{I}_k} \Gamma_i$ and $\mathcal{I}_{-k} = [n] \setminus \mathcal{I}_k$.

Step 2: For each $k \in [K]$, construct the doubly penalized least squares estimator based on a partially linear model \citep{muller2015partial}:
\[
(\widehat \beta^{(-k)}_{plm}, \widehat f^{(-k)}) =
\underset{(\beta,f) \in \R^{d} \times \mathcal{F}}{\argmin}
\left[
\frac{\sum_{i \in \mathcal{I}_{-k}} \Gamma_i \{Y_i - X_{i}^\T \beta - f(Z_i)\}^2}
{\sum_{i \in \mathcal{I}_{-k}} \Gamma_i}
+ \lambda_n \|\beta\|_1 + \eta_n \|f\|_{\mathcal{F}}
\right],
\]
where $\mathcal F$ is a pre-specified functional class, $\|\cdot\|_{\mathcal{F}}$ is an associated norm, and $\lambda_n, \eta_n\geq0$ are regularization parameters. For each $i \in \mathcal{I}_k$, define $\widehat \mu_i^{(-k)} = X_i^\T \widehat \beta^{(-k)}_{plm} + \widehat f^{(-k)}(Z_i)$, $\widehat \mu_{0,i}^{(-k)} = \bar X_0^\T \widehat \beta^{(-k)}_{plm} + \widehat f^{(-k)}(Z_i)$, $\widehat \varepsilon_{i,plm}^{(-k)} = Y_i - \widehat \mu_i^{(-k)}$, and
\[
\widehat \tau_i^{(-k)} = \widehat \beta^{(-k), \T}_{plm} \bar{\Xi}_0 \widehat \beta^{(-k)}_{plm}
+ 2(\bar X_0^\T \widehat \beta^{(-k)}_{plm})\widehat f^{(-k)}(Z_i)
+ \{\widehat f^{(-k)}(Z_i)\}^2.
\]

Step 3: The mean estimator based on a partially linear model is defined as
\[
\widehat \theta_{plm}
= n^{-1}\sum_{k=1}^K\sum_{i\in \mathcal{I}_k}
\left\{\Gamma_i \widehat \mu_i^{(-k)} + (1-\Gamma_i)\widehat \mu_{0,i}^{(-k)}
+ \frac{\Gamma_i}{\widehat \gamma_{k}}\widehat \varepsilon_{i,plm}^{(-k)}\right\}.
\]
When $\bar \Xi_0$ is also observable, the asymptotic variance can be estimated as
\[
\widehat \sigma_{plm}^2
= n^{-1}\sum_{k=1}^K\sum_{i\in\mathcal{I}_{k}}
\left(\Gamma_i \widehat \mu_i^{(-k)} + \frac{\Gamma_i}{\widehat \gamma_k}\widehat \varepsilon_{i,plm}^{(-k)}\right)^2
+ n^{-1}\sum_{k=1}^K\sum_{i\in\mathcal{I}_{k}}(1-\Gamma_i)\widehat \tau_i^{(-k)}
- \widehat \theta_{plm}^2.
\]

Define
\begin{align*}
(\beta^*_{plm}, f^*) &= \argmin_{(\beta, f) \in \R^d \times \mathcal{F}} \E[\{Y - X^\T\beta-f(Z)\}^2],\\
\mu^*(W)&=X^\T\beta^*_{plm}+f^*(Z),\quad
W=(X^\T,Z^\T)^\T,\quad \epsilon = Y - \mu^*(W).
\end{align*}
Let $\widetilde X  = X - \E(X \mid Z)$. We assume the following conditions.

\begin{assumption}
(a) \label{assumption mcar plm (a)} There exist constants $\kappa_l, \kappa_u, \sigma, \sigma_f, \sigma_\epsilon, \delta_\epsilon>0$ such that the following conditions hold: $\inf_{v \in \R^d,\;\|v\|_2=1}\E\{(\widetilde X^\T v)^2\} \geq \kappa_l$, $\sup_{v \in \R^d,\;\|v\|_2=1}\E\{(X^\T v)^2\} \leq \kappa_u$, $\E\{(X^\T \beta^*_{plm})^{4}\} \leq \sigma^{4}$, $\E[\{f^*(Z)\}^{4}] \leq \sigma_f^{4}$, $\E(\epsilon^{4}) \leq \sigma_\epsilon^{4}$, and $\E(\epsilon^{2})\geq \delta_\epsilon$. (b) For each $k\in[K]$, as $n, d \rightarrow \infty$,
$\E_Z[\{\widehat f^{(-k)}(Z) - f^*(Z)\}^2] = o_p(1)$ and $\|\widehat \beta^{(-k)}_{plm} - \beta^*_{plm}\|_2 = o_p(1)$.
\end{assumption}

Assumption~\ref{assumption mcar plm (a)}(a) imposes standard regularity conditions; see, e.g., \citet{xie2009scad, zhang2016statistical, lv2022debiased}. Part (b) only requires nuisance consistency, which holds, for example, when $f^*$ is a low-dimensional smooth function and the linear slope is sparse with $\|\beta^*_{plm}\|_0=o(n\gamma_n/\log d)$ \citep{muller2015partial,zhu2019high}. Theoretical properties of $\widehat \theta_{plm}$ are summarized below.
\begin{theorem} \label{plm Asymptotics theorem}
Let Assumptions~\ref{assumption mcar} and~\ref{assumption mcar plm (a)} hold. Assume further that $\Gamma \perp Z$ and $n\gamma_n\to\infty$. Then, as $n, d \rightarrow \infty$, $\widehat \theta_{plm} - \theta = O_p(n^{-1/2}\gamma_n^{-1/2})$, $\widehat \sigma_{plm}^2 = \sigma_{plm}^2\{1 + o_p(1)\}$, and $\widehat\sigma_{plm}^{-1}n^{1/2}(\widehat \theta_{plm} - \theta)\rightarrow_d \mc{N}(0,1)$, where $\sigma_{plm}^2 = \var(\mu^*(W) + \Gamma\{Y - \mu^*(W)\}/\gamma_n) = \gamma_n^{-1} \var(Y) + (1 - \gamma_n^{-1}) \var(\mu^*(W))$.
\end{theorem}

To compare efficiency, let \(\widehat\theta\) denote the linear-regression-based estimator in \eqref{definition k-th fold ss outcome estimator MCAR}, which uses unlabeled summary statistics of \(W_i=(X_i^\T,Z_i^\T)^\T\), and let \(\widetilde\sigma_n^2\) denote its asymptotic variance given in Theorem~\ref{theorem for the asymptotics under MCAR}. Let \(\widetilde\sigma_{\rm oracle}^2\) denote the oracle efficiency bound established by \citet{zhang2022high} for the ideal semi-supervised setting in which all unlabeled covariates \(W_i\) are individually observed. By construction, $\widetilde\sigma_{\rm oracle}^2 \leq \sigma_{plm}^2 \leq \widetilde\sigma_n^2,$ indicating that additional individualized information can improve efficiency. The partially-linear-model-based estimator \(\widehat\theta_{plm}\) improves upon \(\widehat\theta\) when \(f^*(Z)\) is nonlinear, while access to additional individualized covariates can further improve efficiency when the outcome model is also nonlinear in \(X\). When all individualized covariates are available and \(W_i=Z_i\), our method reduces to the nonparametric-regression-based estimator of \citet{zhang2022high}. In summary, summary statistics are sufficient for covariates that enter the model linearly, whereas individualized covariates are most valuable for capturing substantial nonlinear effects when the additional cost of collecting them is acceptable.

\section{\label{sec: mar} Correcting selection bias under covariate-dependent labeling}

\subsection{Problem setup}

In this section, we consider the mean estimation problem when the labeling mechanism depends on observed covariates. While the methods and results in Section~\ref{sec: mcar} rely on the missing completely at random condition, such an assumption may be restrictive in many practical applications. We therefore allow the labeling mechanism to depend on the covariates and assume the following missing at random condition instead.

\begin{assumption}[Missing at random]
 $\Gamma \perp Y \mid X$ and \(\gamma_n(X) := \P(\Gamma = 1 \mid X)>0\) almost surely.  \label{assumption mar}
\end{assumption}

Assumption~\ref{assumption mar} essentially requires the measurement of all factors that jointly influence the outcome \(Y\) and the labeling indicator \(\Gamma\), commonly referred to as ``no unmeasured confounding'' or ``ignorability'' in the causal inference and missing data literature \citep{rubin1976inference, rosenbaum1983central, imai2015robust}. Under this condition, we allow for potential dependence between \(\Gamma\) and \((X, Y)\). Thus, unlike in Section~\ref{sec: mcar}, the propensity score function \(\gamma_n(X)\) is no longer a constant. Departing from the standard missing data literature, which typically enforces a positivity condition \(\gamma_n(X) > c\) almost surely with a constant \(c > 0\), we adopt a ``decaying missing at random'' framework of \citet{zhang2023semi,zhang2023double,testa2025semiparametric}. Here, we only require \(\gamma_n(X) > 0\) and allow for \(\gamma_n = \P(\Gamma = 1) = \E\{\gamma_n(X)\} \to 0\) as \(n \to \infty\), thereby accommodating settings where the unlabeled data size vastly exceeds that of the labeled data.

When labeling depends on covariates, the mean outcomes differ across labeled and unlabeled samples. We define the generalizability and transportability targets as
\[
\theta_g = \E(Y),\quad \theta_t = \E(Y \mid \Gamma=0),
\]
where \(\theta_g\) represents the mean outcome over the entire population, and \(\theta_t\) denotes the mean outcome for the unlabeled data. Without Assumption~\ref{assumption mcar}, it is likely that \(\theta_g \neq \E(Y \mid \Gamma=1) \neq \theta_t\). Therefore, the sample mean over the labeled sample, $\bar Y_1$, is typically a biased estimate for both $\theta_g$ and $\theta_t$. The goal therefore shifts from improving efficiency of consistent estimators based on labeled data to correcting the bias induced by covariate-dependent labeling. In Section~\ref{sec: mar method}, we introduce estimators for \(\theta_g\) and \(\theta_t\) by leveraging unlabeled summary statistics, with their theoretical properties detailed in Section~\ref{sec: mar theory}. Extensions for estimating causal parameters are further discussed in Section~\ref{sec: causal}.

\subsection{Methodology}\label{sec: mar method}

Since the outcome of interest \(Y_i\) is unavailable among the unlabeled data, directly estimating \(\theta_g\) and \(\theta_t\) via empirical averages is not feasible; this constitutes a typical missing data problem. To address this, we define \(m(X) := \E(Y \mid X)\) as the true outcome regression function. Let \(m^*(\cdot)\) and \(\gamma_n^*(\cdot)\) represent the working models for the outcome regression and propensity score, respectively. These can be viewed as approximations of the nuisance function at the population level. One common approach for robust estimation in this context is to utilize the doubly robust representations \citep{bang2005doubly,funk2011doubly}: as long as either $m^*(\cdot)=m(\cdot)$ or $\gamma_n^*(\cdot)=\gamma_n(\cdot)$,
\begin{align*}
\theta_g&=\E\left[m^*(X)+\frac{\Gamma}{\gamma_n^*(X)}\left\{Y-m^*(X)\right\}\right],\\
\theta_t&=\E\left[\frac{1-\Gamma}{1-\gamma_n}m^*(X)+\frac{\Gamma\left\{1-\gamma_n^*(X)\right\}}{\gamma_n^*(X)(1-\gamma_n)}\left\{Y-m^*(X)\right\}\right].
\end{align*}
Let $\zeta_0:=\E(X\mid\Gamma=0)$. Under a linear outcome regression model, namely $m^*(X)=X^\T\beta_{OR}^*$ for some coefficient vector $\beta_{OR}^*\in\R^d$, the above representations can be rewritten as
\begin{align*}
\theta_g&=(1-\gamma_n)\zeta_0^\T\beta^*_{OR}+\E\left\{\Gamma X^\T\beta_{OR}^*+\frac{\Gamma}{\gamma_n^*(X)}\left(Y-X^\T\beta_{OR}^*\right)\right\},\\
\theta_t&=\E\left[\zeta_0^\T\beta^*_{OR}+\frac{\Gamma\left\{1-\gamma_n^*(X)\right\}}{\gamma_n^*(X)(1-\gamma_n)}\left(Y-X^\T\beta_{OR}^*\right)\right].
\end{align*}
Therefore, given \(m^*(\cdot)\) and \(\gamma_n^*(\cdot)\), we can identify \(\theta_g\) and \(\theta_t\) using the individualized data \((\Gamma X, \Gamma Y)\) from the labeled sample and the mean covariates in the unlabeled data \(\zeta_0\), without requiring access to individual-level unlabeled covariates. The use of a linear outcome working model follows from linearity of expectation: since the expectation operator commutes with linear maps, we have \(\E\{m^*(X)\mid\Gamma=0\}=\zeta_0^\T\beta^*_{\mathrm{OR}}\), enabling identification through the mean \(\zeta_0=\E(X\mid\Gamma=0)\).

Estimation based on doubly robust representations typically yields consistent estimates even when one nuisance model is misspecified. However, valid inference often requires both models to be correctly specified, particularly in high-dimensional settings \citep{farrell2015robust,chernozhukov2017double}. To address the bias induced by model misspecification, \citet{tan2020model,smucler2019unifying,ning2020robust,dukes2021inference} develop tailored loss functions for nuisance estimation, enabling robust inference as long as at least one of the two working models is correctly specified. We extend this idea to the constrained semi-supervised setting and show that constructing such nuisance estimators does not require individualized unlabeled covariates. As in Section~\ref{sec: mcar lasso method}, we implement a cross-fitting procedure after imputing unlabeled covariates by their sample mean \(\bar X_0\). Our proposed estimators for \(\theta_g\) and \(\theta_t\) are presented below.

Step 1: Divide the index set $[n]$ into $K$ disjoint subsets $\mathcal{I}_1, \ldots, \mathcal{I}_K$ with equal sizes such that $n_k := |\mathcal{I}_k| = n/K$ for $k \in [K]$. Let $\widehat \gamma_{k} = n_k^{-1}\sum_{i\in \mathcal{I}_k} \Gamma_i$, $\mathcal{I}_{-k} = [n] \setminus \mathcal{I}_k$, and separate $\mathcal{I}_{-k}$ into two disjoint subsets $\mathcal{I}_{-k,\alpha}, \mathcal{I}_{-k,\beta}$ with equal sizes $M = |\mathcal{I}_{-k,\alpha}| = |\mathcal{I}_{-k,\beta}|$.

Step 2: For each $k \in [K]$, construct the propensity score estimator as
\begin{align}
\widehat \alpha^{(-k)}_{PS} = \argmin_{\alpha \in \R^d} \left[ M^{-1}\sum_{i \in \mathcal{I}_{-k,\alpha}} \{(1-\Gamma_i)\bar{X}_0^\T \alpha + \Gamma_i \exp(-X_i^\T \alpha)\} + \lambda_\alpha \|\alpha\|_1 \right]. \label{mar PS estimator}
\end{align}

Step 3: For each $k \in [K]$, construct the outcome regression estimator as
\begin{align}
\widehat \beta^{(-k)}_{OR} = \argmin_{\beta \in \R^d} \left\{M^{-1}\sum_{i \in \mathcal{I}_{-k,\beta}} \Gamma_i \exp(-X_i^\T \widehat \alpha_{PS}^{(-k)})(Y_i - X_i^\T \beta)^2 + \lambda_\beta \|\beta\|_1\right\}. \label{mar OR estimator}
\end{align}

Step 4: Let $g(t) = e^t/(1+e^t)$ be the logistic function and $\widehat{\varepsilon}_i^{(-k)}=Y_i - X_i^\T\widehat\beta_{OR}^{(-k)}$. We propose
\begin{align}
&\widehat \theta_g = n^{-1}\sum_{k=1}^K\sum_{i\in \mathcal{I}_k} \left[\{\Gamma_iX_i+(1-\Gamma_i)\bar X_0\}^\T\widehat\beta_{OR}^{(-k)}  +  \frac{\Gamma_i\widehat{\varepsilon}_i^{(-k)}}{g\p{X_i^\T \widehat \alpha^{(-k)}_{PS}}}\right],\label{mar generalizability estimator}\\
&\widehat \theta_t = n^{-1}\sum_{k=1}^K\sum_{i\in \mathcal{I}_k}\left\{\frac{1- \Gamma_i}{1-\widehat \gamma_k}\bar X_0^\T\widehat\beta^{(-k)}_{OR}  + \frac{\Gamma_i\exp\p{-X_i^\T \widehat \alpha^{(-k)}_{PS}}\widehat{\varepsilon}_i^{(-k)}}{1-\widehat\gamma_k}\right\}. \label{mar transportability estimator}
\end{align}
When $\bar \Xi_0$ is further observed, denote $b^{(-k)}=\widehat \beta^{(-k),\T}_{OR}\bar \Xi_0 \widehat \beta^{(-k)}_{OR}$ and
\begin{align*}
\widehat \sigma^2_g &= n^{-1}\sum_{k=1}^K \sum_{i \in \mathcal{I}_{k}}\p{1-\Gamma_i}b^{(-k)} - \widehat \theta_g^2 + n^{-1}\sum_{k=1}^K \sum_{i \in \mathcal{I}_{k}} \left\{\Gamma_i X_i^\T\widehat\beta_{OR}^{(-k)} + \frac{ \Gamma_i\widehat{\varepsilon}_i^{(-k)}}{g\p{X_i^\T \widehat \alpha^{(-k)}_{PS}}}\right\}^2,\\
\widehat \sigma_t^2 &= \sum_{k=1}^K\sum_{i\in \mathcal{I}_k}\frac{\p{1-\Gamma_i}\p{b^{(-k)} + \widehat \theta_t^2 - 2\bar X_0^\T\widehat\beta_{OR}^{(-k)}\widehat \theta_t}}{n\p{1-\widehat \gamma_k}^2} + \sum_{k=1}^K\sum_{i\in \mathcal{I}_k}\frac{\Gamma_i\exp\p{-2X_i^\T \widehat \alpha^{(-k)}_{PS}}\p{\widehat{\varepsilon}_i^{(-k)}}^2}{n\p{1-\widehat \gamma_k}^2}.
\end{align*}

The proposed algorithm differs from existing methods \citep{smucler2019unifying,ning2020robust,tan2020model,dukes2021inference}. We apply cross-fitting only within the labeled sample and do not perform cross-fitting on the unlabeled data. As a result, the nuisance estimators \eqref{mar PS estimator}--\eqref{mar OR estimator} and the final doubly robust estimators \eqref{mar generalizability estimator}--\eqref{mar transportability estimator} depend only on unlabeled covariate summaries, thereby eliminating the need for individual-level unlabeled covariates.

\begin{remark}[Required summary statistics under covariate-dependent labeling]\label{remark: mar gen use less info}
Unlike the setting discussed in Remark~\ref{remark:data-MCAR}, constructing robust estimators under covariate-dependent labeling requires access to the full unlabeled sample mean vector \(\bar X_0\).
This contrasts with Section~\ref{sec: mcar}, where the propensity score \(\gamma_n(X) = \P(\Gamma = 1 \mid X) = \P(\Gamma = 1)\) is constant under Assumption~\ref{assumption mcar}, allowing estimation without covariate information. In this setting, however, estimation of the propensity score function depends on \(\bar X_0\), as in \eqref{mar PS estimator}.

The nuisance estimators in \eqref{mar PS estimator}--\eqref{mar OR estimator} can also be constructed after variable screening. If the outcome regression model is correctly specified and all relevant features are retained after screening, estimation accuracy can be preserved while requiring only sample means of the selected features. However, this strategy is less robust, as it relies strongly on correct specification of the outcome regression model. In contrast, the proposed approach is more flexible and remains valid under possible misspecification of the outcome regression model, provided that the propensity score model is correctly specified.

Efficient statistical inference under covariate-dependent labeling, with asymptotic normality established in Theorem~\ref{theorem mar generalizability Asymptotics body}, requires access to the sub-matrix \(\bar{\Xi}_{0,\bar S} = \sum_{i=1}^n(1 - \Gamma_i)X_{i,\bar S}X_{i,\bar S}^\T/\sum_{i=1}^n(1 - \Gamma_i)\), where \(\bar S = \cup_{k=1}^K\mathrm{supp}(\widehat \beta^{(-k)}_{OR})\). In some practical settings, such second-moment information may not be fully available, and only marginal second moments for individual features, corresponding to the diagonal entries of \(\bar{\Xi}_0\), may be observed. Let \(\bar{\Xi}_{0,\bar S}^{\mathrm{diag}}\) denote the diagonal matrix formed by the diagonal elements of \(\bar{\Xi}_{0,\bar S}\), with off-diagonal entries set to zero. In this case, confidence intervals can still be constructed by replacing the quadratic term \(\widehat \beta_{OR}^{(-k),\T} \bar{\Xi}_0 \widehat \beta_{OR}^{(-k)}\) in the definitions of \(\widehat\sigma_g^2\) and \(\widehat\sigma_t^2\) with the conservative quantity \(\|\widehat \beta_{OR}^{(-k)}\|_0 \widehat \beta_{OR,\bar S}^{(-k),\T} \bar{\Xi}_{0,\bar S}^{\rm diag} \widehat \beta_{OR,\bar S}^{(-k)}\). The resulting confidence intervals are inflated by a factor of order \(|\bar S|^{1/2}\) while remaining asymptotically valid.
\end{remark}

\subsection{\label{sec: mar theory} Asymptotic theory}

Define the target nuisance parameters: $\alpha^*_{PS} = \argmin_{\alpha \in \R^d} \E\{(1-\Gamma)X^\T \alpha + \Gamma \exp(-X^\T \alpha)\}$ and $\beta^*_{OR} = \argmin_{\beta \in \R^d} \E\{\Gamma \exp(-X^\T \alpha^*_{PS})(Y - X^\T \beta)^2\}$.
\begin{assumption} \label{assumption mar nuisance (a)}
Suppose that there exist constants $c_0, k_0\in (0,1)$ and $\sigma,\kappa_l,\sigma_w,\delta_w>0$ such that the following conditions hold: (a) The propensity score model satisfies $k_0(1-\gamma_n)/\gamma_n \leq \{1-g(X^\T \alpha^*_{PS})\}/g(X^\T \alpha^*_{PS}) \leq k_0^{-1}(1-\gamma_n)/\gamma_n$ almost surely, $\P(\Gamma=0)\geq c_0$, and $\E\{\gamma_n^q(X)\} \leq \nu \gamma_n^q$ for some $q>1$ and $\nu>0$. (b) For each $j\in\{0,1\}$, and conditional on $\Gamma=j$, $X$ is a sub-Gaussian random vector and $X^\T\beta_{OR}^*$ is a sub-Gaussian random variable, both with parameter $\sigma$. In addition, $\inf_{v \in \R^d,\;\|v\|_2=1}\E\{(X^\T v)^2\mid \Gamma=1\} \geq \kappa_l$. (c) The residual $w_{OR} = Y - X^\T \beta^*_{OR}$ is a sub-Gaussian random variable with parameter $\sigma_w$, $\E(w_{OR}^8\mid \Gamma=1) \leq \sigma^8_w$, and $\E(w_{OR}^2\mid \Gamma=1) \geq \delta_w$.
\end{assumption}

Assumption~\ref{assumption mar nuisance (a)}(a) extends the standard overlap condition to accommodate potential model misspecification and decaying labeling probabilities under missing at random mechanisms. When the propensity score function is correctly specified as \(\gamma_n(X) = g(X^\T \alpha_{PS}^*)\) and \(\gamma_n \asymp 1\), this assumption simplifies to the usual overlap condition, \(k_0 \leq \gamma_n(X) \leq 1-k_0\) for some $k_0 \in (0,1)$ \citep{crump2009dealing, rothe2017robust, khan2010irregular}. Assumptions~\ref{assumption mar nuisance (a)}(b)--(c) impose standard regularity conditions on the distributions of the covariate vector and the residual.

The following theorem demonstrates the asymptotic behavior of $\widehat \theta_g$ and $\widehat \theta_t$. 

\begin{theorem}\label{theorem mar generalizability Asymptotics body}
Suppose that Assumptions~\ref{assumption mar} and~\ref{assumption mar nuisance (a)} hold. Choose $\lambda_\alpha \asymp \lambda_\beta \asymp \{\log d/(n\gamma_n)\}^{1/2}$. Let $n\gamma_n \gg  (\log n)^2\log d$, $\|\alpha^*_{PS}\|_0 \|\beta^*_{OR}\|_0 = o(n\gamma_n/\{\log n(\log d)^2\})$, and either of the following conditions hold: (1) (Correct outcome regression model) $\mu(X) = X^\T \beta^*_{OR}$ or (2) (Correct propensity score model) $\gamma_n(X) = g(X^\T \alpha^*_{PS})$ and $\|\alpha^*_{PS}\|_0 = o((n\gamma_n)^{1/2}/\log d)$. Then as $n, d \rightarrow \infty$,

(a) $\widehat \theta_g -\theta_g = O_p(n^{-1/2}\gamma_n^{-1/2})$, $\widehat \sigma_g^2 = \sigma_g^2\{1 + o_p(1)\}$, and $\widehat\sigma_g^{-1}n^{1/2}(\widehat \theta_g - \theta_g)\rightarrow_d \mc{N}(0,1)$, where $\sigma_g^2 = \var(X^\T \beta^*_{OR} + \Gamma w_{OR}/g(X^\T\alpha^*_{PS}))$.

(b) $\widehat \theta_t -\theta_t = O_p(n^{-1/2}\gamma_n^{-1/2})$, $\widehat \sigma_t^2 = \sigma_t^2\{1 + o_p(1)\}$, and $\widehat\sigma_t^{-1}n^{1/2}(\widehat \theta_t - \theta_t)\rightarrow_d \mc{N}(0,1)$, where $\sigma_t^2 =\E([(1-\Gamma)(X^\T \beta^*_{OR}-\theta_t)/(1-\gamma_n) + \Gamma\{1-g(X^\T \alpha^*_{PS})\}w_{OR}/\{(1-\gamma_n)g(X^\T \alpha^*_{PS})\}]^2)$.
\end{theorem}

Theorem~\ref{theorem mar generalizability Asymptotics body} shows that valid \emph{inference} for both \(\theta_g\) and \(\theta_t\) can be achieved as long as either the propensity score model or the outcome regression model is correctly specified. In particular, correct specification of only one of the two nuisance models is sufficient, highlighting the \emph{model doubly robust} property of the proposed estimator. Moreover, when higher-order moment information of the unlabeled covariates is available, more flexible models based on linear representations after suitable basis transformations can be used to better approximate the nuisance components, which may further improve robustness.

When both nuisance models are correctly specified, the asymptotic variances match those in \citet{bang2005doubly} and \citet{zhang2023double}, aligning with semiparametric efficiency bounds derived in settings where individualized covariates are observed \citep{tsiatis2007semiparametric, muller2012efficient, kennedy2024semiparametric}. Our results show that this efficiency bound can still be attained using only unlabeled covariate summaries, provided the nuisance models are correctly specified.

\begin{remark}[Technical challenges associated with the specialized cross-fitting]\label{remark:challenge of cross-fitting}
To achieve robust inference, we construct an outcome regression estimator \(\widehat \beta^{(-k)}_{OR}\) that incorporates weighting based on an initial propensity score estimator \(\widehat \alpha^{(-k)}_{PS}\). Unlike in settings with covariate-independent labeling, estimating the propensity score as a function of covariates is essential to adjust for the confounding bias in settings where labeling depends on covariates. However, since the construction of \(\widehat \alpha^{(-k)}_{PS}\) involves unlabeled covariates, the conditional independence between the outcome regression estimator \(\widehat \beta^{(-k)}_{OR}\) and the unlabeled summary $\bar X_0$, which holds under covariate-independent labeling as discussed in Remark~\ref{remark:sample-splitting}, no longer holds.

To resolve this dependency, we define an oracle estimator \(\widetilde \beta^{(-k)}_{OR}\) as in \eqref{mar OR estimator}, but replacing \(\widehat \alpha^{(-k)}_{PS}\) with the true parameter \(\alpha_{PS}^*\). Unlike \(\widehat \beta^{(-k)}_{OR}\), the oracle estimator \(\widetilde \beta^{(-k)}_{OR}\) does not depend on the unlabeled sample, which allows us to exploit its conditional independence with \(\bar X_0\) to control the effect of the error \(\widetilde \beta^{(-k)}_{OR} - \beta_{OR}^*\) on the final mean estimator. The remaining estimation error \(\Delta = \widehat \beta^{(-k)}_{OR} - \widetilde \beta^{(-k)}_{OR}\), which depends on \(\bar X_0\), is controlled using uniform bounds based on its \(\ell_1\)-norm, since \(\Delta\) enters the final estimator in a linear form. Lemma~\ref{lemma consistency of betahat and betatilde for product sparsity} shows that \(\|\Delta\|_1 = O_p(\{\|\alpha_{PS}^*\|_0 \|\beta_{OR}^*\|_0 \log d / (n\gamma_n)\}^{1/2})\), indicating that the error remains small provided that either \(\|\alpha_{PS}^*\|_0\) or \(\|\beta_{OR}^*\|_0\) is sufficiently small.

Building on the above analysis, we demonstrate that consistent and asymptotically normal estimation for the mean is achieved under the product sparsity condition \(\|\alpha_{PS}^*\|_0 \|\beta_{OR}^*\|_0 = o(n\gamma_n (\log d)^{-2}(\log n)^{-1} )\) when both nuisance models are correctly specified. With \(n\gamma_n\) interpreted as the effective sample size, this result aligns with existing doubly robust methods \citep{chernozhukov2017double,smucler2019unifying,zhang2023double,kallus2025role}, differing only by an additional \(\log n\) term. For scenarios where the outcome regression model is misspecified, an additional condition \(\|\alpha_{PS}^*\|_0 = o((n\gamma_n)^{1/2}/\log d)\) is required, consistent with \citet{smucler2019unifying}. Importantly, unlike existing approaches that apply cross-fitting across the entire dataset, our method confines this step to the labeled data, bypassing the need for individualized covariates from unlabeled data. This adjustment improves practical applicability while maintaining robustness guarantees.
\end{remark}

\begin{remark}[Technical challenges associated with the decaying probability]
A key feature of our setting is that $\|\alpha_{PS}^*\|_2 \to \infty$ as $\gamma_n \to 0$, since the intercept term diverges to $-\infty$ under Assumption~\ref{assumption mar nuisance (a)}. This leads to a setting that does not satisfy common high-dimensional assumptions that require bounded parameter norms. Despite this, by explicitly accounting for this diverging structure in the estimation procedure and adapting the concentration arguments to this regime, we can characterize its effect and show consistency. The resulting convergence rate is determined by the effective sample size $n\gamma_n$ rather than $n$.
\end{remark}

\subsection{\label{sec: causal} Applications to causal inference}

In this section, we apply the proposed methods to causal inference, with a focus on estimating average treatment effects in the constrained semi-supervised setting.

Consider the potential outcome framework \citep{rubin1974estimating, imbens2015causal}. Let the underlying random variables be $(\Gamma_i, A_i, X_i, Y_i(0), Y_i(1), Y_i)_{i=1}^n$, with $(\Gamma, A, X, Y(0), Y(1), Y)$ as an independent copy. Here, $\Gamma_i \in \{0,1\}$ is the labeling indicator, $A_i \in \{0,1\}$ is the treatment assignment, and $X_i \in \mathbb{R}^d$ is the covariate vector. The potential outcomes $Y_i(a)$ correspond to treatment level $a \in \{0,1\}$, and the observed outcome is $Y_i = Y_i(A_i)$. We observe the labeled data $(\Gamma_i, \Gamma_i A_i, \Gamma_i X_i, \Gamma_i Y_i)_{i=1}^n$, together with the unlabeled summaries $\bar X_0$ and $\bar \Xi_0$ as defined in \eqref{def:summary}. The labeling mechanism is allowed to depend on the covariates. We use standard causal identification conditions, including selection exchangeability, treatment exchangeability and treatment positivity; these conditions are stated formally in Section~\ref{sec: causal'} of the Supplementary Material.

Our parameters of interest are the average treatment effects defined with respect to the full population, corresponding to generalizability, and with respect to the unlabeled population, corresponding to transportability:
\begin{align}\label{def:ATE}
\tau_g = \E\{Y(1) - Y(0)\},\quad \tau_t = \E\{Y(1) - Y(0) \mid \Gamma=0\}.
\end{align}

In the following, we extend the mean estimation methods in Section~\ref{sec: mar method} to estimate $\tau_g$ and $\tau_t$ in the more general covariate-dependent labeling setting. The main construction is summarized below; full estimating procedures and theoretical results are provided in Section~\ref{sec: causal'} of the Supplementary Material.

For \(\tau_g\), it is sufficient to estimate \(\tau_{a,g} = \E\{Y(a)\}\) for each \(a \in \{0,1\}\). Since the potential outcome \(Y(a)\) is observed only when both \(\Gamma = 1\) and \(A = a\), this results in a two-occasion missing data problem. To handle this, we define \(\Gamma_{a} := \Gamma\mathrm{I}_{(A = a)}\), which serves as the ``effective labeling indicator''. Then, \(\tau_{a,g}\) can be estimated by applying Steps 1--4 of Section~\ref{sec: mar method}, replacing \(\Gamma\) with \(\Gamma_{a}\).

For \(\tau_t\), it is sufficient to estimate \(\tau_{a,t} = \E\{Y(a)\mid\Gamma = 0\}\) for each \(a \in \{0,1\}\). In Section~\ref{sec: mar}, the population is partitioned into two groups: (1) the labeled population with \(\Gamma = 1\), where \(Y\) is observed, and (2) the unlabeled target population with \(\Gamma = 0\), to which the results are transported. In the causal setting, the ``effective labeled population'' consists of individuals with \(\Gamma_{a} = \Gamma\mathrm{I}_{(A = a)} = 1\), since the relevant outcome is the potential outcome \(Y(a)\), which is observed only when \(\Gamma_{a} = 1\). Accordingly, the ``effective entire population'' includes those with \(\Omega_{a,i} = 1\), where \(\Omega_{a,i} := 1 - \Gamma_i + \Gamma_{a,i}\). That is, samples with \(\Omega_{a,i} = 0\), i.e., \(\Gamma_i = 1\) and \(A_i \neq a\), are excluded from the estimation of \(\tau_{a,t}\), as they do not contribute information on \(Y(a)\) and fall outside the relevant population. After removing these samples, \(\tau_{a,t}\) can be estimated by repeating Steps 1--4 of Section~\ref{sec: mar method}, replacing \(\Gamma\) with \(\Gamma_{a}\).

\section{Numerical experiments}\label{sec:num}

\subsection{Simulation studies}\label{sec:sim}
We evaluate the performance of the proposed estimators through simulation studies. Our focus is on the generalizability of causal inference in semi-supervised settings under covariate-dependent labeling mechanisms. The target parameter is \(\tau_g\), defined in \eqref{def:ATE}.

We introduce the data-generating process as follows. The covariates are generated as truncated normal variables, where \(X_i[1]=1\) and \(X_i[j]\) are independent and identically distributed as \(N_{\text{trun},2}\) for each \(i \in [n]\) and \(j \in \{2,\ldots,d\}\). Here, \(N_{\text{trun},2} \sim N \mid \{|N| \leq 2\}\) denotes a truncated normal distribution conditional on \(|N| \leq 2\), with \(N \sim \mathcal{N}(0,1)\). Let \(g(t)=e^t/(1+e^t)\) denote the logistic function, and let \(\epsilon_i\) be independent and identically distributed as \(\mathcal{N}(0,1)\). We consider the following two settings:

(a) Linear outcome regression model and non-logistic propensity score model. The treatment $A_i$ and group membership $\Gamma_i$ are generated as Bernoulli random variables with $\P(A_i = 1 \mid X_i) = 0.3\sin(X_i^\T\alpha_1^*) + 0.5$ and $\P(\Gamma_i = 1 \mid X_i, A=a) = g(X_i^\T \alpha^*_a)$ for $a \in \{0,1\}$. The outcomes follow linear models, $Y_i(a) = X_i^\T \beta^*_a + \epsilon_i$ and $Y_i = Y_i(A_i)$.

(b) Non-linear outcome regression model and logistic propensity score model. Consider $\P(A_i = 1 \mid X_i) = \{g(X_i^\T \alpha^*_1) - g(X_i^\T \alpha^*_0)+1\}/2$ and $\P(\Gamma_i=1 \mid X_i, A_i=a) = g(X_i^\T \alpha^*_a)/\P(A_i=a \mid X_i)$. This construction implies the logistic form for the joint propensity scores, $\P(\Gamma_i\mathrm{I}_{\{A_i=a\}}=1 \mid X_i) = g(X_i^\T \alpha^*_a)$. Define $X_i^2 := ((X_i[1])^2, \ldots, (X_i[d])^2)^\T$. The outcomes are generated through quadratic models, $Y_i(a) = 5X_i^\T \beta^*_a+ (X_i^2)^\T v^*_a + \epsilon_i$ and $Y_i = Y_i(A_i)$.

The parameter values are chosen as
\begin{align*}
\alpha^*_1 &= (\alpha_n,1, \boldsymbol{1}_{s_\alpha-1}/(s_\alpha-1), 0, \ldots, 0)^\T \in \R^d,\\
\alpha^*_0 &= (\alpha_n,-1, -\boldsymbol{1}_{s_\alpha-1}/(s_\alpha-1), 0, \ldots, 0)^\T \in \R^d,\\
\beta^*_1 &= 3(1, 1, \boldsymbol{1}_{s_\beta-1}/(s_\beta-1)^{1/2}, 0, \ldots, 0)^\T \in \R^d,\quad
\beta^*_0 = -\beta^*_1,\\
v^*_1 &= 0.5(-24, 1, \boldsymbol{1}_{s_\beta-1}/(s_\beta-1), 0, \ldots, 0)^\T \in \R^d,\quad
v^*_0 = -v^*_1,
\end{align*}
where $\boldsymbol{1}_{s} = (1,\ldots,1)\in\R^s$ for any $s\in \mathbb{Z}_+$ and $\alpha_n$ is chosen such that $\P(\Gamma=1)=\gamma_n$.

\begin{table}[!t]
\tbl{Simulation results for setting (a), with $d=201$, $s_\alpha=6$ and $s_\beta=2$}
{
\begin{tabular}{lccccccccc}
\hline Estimator &  Bias & RMSE & Length & Coverage & &  Bias & RMSE & Length & Coverage \\ 
\hline   &  \multicolumn{4}{c}{ $n=5000$, $\gamma_n=0.2$, $n\gamma_n=1000$} & &\multicolumn{4}{c}{$n=5000$, $\gamma_n=0.5$, $n\gamma_n=2500$}\\ \cline{2-5}\cline{7-10}
    $\hat{\tau}_{\rm aipw}$ & 0.567 & 0.686 & 1.805 & 0.785 & & 0.283 & 0.359 & 0.926 & 0.765\\
    $\hat{\tau}_{\rm cali}$ & 0.027 & 0.881 & 3.745 & 0.975 & & -0.009 & 0.165 & 0.826 & 0.985 \\
    % SS-AIPW & -0.005 & 0.154 & 0.514 & 0.900 & & -0.006 & 0.126 & 0.454 & 0.940\\
    $\hat{\tau}_{\rm brss}$ & -0.006 & 0.162 & 0.557 & 0.905 & & -0.006 & 0.128 & 0.467 & 0.935\\
    $\hat{\tau}_{\rm sbss}$ & -0.004 & 0.180 & 0.593 & 0.945& & -0.003 & 0.128 & 0.477 & 0.945 \\ \cline{2-10}
    &  \multicolumn{4}{c}{ $n=10000$, $\gamma_n=0.2$, $n\gamma_n=2000$} & &\multicolumn{4}{c}{$n=10000$, $\gamma_n=0.5$, $n\gamma_n=5000$}\\ \cline{2-5}\cline{7-10}   
    $\hat{\tau}_{\rm aipw}$ & 0.478 & 0.555 & 1.327 & 0.740 & & 0.246 & 0.278 & 0.666 & 0.740\\
    $\hat{\tau}_{\rm cali}$ & 0.004 & 0.219 & 1.101 & 0.995 & & 0.004 & 0.099 & 0.524 & 0.995\\    
    % SS-AIPW & -0.002 & 0.107 & 0.370 & 0.900 & & 0.002 & 0.083 & 0.323 & 0.935\\
    $\hat{\tau}_{\rm brss}$ & 0.001 & 0.106 & 0.407 & 0.935 & & 0.000 & 0.084 & 0.334 & 0.960\\
    $\hat{\tau}_{\rm sbss}$ & 0.004 & 0.113 & 0.433 & 0.950 & & 0.004 & 0.084 & 0.340 & 0.960\\\hline
\end{tabular}
}
\begin{tabnote}
RMSE, root mean squared error; Length, average length of the 95\% confidence intervals; Coverage, average coverage probability of the 95\% confidence intervals. Bias, RMSE and Length are computed using medians.
\end{tabnote}
\label{table 1}
\end{table}

We evaluate the performance of the proposed summary-based semi-supervised estimator, denoted by $\hat{\tau}_{\rm sbss}$, which combines individualized labeled samples with unlabeled covariate summaries; see Section~\ref{sec: causal}. 
For benchmarking, we compare the proposed estimator with several existing methods: (a) the supervised augmented inverse probability weighting estimator of \citet{chernozhukov2017double}, denoted by $\hat{\tau}_{\rm aipw}$; (b) the calibrated augmented inverse probability weighting estimator adapted from \citet{chu2023targeted}, denoted by $\hat{\tau}_{\rm cali}$, where calibration weights are obtained via empirical likelihood to match the first moment of covariates between the labeled sample and the full population; and (c) the bias-reduced semi-supervised estimator of \citet{zhang2023semi}, denoted by $\hat{\tau}_{\rm brss}$. For methods that allow cross-fitting, we split the data into 5 folds. Nuisance parameters are estimated using $\ell_1$-regularized methods, with tuning parameters selected by 5-fold cross-validation. Notably, $\hat{\tau}_{\rm aipw}$ relies only on labeled data and ignores unlabeled information, which can lead to biased estimates when labeling is covariate-dependent. In contrast, although $\hat{\tau}_{\rm brss}$ is designed to be consistent, it requires access to individual-level unlabeled covariates and is therefore not applicable when only summary-level unlabeled information is available. The estimator $\hat{\tau}_{\rm cali}$ is applicable in this setting, but its performance depends on the accuracy of the calibration weights used to extend inference to the full population. Simulation results over 200 repetitions are presented in Tables~\ref{table 1} and~\ref{table 2}.

\begin{table}[!t]
\tbl{Simulation results for setting (b), with $d=201$, $s_\alpha=6$ and $s_\beta=6$}
{
\begin{tabular}{lccccccccc}
    \hline Estimator &  Bias & RMSE & Length & Coverage & &  Bias & RMSE & Length & Coverage \\ 
    \hline   &  \multicolumn{4}{c}{$n=3000$,  $\gamma_n=0.2$, $n\gamma_n=600$} & &\multicolumn{4}{c}{$n=3000$, $\gamma_n=0.5$, $n\gamma_n=1500$}\\ \cline{2-5}\cline{7-10}
    $\hat{\tau}_{\rm aipw}$ & 0.978 & 2.226 & 8.534 & 0.945 & & 0.110 & 1.194 & 4.883 & 0.950 
    \\
    $\hat{\tau}_{\rm cali}$ & 0.937 & 3.474 & 14.532 & 0.980 & & 0.717 & 1.134 & 4.961 & 0.975 
    \\
    % SS-AIPW & 0.008 & 0.719 & 2.717 & 0.940
%  & & -0.032 & 0.697 & 2.684 & 0.945
% \\
    $\hat{\tau}_{\rm brss}$ & 0.008 & 0.723 & 2.703 & 0.920
 & & -0.004 & 0.689 & 2.681 & 0.950
    \\
    $\hat{\tau}_{\rm sbss}$ & 0.029 & 0.717 & 2.726 & 0.945
 & & -0.007 & 0.688 & 2.687 & 0.950
    \\
    \cline{2-10}   &  \multicolumn{4}{c}{ $n=10000$, $\gamma_n=0.2$, $n\gamma_n=2000$} & &\multicolumn{4}{c}{$n=10000$, $\gamma_n=0.5$, $n\gamma_n=5000$}\\ \cline{2-5}\cline{7-10}
    $\hat{\tau}_{\rm aipw}$ & 0.456 & 1.289 & 5.180 & 0.970
 & & 0.249 & 0.760 & 2.848 & 0.915
    \\
    $\hat{\tau}_{\rm cali}$ & 0.586 & 0.755 & 4.257 & 1.000 & & 0.249 & 0.536 & 2.380 & 0.975
    \\
    % SS-AIPW & -0.094 & 0.396 & 1.494 & 0.940
%  & & -0.049 & 0.385 & 1.474 & 0.950
% \\
    $\hat{\tau}_{\rm brss}$ & -0.063 & 0.388 & 1.486 & 0.945
  & & -0.003 & 0.383 & 1.471 & 0.950
    \\
	    $\hat{\tau}_{\rm sbss}$ & -0.047 & 0.385 & 1.493 & 0.950
	   & &-0.004 & 0.383 & 1.473 & 0.950
	\\\hline
\end{tabular}
}
\begin{tabnote}
The definitions of Bias, RMSE, Length and Coverage are the same as in Table~\ref{table 1}.
\end{tabnote}
\label{table 2}
\end{table}

Among the considered estimators, $\hat{\tau}_{\rm aipw}$ exhibits substantial bias and large RMSE, as it relies only on labeled data and becomes biased when $\tau_g=\E\{Y(1)-Y(0)\}\neq\E\{Y(1)-Y(0)\mid\Gamma=1\}$. In contrast, $\hat{\tau}_{\rm cali}$ incorporates summary-level unlabeled information through calibration weighting, but its validity depends on how accurately the calibration weights recover the target covariate distribution. Under both settings (a)--(b), $\hat{\tau}_{\rm cali}$ can produce RMSEs even larger than those of the biased estimator $\hat{\tau}_{\rm aipw}$ when $n\gamma_n$ is relatively small. As $n\gamma_n$ increases, its estimation accuracy improves, but it remains inferior to $\hat{\tau}_{\rm brss}$ and $\hat{\tau}_{\rm sbss}$.

The proposed estimator $\hat{\tau}_{\rm sbss}$ exhibits performance comparable to $\hat{\tau}_{\rm brss}$, with both estimators showing relatively low bias and RMSE. Their confidence intervals achieve accurate coverage provided that the effective sample size $n\gamma_n$ is sufficiently large, even when one nuisance model is misspecified. However, unlike $\hat{\tau}_{\rm brss}$, which requires access to individualized unlabeled covariates, the proposed estimator $\hat{\tau}_{\rm sbss}$ only requires unlabeled covariate summaries. Despite this restriction, the proposed method attains performance close to $\hat{\tau}_{\rm brss}$, highlighting its practicality, efficiency, and robustness across the simulation settings.

\subsection{Application to HIV drug resistance}
\label{sec:real}

We use the Nucleoside Reverse Transcriptase Inhibitor (NRTI) Mutation Pattern dataset and the Genotype\allowbreak-Phenotype dataset \citep{rhee2004distribution, melikian2012standardized} from the Stanford University HIV Drug Resistance Database \citep{rhee2003human}. The Genotype\allowbreak-Phenotype dataset is treated as the labeled sample, as it contains both gene mutation and drug resistance information. The NRTI Mutation Pattern dataset is used as the unlabeled sample after removing repeated records that appear in the Genotype\allowbreak-Phenotype dataset, and contains a large number of patient records with gene mutation information only. These datasets originate from different studies, and the availability of drug resistance measurements differs across them. As a result, there is a noticeable difference in mutation frequencies between the labeled and unlabeled samples, as shown in Table~\ref{tab:mutation_freq}.

We focus on the log-transformed lamivudine drug resistance as the outcome, and study the average effect of the presence of a specific mutation on drug resistance. The goal is to generalize findings from the labeled dataset to the full population without relying on individual-level information in the large unlabeled dataset. Specifically, we treat each mutation with a frequency exceeding $5\%$ in the full dataset as a treatment variable, and use the remaining mutations with at least 10 occurrences as covariates. This results in a dataset of size $n=247{,}921$, with a labeled sample size $n_L=2{,}120$, unlabeled sample size $n_U=245{,}801$, and covariate dimension $d=25$. 

\begin{table}[t!]
\tbl{Comparison of mutation frequencies between the labeled and unlabeled datasets; the comparison is restricted to frequently occurring mutations}
{
\begin{tabular}{c c c c c c c}
\hline
Mutation & 41L & 67N & 70R & 184V & 210W & 215Y \\
\hline
Labeled & 41.5\% & 33.9\% & 20.1\% & 49.3\% & 28.7\% & 35.2\% \\
Unlabeled & 8.4\% & 7.4\% & 5.4\% & 18.2\% & 4.9\% & 6.4\% \\
\hline
\end{tabular}
}
\label{tab:mutation_freq}
\end{table}

We implement the estimators $\hat{\tau}_{\rm aipw}$, $\hat{\tau}_{\rm cali}$, $\hat{\tau}_{\rm brss}$, and $\hat{\tau}_{\rm sbss}$ considered in Section~\ref{sec:sim}. To improve approximation accuracy, we further apply a basis expansion that incorporates all pairwise interactions among covariate mutations. After removing interaction terms that occur fewer than 10 times, the transformed covariate dimension becomes $d_2=147$. The corresponding estimators based on the expanded feature set are denoted by $\hat{\tau}_{\rm aipw,2}$, $\hat{\tau}_{\rm cali,2}$, $\hat{\tau}_{\rm brss,2}$, and $\hat{\tau}_{\rm sbss,2}$. To account for randomness induced by the cross-fitting procedure, we repeat cross-fitting 10 times and report the median of the resulting estimates as the final values; see also \citet{chernozhukov2017double} for further details. Figure~\ref{fig:real} reports the confidence intervals for the average treatment effects of the selected mutations on log-transformed drug resistance.

The proposed estimators $\hat{\tau}_{\rm sbss}$ and $\hat{\tau}_{\rm sbss,2}$ produce confidence intervals that are close to those obtained from $\hat{\tau}_{\rm brss}$ and $\hat{\tau}_{\rm brss,2}$, respectively. 
However, unlike $\hat{\tau}_{\rm brss}$ and $\hat{\tau}_{\rm brss,2}$, which require access to individual-level unlabeled covariates, the proposed estimators achieve comparable performance using only covariate summaries, showing that our approach improves practical applicability without loss of precision.
The confidence intervals with and without two-way interactions also largely overlap, indicating stable conclusions.

\begin{figure}[b!]
\centering
\includegraphics[width=0.75\textwidth,height=0.45\textwidth]{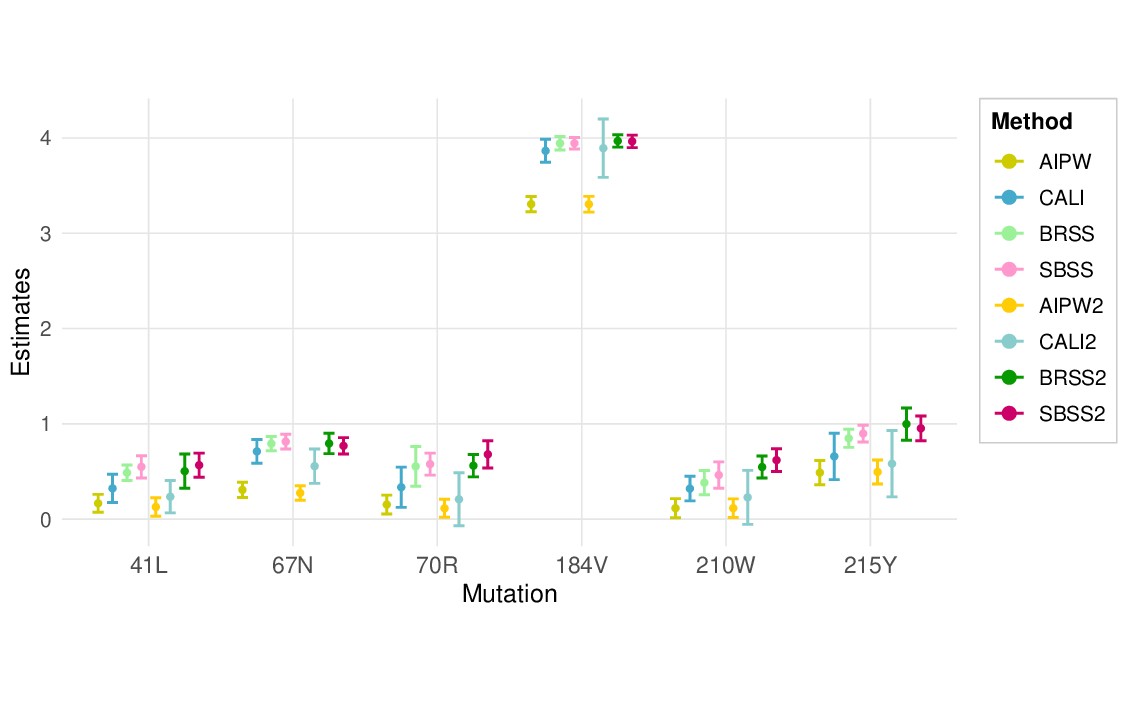}
\caption{Real-data analysis: confidence intervals for the average treatment effects of mutations on log-transformed drug resistance} \label{fig:real}
\end{figure}

On the other hand, $\hat{\tau}_{\rm cali}$ and $\hat{\tau}_{\rm cali,2}$ yield wider confidence intervals, particularly for the effects of mutations 70R and 210W, where $\hat{\tau}_{\rm cali,2}$ fails to detect significant effects while the other methods detect them. This pattern is likely driven by their reliance on calibration weighting, which can become unstable in high-dimensional settings. This finding is consistent with the simulation results; as shown in Tables~\ref{table 1}--\ref{table 2}, their confidence intervals are overly wide and the empirical coverage systematically exceeds the nominal 95\% level.

Lastly, $\hat{\tau}_{\rm aipw}$ and $\hat{\tau}_{\rm aipw,2}$ produce point estimates that differ substantially from those of the other estimators. Since these methods rely only on labeled data, they can be biased when labeling depends on covariates, as suggested by the clear differences in mutation frequencies reported in Table~\ref{tab:mutation_freq}.

\section{Discussion}

In this work, we study the mean estimation problem and the average treatment effect estimation problem by leveraging additional covariate information from unlabeled samples. Our results show that even without access to individual-level unlabeled covariates, reliable and efficient estimation can be achieved by using only covariate summaries, which are often easier to obtain in practice. Under covariate-independent labeling, we find that collecting additional individual-level covariates can further improve estimation efficiency in the presence of nonlinear relationships. We expect that similar gains in efficiency and robustness may also arise under covariate-dependent labeling, although a more detailed analysis is needed.

Given the varying costs of collecting different types of covariates, it is important to develop data-driven strategies to determine which information should be collected from unlabeled samples to balance estimation accuracy and data collection cost. In addition, beyond mean-type parameters, it would be useful to study the estimation of other quantities of interest, such as quantiles and regression coefficients, in more general regression models.

% \section*{Declaration of the use of generative AI and AI-assisted technologies}

% [ONLY INCLUDE IF GENERATIVE AI OR AI-ASSISTED TECHNOLOGIES HAVE BEEN USED IN PRODUCTION OF THE PAPER; SEE \S 4]. During the preparation of this work the author(s) used [NAME TOOL/SERVICE] in order to [REASON]. After using this tool/service the author(s) reviewed and edited the content as necessary and take(s) full responsibility for the content of the publication.

% \section*{Acknowledgement}
% Acknowledgements should appear after the body of the paper but before any appendices and be as brief as possible
% subject to politeness. 
%Information, such as contract numbers, of no interest to readers, should be excluded.

\section*{Data availability}
The data that support the findings in this paper were derived from the Stanford University HIV Drug Resistance Database at \texttt{https://hivdb.stanford.edu}.	

\section*{Supplementary material}

The Supplementary Material provides proofs and technical details. Section~\ref{sec: causal'} extends the causal inference methods of Section~\ref{sec: causal} with full estimation procedures and theoretical results. Proofs of the main results are given in Sections~\ref{sec: proof identifiability}--\ref{sec: proof causal}.

\bibliographystyle{plainnat}
\bibliography{ref}

@article{rhee2003human,
  title={Human immunodeficiency virus reverse transcriptase and protease sequence database},
  author={Rhee, Soo-Yon and Gonzales, Matthew J and Kantor, Rami and Betts, Bradley J and Ravela, Jaideep and Shafer, Robert W},
  journal={Nucleic Acids Research},
  volume={31},
  number={1},
  pages={298--303},
  year={2003},
  publisher={Oxford University Press}
}

@article{rhee2004distribution,
  title={Distribution of human immunodeficiency virus type 1 protease and reverse transcriptase mutation patterns in 4,183 persons undergoing genotypic resistance testing},
  author={Rhee, Soo-Yon and Liu, Tommy and Ravela, Jaideep and Gonzales, Matthew J and Shafer, Robert W},
  journal={Antimicrobial Agents and Chemotherapy},
  volume={48},
  number={8},
  pages={3122--3126},
  year={2004},
  publisher={American Society for Microbiology}
}

@article{melikian2012standardized,
  title={Standardized comparison of the relative impacts of HIV-1 reverse transcriptase (RT) mutations on nucleoside RT inhibitor susceptibility},
  author={Melikian, George L and Rhee, Soo-Yon and Taylor, Jonathan and Fessel, W Jeffrey and Kaufman, David and Towner, William and Troia-Cancio, Paolo V and Zolopa, Andrew and Robbins, Gregory K and Kagan, Ron and others},
  journal={Antimicrobial Agents and Chemotherapy},
  volume={56},
  number={5},
  pages={2305--2313},
  year={2012},
  publisher={American Society for Microbiology 1752 N St., NW, Washington, DC}
}

@article{crump2009dealing,
  title={Dealing with limited overlap in estimation of average treatment effects},
  author={Crump, Richard K and Hotz, V Joseph and Imbens, Guido W and Mitnik, Oscar A},
  journal={Biometrika},
  volume={96},
  number={1},
  pages={187--199},
  year={2009},
  publisher={Oxford University Press}
}

@article{chakrabortty2018efficient,
	Author = {Chakrabortty, Abhishek and Cai, Tianxi},
	Journal = {The Annals of Statistics},
	Number = {4},
	Pages = {1541--1572},
	Publisher = {Institute of Mathematical Statistics},
	Title = {Efficient and adaptive linear regression in semi-supervised settings},
	Volume = {46},
	Year = {2018}}

@article{chernozhukov2017double,
	Author = {Chernozhukov, Victor and Chetverikov, Denis and Demirer, Mert and Duflo, Esther and Hansen, Christian and Newey, Whitney},
	Journal = {American Economic Review},
	Number = {5},
	Pages = {261--65},
	Title = {Double/Debiased/Neyman machine learning of treatment effects},
	Volume = {107},
	Year = {2017}
}

@article{chernozhukov2018double,
  title={Double/debiased machine learning for treatment and structural parameters},
  author={Chernozhukov, Victor and Chetverikov, Denis and Demirer, Mert and Duflo, Esther and Hansen, Christian and Newey, Whitney and Robins, James},
  journal={The Econometrics Journal},
  volume={21},
  number={1},
  pages={C1--C68},
  year={2018},
  publisher={Wiley Online Library}
}

@book{van2000asymptotic,
  title={Asymptotic statistics},
  author={Van der Vaart, Aad W},
  volume={3},
  year={2000},
  publisher={Cambridge University Press}
}

@article{kuchibhotla2022moving,
  title={Moving beyond sub-Gaussianity in high-dimensional statistics: Applications in covariance estimation and linear regression},
  author={Kuchibhotla, Arun Kumar and Chakrabortty, Abhishek},
  journal={Information and Inference: A Journal of the IMA},
  volume={11},
  number={4},
  pages={1389--1456},
  year={2022},
  publisher={Oxford University Press}
}

@article{chakrabortty2019high,
  title={High Dimensional M-Estimation with Missing Outcomes: A Semi-Parametric Framework},
  author={Chakrabortty, Abhishek and Lu, Jiarui and Cai, T Tony and Li, Hongzhe},
  journal={arXiv preprint arXiv:1911.11345},
  year={2019}
}

@book{Vershynin_2018,
author = {Vershynin, Roman},
title = {High dimensional probability: An introduction with applications in data science},
volume={47},
year={2018},
publisher={Cambridge University Press}
}

@book{wainwright2019high,
  title={High-dimensional statistics: A non-asymptotic viewpoint},
  author={Wainwright, Martin J},
  volume={48},
  year={2019},
  publisher={Cambridge University Press}
}

@article{zhang2019semi,
  title={Semi-supervised inference: General theory and estimation of means},
  author={Zhang, Anru and Brown, Lawrence D and Cai, T Tony},
  journal={The Annals of Statistics},
  volume={47},
  number={5},
  pages={2538--2566},
  year={2019},
  publisher={Institute of Mathematical Statistics}
}

@article{bang2005doubly,
  title={Doubly robust estimation in missing data and causal inference models},
  author={Bang, Heejung and Robins, James M},
  journal={Biometrics},
  volume={61},
  number={4},
  pages={962--973},
  year={2005},
  publisher={Wiley Online Library}
}

@article{rothe2017robust,
  title={Robust confidence intervals for average treatment effects under limited overlap},
  author={Rothe, Christoph},
  journal={Econometrica},
  volume={85},
  number={2},
  pages={645--660},
  year={2017},
  publisher={Wiley Online Library}
}

@article{tony2020semisupervised,
  title={Semisupervised inference for explained variance in high dimensional linear regression and its applications},
  author={Cai, T Tony and Guo, Zijian},
  journal={Journal of the Royal Statistical Society: Series B (Statistical Methodology)},
  volume={82},
  number={2},
  pages={391--419},
  year={2020},
  publisher={Wiley Online Library}
}

@book{tsiatis2007semiparametric,
  title={Semiparametric theory and missing data},
  author={Tsiatis, Anastasios},
  year={2007},
  publisher={Springer Science \& Business Media}
}

@article{farrell2015robust,
  title={Robust inference on average treatment effects with possibly more covariates than observations},
  author={Farrell, Max H},
  journal={Journal of Econometrics},
  volume={189},
  number={1},
  pages={1--23},
  year={2015},
  publisher={Elsevier}
}

@article{athey2018approximate,
  title={Approximate residual balancing: debiased inference of average treatment effects in high dimensions},
  author={Athey, Susan and Imbens, Guido W and Wager, Stefan},
  journal={Journal of the Royal Statistical Society Series B},
  volume={80},
  number={4},
  pages={597--623},
  year={2018},
  publisher={Royal Statistical Society}
}

@article{azriel2022semi,
  title={Semi-supervised linear regression},
  author={Azriel, David and Brown, Lawrence D and Sklar, Michael and Berk, Richard and Buja, Andreas and Zhao, Linda},
  journal={Journal of the American Statistical Association},
  volume={117},
  number={540},
  pages={2238--2251},
  year={2022},
  publisher={Taylor \& Francis}
}

@book{imbens2015causal,
  title={Causal inference in statistics, social, and biomedical sciences},
  author={Imbens, Guido W and Rubin, Donald B},
  year={2015},
  publisher={Cambridge University Press}
}

@article{rubin1974estimating,
  title={Estimating causal effects of treatments in randomized and nonrandomized studies.},
  author={Rubin, Donald B},
  journal={Journal of Educational Psychology},
  volume={66},
  number={5},
  pages={688},
  year={1974},
  publisher={American Psychological Association}
}

@article{tan2020model,
  title={Model-assisted inference for treatment effects using regularized calibrated estimation with high-dimensional data},
  author={Tan, Zhiqiang},
  journal={The Annals of Statistics},
  volume={48},
  number={2},
  pages={811--837},
  year={2020},
  publisher={Institute of Mathematical Statistics}
}

@article{smucler2019unifying,
  title={A unifying approach for doubly-robust $\ell_1$ regularized estimation of causal contrasts},
  author={Smucler, Ezequiel and Rotnitzky, Andrea and Robins, James M},
  journal={arXiv preprint arXiv:1904.03737},
  year={2019}
}

@article{dukes2021inference,
  title={Inference for treatment effect parameters in potentially misspecified high-dimensional models},
  author={Dukes, Oliver and Vansteelandt, Stijn},
  journal={Biometrika},
  volume={108},
  number={2},
  pages={321--334},
  year={2021},
  publisher={Oxford University Press}
}

@article{zhang2023double,
  title={Double robust semi-supervised inference for the mean: Selection bias under mar labeling with decaying overlap},
  author={Zhang, Yuqian and Chakrabortty, Abhishek and Bradic, Jelena},
  journal={Information and Inference: A Journal of the IMA},
  volume={12},
  number={3},
  pages={2066--2159},
  year={2023},
  publisher={Oxford University Press}
}

@article{rosenbaum1983central,
  title={The central role of the propensity score in observational studies for causal effects},
  author={Rosenbaum, Paul R and Rubin, Donald B},
  journal={Biometrika},
  volume={70},
  number={1},
  pages={41--55},
  year={1983},
  publisher={Oxford University Press}
}

@article{rubin1976inference,
  title={Inference and missing data},
  author={Rubin, Donald B},
  journal={Biometrika},
  volume={63},
  number={3},
  pages={581--592},
  year={1976},
  publisher={Oxford University Press}
}

@article{cheng2021robust,
  title={Robust and efficient semi-supervised estimation of average treatment effects with application to electronic health records data},
  author={Cheng, David and Ananthakrishnan, Ashwin N and Cai, Tianxi},
  journal={Biometrics},
  volume={77},
  number={2},
  pages={413--423},
  year={2021},
  publisher={Wiley Online Library}
}

@article{ning2020robust,
  title={Robust estimation of causal effects via a high-dimensional covariate balancing propensity score},
  author={Ning, Yang and Sida, Peng and Imai, Kosuke},
  journal={Biometrika},
  volume={107},
  number={3},
  pages={533--554},
  year={2020},
  publisher={Oxford University Press}
}

@article{zhang2022high,
  title={High-dimensional semi-supervised learning: in search of optimal inference of the mean},
  author={Zhang, Yuqian and Bradic, Jelena},
  journal={Biometrika},
  volume={109},
  number={2},
  pages={387--403},
  year={2022},
  publisher={Oxford University Press}
}

@article{chakrabortty2022general,
  title={A General Framework for Treatment Effect Estimation in Semi-Supervised and High Dimensional Settings},
  author={Chakrabortty, Abhishek and Dai, Guorong and Tchetgen, Eric Tchetgen},
  journal={arXiv preprint arXiv:2201.00468},
  year={2022}
}

@article{funk2011doubly,
  title={Doubly robust estimation of causal effects},
  author={Funk, Michele Jonsson and Westreich, Daniel and Wiesen, Chris and St{\"u}rmer, Til and Brookhart, M Alan and Davidian, Marie},
  journal={American Journal of Epidemiology},
  volume={173},
  number={7},
  pages={761--767},
  year={2011},
  publisher={Oxford University Press}
}

@article{zhang2021dynamic,
  title={Dynamic treatment effects: high-dimensional inference under model misspecification},
  author={Zhang, Yuqian and Ji, Weijie and Bradic, Jelena},
  journal={arXiv preprint arXiv:2111.06818},
  year={2021}
}

@article{lesko2017generalizing,
  title={Generalizing study results: a potential outcomes perspective},
  author={Lesko, Catherine R and Buchanan, Ashley L and Westreich, Daniel and Edwards, Jessie K and Hudgens, Michael G and Cole, Stephen R},
  journal={Epidemiology (Cambridge, Mass.)},
  volume={28},
  number={4},
  pages={553},
  year={2017},
  publisher={NIH Public Access}
}

@article{shi2023data,
  title={Data integration in causal inference},
  author={Shi, Xu and Pan, Ziyang and Miao, Wang},
  journal={Wiley Interdisciplinary Reviews: Computational Statistics},
  volume={15},
  number={1},
  pages={e1581},
  year={2023},
  publisher={Wiley Online Library}
}

@article{zhang2023semi,
  title={The Decaying Missing-at-Random Framework: Doubly Robust Causal Inference with Partially Labeled Data},
  author={Zhang, Yuqian and Chakrabortty, Abhishek and Bradic, Jelena},
  journal={arXiv preprint arXiv:2305.12789},
  year={2023}
}

@article{imai2015robust,
  title={Robust estimation of inverse probability weights for marginal structural models},
  author={Imai, Kosuke and Ratkovic, Marc},
  journal={Journal of the American Statistical Association},
  volume={110},
  number={511},
  pages={1013--1023},
  year={2015},
  publisher={Taylor \& Francis}
}

@article{khan2010irregular,
  title={Irregular identification, support conditions, and inverse weight estimation},
  author={Khan, Shakeeb and Tamer, Elie},
  journal={Econometrica},
  volume={78},
  number={6},
  pages={2021--2042},
  year={2010},
  publisher={Wiley Online Library}
}

@article{muller2015partial,
  title={The partial linear model in high dimensions},
  author={M{\"u}ller, Patric and Van de Geer, Sara},
  journal={Scandinavian Journal of Statistics},
  volume={42},
  number={2},
  pages={580--608},
  year={2015},
  publisher={Wiley Online Library}
}

@article{xie2009scad,
author = {Huiliang Xie and Jian Huang},
title = {{SCAD-penalized regression in high-dimensional partially linear models}},
volume = {37},
journal = {The Annals of Statistics},
number = {2},
publisher = {Institute of Mathematical Statistics},
pages = {673--696},
year = {2009}
}

@article{lv2022debiased,
  title={Debiased distributed learning for sparse partial linear models in high dimensions},
  author={Lv, Shaogao and Lian, Heng},
  journal={Journal of Machine Learning Research},
  volume={23},
  number={2},
  pages={1--32},
  year={2022}
}

@inproceedings{zhu2019high,
  title={High dimensional inference in partially linear models},
  author={Zhu, Ying and Yu, Zhuqing and Cheng, Guang},
  booktitle={The 22nd International Conference on Artificial Intelligence and Statistics},
  pages={2760--2769},
  year={2019},
  organization={PMLR}
}

@article{zhang2016statistical,
  title={Statistical inference on restricted partial linear regression models with partial distortion measurement errors},
  author={Zhang, Jun and Zhou, Nanguang and Sun, Zipeng and Li, Gaorong and Wei, Zhenghong},
  journal={Statistica Neerlandica},
  volume={70},
  number={4},
  pages={304--331},
  year={2016},
  publisher={Wiley Online Library}
}

@article{angelopoulos2023prediction,
  title={Prediction-powered inference},
  author={Angelopoulos, Anastasios N and Bates, Stephen and Fannjiang, Clara and Jordan, Michael I and Zrnic, Tijana},
  journal={Science},
  volume={382},
  number={6671},
  pages={669--674},
  year={2023},
  publisher={American Association for the Advancement of Science}
}

@article{colnet2024causal,
  title={Causal inference methods for combining randomized trials and observational studies: a review},
  author={Colnet, B{\'e}n{\'e}dicte and Mayer, Imke and Chen, Guanhua and Dieng, Awa and Li, Ruohong and Varoquaux, Ga{\"e}l and Vert, Jean-Philippe and Josse, Julie and Yang, Shu},
  journal={Statistical Science},
  volume={39},
  number={1},
  pages={165--191},
  year={2024},
  publisher={Institute of Mathematical Statistics}
}

@article{ung2024combining,
  title={Combining an experimental study with external data: study designs and identification strategies},
  author={Ung, Lawson and Wang, Guanbo and Haneuse, Sebastien and Hernan, Miguel A and Dahabreh, Issa J},
  journal={arXiv preprint arXiv:2406.03302},
  year={2024}
}

@article{degtiar2023review,
  title={A review of generalizability and transportability},
  author={Degtiar, Irina and Rose, Sherri},
  journal={Annual Review of Statistics and Its Application},
  volume={10},
  number={1},
  pages={501--524},
  year={2023},
  publisher={Annual Reviews}
}

@article{buchanan2018generalizing,
  title={Generalizing evidence from randomized trials using inverse probability of sampling weights},
  author={Buchanan, Ashley L and Hudgens, Michael G and Cole, Stephen R and Mollan, Katie R and Sax, Paul E and Daar, Eric S and Adimora, Adaora A and Eron, Joseph J and Mugavero, Michael J},
  journal={Journal of the Royal Statistical Society Series A: Statistics in Society},
  volume={181},
  number={4},
  pages={1193--1209},
  year={2018},
  publisher={Oxford University Press}
}

@article{dahabreh2021study,
  title={Study designs for extending causal inferences from a randomized trial to a target population},
  author={Dahabreh, Issa J and Haneuse, Sebastien JP A and Robins, James M and Robertson, Sarah E and Buchanan, Ashley L and Stuart, Elizabeth A and Hern{\'a}n, Miguel A},
  journal={American Journal of Epidemiology},
  volume={190},
  number={8},
  pages={1632--1642},
  year={2021},
  publisher={Oxford University Press}
}

@article{deng2024optimal,
  title={Optimal and safe estimation for high-dimensional semi-supervised learning},
  author={Deng, Siyi and Ning, Yang and Zhao, Jiwei and Zhang, Heping},
  journal={Journal of the American Statistical Association},
  volume={119},
  number={548},
  pages={2748--2759},
  year={2024},
  publisher={Taylor \& Francis}
}

@article{song2024general,
  title={A general m-estimation theory in semi-supervised framework},
  author={Song, Shanshan and Lin, Yuanyuan and Zhou, Yong},
  journal={Journal of the American Statistical Association},
  volume={119},
  number={546},
  pages={1065--1075},
  year={2024},
  publisher={Taylor \& Francis}
}

@article{zrnic2024cross,
  title={Cross-prediction-powered inference},
  author={Zrnic, Tijana and Cand{\`e}s, Emmanuel J},
  journal={Proceedings of the National Academy of Sciences},
  volume={121},
  number={15},
  pages={e2322083121},
  year={2024},
  publisher={National Acad Sciences}
}

@article{dahabreh2020extending,
  title={Extending inferences from a randomized trial to a new target population},
  author={Dahabreh, Issa J and Robertson, Sarah E and Steingrimsson, Jon A and Stuart, Elizabeth A and Hernan, Miguel A},
  journal={Statistics in Medicine},
  volume={39},
  number={14},
  pages={1999--2014},
  year={2020},
  publisher={Wiley Online Library}
}

@article{kennedy2024semiparametric,
  title={Semiparametric doubly robust targeted double machine learning: a review},
  author={Kennedy, Edward H},
  journal={Handbook of Statistical Methods for Precision Medicine},
  pages={207--236},
  year={2024},
  publisher={Chapman and Hall/CRC}
}

@article{zhao2020bayesian,
  title={Bayesian weighted Mendelian randomization for causal inference based on summary statistics},
  author={Zhao, Jia and Ming, Jingsi and Hu, Xianghong and Chen, Gang and Liu, Jin and Yang, Can},
  journal={Bioinformatics},
  volume={36},
  number={5},
  pages={1501--1508},
  year={2020},
  publisher={Oxford University Press}
}

@article{morrison2020mendelian,
  title={Mendelian randomization accounting for correlated and uncorrelated pleiotropic effects using genome-wide summary statistics},
  author={Morrison, Jean and Knoblauch, Nicholas and Marcus, Joseph H and Stephens, Matthew and He, Xin},
  journal={Nature Genetics},
  volume={52},
  number={7},
  pages={740--747},
  year={2020},
  publisher={Nature Publishing Group US New York}
}

@article{burgess2013mendelian,
  title={Mendelian randomization analysis with multiple genetic variants using summarized data},
  author={Burgess, Stephen and Butterworth, Adam and Thompson, Simon G},
  journal={Genetic Epidemiology},
  volume={37},
  number={7},
  pages={658--665},
  year={2013},
  publisher={Wiley Online Library}
}

@article{lin2010relative,
  title={On the relative efficiency of using summary statistics versus individual-level data in meta-analysis},
  author={Lin, Dan-Yu and Zeng, Daniel},
  journal={Biometrika},
  volume={97},
  number={2},
  pages={321--332},
  year={2010},
  publisher={Oxford University Press}
}

@article{liu2015multivariate,
  title={Multivariate meta-analysis of heterogeneous studies using only summary statistics: efficiency and robustness},
  author={Liu, Dungang and Liu, Regina Y and Xie, Minge},
  journal={Journal of the American Statistical Association},
  volume={110},
  number={509},
  pages={326--340},
  year={2015},
  publisher={Taylor \& Francis}
}

@article{zhu2015meta,
  title={Meta-analysis of correlated traits via summary statistics from GWASs with an application in hypertension},
  author={Zhu, Xiaofeng and Feng, Tao and Tayo, Bamidele O and Liang, Jingjing and Young, J Hunter and Franceschini, Nora and Smith, Jennifer A and Yanek, Lisa R and Sun, Yan V and Edwards, Todd L and others},
  journal={The American Journal of Human Genetics},
  volume={96},
  number={1},
  pages={21--36},
  year={2015},
  publisher={Elsevier}
}

@article{muller2012efficient,
author = {Ursula U. M{\"u}ller and Ingrid Van Keilegom},
title = {{Efficient parameter estimation in regression with missing responses}},
volume = {6},
journal = {Electronic Journal of Statistics},
number = {none},
publisher = {Institute of Mathematical Statistics and Bernoulli Society},
pages = {1200--1219},
year = {2012},
}

@article{chu2023targeted,
  title={Targeted optimal treatment regime learning using summary statistics},
  author={Chu, J and Lu, W and Yang, S},
  journal={Biometrika},
  volume={110},
  number={4},
  pages={913--931},
  year={2023},
  publisher={Biometrika Trust}
}

@article{qian2024changepoint,
  title={Changepoint Detection in Complex Models: Cross-Fitting Is Needed},
  author={Qian, Chengde and Wang, Guanghui and Wang, Zhaojun and Zou, Changliang},
  journal={arXiv e-prints},
  pages={arXiv--2411},
  year={2024}
}

@article{deville1992calibration,
  title={Calibration estimators in survey sampling},
  author={Deville, Jean-Claude and S{\"a}rndal, Carl-Erik},
  journal={Journal of the American statistical Association},
  volume={87},
  number={418},
  pages={376--382},
  year={1992},
  publisher={Taylor \& Francis}
}

@article{chen2025semi,
  title={Semi-supervised linear regression: enhancing efficiency and robustness in high dimensions},
  author={Chen, Kai and Zhang, Yuqian},
  journal={Biometrics},
  volume={81},
  number={3},
  pages={ujaf113},
  year={2025},
  publisher={Oxford University Press}
}

@article{hainmueller2012entropy,
  title={Entropy balancing for causal effects: A multivariate reweighting method to produce balanced samples in observational studies},
  author={Hainmueller, Jens},
  journal={Political Analysis},
  volume={20},
  number={1},
  pages={25--46},
  year={2012},
  publisher={Cambridge University Press}
}

@article{zubizarreta2015stable,
  title={Stable weights that balance covariates for estimation with incomplete outcome data},
  author={Zubizarreta, Jos{\'e} R},
  journal={Journal of the American Statistical Association},
  volume={110},
  number={511},
  pages={910--922},
  year={2015},
  publisher={Taylor \& Francis}
}

@article{kallus2025role,
  title={On the role of surrogates in the efficient estimation of treatment effects with limited outcome data},
  author={Kallus, Nathan and Mao, Xiaojie},
  journal={Journal of the Royal Statistical Society Series B: Statistical Methodology},
  volume={87},
  number={2},
  pages={480--509},
  year={2025},
  publisher={Oxford University Press UK}
}

@article{testa2025semiparametric,
  title={Semiparametric semi-supervised learning for general targets under distribution shift and decaying overlap},
  author={Testa, Lorenzo and Xu, Qi and Lei, Jing and Roeder, Kathryn},
  journal={arXiv preprint arXiv:2505.06452},
  year={2025}
}

@article{hu2026semiparametric,
  title={Semiparametric efficient fusion of individual data and summary statistics},
  author={Hu, Wenjie and Wang, Ruoyu and Li, Wei and Miao, Wang},
  journal={Journal of the American Statistical Association},
  number={just-accepted},
  pages={1--28},
  year={2026},
  publisher={Taylor \& Francis}
}

@article{zhao2017entropy,
  title={Entropy Balancing is Doubly Robust},
  author={Zhao, Qingyuan and Percival, Daniel},
  journal={Journal of Causal Inference},
  volume={5},
  number={1},
  pages={1--19},
  year={2017},
  publisher={De Gruyter}
}

\newpage

\appendix

\begin{center}
    \Large\textbf{Supplementary material for ``Semi-supervised inference using unlabeled summary statistics''}
\end{center}

\section{Details on the applications to causal inference}\label{sec: causal'}

In this section, we present further details on the causal inference problems discussed in Section~\ref{sec: causal}.

To identify the average treatment effect (ATE) parameters defined in \eqref{def:ATE}, we assume the following identification conditions.

\begin{assumption} \label{assumption causal}
The following conditions hold: (a) $\Gamma \perp \{Y(1), Y(0)\} \mid X$. (b) $A \perp \{Y(1), Y(0)\} \mid X, \Gamma=1$. (c) $\eta_0 \leq \P(A=1 \mid X, \Gamma=1) \leq 1-\eta_0$ for some constant $0<\eta_0<1$.
\end{assumption}

Assumption~\ref{assumption causal} includes standard identification conditions in causal inference and semi-supervised learning \citep{lesko2017generalizing,dahabreh2020extending,shi2023data}. Specifically, Assumption~\ref{assumption causal}(a) corresponds to ``ignorability'' or ``selection exchangeability,'' ensuring that the group assignment is independent of potential outcomes given covariates. Assumption~\ref{assumption causal}(b), commonly known as ``no unmeasured confounding'' or ``treatment exchangeability,'' implies that all relevant confounders between treatments and potential outcomes have been accounted for.
%To satisfy Assumption~\ref{assumption causal}(a)--(b), unless there is complete control over both group and treatment assignments, it is generally necessary to collect as many potentially relevant baseline covariates as possible. This often results in a high-dimensional scenario where the number of covariates, \(d\), may exceed the labeled sample size, \(n_L\). Our framework is designed to address such challenges, providing robust solutions even when the dimensionality surpasses the sample size, a common issue in modern causal inference tasks. Lastly, 
Assumption~\ref{assumption causal}(c), referred to as the ``overlap'' or ``positivity'' condition, ensures that each individual has a non-zero probability of receiving either treatment. Notably, Assumption~\ref{assumption causal} allows the labeling mechanism to depend on covariates and accommodates heterogeneity between distinct treatment groups.

\subsection{Achieving generalizability using summary statistics from unlabeled data}\label{sec: 4.A}

We begin by focusing on the generalization of causal conclusions to a broader population and consider the estimation of \(\tau_g = \E[Y(1) - Y(0)]\) over the entire population.

As discussed in Section~\ref{sec: causal}, we treat
\(\Gamma_{a} := \Gamma\mathbbm{1}_{\{A = a\}}\) as the ``effective labeling indicator'' for estimating the expected potential outcome \(\tau_{a,g} = \E[Y(a)]\) for each $a\in\{0,1\}$. Notably, since we have access to individualized information from the labeled data, we can also compute the covariate sample average for the group with \(\Gamma_{a,i} = \Gamma_i \mathbbm{1}_{\{A_i = a\}} = 0\), which includes individuals who either (1) belong to the unlabeled sample (\(\Gamma_i = 0\)) or (2) are from the labeled sample but received a different treatment (\(\Gamma_i = 1\) and \(A_i \neq a\)):
\[\bar X_{a,0} = \frac{\sum_{i=1}^n\p{1-\Gamma_{a,i}}X_i}{\sum_{i=1}^n\p{1-\Gamma_{a,i}}} = \frac{\p{n-\sum_{i=1}^n \Gamma_i}\bar{X}_{0} + \sum_{i=1}^n \Gamma_i\mathbbm{1}_{(A_i\neq a)}X_i}{n-\sum_{i=1}^n\Gamma_{a,i}}.\]

In the following, we extend the mean estimation method from Section~\ref{sec: mar method} to the ATE estimation problem by replacing the original labeling indicator \(\Gamma_i\) with \(\Gamma_{a,i}\) for each treatment arm \(a \in \{0,1\}\) and imputing the individualized covariates for the group with \(\Gamma_{a,i} = 0\) using the sample mean \(\bar{X}_{a,0}\). The detailed construction is outlined below.

Step 1: Divide the index set $[n]$ into $K$ disjoint subsets $\mathcal{I}_1, \dots, \mathcal{I}_K$ with equal sizes such that $n_k := \abs{\mathcal{I}_k} = n/K \text{ for } k \in [K]$. Let $\mathcal{I}_{-k} = [n] \setminus \mathcal{I}_k$ and separate $\mathcal{I}_{-k}$ into two disjoint subsets $\mathcal{I}_{-k,\alpha}, \mathcal{I}_{-k,\beta}$ with the same size $M = \abs{\mathcal{I}_{-k,\alpha}} = \abs{\mathcal{I}_{-k,\beta}}$.

Step 2: For each $a\in \{0,1\}$ and $k \in [K]$, construct the propensity score estimator as
\[\widehat \alpha^{(-k)}_{a,g} = \argmin_{\alpha \in \R^d} \left[M^{-1}\sum_{i \in \mathcal{I}_{-k,\alpha}} \left\{(1-\Gamma_{a,i})\bar{X}_{a,0}^\top \alpha + \Gamma_{a,i} \exp(-X_i^\top \alpha)\right\} + \lambda_\alpha \norm{\alpha}_1\right].\]

Step 3: For each $a\in \{0,1\}$ and $k \in [K]$, construct the outcome regression estimator as
\[\widehat \beta^{(-k)}_{a,g} = \argmin_{\beta \in \R^d} \left\{ M^{-1}\sum_{i \in \mathcal{I}_{-k,\beta}} \Gamma_{a,i} \exp(-X_i^\top \widehat \alpha_{a,g}^{(-k)})(Y_i - X_i^\top \beta)^2 + \lambda_\beta \norm{\beta}_1 \right\}.\]

Step 4: Let $g(t)=e^t/(1+e^t)$ be the logistic function. For each $a \in \{0,1\}$, define the cross-fitted estimator of $\tau_{a,g}$ as
\[\widehat \tau_{a,g} = n^{-1}\sum_{k=1}^K\sum_{i\in \mathcal{I}_k} \p{\Gamma_{a,i}X_i + (1-\Gamma_{a,i})\bar{X}_{a,0}}^\top \widehat \beta^{(-k)}_{a,g}  + n^{-1}\sum_{k=1}^K\sum_{i\in \mathcal{I}_k} \frac{\Gamma_{a,i}\p{Y_i - X_i^\top\widehat \beta^{(-k)}_{a,g}}}{g\p{X_i^\top \widehat \alpha^{(-k)}_{a,g}}}.\]

Step 5: We propose the estimator for the ATE over the entire population:
\begin{equation}\label{def:tau_g}
    \widehat \tau_{g}=\widehat \tau_{1,g}-\widehat \tau_{0,g}.
\end{equation}
When we have access to second-moment information from the unlabeled data, we can further estimate the corresponding asymptotic variance as
\begin{align*}
    \widehat \Sigma_{g}^2 &= n^{-1}\sum_{k=1}^K \sum_{i \in \mathcal{I}_k}\p{1-\Gamma_i}\p{\widehat \beta_{1,g}^{(-k)} -  \widehat \beta_{0,g}^{(-k)}}^\top\bar\Xi_0\p{\widehat \beta_{1,g}^{(-k)} -  \widehat \beta_{0,g}^{(-k)}}\\
    &\qquad + n^{-1}\sum_{k=1}^K \sum_{i \in \mathcal{I}_k}\Gamma_iA_i\left\{X_i^\top \widehat \beta^{(-k)}_{1,g} - X_i^\top \widehat \beta^{(-k)}_{0,g} + \frac{Y_i - X_i^\top \widehat \beta^{(-k)}_{1,g} }{g\p{X_i^\top \widehat \alpha^{(-k)}_{1,g}}}\right\}^2\\
    &\qquad + n^{-1}\sum_{k=1}^K \sum_{i \in \mathcal{I}_k}\Gamma_i\p{1-A_i}\left\{X_i^\top \widehat \beta^{(-k)}_{1,g} - X_i^\top \widehat \beta^{(-k)}_{0,g} - \frac{Y_i - X_i^\top \widehat \beta^{(-k)}_{0,g} }{g\p{X_i^\top \widehat \alpha^{(-k)}_{0,g}}}\right\}^2 - \widehat \tau_g^2.
\end{align*}

To establish the theoretical properties of the proposed methods, for each $a \in \{0,1\}$, we define the target nuisance parameters as
\begin{align*}
    \alpha^*_{a,g} &= \argmin_{\alpha \in \R^d} \E\{\p{1-\Gamma_{a}}X^\top \alpha + \Gamma_{a} \exp\p{-X^\top \alpha} \},\\
    \beta^*_{a,g} &= \argmin_{\beta \in \R^d} \E\{\Gamma_{a} \exp\p{-X^\top  \alpha^*_{a,g}}\p{Y - X^\top \beta}^2\}.
\end{align*}
Let $w_{a,g} = Y\p{a} - X^\top \beta^*_{a,g}$ be the residuals from the outcome regression models. The following theorem characterizes the asymptotic behavior of the proposed ATE estimator $\widehat \tau_{g}$ of $\tau_g$.

\begin{theorem}\label{theorem causal generalizability Asymptotics body}
Let Assumption~\ref{assumption causal} hold, and suppose that, for each \(a \in \{0,1\}\), Assumption~\ref{assumption mar nuisance (a)} holds with \((\alpha^*_{PS}, \beta^*_{OR}, w_{OR})\) replaced by \((\alpha^*_{a,g}, \beta^*_{a,g}, w_{a,g})\). Assume that \(n\gamma_n \gg  (\log n)^2\log d\),
\[
    \|\alpha^*_{a,g}\|_0 \|\beta^*_{a,g}\|_0
    = o((n\gamma_n)^{1/2}/\{\log n(\log d)^2\}),
\]
and that either of the following conditions holds: (1) (correct OR model) \(\E\{Y(a) \mid X\} = X^\top \beta^*_{a,g}\), or (2) (correct PS model) \(\P(\Gamma_{a} = 1 \mid X) = g\p{X^\top \alpha^*_{a,g}}\) and \(\|\alpha^*_{a,g}\|_0 = o((n\gamma_n)^{1/2}/ \log d)\). Then as \(n, d \rightarrow \infty\), \(\widehat \tau_{g} -\tau_{g} = O_p \p{(n\gamma_n)^{-1/2}}\), \(\widehat \Sigma^2_{g} = \Sigma_g^2\{1 + o_p(1)\}\), and \(\widehat\Sigma^{-1}_{g}n^{1/2}(\widehat \tau_{g} - \tau_{g})\xrightarrow{d} \mathcal{N}(0,1)\), where
\[
    \Sigma_g^2 = \Var\left(X^\top \beta^*_{1,g}- X^\top \beta^*_{0,g}
    + \frac{\Gamma A}{g\p{X^\top \alpha^*_{1,g}}}(Y - X^\top \beta^*_{1,g})
    - \frac{\Gamma\p{1-A}}{g\p{X^\top \alpha^*_{0,g}}}(Y - X^\top \beta^*_{0,g})\right).
\]
\end{theorem}

Theorem~\ref{theorem causal generalizability Asymptotics body} establishes robust inference for \(\tau_g\) under the condition that at least one nuisance model is correctly specified for each treatment arm \(a \in \{0,1\}\). When all nuisance models are correctly specified, the asymptotic variance aligns with that of existing studies \citep{kallus2025role,zhang2023double}, yet our approach does not require full access to individualized covariates from the unlabeled data.

Notably, our framework is well-suited for settings where the labeled data arise from randomized trials (individuals with \(\Gamma_i = 1\)), while unlabeled non-participants (\(\Gamma_i = 0\)) provide only summary statistics. This approach also extends to cases where the labeled data come from observational studies. In such instances, ensuring the ``no unmeasured confounding'' condition (Assumption~\ref{assumption causal}(b)) often requires the collection of a large number of covariates, resulting in a high-dimensional problem.

\subsection{Reaching transportability using unlabeled summary statistics}\label{sec: 4.B}

We now shift our focus to the transport of findings from the labeled data to the unlabeled population to estimate the ATE parameter \(\tau_t = \E\{Y(1) - Y(0) \mid \Gamma = 0\}\), commonly known as the transportability problem.

As discussed in Section~\ref{sec: causal}, we treat \(\Gamma_{a} = \Gamma \mathbbm{1}_{(A = a)}\) as the ``effective labeling indicator'' and exclude the samples with \(\Omega_{a,i} = 0\), where \(\Omega_{a,i} = 1 - \Gamma_i + \Gamma_{a,i}\). Below, we propose estimates for \(\tau_{a,t}\) and \(\tau_t=\tau_{1,t}-\tau_{0,t}\) based on the procedure outlined in Section~\ref{sec: mar method}.

Step 1: Divide the index set $[n]$ into $K$ disjoint subsets $\mathcal{I}_1, \dots, \mathcal{I}_K$ with equal sizes such that $n_k := \abs{\mathcal{I}_k} = n/K \text{ for } k \in [K]$. Let $\widehat\gamma_k = n_k^{-1}\sum_{i\in \mathcal{I}_k} \Gamma_i$, $\mathcal{I}_{-k} = [n] \setminus \mathcal{I}_k$, and separate $\mathcal{I}_{-k}$ into two disjoint subsets $\mathcal{I}_{-k,\alpha}, \mathcal{I}_{-k,\beta}$ with the same size $M = \abs{\mathcal{I}_{-k,\alpha}} = \abs{\mathcal{I}_{-k,\beta}}$.

Step 2: For each $a\in \{0,1\}$ and $k \in [K]$, construct the propensity score estimator as
\[\widehat \alpha^{(-k)}_{a,t} = \argmin_{\alpha \in \R^d} \left[ M^{-1}\sum_{i \in \mathcal{I}_{-k,\alpha}} \br{\p{1-\Gamma_i}\bar{X}_0^\top \alpha + \Gamma_{a,i} \exp\p{-X_i^\top \alpha}} + \lambda_\alpha \norm{\alpha}_1 \right].\]

Step 3: For each $a\in \{0,1\}$ and $k \in [K]$, construct the outcome regression estimator as
\[\widehat \beta^{(-k)}_{a,t} = \argmin_{\beta \in \R^d} \left\{ M^{-1}\sum_{i \in \mathcal{I}_{-k,\beta}} \Gamma_{a,i} \exp\p{-X_i^\top \widehat \alpha_{a,t}^{(-k)}}\p{Y_i - X_i^\top \beta}^2 + \lambda_\beta \norm{\beta}_1 \right\}.\]

Step 4: For each $a \in \{0,1\}$, define the cross-fitted estimator of $\tau_{a,t}$ as
\[\widehat \tau_{a,t} = K^{-1}\sum_{k=1}^K\bar{X}_0^\top \widehat \beta^{(-k)}_{a,t}  + n^{-1}\sum_{k=1}^K\sum_{i\in \mathcal{I}_k} \frac{\Gamma_{a,i}\exp\p{-X_i^\top \widehat \alpha^{(-k)}_{a,t}}}{1-\widehat \gamma_k}\p{Y_i - X_i^\top\widehat \beta^{(-k)}_{a,t}}.\]

Step 5: We propose the estimator for the ATE over the unlabeled population:
\[\widehat\tau_{t}=\widehat\tau_{1,t}-\widehat\tau_{0,t}.\]
When we have access to the matrix $\bar\Xi_0$, the corresponding asymptotic variance can be estimated as
\begin{align*}
    \widehat \Sigma_t^2 &= K^{-1}\sum_{k=1}^K \frac{\p{\widehat \beta_{1,t}^{(-k)} -  \widehat \beta_{0,t}^{(-k)}}^\top\bar \Xi_0\p{\widehat \beta_{1,t}^{(-k)} -  \widehat \beta_{0,t}^{(-k)}}-2 \widehat \tau_t \bar X_0^\top\p{\widehat \beta_{1,t}^{(-k)}-\widehat \beta_{0,t}^{(-k)}} + \widehat \tau_t^2}{1-\widehat \gamma_k}
    \\  
    &\quad+ n^{-1}\sum_{k=1}^{K}\sum_{i\in \mathcal{I}_k} \Gamma_iA_i\p{1-\widehat \gamma_k}^{-2}\exp\p{-2X_i^\top \widehat \alpha^{(-k)}_{1,t}}\p{Y_i - X_i^\top \widehat \beta_{1,t}^{(-k)}}^2 \\
    &\quad + n^{-1}\sum_{k=1}^{K}\sum_{i\in \mathcal{I}_k} \Gamma_i\p{1-A_i}\p{1-\widehat \gamma_k}^{-2}\exp\p{-2X_i^\top \widehat \alpha^{(-k)}_{0,t}}\p{Y_i - X_i^\top \widehat \beta_{0,t}^{(-k)}}^2.
\end{align*}

For each $a \in \{0,1\}$,  we define the target nuisance parameters as
\begin{align*}
    \alpha^*_{a,t} &= \argmin_{\alpha \in \R^d} \E\{\p{1-\Gamma}X^\top \alpha + \Gamma_{a} \exp\p{-X^\top \alpha} \},\\
    \beta^*_{a,t} &= \argmin_{\beta \in \R^d} \E\{\Gamma_{a} \exp\p{-X^\top  \alpha^*_{a,t}}\p{Y - X^\top \beta}^2\}.
\end{align*}
Let $w_{a,t} = Y\p{a} - X^\top \beta^*_{a,t}$ be the corresponding residuals and \(\Omega_a = 1-\Gamma+\Gamma_a\) be an independent copy of $\Omega_{a,i}$. The following theorem demonstrates the properties of the proposed estimator $\widehat\tau_{t}$ of $\tau_t$.

\begin{theorem}\label{theorem causal transportability Asymptotics body}
Let Assumption~\ref{assumption causal} hold, and suppose that, for each \(a\in\{0,1\}\), Assumption~\ref{assumption mar nuisance (a)} holds with \((\alpha^*_{PS}, \beta^*_{OR}, w_{OR})\) replaced by \((\alpha^*_{a,t}, \beta^*_{a,t}, w_{a,t})\). Assume that \(n\gamma_n \gg  (\log n)^2\log d\),
\[
    \|\alpha^*_{a,t}\|_0 \|\beta^*_{a,t}\|_0
    = o((n\gamma_n)^{1/2}/\{\log n(\log d)^2\}),
\]
and that either of the following conditions holds: (1) (correct OR model) \(\E\{Y(a) \mid X\} = X^\top \beta^*_{a,t}\), or (2) (correct PS model) \(\P(\Gamma_{a}=1 \mid X, \Omega_{a}=1) = g\p{X^\top \alpha^*_{a,t}}\) and \(\|\alpha^*_{a,t}\|_0 = o((n\gamma_n)^{1/2}/ \log d)\). Then as \(n, d \rightarrow \infty\),
\[
    \widehat \tau_{t} -\tau_{t} = O_p \p{(n\gamma_n)^{-1/2}}, \qquad
    \widehat \Sigma^2_{t} = \Sigma_t^2\{1 + o_p(1)\},
\]
and \(\widehat\Sigma^{-1}_{t}n^{1/2}(\widehat \tau_{t} - \tau_{t})\xrightarrow{d} \mathcal{N}(0,1)\), where
\[
    \begin{aligned}
    \Sigma_t^2 = \Var\bigg(&
    \frac{1-\Gamma}{1-\gamma_n}(X^\top \beta^*_{1,t} - X^\top \beta^*_{0,t} - \tau_{t})+ \frac{\Gamma A\exp\p{-X^\top \alpha^*_{1,t}}}{1-\gamma_n}(Y - X^\top \beta^*_{1,t})\\
    &\qquad- \frac{\Gamma \p{1-A}\exp\p{-X^\top \alpha^*_{0,t}}}{1-\gamma_n} (Y - X^\top \beta^*_{0,t})\bigg).
    \end{aligned}
\]
\end{theorem}

\section{Identifiability of target nuisance parameters}\label{sec: proof identifiability}

In this section, we discuss the identifiability of the target nuisance parameters.

\begin{lemma}\label{lemma least squared residual properties}
    Let Assumption~\ref{assumption mcar lasso (a)} hold. It holds that $\beta^* = \E[XX^\top]^{-1}\E[X^\top Y]$ and for $w = Y - X^\top \beta^*$,
    \begin{align}
        \E[w]=0 \quad \text{and} \quad \E[Xw] = \boldsymbol{0}.
    \end{align}
\end{lemma}
Lemma~\ref{lemma least squared residual properties} follows from Lemma 1 of \cite{zhang2019semi}.

\begin{lemma}\label{lemma oracle plm properties}
    Let Assumption~\ref{assumption mcar plm (a)} hold. For $\widetilde X =  X - \E\sbr{X \mid Z}$ and $\widetilde Y = Y - \E\sbr{Y \mid Z}$, it holds that 
    \begin{align}
    \beta^*_{plm} = \E\sbr{\widetilde X\widetilde X^\top}^{-1}\E\sbr{ X\widetilde Y} \text{ and } f^*\p{Z} = \E\sbr{Y - X^\top\beta^*_{plm} \mid Z}.\label{solution to plm oracle problem}
\end{align}
In addition, for $\epsilon = Y - X^\top \beta^*_{plm} - f^*(Z)$, it holds that
    \begin{align}
        \E\sbr{\epsilon \mid Z} = 0 \quad \text{and}\quad \E\sbr{ X\epsilon} = \boldsymbol{0}.
    \end{align}
\end{lemma}
\begin{proof}
Since $\p{\beta^*_{plm}, f^*} = \argmin_{\p{\beta, f} \in \R^d \times \mathcal{F}} \E\sbr{\p{Y - X^\top\beta-f\p{Z}}^2}$,
\begin{align*}
    f^* = \argmin_{f \in \R^d \times \mathcal{F}}\E\sbr{\p{Y - X^\top\beta^*_{plm}-f\p{Z}}^2}.
\end{align*}
By bias-variance decomposition, we have $f^*(Z) = \E\sbr{Y - X^\top\beta^*_{plm} \mid Z}$. On the other hand,
\begin{align*}
    \beta^*_{plm} = \argmin_{\beta \in \R^d} \E\sbr{\p{Y - X^\top\beta-f^*\p{Z}}^2}.
\end{align*}
By first order optimality, we have $\beta^*_{plm} = \E[\widetilde X\widetilde X^\top]^{-1}\E[X\widetilde Y]$, where $\widetilde X =  X - \E\sbr{X \mid Z}$ and $\widetilde Y = Y - \E\sbr{Y \mid Z}$.
Then it follows that
\begin{align*}
    \E\sbr{\epsilon \mid Z} &= \E\sbr{Y - X^\top \beta^*_{plm} - f^*\p{Z} \mid Z}=\E\sbr{Y - X^\top \beta^*_{plm} \mid Z} - f^*\p{Z} = 0,
\end{align*}
and
    \begin{align*}
        \E\sbr{X\epsilon} &= \E\sbr{\p{\widetilde X + \E\sbr{X \mid Z}}\p{\widetilde Y + \E\sbr{Y \mid Z} - \p{\widetilde X + \E\sbr{X \mid Z}}^\top \beta^*_{plm} - f^*\p{Z}}}\\
        &=\E\sbr{\p{\widetilde X + \E\sbr{X \mid Z}}\p{\widetilde Y -\widetilde X^\top \beta^*_{plm}}}\\
        &=\E\sbr{\widetilde X\p{\widetilde Y -\widetilde X^\top \beta^*_{plm}}} + \E\sbr{\E\sbr{X \mid Z}\p{\widetilde Y -\widetilde X^\top \beta^*_{plm}}}\\
        &=\E\sbr{\widetilde X\widetilde Y} - \E\sbr{X\widetilde Y} + \E\sbr{\E\sbr{X \mid Z}\E\sbr{\widetilde Y -\widetilde X^\top \beta^*_{plm} \mid Z}}\\
        &=-\E\sbr{\E\sbr{X \mid Z}\E\sbr{\widetilde Y \mid Z}} + \E\sbr{\E\sbr{X \mid Z}\E\sbr{\widetilde Y -\widetilde X^\top \beta^*_{plm} \mid Z}}=\boldsymbol{0},
    \end{align*}
    where the last equation comes from the fact that $\E[\widetilde X \mid Z] = \boldsymbol{0}$ and $ \E[\widetilde Y \mid Z] = 0$.
\end{proof}

\begin{lemma}\label{lemma oracle mar properties}
    Suppose that Assumption~\ref{assumption mar nuisance (a)} holds. Then
    \begin{align}
        \E\sbr{\Gamma\exp\p{-X^\top \alpha^*_{PS}}XX^\top}  \succcurlyeq k_0c_0\kappa_l\mathbf{I}_d,
    \end{align}
    which implies the uniqueness of \(\alpha_{PS}^*\) and \(\beta_{OR}^*\).
\end{lemma}
\begin{proof}
By Assumption~\ref{assumption mar nuisance (a)},
\begin{align*}
    \exp\p{-X^\top \alpha^*_{PS}} = \frac{1-g(X^\top \alpha^*_{PS})}{g(X^\top \alpha^*_{PS})} \geq k_0\frac{1-\gamma_n}{\gamma_n} \geq k_0c_0\gamma_n^{-1}.
\end{align*}
Thus,
    \begin{align*}
        \E\sbr{\Gamma\exp\p{-X^\top \alpha^*_{PS}}XX^\top} \succcurlyeq k_0c_0\gamma_n^{-1}\E\sbr{\Gamma XX^\top} = k_0c_0\E\sbr{XX^\top \mid \Gamma=1} \succcurlyeq k_0c_0\kappa_l\mathbf{I}_d.
    \end{align*}
Thus, strong convexity guarantees the uniqueness of \(\alpha_{PS}^*\) and \(\beta_{OR}^*\).
\end{proof}

\section{Consistency of nuisance estimators} \label{sec: appendix B}
In this section, we discuss the consistency of nuisance estimators in Sections~\ref{sec: mcar lasso method} and~\ref{sec: mar method}. Let the sparsity levels of $\beta^*$, $\alpha^*_{PS}$, and $\beta^*_{OR}$ be $s$, $s_\alpha$, and $s_\beta$, respectively. Propositions~\ref{proposition lasso consistency body} and~\ref{proposition mar nuisance consistency body} below characterize the convergence rates of the lasso estimator $\widehat \beta^{(-k)}$ and nuisance estimators $\widehat \alpha^{(-k)}_{PS}$ and $\widehat \beta^{(-k)}_{OR}$. The following analysis allows settings where $\gamma_n=\P(\Gamma=1)\to0$ as $n$ increases, and therefore differs from the standard lasso literature.

The proposition below demonstrates the convergence rate of the lasso estimator under missing completely at random (MCAR).
\begin{proposition}\label{proposition lasso consistency body}
   Let Assumptions~\ref{assumption mcar} and~\ref{assumption mcar lasso (a)} hold.
    Choose $\lambda_n \asymp (\log d/(n\gamma_n))^{1/2}$. If $$n\gamma_n \gg \max\br{s, \p{\log n}^2}\log d,$$ then as \(n, d \rightarrow \infty\),
    \[
        \norm{\widehat \beta^{(-k)} - \beta^*}_2
        = O_p\p{(s\log d/(n\gamma_n))^{1/2}}.
    \]
\end{proposition}

To obtain the product sparsity condition under MAR, we define the oracle outcome estimator as
\begin{align*}
    \widetilde \beta_{OR}^{(-k)} = \argmin_{\beta \in \R^d} \left\{ M^{-1}\sum_{i \in \mathcal{I}_{-k,\beta}} \Gamma_i \exp\p{-X_i^\top \alpha^*_{PS}}\p{Y_i - X_i^\top \beta}^2 + \lambda_\beta \norm{\beta}_1 \right\}.
\end{align*}
The following proposition gives the convergence rate for nuisance estimators under MAR. 
\begin{proposition}\label{proposition mar nuisance consistency body}
    Let Assumption~\ref{assumption mar nuisance (a)} hold. Then as $n, d \rightarrow \infty$, it holds that

    (a) Choose $\lambda_\alpha \asymp (\log d/(n\gamma_n))^{1/2}$. If $n\gamma_n \gg \max\br{s_\alpha, \log n} \log d$, then
    \begin{align*}
        \norm{\widehat \alpha_{PS}^{(-k)} - \alpha^*_{PS}}_1
        &= O_p\p{s_\alpha(\log d/(n\gamma_n))^{1/2}},\\
        \norm{\widehat \alpha_{PS}^{(-k)} - \alpha^*_{PS}}_2
        &= O_p\p{(s_\alpha \log d/(n\gamma_n))^{1/2}}.
    \end{align*}

    (b) Choose $\lambda_\beta \asymp (\log d/(n\gamma_n))^{1/2}$. If $n\gamma_n \gg \max\br{s_\beta, (\log n)^2} \log d$, then
    \begin{align*}
        \norm{\widetilde\beta_{OR}^{(-k)} - \beta^*_{OR}}_1
        &= O_p\p{s_\beta (\log d/(n\gamma_n))^{1/2}},\\
        \norm{\widetilde\beta_{OR}^{(-k)} - \beta^*_{OR}}_2
        &= O_p\p{(s_\beta \log d/(n\gamma_n))^{1/2}}.
    \end{align*}
    
    (c) Choose $\lambda_\alpha \asymp \lambda_\beta \asymp (\log d/(n\gamma_n))^{1/2}$. If $n\gamma_n \gg \max\br{s_\alpha, s_\beta\log n , (\log n)^2} \log d$ and $s_\alpha s_\beta \ll {(n\gamma_n)^{3/2}}/{\log n(\log d)^2}$, then 
    \begin{align*}
    \norm{\widehat \beta_{OR}^{(-k)} - \widetilde \beta_{OR}^{(-k)}}_1 = O_p\p{(s_\alpha s_\beta \log d/(n\gamma_n))^{1/2}} \text{ and }\norm{\widehat \beta_{OR}^{(-k)} - \widetilde \beta_{OR}^{(-k)}}_2 = O_p\p{ (s_\alpha \log d/(n\gamma_n))^{1/2}}.
    \end{align*}
\end{proposition}

\section{Proof of propositions in Section~\ref{sec: appendix B}}

\subsection{Auxiliary lemmas for Proposition~\ref{proposition lasso consistency body}}
\begin{lemma}\label{sub-gaussian properties}
The following are some useful properties regarding the $\psi_\alpha$-norms.

(a) Let $X,Y \in \mathbb{R}$ be random variables. If $|X| \leq|Y|$ a.s., then $\|X\|_{\psi_2} \leq\|Y\|_{\psi_2}$. If $|X| \leq C$ a.s. for some constant $C>0$, then $\|X\|_{\psi_2} \leq\{\log (2)\}^{-1 / 2} C$.

(b) Let $X \in \mathbb{R}$ be a random variable. If $\|X\|_{\psi_2} \leq \sigma$, then $\mathbb{P}(|X|>t) \leq 2 \exp \left(-t^2 / \sigma^2\right)$ for all $t \geq 0$.

(c) Let $X \in \mathbb{R}$ be a random variable. If $\|X\|_{\psi_\alpha} \leq \sigma$ for some $(\alpha, \sigma)>0$, then $\mathbb{E}\left(|X|^m\right) \leq C_\alpha^m \sigma^m m^{m / \alpha}$ for all $m \geq 1$, for some constant $C_\alpha$ depending only on $\alpha$. In particular, if $\|X\|_{\psi_2} \leq \sigma, \mathbb{E}\left(|X|^m\right) \leq$ $2 \sigma^m \Gamma(m / 2+1)$, for all $m \geq 1$, where $\Gamma(a):=\int_0^{\infty} x^{a-1} \exp (-x) d x$ denotes the Gamma function. Hence, $\mathbb{E}(|X|) \leq \sigma \pi^{1/2}$ and $\mathbb{E}\left(|X|^m\right) \leq 2 \sigma^m(m / 2)^{m / 2}$ for $m \geq 2$.

(d) Let $X,Y \in \mathbb{R}$ be random variables. For any $\alpha, \beta>0$, let $\gamma:=\left(\alpha^{-1}+\beta^{-1}\right)^{-1}$. Then, for any $X, Y$ with $\|X\|_{\psi_\alpha}<\infty$ and $\|Y\|_{\psi_\beta}<\infty,\|X Y\|_{\psi_\gamma}<\infty$ and $\|X Y\|_{\psi_\gamma}<\|X\|_{\psi_\alpha}\|Y\|_{\psi_\beta}$.

(e) Let $X \in \mathbb{R}^d$ be a random vector with $\sup _{1 \leq j \leq d}\|X[j]\|_{\psi_\alpha} \leq \sigma$. Then, $\|\| X\left\|_{\infty}\right\|_{\psi_\alpha} \leq$ $\sigma\{\log (d)+2\}^{1 / \alpha}$.

(f) Let $X \in \mathbb{R}$ be a random variable. If $\norm{X}_{\psi_2}\leq \sigma$, then $\E\sbr{\exp\p{tX}} \leq \exp\p{C\sigma^2 t^2}$ for some constant $C$ and any $t \in \R$.
\end{lemma}
Lemma~\ref{sub-gaussian properties}(a)-(e) follow from Lemma D.1 of \cite{chakrabortty2019high}, and Lemma~\ref{sub-gaussian properties}(f) follows from Proposition 2.5.2 of \cite{Vershynin_2018}.

\begin{lemma}\label{lemma convergence of conditional random variable} (Lemma D.1 of \cite{zhang2023double})
    Let $\left(X_N\right)_{N \geq 1}$ and $\left(Y_N\right)_{N \geq 1}$ be sequences of random variables in $\mathbb{R}$. If 
    $\mathbb{E}\left(\left|X_N\right|^r \mid Y_N\right)=O_p(1)$ for some $r \geq 1,$ then $X_N=O_p(1)$.
\end{lemma}

\begin{lemma} \label{lemma concentrate gamma} (Lemma F.2 of \cite{zhang2023semi})
   Let $\Gamma_i \overset{i.i.d.}{\sim} Bernoulli(\gamma_n)$, then (a) for any $t>0$, $\mathbb{P}\left(\abs{n^{-1}\sum_{i=1}^n \Gamma_i - \gamma_n} \leq 2(t\gamma_n/n)^{1/2} + {t}/{n} \right) \geq 1-2 e^{-t}$; 
(b) for any $t>0$ such that $t<0.01n\gamma_n$,
$
\mathbb{P}\left(0.79 \gamma_n \leq n^{-1}\sum_{i=1}^n \Gamma_i \leq 1.21 \gamma_n\right)\geq1-2 e^{-t}.
$
In addition, (a) and (b) hold for $1-\Gamma_i$ with $\gamma_n$ replaced by $1-\gamma_n$.
\end{lemma}

\begin{lemma}\label{Lemma Tail Bounds for Maximums}(Theorem 3.4 of \cite{kuchibhotla2022moving})
    Suppose $X_1, \ldots, X_n$ are independent mean zero random vectors in $\mathbb{R}^d$ such that for some $\alpha>0$ and $K_{n, d}>0$,
$$
\max _{1 \leq i \leq n} \max _{1 \leq j \leq d}\left\|X_i(j)\right\|_{\psi_\alpha} \leq K_{n, d}\;\; \text { and } \;\;\Gamma_{n, d}:=\max _{1 \leq j \leq d} \frac{1}{n} \sum_{i=1}^n \mathbb{E}\left[X_i^2(j)\right].
$$
Then for any $t \geq 0$, with probability at least $1-3 e^{-t}$,
$$
\left\|\frac{1}{n} \sum_{i=1}^n X_i\right\|_{\infty} \leq 7 (\Gamma_{n, d}(t+\log d)/n)^{1/2}+\frac{C_\alpha K_{n, d}(\log (2 n))^{1 / \alpha}(t+\log d)^{1 / \alpha^*}}{n},
$$
where $\alpha^*:=\min \{\alpha, 1\}$ and $C_\alpha>0$ is some constant depending only on $\alpha$.
\end{lemma}

\begin{lemma} \label{concentrate lem RE condition MCAR}
    Let Assumptions~\ref{assumption mcar} and \ref{assumption mcar lasso (a)} hold. Then, there exist some constants $\kappa_1, \kappa_2, C_1, C_2, c_1, c_2 > 0$ such that for any $\Delta \in \R^d \setminus \br{\boldsymbol{0}}$, if $n \gamma_n > \max\br{C_1, C_2\log n \log d}$,
$$
\P\p{\frac{n^{-1}\sum_{i=1}^n\Gamma_i (X_i^\top \Delta)^2}{n^{-1}\sum_{i=1}^n\Gamma_i} \geq   \kappa_1 \norm{\Delta}_2^2 - \kappa_2 \frac{\log d}{n\gamma_n}\norm{\Delta}_1^2 } \geq 1 - c_1e^{-c_2n\gamma_n}.
$$

\end{lemma}

\begin{proof}
    By Lemma F.1 of \cite{zhang2023semi}, for $a=1, v=0, \phi(\cdot)\equiv 1$, there exists some constant $\kappa_0, \kappa_1', \kappa_2', C_1, c_1', c_2' > 0$ such that when $n\gamma_n > C_1\max\br{1, \log n \log d}$, for any $\Delta \in \R^d$ with $\norm{\Delta}_2 \leq \kappa_0$, 
    \begin{align*}
n^{-1}\sum_{i=1}^n\Gamma_i (X_i^\top \Delta)^2 \geq \gamma_n \Bigl(\kappa_1' \norm{\Delta}_2^2 - \kappa_2' \frac{\log d}{n\gamma_n}\norm{\Delta}_1^2\Bigr),
    \end{align*}
with probability at least $1- c_1'e^{-c_2'n\gamma_n}$. 

Let $c_2 = \min\br{0.01, c_2'}$. Together with Lemma~\ref{lemma concentrate gamma}, it holds that
$$
\frac{n^{-1}\sum_{i=1}^n\Gamma_i (X_i^\top \Delta)^2}{n^{-1}\sum_{i=1}^n\Gamma_i} \geq  \frac{\kappa_1'}{1.21}\norm{\Delta}_2^2 - \frac{\kappa_2' }{0.79}\frac{\log d}{n\gamma_n}\norm{\Delta}_1^2,
$$
with probability at least $1- (c_1'+2)e^{-c_2n\gamma_n}$.

Since for any $\Delta \in \R^d \setminus \br{\boldsymbol{0}}$, we have $\norm{{\kappa_0 \Delta}/{\norm{\Delta}_2}}_2 \leq \kappa_0$. Let $\kappa_1 = \kappa_1'/1.21$, $\kappa_2 = \kappa_2'/0.79$, and $c_1 = c_1'+2$. Thus, for any $\Delta \in \R^d \setminus \br{\boldsymbol{0}}$, 
$$
\frac{n^{-1}\sum_{i=1}^n\Gamma_i (X_i^\top \Delta)^2}{n^{-1}\sum_{i=1}^n\Gamma_i} \geq   \kappa_1 \norm{\Delta}_2^2 - \kappa_2 \frac{\log d}{n\gamma_n}\norm{\Delta}_1^2,
$$
with probability at least $1 - c_1e^{-c_2n\gamma_n}$.
\end{proof}

\begin{lemma} \label{concentrate lem infty norm bound MCAR}
Let Assumptions~\ref{assumption mcar} and \ref{assumption mcar lasso (a)} hold. If $d>1$ and $n\gamma_n > \p{\log n}^2 \log d $, then for any $0 < t < {n\gamma_n}/{100+\p{\log n}^2 \log d}$, there exists some constants $\kappa_3>0$ such that
\begin{align*}
    \P\p{\norm{\frac{n^{-1}\sum_{i=1}^n \Gamma_i w_i X_i}{n^{-1}\sum_{i=1}^n \Gamma_i}}_\infty > \kappa_3 ((t + \log d)/(n\gamma_n))^{1/2}} \leq 5e^{-t}.
\end{align*}
\end{lemma}

\begin{proof} 
By Lemma~\ref{lemma least squared residual properties}, we have $\E[X_iw_i]=\boldsymbol{0} \in \R^d$. 
Then, under Assumption~\ref{assumption mcar}, we have $\E[\Gamma_i X_i w_i]=\boldsymbol{0}$. Under Assumption~\ref{assumption mcar lasso (a)}, by Lemma~\ref{sub-gaussian properties}(a) and (d), we have 
$$
\sup_{j = 1,\dots, d}\norm{\Gamma_i w_i X_i[j] }_{\psi_1} \leq \sup_{j = 1,\dots, d}\norm{w_i  }_{\psi_2}\norm{X_i[j]}_{\psi_2} \leq  \sigma\sigma_w.
$$
Then by Lemma~\ref{sub-gaussian properties}(c), there exists some constant $C_1$ such that
$$
\sup_{j = 1,\dots, d}\E[(\Gamma_i w_i X_i[j])^2]   \leq C_1 \sigma^2\sigma_w^2\gamma_n.
$$
Thus, by Lemma~\ref{Lemma Tail Bounds for Maximums}, for some constants $C_2 > 0$, 
\begin{align*}
    \P\p{\norm{n^{-1}\sum_{i=1}^n \Gamma_i w_i X_i}_\infty > 7 C_1 \sigma\sigma_w(\gamma_n(t + \log d)/n)^{1/2} + C_2 \sigma\sigma_w \frac{(t+\log d)\log n}{n}} \leq 3e^{-t}.
\end{align*}
Together with Lemma~\ref{lemma concentrate gamma}, for $0 < t < {n\gamma_n}/{100 + \p{\log n}^2\log d }$, we have
\begin{align*}
    \P\p{\norm{\frac{n^{-1}\sum_{i=1}^n \Gamma_i w_i X_i}{n^{-1}\sum_{i=1}^n \Gamma_i}}_\infty > \kappa_1 ((t + \log d)/(n\gamma_n))^{1/2} + \kappa_2 \frac{(t+\log d){\log n}}{n\gamma_n}} \leq 5e^{-t},
\end{align*}
where $\kappa_1 = 7 C_1 \sigma\sigma_w/0.79$ and $\kappa_2 = C_2 \sigma\sigma_w/0.79$. 

In addition, if $n\gamma_n > \p{\log n}^2 \log d $, we have
$$
\frac{(t+\log d){\log n}}{n\gamma_n} < ((t + \log d)/(n\gamma_n))^{1/2} (\frac{\p{\log n}^2}{100 + \p{\log n}^2\log d} + \frac{\p{\log n}^2 \log d}{n\gamma_n})^{1/2}.
$$
Let $\kappa_3 = \kappa_1+ \kappa_2 (\p{\log 2}^{-1}+1)^{1/2}$, we have
\begin{align*}
    \P\p{\norm{\frac{n^{-1}\sum_{i=1}^n \Gamma_i w_i X_i}{n^{-1}\sum_{i=1}^n \Gamma_i}}_\infty > \kappa_3 ((t + \log d)/(n\gamma_n))^{1/2}} \leq 5e^{-t}.
\end{align*}

\end{proof}

\subsection{Proof of Proposition~\ref{proposition lasso consistency body}}
\begin{proof}
    Define the Lagrangian function
    $$
    \mathcal{L}(\beta; \mathcal{I}_{-k},\mathcal{D}_n) = \frac{\sum_{i \in \mathcal{I}_{-k}} \Gamma_i\p{Y_i - X_i^\top\beta}^2}{\sum_{i \in \mathcal{I}_{-k}} \Gamma_i }   + \lambda_n \norm{\beta}_1.
    $$
    By definition, we have $\mathcal{L}(\widehat \beta^{(-k)}; \mathcal{I}_{-k},\mathcal{D}_n) \leq \mathcal{L}(\beta^*; \mathcal{I}_{-k},\mathcal{D}_n)$, i.e.,
    $$
    \frac{\sum_{i \in \mathcal{I}_{-k}} \Gamma_i\p{Y_i - X_i^\top\widehat \beta^{(-k)}}^2}{\sum_{i \in \mathcal{I}_{-k}} \Gamma_i }   + \lambda_n \norm{\widehat \beta^{(-k)}}_1 \leq \frac{\sum_{i \in \mathcal{I}_{-k}} \Gamma_i\p{Y_i - X_i^\top\beta^*}^2}{\sum_{i \in \mathcal{I}_{-k}} \Gamma_i }   + \lambda_n \norm{\beta^*}_1.
    $$
    Rearranging it and let $\widehat \Delta^{(-k)} = \widehat \beta^{(-k)} - \beta^*$, we have
    \begin{align*}
         \frac{\sum_{i \in \mathcal{I}_{-k}} \Gamma_i\p{X_i^\top\widehat \Delta^{(-k)}}^2}{\sum_{i \in \mathcal{I}_{-k}} \Gamma_i } \leq  \frac{2\sum_{i \in \mathcal{I}_{-k}} \Gamma_iw_iX_i^\top \widehat \Delta^{(-k)}}{\sum_{i \in \mathcal{I}_{-k}} \Gamma_i } + \lambda_n \p{\norm{\beta^*}_1 - \norm{\widehat \beta^{(-k)}}_1}. 
    \end{align*}
Note that 
\begin{align}
    &\norm{\beta^*}_1 - \norm{\widehat \beta^{(-k)}}_1 = \norm{\beta^*}_1 - \norm{\widehat \Delta^{(-k)} + \beta^*}_1\notag \\
    &\qquad=\norm{\beta^*}_1 - \norm{{\widehat \Delta^{(-k)}}_S + \beta^*}_1 - \norm{{\widehat \Delta^{(-k)}}_{S^c}}_1\leq \norm{{\widehat \Delta^{(-k)}}_{S}}_1 - \norm{{\widehat \Delta^{(-k)}}_{S^c}}_1. \label{sparsity technique}
\end{align}
Thus, 
\begin{align}\label{lasso proof ineq}
         \frac{\sum_{i \in \mathcal{I}_{-k}} \Gamma_i\p{X_i^\top\widehat \Delta^{(-k)}}^2}{\sum_{i \in \mathcal{I}_{-k}} \Gamma_i } \leq&  2\norm{\frac{\sum_{i \in \mathcal{I}_{-k}} \Gamma_iw_iX_i}{\sum_{i \in \mathcal{I}_{-k}} \Gamma_i }}_\infty \norm{\widehat \Delta^{(-k)}}_1 + \lambda_n \p{\norm{{\widehat \Delta^{(-k)}}_{S}}_1 - \norm{{\widehat \Delta^{(-k)}}_{S^c}}_1}. 
\end{align}

By Lemma~\ref{concentrate lem RE condition MCAR} and Lemma~\ref{concentrate lem infty norm bound MCAR}, there exist some constants $\kappa_1, \kappa_2, \kappa_3, C_1,C_2, c_1',c_2' >0$ such that if $|\mathcal I_{-k}| \gamma_n > \max\br{C_1, C_2\log |\mathcal I_{-k}| \log d, \p{\log |\mathcal I_{-k}|}^2 \log d}$, for any $0 < t < {|\mathcal I_{-k}|\gamma_n}/{100+\p{\log |\mathcal I_{-k}|}^2 \log d}$, with probability at least $1-c_1'e^{-c_2'|\mathcal I_{-k}|\gamma_n}-5e^{-t}$, the following statements hold simultaneously:
    \begin{align*}
        \frac{|\mathcal I_{-k}|^{-1}\sum_{i \in \mathcal{I}_{-k}}{\Gamma_i (X_i^\top \widehat \Delta^{(-k)})^2 }}{|\mathcal I_{-k}|^{-1}\sum_{i \in \mathcal{I}_{-k}} \Gamma_i} &\geq \kappa_1 \norm{\widehat \Delta^{(-k)}}_2^2 - \kappa_2 \frac{\log d}{|\mathcal I_{-k}|\gamma_n}\norm{\widehat \Delta^{(-k)}}_1^2,\\
            \norm{\frac{|\mathcal I_{-k}|^{-1}\sum_{i \in \mathcal{I}_{-k}} \Gamma_i w_i X_i}{|\mathcal I_{-k}|^{-1}\sum_{i \in \mathcal{I}_{-k}} \Gamma_i}}_\infty &\leq  \kappa_3 ((t + \log d)/(|\mathcal I_{-k}|\gamma_n))^{1/2}.
    \end{align*}
Let $\lambda_n = 4\kappa_3 ((t + \log d)/(|\mathcal I_{-k}|\gamma_n))^{1/2}$. Then we have
\begin{align}
    \frac{\sum_{i \in \mathcal{I}_{-k}}{\Gamma_i (X_i^\top \widehat \Delta^{(-k)})^2 }}{\sum_{i \in \mathcal{I}_{-k}} \Gamma_i} &\geq \kappa_1\norm{\widehat \Delta^{(-k)}}_2^2 - \kappa_2 \frac{\log d}{|\mathcal I_{-k}|\gamma_n}\norm{\widehat \Delta^{(-k)}}_1^2,\label{inequality of RSC 1 MCAR}\\
    \lambda_n = 4\kappa_3 ((t + \log d)/(|\mathcal I_{-k}|\gamma_n))^{1/2} &\geq 4\norm{\frac{\sum_{i \in \mathcal{I}_{-k}} \Gamma_iw_iX_i}{\sum_{i \in \mathcal{I}_{-k}} \Gamma_i }}_\infty. \label{inequality of infty norm}
\end{align}
By \eqref{lasso proof ineq} and \eqref{inequality of infty norm},
$
    0 \leq {\lambda_n}/{2} \p{3\norm{{\widehat \Delta^{(-k)}}_{S}}_1 - \norm{{\widehat \Delta^{(-k)}}_{S^c}}_1},
$
which implies that $\norm{{\widehat \Delta^{(-k)}}_{S^c}}_1 \leq 3\norm{{\widehat \Delta^{(-k)}}_{S}}_1$. Then, 
$
    \norm{{\widehat \Delta^{(-k)}}}_1 \leq 4\norm{{\widehat \Delta^{(-k)}}_{S}}_1 \leq 4s^{1/2}\norm{{\widehat \Delta^{(-k)}}}_2.
$
By \eqref{lasso proof ineq}, \eqref{inequality of RSC 1 MCAR}, and \eqref{inequality of infty norm}, we have
\begin{align*}
    \norm{\widehat \Delta^{(-k)}}_2^2  &\leq \frac{6\kappa_3}{\kappa_1} ((t + \log d)/(|\mathcal I_{-k}|\gamma_n))^{1/2} \norm{{\widehat \Delta^{(-k)}}_{S}}_1 +  \frac{\kappa_2}{\kappa_1}\frac{\log d}{|\mathcal I_{-k}|\gamma_n}\norm{\widehat \Delta^{(-k)}}_1^2\\
    &\leq \frac{6\kappa_3}{\kappa_1} ((st + s\log d)/(|\mathcal I_{-k}|\gamma_n))^{1/2} \norm{{\widehat \Delta^{(-k)}}}_2 +  \frac{16\kappa_2}{\kappa_1}\frac{s\log d}{|\mathcal I_{-k}|\gamma_n}\norm{\widehat \Delta^{(-k)}}_2^2.
\end{align*}
Then if
\[
    |\mathcal I_{-k}| \gamma_n
    > \max\br{32(\kappa_1)^{-1}\kappa_2 s\log d, C_1,
    C_2\log |\mathcal I_{-k}| \log d, \p{\log |\mathcal I_{-k}|}^2 \log d},
\]
for any \(0 < t < |\mathcal I_{-k}|\gamma_n/\{100+\p{\log |\mathcal I_{-k}|}^2 \log d\}\), it holds that with probability at least \(1-c_1'e^{-c_2'|\mathcal I_{-k}|\gamma_n}-5e^{-t}\),
$
    \norm{\widehat \Delta^{(-k)}}_2  \leq {12\kappa_3}/{\kappa_1} ((st + s\log d)/(|\mathcal I_{-k}|\gamma_n))^{1/2}.
$
It follows that, if $n\gamma_n \gg \p{s\vee \p{\log n}^2}\log d$ and $\lambda_n \asymp (\log d/(n\gamma_n))^{1/2}$, as $n, d \rightarrow \infty$,
$
    \norm{\widehat \beta^{(-k)} - \beta^*}_2 = O_p\p{(s\log d/(n\gamma_n))^{1/2}}.
$

\end{proof}

\subsection{Auxiliary lemmas for Proposition~\ref{proposition mar nuisance consistency body}}
Define 
\begin{align*}
    &l_{1,i}\p{\alpha} := \p{1-\Gamma_i}\bar X_0^\top \alpha + \Gamma_i \exp\p{-X_i^\top\alpha},\quad
    l_{2,i}\p{\alpha, \beta} := \Gamma_i \exp\p{-X_i^\top\alpha}\p{Y_i - X_i^\top \beta}^2.
\end{align*}
We define the empirical summations of $l_{1,i}\p{\alpha}$ and $l_{2,i}\p{\alpha, \beta}$ as
\begin{align*}
    &\bar l_1^{(-k)}\p{\alpha} := M^{-1} \sum_{i \in \mathcal{I}_{-k,\alpha}} l_{1,i}\p{\alpha},\quad
    \bar l_2^{(-k)}\p{\alpha, \beta} := M^{-1} \sum_{i \in \mathcal{I}_{-k,\beta}} l_{2,i}\p{\alpha, \beta}.
\end{align*}
In addition, for any $\Delta \in \R^d$, we define
\begin{align*}
    &\delta \bar l_1^{(-k)}\p{\alpha, \Delta} := \bar l_1^{(-k)}\p{\alpha+\Delta} - \bar l_1^{(-k)}\p{\alpha} - \nabla_\alpha \bar l_1^{(-k)}\p{\alpha}^\top \Delta,\\
    &\delta \bar l_2^{(-k)}\p{\alpha, \beta, \Delta} := \bar l_2^{(-k)}\p{\alpha, \beta+\Delta} - \bar l_2^{(-k)}\p{\alpha, \beta} - \nabla_\beta \bar l_2^{(-k)}\p{\alpha, \beta}^\top \Delta.
\end{align*}

\begin{lemma}\label{lemma conditional independence}
For any index set $A \subseteq [n]$, let $\mc{M}_A := \{(1-\Gamma_j)X_j \mid j \in A\}$, $\mathcal{D}_{A} := \{\p{\Gamma_j, \Gamma_jX_j, \Gamma_jY_j} \mid j \in A\}$, and $\Gamma_{A}  := \{\Gamma_j \mid j \in A\}$. It holds that

(a) For any $A_1, A_2 \subseteq [n]$ and $B_1\subseteq [n]\setminus A_1$, $\mathcal{D}_{A_1} \perp \mc{M}_{A_2} \mid \Gamma_{A_1}$ and $\mathcal{D}_{A_1} \perp \mathcal{D}_{B_1} \mid \Gamma_{A_1}$.

(b) Let $f_1$ be any function of $\mc{M}_{[n]}$ and $\mathcal{D}_{\mathcal{I}_{-k}}$, then $f_1(\mc{M}_{[n]}, \mathcal{D}_{\mathcal{I}_{-k}})\perp \mathcal{D}_{\mathcal{I}_{k}} \mid \Gamma_{\mathcal{I}_{k}}$.

(c) Let $f_2$ be any function of $\mc{M}_{[n]}$ and $\mathcal{D}_{\mathcal{I}_{-k,\alpha}}$, then $f_2(\mc{M}_{[n]}, \mathcal{D}_{\mathcal{I}_{-k,\alpha}})\perp \mathcal{D}_{\mathcal{I}_{-k,\beta}} \mid \Gamma_{\mathcal{I}_{-k,\beta}}$.

(d) Let $f_3$ be any function of $\mc{M}_{[n]}$, $\mathcal{D}_{\mathcal{I}_{-k,\alpha}}$, and $\mathcal{D}_{\mathcal{I}_{-k,\beta}}$, then $f_3(\mc{M}_{[n]}, \mathcal{D}_{\mathcal{I}_{-k,\alpha}}, \mathcal{D}_{\mathcal{I}_{-k,\beta}})\perp \mathcal{D}_{\mathcal{I}_{k}} \mid \Gamma_{\mathcal{I}_{k}}$.
\end{lemma}
\begin{proof}
    (a) For $\mathcal{D}_{A_1} \perp \mc{M}_{A_2} \mid \Gamma_{A_1}$, it suffices to prove that for each $i\in[n]$, $(1-\Gamma_i)X_i \perp \mathcal{D}_n \mid \Gamma_{i}$. If $\Gamma_i = 1$, $(1-\Gamma_i)X_i=0$, which is conditional independent of $\mathcal{D}_n$. If $\Gamma_i=0$, then $\mathcal{D}_n = \br{\p{\Gamma_j, \Gamma_jX_j, \Gamma_jY_j} \mid j \in [n]\setminus\br{i}} \cup \{\p{ 0, 0, 0}\}$, which is conditionally independent of $(1-\Gamma_i)X_i$. Thus, $(1-\Gamma_i)X_i \perp \mathcal{D}_n \mid \Gamma_{i}$ for any $i = 1,\dots, n$. In addition, since $\Gamma_{A_1} \subset \mathcal{D}_{A_1}$ and $\mathcal{D}_{A_1} \perp \mathcal{D}_{B_1}$ when $A_1 \cap B_1 = \emptyset$, we have that $\mathcal{D}_{A_1} \perp \mathcal{D}_{B_1} \mid \Gamma_{A_1}$ holds.
    
    (b), (c), (d) follow directly from (a).
\end{proof}

\begin{lemma} \label{lemma Lemma C.5 of zhang 2021}
    Suppose that $\mathbb{S}^{\prime}=\left(\mathbf{U}_i\right)_{i \in \mathcal{J}}$ are independent and identically distributed sub-Gaussian random vectors, i.e., $\left\|\mathbf{a}^{\top} \mathbf{U}\right\|_{\psi_2} \leq \sigma_{\mathbf{U}}\|\mathbf{a}\|_2$ for all $\mathbf{a} \in \mathbb{R}^d$ with some constant $\sigma_{\mathbf{U}}>0$. Additionally, suppose the smallest eigenvalue of $\mathbb{E}\left(\mathbf{U U}^{\top}\right)$ is bounded bellow by some constant $\lambda_{\mathbf{U}}>0$. Let $M=|\mathcal{J}|$. For any continuous function $\phi: \mathbb{R} \rightarrow(0, \infty)$, $v \in[0,1]$, and $\boldsymbol{\eta} \in \mathbb{R}^d$ satisfying $\mathbb{E}\left\{\left|\mathbf{U}^{\top} \boldsymbol{\eta}\right|^c\right\}<C$ with some constants $c, C>0$, there exists constants $\kappa_1, \kappa_2, c_1, c_2>0$, such that
\begin{align*}
    \mathbb{P}_{\mathbb{S}^{\prime}}\left(M^{-1} \sum_{i \in \mathcal{J}} \phi\left(\mathbf{U}_i^{\top}(\boldsymbol{\eta}+v \boldsymbol{\Delta})\right)\left(\mathbf{U}_i^{\top} \boldsymbol{\Delta}\right)^2 \geq \kappa_1\|\boldsymbol{\Delta}\|_2^2-\kappa_2 \frac{\log d}{M}\|\boldsymbol{\Delta}\|_1^2, \;\; \forall\|\boldsymbol{\Delta}\|_2 \leq 1\right)\geq 1-c_1\exp{-c_2M}.
\end{align*}
\end{lemma}
Lemma~\ref{lemma Lemma C.5 of zhang 2021} follows from Lemma C.5 of \cite{zhang2021dynamic}.

% \begin{lemma}
%     For any random variables $X, Y \in \R$, if $\E\sbr{X \mid Y}$
% \end{lemma}
\begin{lemma}\label{lemma mar RSC condition}
Let Assumption~\ref{assumption mar nuisance (a)} hold. Then there exist some constants $\kappa_1, \kappa_2$, $c_1, c_2>0$ such that
\begin{align*}
    \P\p{\delta \bar l_1 \p{\alpha^*, \Delta} \geq \kappa_1 \norm{\Delta}_2^2 -  \kappa_2 \frac{\log d}{M\gamma_n} \norm{\Delta}_1^2,\;\; \forall \norm{\Delta }_2 \leq 1}\geq 1 - c_1 \exp\p{-c_2M\gamma_n}.
\end{align*}
\end{lemma}
\begin{proof}
    By Taylor's theorem, with some $t_1 \in \p{0,1}$, 
    \begin{align}
        \delta \bar l_1\p{\alpha^*, \Delta} &= \frac{1}{2} M^{-1}\sum_{i \in \mathcal{I}_{-k,\alpha}} \Gamma_i \exp\p{ -X_i^\top\alpha^*-t_1X_i^\top\Delta}\p{X_i^\top\Delta}^2.\label{taylor expansion of delta l1}
    \end{align}
    Let $\widehat \gamma_{-k,\alpha} = M^{-1}\sum_{i \in \mathcal{I}_{-k,\alpha}} \Gamma_i$ and
\[
        L_1 = M^{-1}\sum_{i \in \mathcal{I}_{-k,\alpha}} \Gamma_i \exp\p{ -t_1X_i^\top\Delta}\p{X_i^\top\Delta}^2.
\]
    Under Assumption~\ref{assumption mar nuisance (a)}(a), since
\[
\exp\p{-X_i^\top \alpha^*} = \{1-g(X_i^\top \alpha^*)\}/g(X_i^\top \alpha^*),
\]
we have
\begin{align}
     k_0c_0\gamma_n^{-1}\leq k_0\frac{1-\gamma_n}{\gamma_n}\leq \exp\p{-X_i^\top \alpha^*} \leq k_0^{-1}\frac{1-\gamma_n}{\gamma_n} \leq k_0^{-1}\gamma_n^{-1}, \label{exp(-X'alpha = gamma-1)}
\end{align}
which implies
\begin{align}
    \delta \bar l_1\p{\alpha^*, \Delta} \geq \frac{k_0c_0}{2}\gamma_n^{-1} L_1. \label{RSC t1}
\end{align}
     By Lemma~\ref{lemma Lemma C.5 of zhang 2021}, 
    \begin{align*}
        \P\p{\widehat \gamma_{-k,\alpha}^{-1} L_1 \geq \kappa_1 \norm{\Delta}_2^2 - \kappa_2 \frac{\log d}{M\widehat \gamma_{-k,\alpha}} \norm{\Delta}_1^2,\;\;\forall \norm{\Delta
        }_2 \leq 1 \mid \Gamma_{\mathcal{I}_{-k,\alpha}}}\geq 1-c_1\exp\p{-c_2M\widehat \gamma_{-k,\alpha}}.
    \end{align*}
Then
\begin{align*}
    \P&\p{\widehat \gamma^{-1}_{-k,\alpha}L_1 < \kappa_1 \norm{\Delta}_2^2 - \kappa_2 \frac{\log d}{M\widehat \gamma_{-k,\alpha}} \norm{\Delta}_1^2, \forall \norm{\Delta
        }_2 \leq 1 }\\
        &=\E\sbr{\P\p{\widehat \gamma^{-1}_{-k,\alpha}L_1 < \kappa_1 \norm{\Delta}_2^2 - \kappa_2 \frac{\log d}{M\widehat \gamma_{-k,\alpha}} \norm{\Delta}_1^2, \forall \norm{\Delta
        }_2 \leq 1 \mid \Gamma_{\mathcal{I}_{-k,\alpha}}}}\leq \E\sbr{c_1\exp\p{-c_2M\widehat \gamma_{-k,\alpha}}}\\
        &\overset{(i)}{=}c_1\p{1-\gamma_n + \gamma_n e^{-c_2}}^M\overset{(ii)}{\leq} c_1 \exp\br{-\p{1-e^{-c_2}}M\gamma_n},
\end{align*}
where $(i)$ uses the moment-generating function of a binomial random variable,
\[
\E\{\exp(tB)\}=(1-p+pe^t)^n,\quad B\sim B(n,p),
\]
and $(ii)$ uses the inequality $1-t \leq e^{-t}$. By Lemma~\ref{lemma concentrate gamma}, for $0<t<0.01M\gamma_n$,
\begin{align*}
    \P\p{0.79\gamma_n \leq \widehat \gamma_{-k,\alpha} \leq 1.21\gamma_n} \geq 1-2e^{-t}.
\end{align*}
Together we have
\begin{align*}
    \P&\p{L_1 < 0.79\gamma_n\br{\kappa_1 \norm{\Delta}_2^2 - \kappa_2 \frac{\log d}{0.79M\gamma_n} \norm{\Delta}_1^2}, \forall \norm{\Delta
        }_2 \leq 1 }\\
        &\leq \P\p{\widehat \gamma^{-1}_{-k,\alpha}L_1 < \kappa_1 \norm{\Delta}_2^2 - \kappa_2 \frac{\log d}{M\widehat \gamma_{-k,\alpha}} \norm{\Delta}_1^2, \forall \norm{\Delta
        }_2 \leq 1 }+\P\p{ \abs{\frac{\widehat \gamma_{-k,\alpha}}{\gamma_n}-1} > 0.21}\\
        &\leq c_1\exp\br{-\p{1-e^{-c_2}}M\gamma_n} + 2\exp\p{-0.009M\gamma_n}.
\end{align*}
Let $\kappa_1' = 2\cdot 0.79\kappa_1, \kappa_2' = 2\kappa_2, c_1' = c_1+2, c_2' = \p{1-e^{-c_2} }\wedge 0.009$, we have
\begin{align*}
    \P\p{L_1 \geq \frac{\gamma_n}{2}\br{\kappa_1' \norm{\Delta}_2^2 -  \kappa_2' \frac{\log d}{M\gamma_n} \norm{\Delta}_1^2}, \forall \norm{\Delta }_2 \leq 1} \geq 1 - c_1' \exp\p{-c_2'M\gamma_n}.
\end{align*}
With \eqref{RSC t1}, it follows that
\begin{align*}
    \P\p{\delta \bar l_1\p{\alpha^*, \Delta} \geq \frac{k_0c_0}{4}\br{\kappa_1' \norm{\Delta}_2^2 -  \kappa_2' \frac{\log d}{M\gamma_n} \norm{\Delta}_1^2}, \forall \norm{\Delta }_2 \leq 1}\geq 1 - c_1' \exp\p{-c_2'M\gamma_n}.
\end{align*}
\end{proof}

\begin{lemma}\label{lemma mar gradient infty norm}
    Let Assumption~\ref{assumption mar nuisance (a)} hold. Then there exist some constants $c_1,C_1,$ $C_2>0$ such that if $M\gamma_n > C_1\log M \log d$, for any $0<t<C_2M\gamma_n/\log M$,
\begin{align*}
    \P\p{\norm{\nabla_\alpha \bar l_1\p{\alpha^*}}_\infty \leq c((t + \log d)/(M\gamma_n))^{1/2}} \geq 1 - 13\exp\p{-t}.
\end{align*}
\end{lemma}

\begin{proof}
Define 
\begin{align*}
    l_{1,i}^*\p{\alpha} &= \p{1-\Gamma_i}X_i^\top \alpha + \Gamma_i \exp\p{-X_i^\top \alpha},\quad
    \bar l_1^*\p{\alpha} = M^{-1}\sum_{i \in \mathcal{I}_{-k,\alpha}} l_{1,i}^*\p{\alpha}.
\end{align*}
Then
$
    \nabla_\alpha \bar l_1^*\p{\alpha^*} = M^{-1} \sum_{i \in \mathcal{I}_{-k,\alpha}} \br{ 1-\Gamma_i - \Gamma_i \exp\p{-X_i^\top\alpha^*}}X_i.
$

First, we show that there exists some constant $c>0$ such that, if $M\gamma_n > \log M \log d$, then for any $0 < t < 0.01M\gamma_n/\log M$,
\[
    \P\p{\norm{\nabla_\alpha \bar l_1^*\p{\alpha^*}}_{\infty} > c((t + \log d)/(M\gamma_n))^{1/2}} \leq 3\exp\p{-t}.
\]
Since $\alpha^* = \argmin_{\alpha \in \R^d} \E[l_{1,i}^*\p{\alpha^*}]$, by first order optimality we have
\[
    \E\sbr{\br{1-\Gamma_i - \Gamma_i \exp\p{-X_i^\top\alpha^*}}X_i} = \boldsymbol{0}.
\]
Denote the k-th element of $X_i$ by $X_i^\top e_k$, then by \eqref{exp(-X'alpha = gamma-1)},
\begin{align*}
    \abs{\br{1-\Gamma_i - \Gamma_i \exp\p{-X_i^\top\alpha^*}}X_i^\top e_k} \leq \p{1+k_0^{-1}\gamma_n^{-1}}\abs{X_i^\top e_k}.
\end{align*}
By Lemma~\ref{sub-gaussian properties}(a), 
\begin{align*}
    &\max_{i \in \mathcal{I}_{-k,\alpha}}\max_{1\leq k \leq d}\norm{\br{1-\Gamma_i\p{g\p{X_i^\top \alpha^*}}^{-1}}X_i^\top e_k}_{\psi_2} \\
&\qquad\leq \max_{i \in \mathcal{I}_{-k,\alpha}}\max_{1\leq k \leq d} \p{1+k_0^{-1}\frac{1-\gamma_n}{\gamma_n}}\norm{X_i^\top e_k}_{\psi_2} \leq \p{1+k_0^{-1}\gamma_n^{-1}}\sigma.
\end{align*}
By \eqref{exp(-X'alpha = gamma-1)},
\begin{align}
    k_0^2c_0^2\gamma_n^{-2}\leq \exp\p{-2X_i^\top \alpha^*} \leq k_0^{-2}\gamma_n^{-2}.\label{exp(-2X'alpha = gamma-2)}
\end{align}
Thus, by Assumption~\ref{assumption mar nuisance (a)}(b) and Lemma~\ref{sub-gaussian properties} (c),
\begin{align*}
    &\E\sbr{{\br{1-\Gamma_i-\Gamma_i\exp\p{-X_i^\top \alpha^*}}^2\p{X_i^\top e_k}}^2}=\E\sbr{\br{1-\Gamma_i + \Gamma_i\exp\p{-2X_i^\top \alpha^*}}\p{X_i^\top e_k}^2}\\
    &\qquad\leq \E\sbr{\p{X_i^\top e_k}^2} + k_0^{-2}\gamma_n^{-2}\E\sbr{\Gamma_i\p{X_i^\top e_k}^2}=\E\sbr{\p{X_i^\top e_k}^2} + k_0^{-2}\gamma_n^{-1}\E\sbr{\p{X_i^\top e_k}^2 \mid \Gamma_i = 1}\\
    &\qquad\leq 2\p{1 + k_0^{-2}\gamma_n^{-1}}\sigma^2,
\end{align*}
which implies 
\begin{align*}
    \max_{1\leq k \leq d} M^{-1} \sum_{i \in \mathcal{I}_{-k,\alpha}} \E\sbr{\p{{\br{1-\Gamma_i\p{g\p{X_i^\top \alpha^*}}^{-1}}X_i^\top e_k}}^2} \leq 2\p{1 + k_0^{-2}\gamma_n^{-1}}\sigma^2.
\end{align*}
By Lemma~\ref{Lemma Tail Bounds for Maximums}, for some constant $C>0$, with probability at least $1-3\exp\p{-t}$,
\begin{align*}
    \norm{\nabla_\alpha \bar l_1^*\p{\alpha^*}}_{\infty} &\leq 7(\frac{2\p{1 + k_0^{-2}\gamma_n^{-1}}\sigma^2\p{t + \log d}}{M})^{1/2} + \frac{C\p{1+k_0^{-1}\gamma_n^{-1}}\sigma\p{t + \log d}(\log \p{2M})^{1/2}}{M}.
\end{align*}
Let $c_1 = 7(2\p{1+ k_0^{-2}}\sigma^2)^{1/2}$ and $c_2 = C\p{1+k_0^{-1}}\sigma(2)^{1/2}$. Then with probability at least $1-3\exp\p{-t}$,
\begin{align*}
    \norm{\nabla_\alpha \bar l_1^*\p{\alpha^*}}_{\infty} \leq c_1((t + \log d)/(M\gamma_n))^{1/2} + c_2\frac{\p{t + \log d}(\log M)^{1/2}}{M\gamma_n}.
\end{align*}
If $M\gamma_n > \log M \log d$ and $0 < t < 0.01M\gamma_n/\log M $, we have
\begin{align}
    &c_1((t + \log d)/(M\gamma_n))^{1/2} + c_2\frac{\p{t + \log d}(\log M)^{1/2}}{M\gamma_n}\\
    =& c_1((t + \log d)/(M\gamma_n))^{1/2} + c_2((t + \log d)/(M\gamma_n))^{1/2}(\frac{t\log M}{M\gamma_n} + \frac{\log d\log M}{M\gamma_n})^{1/2}\notag\\
    \leq& c_1((t + \log d)/(M\gamma_n))^{1/2} + c_2((t + \log d)/(M\gamma_n))^{1/2}(\frac{0.01M\gamma_n\log M}{M\gamma_n\log M } + 1)^{1/2}\\
    \leq& c_3((t + \log d)/(M\gamma_n))^{1/2}, \label{simplify 20}
\end{align}
where $c_3 = \max\br{c_1, c_2(1.01)^{1/2}}$. Thus, we have
\[
    \P\p{\norm{\nabla_\alpha \bar l_1^*\p{\alpha^*}}_{\infty} > c_3((t + \log d)/(M\gamma_n))^{1/2}} \leq 3\exp\p{-t}.
\]

Now we back to this lemma. Note that
\begin{align*}
        &\nabla_\alpha \bar l_1\p{\alpha^*} - \nabla_\alpha \bar l_1^*\p{\alpha^*} = M^{-1} \sum_{i \in \mathcal{I}_{-k,\alpha}} \p{1-\Gamma_i} \p{\bar X_0 - X_i}\\
        &\qquad= \p{1-\widehat \gamma_{-k,\alpha}}\p{\bar X_0 - \E\sbr{X_i \mid \Gamma_i=0}} - M^{-1} \sum_{i \in \mathcal{I}_{-k,\alpha}}\p{1-\Gamma_i}\p{X_i - \E\sbr{X_i \mid \Gamma_i=0}}= B_1 - B_2,
    \end{align*}
where
\begin{align*}
    B_1 &= \p{1-\widehat \gamma_{-k,\alpha}}\p{1-\widehat \gamma}^{-1}n^{-1}\sum_{i=1}^n\p{1-\Gamma_i}\p{X_i - \E\sbr{X_i \mid \Gamma_i=0}},\\
    B_2 &= M^{-1} \sum_{i \in \mathcal{I}_{-k,\alpha}}\p{1-\Gamma_i}\p{X_i - \E\sbr{X_i \mid \Gamma_i=0}}.
\end{align*}
For $B_2$,
\begin{align*}
    \E\sbr{(1-\Gamma_i)\p{X_i - \E\sbr{X_i \mid \Gamma_i=0}}} = \p{1-\gamma_n}\E\sbr{\p{X_i - \E\sbr{X_i \mid \Gamma_i=0}} \mid \Gamma_i=0} = \boldsymbol{0}.
\end{align*}
By Lemma~\ref{sub-gaussian properties}(a) and (c), there exists some constant $C_1$ such that
\begin{align*}
    &\max_{i \in \mathcal{I}_{-k,\alpha}}\max_{1 \leq k \leq d} \norm{(1-\Gamma_i)\p{X_i^\top e_k - \E\sbr{X_i^\top e_k \mid \Gamma_i=0}}}_{\psi_2} \\
    &\leq \max_{i \in \mathcal{I}_{-k,\alpha}}\max_{1 \leq k \leq d}\norm{X_i^\top e_k - \E\sbr{X_i^\top e_k \mid \Gamma_i=0}}_{\psi_2}\\
    &\leq \max_{i \in \mathcal{I}_{-k,\alpha}}\max_{1 \leq k \leq d}\norm{X_i^\top e_k}_{\psi_2} + \norm{\E\sbr{X_i^\top e_k \mid \Gamma_i=0}}_{\psi_2}\leq C_1\sigma.
\end{align*}
Also by Lemma~\ref{sub-gaussian properties}(c), we have
\begin{align*}
    &\E\sbr{\br{(1-\Gamma_i)\p{X_i^\top e_k - \E\sbr{X_i^\top e_k \mid \Gamma_i=0}}}^2} = \E\sbr{\p{1-\Gamma_i}\p{X_i^\top e_k - \E\sbr{X_i^\top e_k \mid \Gamma_i=0}}^2}\\
    &\qquad=\p{1-\gamma_n}\E\sbr{\p{X_i^\top e_k - \E\sbr{X_i^\top e_k \mid \Gamma_i=0}}^2 \mid \Gamma_i=0}\leq \E\sbr{\p{X_i^\top e_k}^2 \mid \Gamma_i=0}\leq 2\sigma^2,
\end{align*}
which implies
\begin{align*}
    \max_{1\leq k \leq d} M^{-1} \sum_{i \in \mathcal{I}_{-k,\alpha}}\E\sbr{\br{(1-\Gamma_i)\p{X_i^\top e_k - \E\sbr{X_i^\top e_k \mid \Gamma_i=0}}}^2} \leq 2\sigma^2.
\end{align*}
By Lemma~\ref{Lemma Tail Bounds for Maximums}, with probability at least $1-3\exp\p{-t}$,
\begin{align*}
    \norm{B_2}_{\infty} \leq 7(\frac{2\sigma^2\p{t + \log d}}{M})^{1/2} + \frac{CC_1\sigma\p{t + \log d}(\log \p{2M})^{1/2}}{M}.
\end{align*}
If $M > \log M \log d$ and $0 < t < 0.01M/\log M $, with identical analysis to \eqref{simplify 20}, there exists some constant $c_4>0$ such that
$
    \P\p{\norm{B_2}_{\infty} > c_4({t + \log d}/{M})^{1/2}} \leq 3\exp\p{-t}.
$

For $B_1$, let
$
B_{1a} = n^{-1}\sum_{i=1}^n\p{1-\Gamma_i}\p{X_i - \E\sbr{X_i \mid \Gamma_i=0}}.
$
Similar to $B_2$, if $n > \log n \log d$ and $0 < t < 0.01n/\log n $, there exists some constant $c_5 >0$ such that
$
\P\p{\norm{B_{1a}}_{\infty} > c_5({t  + \log d}/{n})^{1/2}} \leq 3\exp\p{-t}.
$

By Lemma~\ref{lemma concentrate gamma}, for any $0<t<0.01M\gamma_n$, we have
\begin{align*}
    \P\p{0.79\p{1-\gamma_n} \leq 1-\widehat \gamma_{-k,\alpha} \leq 1.21 
    \p{1-\gamma_n} } \geq 1-2\exp\p{-t},
\end{align*}
and for any $0<t<0.01n\gamma_n$, we have
\begin{align*}
    \P\p{0.79\p{1-\gamma_n} \leq 1-\widehat \gamma \leq 1.21 
    \p{1-\gamma_n}} \geq 1-2\exp\p{-t}.
\end{align*}
Define events $\mathcal{A}_1\p{t}, \mathcal{A}_2\p{t}, \mathcal{A}_3\p{t}$ as 
\begin{align*}
    \mathcal{A}_1\p{t} &:= \br{\norm{B_{1a}}_{\infty} \leq c_5(\frac{t + \log d}{n})^{1/2}},\quad
    \mathcal{A}_2\p{t} := \br{0.79\p{1-\gamma_n} \leq 1-\widehat \gamma_{-k,\alpha} \leq 1.21 
    \p{1-\gamma_n}},\\
    \mathcal{A}_3\p{t} &:= \br{0.79\p{1-\gamma_n} \leq 1-\widehat \gamma \leq 1.21 
    \p{1-\gamma_n}}.
\end{align*}
Since $M = (K-1)n/\p{2K}$, there exists some constant $C_1' >0$ and $0< C_2' \leq 0.01$ such that if $M > C_1'\log M \log d$, we have $M > \log M \log d$ and $n > \log n \log d$, which implies $C_2'M\gamma_n/\log M < 0.01n\gamma_n/\log n$.
Since
\begin{align*}
    B_1 &= \p{1-\widehat \gamma_{-k,\alpha}}\p{1-\widehat \gamma}^{-1}n^{-1}\sum_{i=1}^n\p{1-\Gamma_i}\p{X_i - \E\sbr{X_i \mid \Gamma_i=0}},
\end{align*}
If $M > C_1'\log M \log d$, for any $0<t<C_2'M\gamma_n/\log M$, on event $\mathcal{A}_1\p{t}\cap \mathcal{A}_2\p{t} \cap \mathcal{A}_3\p{t}$, we have 
$
    \norm{B_1}_\infty \leq {1.21\p{1-\gamma_n}}/{0.79\p{1-\gamma_n}}c_5({t + \log d}/{n})^{1/2}.
$
Let $c_6 = {1.21}/{0.79}c_5$. Thus, if $M > C_1'\log M \log d$, for any $0<t<C_2'M\gamma_n/\log M$,
\begin{align*}
\P\p{\norm{B_1}_\infty \leq c_6(\frac{t  + \log d}{n})^{1/2}} \geq \P\p{\mathcal{A}_1\cap \mathcal{A}_2 \cap \mathcal{A}_3}\geq 1 - \P\p{\mathcal{A}_1^c} - \P\p{\mathcal{A}_2^c} - \P\p{\mathcal{A}_3^c}\geq 1 - 7\exp\p{-t}.
\end{align*}
In addition, define events $\mathcal{B}_1\p{t}, \mathcal{B}_2\p{t},\mathcal{B}_3\p{t}$ as
\begin{align*}
    \mathcal{B}_1\p{t} &:=\br{\norm{\nabla_\alpha \bar l_1\p{\alpha^*}}_{\infty} \leq c_3(\frac{t  + \log d}{M\gamma_n})^{1/2}},\quad
\mathcal{B}_2\p{t} :=\br{\norm{B_2}_{\infty} \leq c_4(\frac{t  + \log d}{M})^{1/2}},\\
\mathcal{B}_3\p{t} &:=\br{\norm{B_1}_\infty \leq c_6(\frac{t  + \log d}{n})^{1/2}}.
\end{align*}
Note that
\begin{align*}
    \nabla_\alpha \bar l_1\p{\alpha^*} &= \nabla_\alpha \bar l_1\p{\alpha^*} - \nabla_\alpha \bar l_1^*\p{\alpha^*} +  \nabla_\alpha \bar l_1^*\p{\alpha^*}=B_1 - B_2 + \nabla_\alpha \bar l_1^*\p{\alpha^*}.
\end{align*}
If $M > C_1'\log M \log d$, for any $0<t<C_2'M\gamma_n/\log M$, on event $\mathcal{B}_1\p{t} \cap \mathcal{B}_2\p{t} \cap \mathcal{B}_3\p{t}$, we have
\begin{align*}
    &\norm{\nabla_\alpha \bar l_1\p{\alpha^*}}_\infty \leq \norm{B_1}_\infty + \norm{B_2}_\infty + \norm{\nabla_\alpha \bar l_1^*\p{\alpha^*}}_\infty\\
    &\qquad\leq c_6(\frac{t  + \log d}{n})^{1/2} + c_4(\frac{t  + \log d}{M})^{1/2} + c_3(\frac{t  + \log d}{M\gamma_n})^{1/2}\\
    &\qquad\leq c_6\frac{K-1}{2K}(\frac{t  + \log d}{M\gamma_n})^{1/2} + c_4(\frac{t +  \log d}{M\gamma_n})^{1/2} + c_3(\frac{t  + \log d}{M\gamma_n})^{1/2}\leq c_1'(\frac{t  + \log d}{M\gamma_n})^{1/2},
\end{align*}
where $c_1' = c_6/2 + c_4 + c_3$. Thus, if $M\gamma_n > C_1'\log M \log d$, for any $0<t<C_2'M\gamma_n/\log M$,
\begin{align*}
    \P\p{\norm{\nabla_\alpha \bar l_1\p{\alpha^*}}_\infty \leq c_1'((t + \log d)/(M\gamma_n))^{1/2}} \geq 1 - 13\exp\p{-t}.
\end{align*}

\end{proof}

\begin{lemma}\label{lemma mar prob. of Delta in the ball}
    Let Assumption~\ref{assumption mar nuisance (a)} hold. For any $0<t<0.01M\gamma_n/\log M$, choose
\[
\lambda_\alpha \asymp ((t + \log d)/(M\gamma_n))^{1/2}.
\]
If
\[
M\gamma_n > C\max\br{s_{\alpha},\log M}\p{t+\log d}
\]
for some constant $C>0$, then there exist some constants $c_1, c_2>0$ such that
\begin{align*}
    \P\p{\norm{\widehat \alpha - \alpha^*}_2 \leq 1} \geq 1 - c_1 \exp\p{-c_2M\gamma_n} - 13\exp\p{-t}.
\end{align*}
\end{lemma}
\begin{proof}
    By the construction of $\widehat \alpha$, we have
    \begin{align*}
        \bar l_1\p{\widehat \alpha} + \lambda_\alpha \norm{\widehat \alpha}_1 \leq \bar l_1\p{\alpha^*} + \lambda_\alpha \norm{\alpha^*}_1.
    \end{align*}
Let $\widehat \Delta_\alpha = \widehat \alpha - \alpha^*$. Then we have
\begin{align*}
    \bar l_1\p{\alpha^* + \widehat \Delta_\alpha} - \bar l_1\p{\alpha^*} - \nabla_\alpha \bar l_1\p{\alpha^*}^\top \widehat \Delta_\alpha + \nabla_\alpha \bar l_1\p{\alpha^*}^\top \widehat \Delta_\alpha + \lambda_\alpha\p{\norm{\widehat \alpha}_1 - \norm{\alpha^*}_1} \leq 0,
\end{align*}
which implies 
\begin{align*}
    \delta \bar l_1\p{\alpha^*, \widehat \Delta_\alpha} + \lambda_\alpha\p{\norm{ \alpha^* + \widehat \Delta_\alpha}_1 - \norm{\alpha^*}_1} \leq -\nabla_\alpha \bar l_1\p{\alpha^*}^\top \widehat \Delta_\alpha.
\end{align*}
Since
\begin{align*}
    &\norm{ \alpha^* + \widehat \Delta_\alpha}_1 = \norm{ \alpha^* + \widehat \Delta_{\alpha, S_\alpha}}_1 + \norm{ \widehat \Delta_{\alpha, S^c_\alpha}}_1\geq \norm{ \alpha^*}_1 - \norm{ \widehat \Delta_{\alpha, S_\alpha}}_1 + \norm{ \widehat \Delta_{\alpha, S^c_\alpha}}_1\\
    &\qquad=\norm{ \alpha^*}_1 - 2\norm{ \widehat \Delta_{\alpha, S_\alpha}}_1 + \norm{ \widehat \Delta_{\alpha}}_1 \geq \norm{ \alpha^*}_1 - 2(s_\alpha)^{1/2}\norm{ \widehat \Delta_{\alpha}}_2 + \norm{ \widehat \Delta_{\alpha}}_1,
\end{align*}
and
$-\nabla_\alpha \bar l_1\p{\alpha^*}^\top \widehat \Delta_\alpha \leq \norm{\nabla_\alpha \bar l_1\p{\alpha^*}}_\infty\norm{\widehat \Delta_\alpha}_1,$
we have
\begin{align*}
    \delta \bar l_1\p{\alpha^*, \widehat \Delta_\alpha} \leq \p{\norm{\nabla_\alpha \bar l_1\p{\alpha^*}}_\infty - \lambda_\alpha}\norm{\widehat \Delta_\alpha}_1 + 2\lambda_\alpha(s_\alpha)^{1/2}\norm{ \widehat \Delta_{\alpha}}_2.
\end{align*}
Define events $\mathcal{C}_1, \mathcal{C}_2(t)$ as
\begin{align}
    \mathcal{C}_1 &:= \br{\delta \bar l_1 \p{\alpha^*, \Delta} \geq \kappa_1 \norm{\Delta}_2^2 -  \kappa_2 \frac{\log d}{M\gamma_n} \norm{\Delta}_1^2, \forall \norm{\Delta }_2 \leq 1 }, \label{definition event C1}\\
    \mathcal{C}_2(t) &:= \br{\norm{\nabla_\alpha \bar l_1\p{\alpha^*}}_\infty \leq c((t + \log d)/(M\gamma_n))^{1/2}}.\label{definition event C2}
\end{align}
For $0 < t < C_1M\gamma_n/\log M $, if $M\gamma_n > C_2\log M \log d$ for some constant $C_1, C_2>0$, By Lemma~\ref{lemma mar RSC condition} and Lemma~\ref{lemma mar gradient infty norm}, we have 
\begin{align}
    \P\p{\mathcal{C}_1 \cap \mathcal{C}_2(t)} \geq 1 - c_1 \exp\p{-c_2M\gamma_n} - 13\exp\p{-t}. \label{prob. for C1 cap C2}
\end{align}
Choose $\lambda_\alpha \geq  2c((t + \log d)/(M\gamma_n))^{1/2}$. On event $\mathcal{C}_2(t)$, we have 
$\norm{\nabla_\alpha \bar l_1\p{\alpha^*}}_\infty \leq {\lambda_\alpha}/{2},$
which implies
\begin{align*}
    \delta \bar l_1\p{\alpha^*, \widehat \Delta_\alpha} + \frac{\lambda_\alpha}{2}\norm{\widehat \Delta_\alpha}_1 - 2\lambda_\alpha(s_\alpha)^{1/2}\norm{ \widehat \Delta_{\alpha}}_2 \leq 0.
\end{align*}
Define 
\begin{align}
    F\p{\Delta; \alpha^*} = \delta \bar l_1\p{\alpha^*, \Delta} + \frac{\lambda_\alpha}{2}\norm{\Delta}_1 - 2\lambda_\alpha(s_\alpha)^{1/2}\norm{ \Delta}_2. \label{definition of F}
\end{align}
Then $F\p{\widehat \Delta_\alpha; \alpha^*} \leq 0$ on event $\mathcal{C}_2(t)$. 

Now we show by contradiction that when $n$ is large enough, $\norm{\widehat \Delta_\alpha}_2 \leq 1$ on event $\mathcal{C}_1 \cap \mathcal{C}_2(t)$. Suppose that $\eta_\alpha^{-1} = \norm{\widehat \Delta_\alpha}_2 > 1$. Let 
$\widetilde \Delta_\alpha = \widehat \Delta_\alpha/\norm{\widehat \Delta_\alpha}_2 = \eta_\alpha\widehat \Delta_\alpha + \p{1- \eta_\alpha}\boldsymbol{0}.$
Then $\norm{\widetilde \Delta_\alpha}_2 = 1$. By \eqref{taylor expansion of delta l1}, $\delta \bar l_1\p{\alpha^*, \widetilde \Delta_\alpha} \geq 0$. Since $\delta \bar l_1\p{\alpha^*, \Delta} = \bar l_1\p{\alpha^*+\Delta} - \bar l_1\p{\alpha^*} - \nabla_\alpha \bar l_1\p{\alpha^*}^\top \Delta$ is a convex function of $\Delta$, it is clear that $F\p{\Delta; \alpha^*}$ is also a convex function of $\Delta$. Thus,
\begin{align}
    &\frac{\lambda_\alpha}{2}\norm{\widetilde \Delta_\alpha}_1 - 2\lambda_\alpha(s_\alpha)^{1/2}\norm{\widetilde \Delta_\alpha}_2 \leq F\p{\widetilde \Delta_\alpha; \alpha^*} = F\p{\eta_\alpha\widehat \Delta_\alpha + \p{1- \eta_\alpha}\boldsymbol{0}; \alpha^*} \notag\\
    &\qquad\leq \eta_\alpha F\p{\widehat \Delta_\alpha; \alpha^*} + \p{1- \eta_\alpha}F\p{\boldsymbol{0}; \alpha^*}= \eta_\alpha F\p{\widehat \Delta_\alpha; \alpha^*} \leq 0. \label{contradiction 1}
\end{align}
It follows that
\begin{align}
    \norm{\widetilde \Delta_\alpha}_1 \leq  4(s_\alpha)^{1/2}\norm{\widetilde \Delta_\alpha}_2 = 4(s_\alpha)^{1/2}.\label{mar cone}
\end{align}
Therefore, on event $\mathcal{C}_1$ we have
\begin{align*}
    &F\p{\widetilde \Delta_\alpha; \alpha^*} = \delta \bar l_1\p{\alpha^*, \widetilde \Delta_\alpha} + \frac{\lambda_\alpha}{2}\norm{\widetilde \Delta_\alpha}_1 - 2\lambda_\alpha(s_\alpha)^{1/2}\norm{ \widetilde \Delta_\alpha}_2\\
    &\qquad\geq \kappa_1 \norm{\widetilde \Delta_\alpha}_2^2 -  \kappa_2 \frac{\log d}{M\gamma_n} \norm{\widetilde \Delta_\alpha}_1^2+ \frac{\lambda_\alpha}{2}\norm{\widetilde \Delta_\alpha}_1 - 2\lambda_\alpha(s_\alpha)^{1/2}\norm{ \widetilde \Delta_\alpha}_2\\
    &\qquad\geq \kappa_1 - 2\lambda_\alpha(s_\alpha)^{1/2}  - \kappa_2\frac{\log d}{M\gamma_n}\norm{\widetilde \Delta_\alpha}_1^2\geq \frac{\kappa_1}{2} - 2\lambda_\alpha(s_\alpha)^{1/2}- 16\kappa_2\frac{s_\alpha\log d}{M\gamma_n}.
\end{align*}
Choose $\lambda_\alpha =  2c((t + \log d)/(M\gamma_n))^{1/2}$. For any $0<t<C_1M\gamma_n/\log M$, if 
${M\gamma_n}/{s_{\alpha}\p{t+\log d}} > \max\br{128{\kappa_2}/{\kappa_1}, 128{c^2}/{\kappa_1}},$
then
\begin{align*}
    F\p{\widetilde \Delta_\alpha; \alpha^*} \geq \frac{\kappa_1}{2} - 4c(\frac{s_\alpha(t + \log d)}{M\gamma_n})^{1/2}- 16\kappa_2\frac{s_\alpha\log d}{M\gamma_n} \geq \frac{\kappa_1}{2} - \frac{\kappa_1}{8} - \frac{\kappa_1}{8} = \frac{\kappa_1}{4} > 0,
\end{align*}
which contradicts with \eqref{contradiction 1}. It follows that $\norm{\widehat \Delta_\alpha}_2 \leq 1$ on event $\mathcal{C}_1 \cap \mathcal{C}_2(t)$ and 
\begin{align}
    \P\p{\norm{\widehat \Delta_\alpha}_2 \leq 1} \geq \P\p{\mathcal{C}_1 \cap \mathcal{C}_2(t)} \geq 1 - c_1 \exp\p{-c_2M\gamma_n} - 13\exp\p{-t}. \label{Delta alpha < 1 on C1 cap C2}
\end{align}
\end{proof}

\begin{lemma}\label{lemma mar delta l2 concentration}
    Let Assumption~\ref{assumption mar nuisance (a)} hold. Then there exist some constants $\kappa_1', \kappa_2'$, $c_1', c_2'>0$ such that for any $\beta \in \R^d$,
\begin{align*}
    \P\p{\delta \bar l_2\p{\alpha^*, \beta, \Delta} \geq \kappa_1' \norm{\Delta}_2^2 -  \kappa_2' \frac{\log d}{M\gamma_n} \norm{\Delta}_1^2, \forall \norm{\Delta
        }_2 \in \R^d}\geq 1 - c_1' \exp\p{-c_2'M\gamma_n}.
\end{align*}
\end{lemma}

\begin{proof}
By Taylor's theorem,
\begin{align} \label{delta l_2 taylor expansion}
    \delta \bar l_2\p{\alpha^*, \beta, \Delta} &= M^{-1} \sum_{i \in \mathcal{I}_{-k,\beta}} \Gamma_i \exp\p{-X_i^\top \alpha^*}\p{X_i^\top \Delta}^2. 
\end{align}
Then by \eqref{exp(-X'alpha = gamma-1)}, 
\begin{align}
    \delta \bar l_2\p{\alpha^*, \beta, \Delta} &\geq k_0c_0\gamma_n^{-1} L_2, \label{RSC t2'}
\end{align}
where
$L_2 = M^{-1} \sum_{i \in \mathcal{I}_{-k,\beta}} \Gamma_i \p{X_i^\top \Delta}^2.$
Let $\widehat \gamma_{-k,\beta} = M^{-1}\sum_{i \in \mathcal{I}_{-k,\beta}} \Gamma_i$. 
Under Assumption~\ref{assumption mar nuisance (a)}(b), by Lemma~\ref{lemma Lemma C.5 of zhang 2021} we have
\begin{align*}
    \P\p{\widehat \gamma^{-1}_{-k,\beta} L_2 \geq \kappa_1 \norm{\Delta}_2^2 - \kappa_2 \frac{\log d}{M\widehat \gamma_{-k,\beta}} \norm{\Delta}_1^2, \forall \norm{\Delta
        }_2 \in \R^d \mid \Gamma_{\mathcal{I}_{-k,\beta}}}\geq 1- c_1\exp\p{-c_2M\widehat \gamma_{-k,\beta}}.
\end{align*}
With identical analysis to \eqref{RSC t1}, we have that for some constants $\kappa_1', \kappa_2', c_1', c_2'>0$,
\begin{align*}
    \P\p{\delta \bar l_2\p{\alpha^*, \beta, \Delta} \geq \frac{k_0c_0}{2}\br{\kappa_1' \norm{\Delta}_2^2 -  \kappa_2' \frac{\log d}{M\gamma_n} \norm{\Delta}_1^2}, \forall \norm{\Delta
        }_2 \in \R^d} \geq 1 - c_1' \exp\p{-c_2'M\gamma_n}.
\end{align*}
\end{proof}

\begin{lemma}\label{lemma mar gradient infty norm for beta}
    Let Assumption~\ref{assumption mar nuisance (a)} hold. If $M\gamma_n > \p{\log M}^2\log d$, then for any $0<t<0.01M\gamma_n/\p{\log M}^2$, there exists some constant $c>0$ such that
\begin{align*}
        \P\p{\norm{\nabla_\beta \bar l_2\p{\alpha^*, \beta^*_{OR}}}_\infty \leq c((t + \log d)/(M\gamma_n))^{1/2}} \geq 1 - 3\exp\p{-t}.
\end{align*}
\end{lemma}
\begin{proof}
    Let $w_{OR,i} = Y_i - X_i^\top \beta^*_{OR}$. For simplicity, we ignore the subscript $OR$ in this proof. 
 Then
\begin{align*}
    \nabla_\beta \bar l_2\p{\alpha^*, \beta^*} = -2M^{-1} \sum_{i \in \mathcal{I}_{-k,\beta}} \Gamma_i \exp\p{-X_i^\top \alpha^*}w_iX_i.
\end{align*}
First, by the construction of $\beta^*$ we have
$\E\sbr{\Gamma_i \exp\p{-X_i^\top \alpha^*}w_iX_i} = \boldsymbol{0}.$
Under Assumption~\ref{assumption mar nuisance (a)}, by \eqref{exp(-X'alpha = gamma-1)} and Lemma~\ref{sub-gaussian properties}(d), for any $1 \leq k \leq d$,
\begin{align*}
    \norm{-2\Gamma_i \exp\p{-X_i^\top \alpha^*}w_iX_i^\top e_k}_{\psi_1} \leq 2k_0^{-1}\frac{1-\gamma_n}{\gamma_n}\norm{w_i}_{\psi_2} \norm{X_i^\top e_k}_{\psi_2} \leq 2k_0^{-1}\frac{1-\gamma_n}{\gamma_n}\sigma_{w}\sigma,
\end{align*}
which implies
\begin{align*}
    \max_{i \in \mathcal{I}_{-k,\beta}} \max_{1 \leq k \leq d} \norm{-2\Gamma_i \exp\p{-X_i^\top \alpha^*}w_iX_i^\top e_k}_{\psi_1} \leq 2k_0^{-1}\frac{1-\gamma_n}{\gamma_n}\sigma_{w}\sigma.
\end{align*}
Moreover, by \eqref{exp(-2X'alpha = gamma-2)} and Lemma~\ref{sub-gaussian properties} (c),
\begin{align*}
    &\E\sbr{\br{-2\Gamma_i \exp\p{-X_i^\top \alpha^*}w_iX_i^\top e_k}^2}= \E\sbr{4\Gamma_i\exp\p{-2X_i^\top \alpha^*}\p{w_iX_i^\top e_k}^2}\\
    &\qquad\leq 4k_0^{-2}\p{\frac{1-\gamma_n}{\gamma_n}}^2\E\sbr{\Gamma_i\p{w_iX_i^\top e_k}^2}= 4k_0^{-2}\frac{(1-\gamma_n)^2}{\gamma_n}\E\sbr{\p{w_iX_i^\top e_k}^2\mid \Gamma_i=1}\\
    &\qquad\leq 4k_0^{-2}\frac{(1-\gamma_n)^2}{\gamma_n} \E\sbr{w_i^4\mid \Gamma_i=1}^{1/2}\E\sbr{\p{X_i^\top e_k}^4\mid \Gamma_i=1}^{1/2}\leq 4k_0^{-2} C_2\sigma_w^2\sigma^2\frac{(1-\gamma_n)^2}{\gamma_n}.
\end{align*}
Thus,
\begin{align*}
    \max_{1\leq k \leq d} M^{-1} \sum_{i \in \mathcal{I}_{-k,\beta}} \E\sbr{\br{-2\Gamma_i \exp\p{-X_i^\top \alpha^*}w_iX_i^\top e_k}^2} \leq 4k_0^{-2} C_2\sigma_w^2\sigma^2\frac{(1-\gamma_n)^2}{\gamma_n}.
\end{align*}
By Lemma~\ref{Lemma Tail Bounds for Maximums}, there exist some constants $c_1, c_2>0$ such that with probability at least $1-3\exp\p{-t}$,
\begin{align*}
    \norm{\nabla_\beta \bar l_2\p{\alpha^*, \beta^*}}_\infty \leq c_0\br{c_1((t + \log d)/(M\gamma_n))^{1/2} + c_2\frac{\p{t+ \log d}\log M}{M\gamma_n}}.
\end{align*}
If $M\gamma_n > \p{\log M}^2\log d$, then for any $0<t<0.01M\gamma_n/\p{\log M}^2$,
\begin{align*}
    c_1((t + \log d)/(M\gamma_n))^{1/2} + c_2\frac{\p{t+ \log d}\log M}{M\gamma_n} \leq \p{c_1 + c_2(2)^{1/2}}((t + \log d)/(M\gamma_n))^{1/2}.
\end{align*}
Let $c_3 = c_0(c_1 + c_2(2)^{1/2})$, we have
$
    \P\p{\norm{\nabla_\beta \bar l_2\p{\alpha^*, \beta^*}}_\infty\leq c_3((t + \log d)/(M\gamma_n))^{1/2}} \geq 1 - 3\exp\p{-t}.
$
\end{proof}

\begin{lemma}\label{lemma sparsity for norm1 and norm2}
   Let $s_{\tilde\beta} = \norm{ \widetilde \beta_{OR}}_0$. If $\lambda_\beta > \max\br{2\norm{\nabla_\beta \bar l_2(\alpha^*, \beta^*)}_\infty, 2\|\nabla_\beta \bar l_2(\widehat \alpha, \widetilde \beta)\|_\infty}$, then
    \begin{align*}
        \norm{\widetilde \beta_{OR} - \beta^*_{OR}}_1 \leq 4(s_\beta)^{1/2}\norm{\widetilde \beta_{OR} - \beta^*_{OR}}_2,\quad\norm{\widehat \beta_{OR} - \widetilde \beta_{OR}}_1 \leq 4(s_{\tilde\beta})^{1/2}\norm{\widehat \beta_{OR} - \widetilde \beta_{OR}}_2.
    \end{align*}
\end{lemma}
\begin{proof}
For simplicity, we omit the subscript $OR$ in this proof. By the construction of $\widetilde \beta$,
\begin{align*}
    \bar l_2(\alpha^*, \widetilde \beta) + \lambda_\beta \norm{\widetilde \beta}_1 \leq \bar l_2(\alpha^*, \beta^*) + \lambda_\beta \norm{\beta^*}_1.
\end{align*}
Let $\widetilde \Delta_\beta = \widetilde \beta - \beta^*$ and $w_i = Y_i - X_i^\top \beta^*$. Then
\begin{align}
    &M^{-1}\sum_{i \in \mathcal{I}_{-k,\beta}} \Gamma_i \exp\p{-X_i^\top \alpha^*}\p{X_i^\top\widetilde \Delta_\beta}^2  \notag \\
    &\qquad\leq2M^{-1}\sum_{i \in \mathcal{I}_{-k,\beta}} \Gamma_i \exp\p{-X_i^\top \alpha^*}w_i\p{X_i^\top\widetilde \Delta_\beta}+ \lambda_\beta\p{\norm{\beta^*}_1 - \norm{\widetilde \beta}_1}. \label{qwer1}
\end{align}
Let $S_\beta$ be the support of $\beta^*$. Similar to \eqref{sparsity technique}, we have
$
    \norm{\beta^*}_1 - \norm{\widetilde \beta}_1 \leq \norm{{\widetilde \Delta}_{\beta, S_\beta}}_1 - \norm{{\widetilde \Delta}_{\beta, S_\beta^c}}_1.
$
If $\lambda_\beta \geq 2\norm{\nabla_\beta \bar l_2(\alpha^*, \beta^*)}_\infty$, then
\begin{align*}
    0 &\leq 2M^{-1}\sum_{i \in \mathcal{I}_{-k,\beta}} \Gamma_i \exp\p{-X_i^\top \alpha^*}w_i\p{X_i^\top\widetilde \Delta_\beta} + \lambda_\beta\p{\norm{\beta^*}_1 - \norm{\widetilde \beta}_1}\\
    &\leq \norm{\nabla_\beta \bar l_2(\alpha^*, \beta^*)}_\infty\norm{\widetilde \Delta_\beta}_1 + \lambda_\beta\p{\norm{{\widetilde \Delta}_{\beta, S_\beta}}_1 - \norm{{\widetilde \Delta}_{\beta, S_\beta^c}}_1}\\
    &\leq \frac{\lambda_\beta}{2}\norm{\widetilde \Delta_\beta}_1 + \lambda_\beta\p{\norm{{\widetilde \Delta}_{\beta, S_\beta}}_1 - \norm{{\widetilde \Delta}_{\beta, S_\beta^c}}_1}= \frac{3\lambda_\beta}{2}\norm{{\widetilde \Delta}_{\beta, S_\beta}}_1 - \frac{\lambda_\beta}{2}\norm{{\widetilde \Delta}_{\beta, S_\beta^c}}_1,
\end{align*}
which implies $\norm{{\widetilde \Delta}_{\beta, S_\beta^c}}_1 \leq 3\norm{{\widetilde \Delta}_{\beta, S_\beta}}_1$. Thus,
$
    \norm{\widetilde \Delta_\beta}_1 = \norm{{\widetilde \Delta}_{\beta, S_\beta^c}}_1 + \norm{{\widetilde \Delta}_{\beta, S_\beta}}_1 \leq 4\norm{{\widetilde \Delta}_{\beta, S_\beta}}_1\leq 4(s_\beta)^{1/2}\norm{\widetilde \Delta_\beta}_2.
$
In addition, for \eqref{qwer1}, by the analysis above we also have
\begin{align}
    M^{-1}\sum_{i \in \mathcal{I}_{-k,\beta}} \Gamma_i \exp\p{-X_i^\top \alpha^*}\p{X_i^\top\widetilde \Delta_\beta}^2 &\leq \frac{3\lambda_\beta}{2}\norm{{\widetilde \Delta}_{\beta, S_\beta}}_1- \frac{\lambda_\beta}{2}\norm{{\widetilde \Delta}_{\beta, S_\beta^c}}_1\leq \frac{3\lambda_\beta (s_\beta)^{1/2}}{2}\norm{\widetilde \Delta_\beta}_2. \label{qwer2}
\end{align}

Similarly, for $\Delta = \widehat \beta - \widetilde \beta$, consider the construction of $\widehat \beta$. If $\lambda_\beta \geq 2\norm{\nabla_\beta \bar l_2(\widehat \alpha, \widetilde \beta)}_\infty$, then
$
    \norm{\Delta}_1 \leq 4(s_{\tilde \beta})^{1/2}\norm{\Delta}_2.
$
\end{proof}

\begin{lemma}\label{lemma consistency for alpha and betatilde}
Let Assumption~\ref{assumption mar nuisance (a)} hold. Then as $n, d \rightarrow \infty$, it holds that

    (a) Choose $\lambda_\alpha \asymp (\log d/(n\gamma_n))^{1/2}$. If $n\gamma_n \gg \max\br{s_\alpha, \log n} \log d$, then
    $$\norm{\widehat \alpha_{PS} - \alpha^*_{PS}}_2 = O_p\p{(s_\alpha \log d/(n\gamma_n))^{1/2}} \quad \text{and} \quad \norm{\widehat \alpha_{PS} - \alpha^*_{PS}}_1 = O_p\p{s_\alpha(\log d/(n\gamma_n))^{1/2}}.$$

    (b) Choose $\lambda_\beta \asymp (\log d/(n\gamma_n))^{1/2}$. If $n\gamma_n \gg \max\br{s_\beta, (\log n)^2} \log d$, then
    $$\norm{\widetilde\beta_{OR} - \beta^*_{OR}}_2 = O_p\p{(s_\beta \log d/(n\gamma_n))^{1/2}} \quad \text{and} \quad \norm{\widetilde\beta_{OR} - \beta^*_{OR}}_1 = O_p\p{s_\beta (\log d/(n\gamma_n))^{1/2}}.$$
\end{lemma}

\begin{proof}
    (a) follows by Lemma~\ref{lemma mar RSC condition}, Lemma~\ref{lemma mar gradient infty norm}, and Corollary 9.20 of \cite{wainwright2019high}. (b) follows by Lemma~\ref{lemma mar delta l2 concentration}, Lemma~\ref{lemma mar gradient infty norm for beta}, and Corollary 9.20 of \cite{wainwright2019high}.
\end{proof}

\begin{lemma}\label{lemma sparsity betatilde -> beta}
    Let Assumption~\ref{assumption mar nuisance (a)} hold. Choose $\lambda_\beta \asymp (\log d/(n\gamma_n))^{1/2}$. If $n\gamma_n \gg \max\br{s_\beta , \log n}\log n \log d$, then as $n, d\rightarrow \infty$, $s_{\tilde \beta} = O_p\p{s_{\beta}}$.
\end{lemma}

\begin{proof}
For simplicity, we ignore the subscript $OR$ in this proof. By the construction of $\widetilde \beta$, with the first order optimality, we have
$
\nabla_{\beta} \bar l_2(\alpha^*, \widetilde \beta) + \lambda_{\beta} \tilde \xi = \boldsymbol{0},
$
where  $\tilde \xi \in \partial \|\widetilde \beta\|_1$ such that
\begin{align*}
\tilde \xi_j = \begin{cases}
         {\sf sign}(\widetilde \beta_{j}) \text{, if } \widetilde \beta_{j} \neq 0,\\
        \xi \in [-1, 1] \text{, otherwise}. 
    \end{cases}
\end{align*}
Let $w_i = Y_i - X_i^\top \beta^*$ and $\widetilde \Delta_\beta = \widetilde\beta - \beta^*$. Denote $e_j$ a vector whose $j$-th coordinate is one while all other coordinates are zeros and $\tilde e_j = \tilde \xi_j e_j$. Define the support of $\widetilde \beta$ as $\tilde S = \{j \in [d]: \widetilde\beta_j \neq 0\}$. Then for any $j \in \tilde S$, by Taylor's theorem, for some $t \in (0,1)$ and $\bar \beta = t\widetilde \beta + (1-t)\beta^*$,
    \begin{align*}
         \lambda_{\beta} &= -\nabla_{\beta} \bar l_2(\alpha^*, \widetilde \beta)^\top \tilde e_j=-\nabla_{\beta} \bar l_2(\alpha^*, \beta^*)^\top \tilde e_j -\widetilde \Delta_\beta^\top \nabla_{\beta}^2 \bar l_2(\alpha^*, \bar \beta) \tilde e_j,
    \end{align*}
Let $\tilde e_{\tilde S} = \sum_{j \in \tilde S} \tilde e_j$, then $\norm{\tilde e_{\tilde S}}_1 = \norm{\tilde e_{\tilde S}}_2^2 = s_{\tilde \beta}$. Choose $\lambda_\beta \geq 2\norm{\nabla_{\beta} \bar l_2(\alpha^*, \beta^*)}_\infty$. It follows that 
\begin{align*}
    \widetilde \Delta_\beta^\top \nabla_{\beta}^2 \bar l_2(\alpha^*, \bar \beta) \tilde e_{\tilde S} &= s_{\tilde \beta}\lambda_{\beta} + \nabla_{\beta} \bar l_2(\alpha^*, \beta^*)^\top \tilde e_{\tilde S}\geq s_{\tilde \beta}\lambda_{\beta} - \norm{\nabla_{\beta} \bar l_2(\alpha^*, \beta^*)}_\infty\norm{\tilde e_{\tilde S}}_1 \geq \frac{s_{\tilde \beta}\lambda_\beta}{2}.
\end{align*}
On the other hand, by H{\"o}lder inequality,
\begin{align*}
    \widetilde \Delta_\beta^\top \nabla_{\beta}^2 \bar l_2(\alpha^*, \bar \beta) \tilde e_{\tilde S} &= 2M^{-1}\sum_{\mc{I}_{-k,\beta}}\Gamma_i\exp\p{-X_i^\top\alpha^*}\p{X_i^\top\widetilde \Delta_\beta}\p{X_i^\top \tilde e_{\tilde S}}\leq 2(Q_1Q_2)^{1/2},
\end{align*}
where by \eqref{delta l_2 taylor expansion},
\begin{align*}
    Q_1 &= M^{-1}\sum_{\mc{I}_{-k,\beta}}\Gamma_i\exp\p{-X_i^\top\alpha^*}\p{X_i^\top\widetilde \Delta_\beta}^2 = \delta \bar l_2(\alpha^*, \beta^*, \widetilde \Delta_\beta),\\
    Q_2 &= M^{-1}\sum_{\mc{I}_{-k,\beta}}\Gamma_i\exp\p{-X_i^\top\alpha^*}\p{X_i^\top \tilde e_{\tilde S}}^2.
\end{align*}
For $Q_1$, by \eqref{qwer2}, we have
\begin{align}
    Q_1 \leq \frac{3\lambda_\beta(s_\beta)^{1/2}}{2}\norm{\widetilde \Delta_\beta}_2. \label{Q1}
\end{align}
For $Q_2$, by \eqref{exp(-X'alpha = gamma-1)}, we have
$Q_2 \leq k_0^{-1}\p{M\gamma_n}^{-1}\sum_{\mc{I}_{-k,\beta}}\Gamma_i\p{X_i^\top\tilde e_{\tilde S}}^2.$
Then by Lemma E.6(d) of \cite{zhang2023semi}, with probability at least $1-2\exp(-t)$, for some constant $c_1>0$, it holds that
\begin{align*}
    Q_2 &\leq c_1\p{s_{\tilde \beta}^{-1}\norm{\tilde e_{\tilde S}}_1^2 + \norm{\tilde e_{\tilde S}}_2^2}\p{1+(\frac{s_{\tilde \beta}\log d}{M\gamma_n})^{1/2} + \frac{s_{\tilde \beta}\log M \log d}{M\gamma_n}}\\
    &= 2c_1s_{\tilde \beta}\p{1+(\frac{s_{\tilde \beta}\log d}{M\gamma_n})^{1/2} + \frac{s_{\tilde \beta}\log M \log d}{M\gamma_n}}.
\end{align*}

In conclusion, when $\lambda_\beta \geq 2\norm{\nabla_{\beta} \bar l_2(\alpha^*, \beta^*)}_\infty$, with probability at least $1-2\exp(-t)$,
\begin{align*}
    \frac{s_{\tilde \beta}\lambda_\beta}{2} \leq 2\p{3c_1\lambda_\beta(s_\beta)^{1/2}\norm{\widetilde \Delta_\beta}_2s_{\tilde \beta}\p{1+(\frac{s_{\tilde \beta}\log d}{M\gamma_n})^{1/2} + \frac{s_{\tilde \beta}\log M \log d}{M\gamma_n}}}^{1/2},
\end{align*}
which implies that for $c_n = 36c_1\lambda_\beta^{-1}(s_\beta)^{1/2}\|\widetilde \Delta_\beta\|_2$, $b_n = c_n({\log d}/{M\gamma_n})^{1/2}$, and $a_n = c_n{\log M\log d}/{M\gamma_n}$,
$
    \p{1 - a_n}s_{\tilde \beta} - b_n(s_{\tilde \beta})^{1/2} - c_n \leq 0.
$
If $1 - a_n > 0$, we have
$
    s_{\tilde \beta} \leq \p{{b_n + (b_n^2 + 4\p{1 - a_n}c_n )^{1/2}}/{2\p{1 - a_n}}}^2.
$
For any $0<t<0.01M\gamma_n/(\log M)^2$, choose $\lambda_\beta = 2c((t+\log d)/(M\gamma_n))^{1/2}$. By Lemma~\ref{lemma mar gradient infty norm for beta}, if $M\gamma_n > (\log M)^2\log d$, then with probability at least $1-3\exp\p{-t}$, we have
$
    \lambda_\beta \geq 2\norm{\nabla_{\beta} \bar l_2(\alpha^*, \beta^*)}_\infty.
$

When $n\gamma_n \gg \max\br{s_\beta , \log n}\log n \log d$, we have $\lambda_\beta \asymp (\log d/(n\gamma_n))^{1/2}$, and by Lemma~\ref{lemma consistency for alpha and betatilde},
$
    \norm{\widetilde \Delta_\beta}_2 =O_p\p{(s_\beta \log d/(n\gamma_n))^{1/2}},
$
which implies $c_n = O_p(s_\beta)$, $b_n = o_p((s_\beta)^{1/2})$, and $a_n = o_p(1)$. Thus,
\begin{align*}
    s_{\tilde \beta} &= O_p\p{s_\beta\p{\frac{s_\beta^{-1/2}b_n + (s_\beta^{-1}b_n^2 + 4s_\beta^{-1}\p{1 - a_n}c_n )^{1/2}}{2\p{1 - a_n}}}^2}=O_p\p{s_\beta}.
\end{align*}
\end{proof}

\begin{lemma} \label{lemma Exp(Xn2) ineq}
Let Assumption~\ref{assumption mar nuisance (a)} hold. Then for any $p>1$, $u \in \R$, and $\Delta \in \R^d$, there exist some constant $c, c_p>0$ such that
\begin{align*}
    \E\sbr{\exp\p{u\abs{X_i^\top \Delta}}\p{X_i^\top \Delta}^p \mid \Gamma_i} \leq c_p\exp\p{cu^2\norm{\Delta}_2^2}\norm{\Delta}_2^{p} \text{ almost surely}.
\end{align*}
In addition, choose $\lambda_\alpha \asymp ((t + \log d)/(M\gamma_n))^{1/2}$. If $M\gamma_n > C\max\br{s_{\alpha},\log M}\p{t+\log d}$ for some constant $C>0$, there exist some constants $c_1, c_2, c_{u,p}>0$ such that for any $0<t<0.01M\gamma_n/\log M$, 
\begin{align*}
    &\P\p{M^{-1}\sum_{i \in \mathcal{I}_{-k,\beta}} \Gamma_i\exp\p{u\abs{X_i^\top \p{\widehat \alpha_{PS} - \alpha^*_{PS}}}}\p{X_i^\top \p{\widehat \alpha_{PS} - \alpha^*_{PS}}}^p \leq c_{u,p}t\gamma_n\norm{\widehat \alpha_{PS} - \alpha^*_{PS}}_2^p} \\
    &\qquad\geq 1- t^{-1} - c_1 \exp\p{-c_2M\gamma_n} - 15\exp\p{-t}.
\end{align*}

\end{lemma}
\begin{proof}
Let $\widetilde \Delta = \Delta / \norm{\Delta}_2$. By Lemma~\ref{sub-gaussian properties}(c) and (f), for any $p>1$ and $t\in\R$, there exists some constant $C_{1,p}, C_{2}>0$ such that
\begin{align*}
    \E\sbr{\abs{\p{X_i^\top \Delta}^{2p}} \mid \Gamma_i } &= \norm{ \Delta}_2^{2p}\E\sbr{\abs{\p{X_i^\top \widetilde \Delta}^{2p}} \mid \Gamma_i} \leq C_{1,p}^2\norm{\Delta}_2^{2p}\sigma^{2p},\\
\E\sbr{\exp\p{2t\norm{\Delta}_2\abs{X_i^\top \widetilde \Delta}}\mid \Gamma_i} &\leq \exp\p{4C_{2}t^2\sigma^2\norm{\Delta}_2^2}.
\end{align*}
Thus, by H{\"o}lder inequality, for any $p>0$,
\begin{align}
&\E\sbr{\exp\p{u\abs{X_i^\top \Delta}}\p{X_i^\top \Delta}^p \mid \Gamma_i} \leq \E\sbr{\exp\p{2u\abs{X_i^\top \Delta}} \mid \Gamma_i}^{1/2}\E\sbr{\p{X_i^\top \Delta}^{2p} \mid \Gamma_i}^{1/2} \notag\\
    &\qquad= \E\sbr{\exp\p{2u\norm{\Delta}_2\abs{X_i^\top \widetilde \Delta}}\mid \Gamma_i}^{1/2}\E\sbr{\p{X_i^\top \Delta}^{2p} \mid \Gamma_i}^{1/2} \notag\\
    &\qquad\leq \br{\exp\p{4C_{2}u^2\sigma^2\norm{\Delta}_2^2}}^{1/2}\E\sbr{\p{X_i^\top \Delta}^{2p} \mid \Gamma_i}^{1/2}\leq \exp\p{2C_{2}u^2\sigma^2\norm{\Delta}_2^2}C_{1,p}\sigma^{p}\norm{\Delta}_2^{p}. \label{exp(4XDelta)XDelta4}
\end{align}
Let $G = M^{-1}\sum_{i \in \mathcal{I}_{-k,\beta}} \Gamma_i\exp\p{u\abs{X_i^\top \widehat\Delta_\alpha}}\p{X_i^\top \widehat\Delta_\alpha}^p$ where $\widehat\Delta_\alpha = \widehat \alpha - \alpha^*$. Then by Lemma~\ref{lemma conditional independence}, for $i \in \mathcal{I}_{-k,\beta}$, $\Gamma_iX_i \perp \widehat \alpha \mid \Gamma_{i}$ and 
\begin{align*}
    \E\sbr{G \mid \Gamma_{\mathcal{I}_{-k,\beta}}, \widehat \alpha}&=M^{-1}\sum_{i \in \mathcal{I}_{-k,\beta}} \Gamma_i\E\sbr{\exp\p{u\abs{X_i^\top \widehat\Delta_\alpha}}\p{X_i^\top \widehat\Delta_\alpha}^p \mid \Gamma_{\mathcal{I}_{-k,\beta}}, \widehat \alpha}\\
    &=M^{-1}\sum_{i \in \mathcal{I}_{-k,\beta}} \Gamma_i\E\sbr{\exp\p{u\abs{X_i^\top \widehat\Delta_\alpha}}\p{X_i^\top \widehat\Delta_\alpha}^p \mid \Gamma_i, \widehat \alpha}\\
    &\leq M^{-1}\sum_{i \in \mathcal{I}_{-k,\beta}} \Gamma_i\exp\p{2C_{2}u^2\sigma^2\norm{\widehat\Delta_\alpha}_2^2}C_{1,p}\sigma^{p}\norm{\widehat\Delta_\alpha}_2^p.
\end{align*}
Define events $\mc{E}_1(t), \mc{E}_2(t)$ as
\begin{align*}
    \mc{E}_1(t) &:= \br{\norm{\widehat\Delta_\alpha}_2 \leq 1},\quad\mc{E}_2(t) := \br{0.79\gamma_n \leq M^{-1}\sum_{i \in \mathcal{I}_{-k,\beta}} \Gamma_i \leq 1.21\gamma_n}.
\end{align*}
Choose $\lambda_\alpha \asymp ((t + \log d)/(M\gamma_n))^{1/2}$. If $M\gamma_n > C\max\br{s_{\alpha},\log M}\p{t+\log d}$ for some constant $C>0$, then by Lemma~\ref{lemma mar prob. of Delta in the ball} and Lemma~\ref{lemma concentrate gamma}, there exist some constants $c_1, c_2>0$ such that for any $0<t<0.01M\gamma_n/\log M$, 
\begin{align*}
    &\P\p{\mc{E}_1(t)} \geq 1 - c_1 \exp\p{-c_2M\gamma_n} - 13\exp\p{-t},\quad\P\p{\mc{E}_2(t)} \geq 1-2\exp\p{-t}.
\end{align*}
Let $c_{u,p} = \exp\p{2C_{2}u^2\sigma^2}C_{1,p}\sigma^{p}$. Then by Markov inequality,
\begin{align*}
    &\P\p{G \geq c_pt\gamma_n\norm{\widehat\Delta_\alpha}_2^p}=\E\sbr{\mathbbm{1}_{\br{G \geq c_pt\gamma_n\norm{\widehat\Delta_\alpha}_2^p}}\mathbbm{1}_{\mc{E}_1(t)\cap \mc{E}_2(t)}} + \E\sbr{\mathbbm{1}_{\br{G \geq c_pt\gamma_n\norm{\widehat\Delta_\alpha}_2^p}}\mathbbm{1}_{\br{\mc{E}_1(t)\cap \mc{E}_2(t)}^c}}\\
    &\qquad\leq \E\sbr{\mathbbm{1}_{\br{G \geq c_pt\gamma_n\norm{\widehat\Delta_\alpha}_2^p}}\mathbbm{1}_{\mc{E}_1(t)\cap \mc{E}_2(t)}} + \E\sbr{\mathbbm{1}_{\br{\mc{E}_1(t)\cap \mc{E}_2(t)}^c}}\leq t^{-1} + c_1 \exp\p{-c_2M\gamma_n} + 15\exp\p{-t}.
\end{align*}
\end{proof}

\begin{lemma}\label{lemma constant concentration of w}
    Let Assumption~\ref{assumption mar nuisance (a)} hold. Then for any $1\leq p\leq 8$ and some constant $C>0$, with probability at least $1-t^{-1}$,
    \begin{align*}
        M^{-1} \sum_{i \in \mathcal{I}_{-k,\beta}} \Gamma_i \exp\p{-X_i^\top \alpha^*}w_i^p \leq C.
    \end{align*}
\end{lemma}

\begin{proof}
    Under Assumption~\ref{assumption mar nuisance (a)}(c) and by \eqref{exp(-X'alpha = gamma-1)},
    \begin{align*}
        \E\sbr{\Gamma_i \exp\p{-X_i^\top \alpha^*}w_i^p} &\leq k_0^{-1}\gamma_n^{-1}\E\sbr{\Gamma_i w_i^p}=k_0^{-1}\E\sbr{w_i^p \mid \Gamma=1}\leq k_0^{-1}\sigma_w^p \leq k_0^{-1}\max_{1\leq p\leq 8}\sigma_w^p.
    \end{align*}
    By Markov inequality, we finished the proof.
\end{proof}

\begin{lemma}\label{lemma consistency of betahat and betatilde for product sparsity}
Let Assumption~\ref{assumption mar nuisance (a)} hold. Choose $\lambda_\alpha \asymp \lambda_\beta \asymp (\log d/(n\gamma_n))^{1/2}$. If
\[
    n\gamma_n \gg \max\br{s_\alpha, s_\beta\log n , (\log n)^2} \log d,
    \qquad
    s_\alpha s_\beta \ll (n\gamma_n)^{3/2}/\{\log n(\log d)^2\},
\]
then as \(n, d \rightarrow \infty\),
    \begin{align*}
        \norm{\widehat \beta_{OR} - \widetilde \beta_{OR}}_1 = O_p\p{(s_\alpha s_\beta \log d/(n\gamma_n))^{1/2}} \text{ and }\norm{\widehat \beta_{OR} - \widetilde \beta_{OR}}_2 = O_p\p{ (s_\alpha \log d/(n\gamma_n))^{1/2}}.
    \end{align*}
\end{lemma}

\begin{proof}
    For simplicity, we ignore the subscript $OR$ in this proof. By the definition of $\widehat \beta$ and $\widetilde \beta$, with the first order optimality, we have
    \begin{align}
        \nabla_{\beta} \bar l_2(\widehat \alpha, \widehat \beta) + \lambda_{\beta} \hat \xi = \boldsymbol{0}, \label{grad for betahat}\\
        \nabla_{\beta} \bar l_2(\alpha^*, \widetilde \beta) + \lambda_{\beta} \tilde \xi =  \boldsymbol{0} \label{grad for betatilde},
    \end{align}
where $\hat \xi \in \partial \|\widehat \beta\|_1$ and $\tilde \xi \in \partial \|\widetilde \beta\|_1$ such that
\begin{align*}
    \hat \xi_j = \begin{cases}
         {\sf sign}(\widehat \beta_{j}) \text{, if } \widehat \beta_{j} \neq 0,\\
        \xi \in [-1, 1] \text{, otherwise}, 
    \end{cases} \text{ and}\quad \tilde \xi_j = \begin{cases}
         {\sf sign}(\widetilde \beta_{j}) \text{, if } \widetilde \beta_{j} \neq 0,\\
        \xi \in [-1, 1] \text{, otherwise}. 
    \end{cases}
\end{align*}
Let $\Delta = \widehat \beta - \widetilde \beta$. By the definition of $\hat \xi$ and $\tilde \xi$, it is clear that
\begin{align*}
   \p{ \hat \xi - \tilde \xi}^\top \Delta &=\norm{\widehat \beta}_1 + \norm{\widetilde \beta}_1 - \hat \xi^\top\widetilde \beta - \tilde \xi^\top \widehat \beta\geq \norm{\widehat \beta}_1 + \norm{\widetilde \beta}_1  - \norm{\widetilde \beta}_1 - \norm{\widehat \beta}_1 = 0.
\end{align*}
By $\eqref{grad for betahat}$ and $ \eqref{grad for betatilde}$ and Taylor's theorem, we have
\begin{align*}
    0 &\geq -\lambda_\beta\p{ \hat \xi - \tilde \xi}^\top \Delta\\
    &= \p{\nabla_{\beta} \bar l_2(\widehat \alpha, \widehat \beta)
    - \nabla_{\beta} \bar l_2(\alpha^*, \widehat \beta)
    + \nabla_{\beta} \bar l_2(\alpha^*, \widehat \beta)
    - \nabla_{\beta} \bar l_2(\alpha^*, \widetilde \beta)}^\top \Delta\\
&=\p{\nabla_{\beta} \bar l_2(\widehat \alpha, \widehat \beta) - \nabla_{\beta} \bar l_2(\alpha^*, \widehat \beta)}^\top \Delta
 + \Delta^\top\nabla_{\beta}^2 \bar l_2(\alpha^*, \bar \beta)\Delta\\
&=\p{\nabla_{\beta} \bar l_2(\widehat \alpha, \widehat \beta) - \nabla_{\beta} \bar l_2(\alpha^*, \widehat \beta)}^\top \Delta
 + 2\delta \bar l_2(\alpha^*, \widetilde \beta, \Delta).
\end{align*}
Let $\widehat \Delta_\alpha = \widehat \alpha - \alpha^*$ and $\widetilde \Delta_\beta = \widetilde \beta - \beta^*$. Let $w_{i} = Y_i - X_i^\top \beta^*$. By Taylor's theorem, for some $t \in (0,1)$ and $\bar \alpha = t\widehat \alpha + (1-t)\alpha^*$,
\begin{align*}
    &\delta \bar l_2(\alpha^*, \widetilde \beta, \Delta)\leq -\frac{1}{2}\p{\nabla_{\beta} \bar l_2(\widehat \alpha, \widehat \beta) - \nabla_{\beta} \bar l_2(\alpha^*, \widehat \beta)}^\top \Delta\\
    &\qquad= -\frac{1}{2} \p{\nabla_{\alpha}\nabla_{\beta} \bar l_2(\alpha^*, \widehat \beta)^\top\widehat \Delta_\alpha + \nabla_{\alpha}^2\nabla_{\beta} \bar l_2(\bar \alpha, \widehat \beta) \widehat \Delta_\alpha}^\top \Delta=A_1 - A_2 - A_3,
\end{align*}
where we use $Y_i - X_i^\top \widehat \beta = w_i - X_i^\top \widetilde \Delta_\beta -X_i^\top \Delta$, and define
\begin{align*}
    A_{1} &= M^{-1} \sum_{i \in \mathcal{I}_{-k,\beta}} \Gamma_i \exp\p{-X_i^\top \alpha^*}\p{X_i^\top \widehat \Delta_\alpha}w_i\p{X_i^\top \Delta},\\
    A_2 &= M^{-1} \sum_{i \in \mathcal{I}_{-k,\beta}} \Gamma_i \exp\p{-X_i^\top \alpha^*}\p{X_i^\top \widehat \Delta_\alpha}\p{X_i^\top \widetilde \Delta_\beta + X_i^\top \Delta}\p{X_i^\top \Delta},\\
    A_3 &= M^{-1} \sum_{i \in \mathcal{I}_{-k,\beta}} \Gamma_i \exp\p{-X_i^\top \bar \alpha}\p{X_i^\top \widehat \Delta_\alpha}^2\p{w_i - X_i^\top \widetilde \Delta_\beta -X_i^\top \Delta}\p{X_i^\top \Delta}.
\end{align*}

For $A_1$, by H{\"o}lder inequality and \eqref{exp(-X'alpha = gamma-1)},
\begin{align*}
    A_1 &= M^{-1} \sum_{i \in \mathcal{I}_{-k,\beta}} \Gamma_i \exp\p{-X_i^\top \alpha^*}\p{X_i^\top \widehat \Delta_\alpha}w_i\p{X_i^\top \Delta}\\
    &\leq \br{M^{-1} \sum_{i \in \mathcal{I}_{-k,\beta}} \Gamma_i \exp\p{-X_i^\top \alpha^*}\p{X_i^\top \widehat \Delta_\alpha}^4M^{-1} \sum_{i \in \mathcal{I}_{-k,\beta}} \Gamma_i \exp\p{-X_i^\top \alpha^*}w_i^4}^{1/4}\p{\delta \bar l_2(\alpha^*, \widetilde \beta, \Delta)}^{1/2}\\
    &\leq k_0^{-1/2}\gamma_n^{-1/2}\br{M^{-1} \sum_{i \in \mathcal{I}_{-k,\beta}} \Gamma_i \p{X_i^\top \widehat \Delta_\alpha}^4}^{1/4}\br{M^{-1} \sum_{i \in \mathcal{I}_{-k,\beta}} \Gamma_i w_i^4}^{1/4}\p{\delta \bar l_2(\alpha^*, \widetilde \beta, \Delta)}^{1/2}.
\end{align*}
By Lemma~\ref{lemma conditional independence}, $\widehat \Delta_\alpha \perp \Gamma_i X_i \mid \Gamma_i$ for $i \in  \mathcal{I}_{-k,\beta}$. Then by Assumption~\ref{assumption mar nuisance (a)}(b) and Lemma~\ref{sub-gaussian properties}(c), for some constant $K_1>0$,
\begin{align*}
    \E\sbr{M^{-1} \sum_{i \in \mathcal{I}_{-k,\beta}} \Gamma_i \p{X_i^\top \widehat \Delta_\alpha}^4\mid \Gamma_{\mc{I}_{-k,\beta}}, \widehat\alpha} \leq  K_1\norm{\widehat \Delta_\alpha}^4M^{-1} \sum_{i \in \mathcal{I}_{-k,\beta}} \Gamma_i.
\end{align*}
By Lemma~\ref{lemma concentrate gamma}, with probability at least $1-0.01M\gamma_n$,
$
    M^{-1} \sum_{i \in \mathcal{I}_{-k,\beta}} \Gamma_i \leq 1.21 \gamma_n.
$
Similar to Lemma~\ref{lemma Exp(Xn2) ineq}, we can show that with probability at least $1-t^{-1}-0.01M\gamma_n$,
\begin{align}
    M^{-1} \sum_{i \in \mathcal{I}_{-k,\beta}} \Gamma_i \p{X_i^\top \widehat \Delta_\alpha}^4 \leq 1.21 tK_1\gamma_n\norm{\widehat \Delta_\alpha}^4. \label{kkkk1}
\end{align}
Note that $\norm{w_i}_{\psi_2} \leq \sigma_w$. Similarly, for some $K_2>0$, we have with probability at least $1-t^{-1}-0.01M\gamma_n$,
$
    M^{-1} \sum_{i \in \mathcal{I}_{-k,\beta}} \Gamma_i w_i^4 \leq tK_2\gamma_n.
$
Thus, with probability at least $1-2t^{-1}-0.02M\gamma_n$, for some constant $K_3$
\begin{align*}
    A_1 \leq K_3t^{1/2}\norm{\widehat \Delta_\alpha}\p{\delta \bar l_2(\alpha^*, \widetilde \beta, \Delta)}^{1/2}.
\end{align*}

For $A_2$, by \eqref{exp(-X'alpha = gamma-1)} and H{\"o}lder inequality,
\begin{align*}
    A_2 &\leq k_0^{-1/2}\gamma_n^{-1/2} M^{-1} \sum_{i \in \mathcal{I}_{-k,\beta}} \Gamma_i \exp\p{-\frac{1}{2}X_i^\top \alpha^*}\p{X_i^\top \widehat \Delta_\alpha}\p{X_i^\top \widetilde \Delta_\beta + X_i^\top \Delta}\p{X_i^\top \Delta}\\
    &\leq k_0^{-1/2}\gamma_n^{-1/2}\br{M^{-1} \sum_{i \in \mathcal{I}_{-k,\beta}} \Gamma_i \p{X_i^\top \widehat \Delta_\alpha}^{4}M^{-1} \sum_{i \in \mathcal{I}_{-k,\beta}} \Gamma_i\p{X_i^\top \widetilde \Delta_\beta + X_i^\top \Delta}^{4}}^{1/4}
     \p{\delta \bar l_2(\alpha^*, \widetilde \beta, \Delta)}^{1/2}\\
     &\leq k_0^{-1/2}\gamma_n^{-1/2}\br{M^{-1} \sum_{i \in \mathcal{I}_{-k,\beta}} \Gamma_i \p{X_i^\top \widehat \Delta_\alpha}^{4}}^{1/4}\\
    &\quad \times\sbr{\br{M^{-1} \sum_{i \in \mathcal{I}_{-k,\beta}} \Gamma_i\p{X_i^\top \widetilde \Delta_\beta}^{4}}^{1/4} + \br{M^{-1} \sum_{i \in \mathcal{I}_{-k,\beta}} \Gamma_i\p{X_i^\top \Delta}^{4}}^{1/4}}
     \p{\delta \bar l_2(\alpha^*, \widetilde \beta, \Delta)}^{1/2}.
\end{align*}
By Lemma E.6(d) of \cite{zhang2023semi}, with probability at least $1-4\exp(-t)$, for some $K_3>0$,
\begin{align*}
    &M^{-1} \sum_{i \in \mathcal{I}_{-k,\beta}} \Gamma_i\p{X_i^\top \widetilde \Delta_\beta}^{4} \leq K_3\gamma_n\p{1 + (\frac{s_\beta \log d}{M\gamma_n})^{1/2} + \frac{(s_\beta\log M \log d)^2}{M\gamma_n}}\p{s_\beta^{-2}\norm{\widetilde \Delta_\beta}_1^4 + \norm{\widetilde \Delta_\beta}_2^4},\\
    &M^{-1} \sum_{i \in \mathcal{I}_{-k,\beta}} \Gamma_i\p{X_i^\top \Delta}^{4} \leq K_3\gamma_n\p{1 + (\frac{s_{\tilde \beta} \log d}{M\gamma_n})^{1/2} + \frac{(s_{\tilde \beta}\log M \log d)^2}{M\gamma_n}}\p{s_{\tilde \beta}^{-2}\norm{\Delta}_1^4 + \norm{\Delta}_2^4}.
\end{align*}
Let
\[
a_n = 1 + ({s_\beta \log d}/{M\gamma_n})^{1/2} + {(s_\beta\log M \log d)^2}/{M\gamma_n}
\]
and
\[
\tilde a_n = 1 + ({s_{\tilde \beta} \log d}/{M\gamma_n})^{1/2} + {(s_{\tilde \beta}\log M \log d)^2}/{M\gamma_n}.
\]
Together with \eqref{kkkk1}, we have for some $K_4>0$,  with probability at least $1-t^{-1}-4\exp(-t)$,
\begin{align*}
    A_2 \leq K_4 \p{a_n\p{s_\beta^{-2}\norm{\widetilde \Delta_\beta}_1^4 + \norm{\widetilde \Delta_\beta}_2^4} + \tilde a_n\p{s_{\tilde \beta}^{-2}\norm{\Delta}_1^4 + \norm{\Delta}_2^4}}^{1/4}\p{\delta \bar l_2(\alpha^*, \widetilde \beta, \Delta)}^{1/2}.
\end{align*}

For $A_3$, by H{\"o}lder inequality and Minkowski's inequality, we have
\begin{align*}
    A_3 &\leq \br{M^{-1} \sum_{i \in \mathcal{I}_{-k,\beta}} \Gamma_i \frac{\exp\p{-2X_i^\top \bar \alpha}}{\exp\p{-X_i^\top \alpha^*}}\p{X_i^\top \widehat \Delta_\alpha}^4\p{w_i - X_i^\top \widetilde \Delta_\beta -X_i^\top \Delta}^2}^{1/2}\p{\delta \bar l_2(\alpha^*, \widetilde \beta, \Delta)}^{1/2}\\
    &\leq \br{P_1^{1/2} + P_2^{1/2} + P_3^{1/2}}\p{\delta \bar l_2(\alpha^*, \widetilde \beta, \Delta)}^{1/2},
\end{align*}
where
\begin{align*}
    P_1 &= M^{-1} \sum_{i \in \mathcal{I}_{-k,\beta}} \Gamma_i \frac{\exp\p{-2X_i^\top \bar \alpha}}{\exp\p{-X_i^\top \alpha^*}}\p{X_i^\top \widehat \Delta_\alpha}^4 w_i^2,\\
    P_2 &= M^{-1} \sum_{i \in \mathcal{I}_{-k,\beta}} \Gamma_i \frac{\exp\p{-2X_i^\top \bar \alpha}}{\exp\p{-X_i^\top \alpha^*}}\p{X_i^\top \widehat \Delta_\alpha}^4 \p{X_i^\top \widetilde \Delta_\beta}^2,\\
    P_3 &= M^{-1} \sum_{i \in \mathcal{I}_{-k,\beta}} \Gamma_i \frac{\exp\p{-2X_i^\top \bar \alpha}}{\exp\p{-X_i^\top \alpha^*}}\p{X_i^\top \widehat \Delta_\alpha}^4 \p{X_i^\top \Delta}^2.
\end{align*}

For $P_1$, by H{\"o}lder inequality, we have
\begin{align*}
    P_1 \leq \p{M^{-1} \sum_{i \in \mathcal{I}_{-k,\beta}} \Gamma_i \frac{\exp\p{-4X_i^\top \bar \alpha}}{\exp\p{-3X_i^\top \alpha^*}}\p{X_i^\top \widehat \Delta_\alpha}^8}^{1/2}\p{M^{-1} \sum_{i \in \mathcal{I}_{-k,\beta}} \Gamma_i \exp\p{-X_i^\top \alpha^*}w_i^4}^{1/2}.
\end{align*}
By \eqref{exp(-X'alpha = gamma-1)} and Lemma~\ref{lemma Exp(Xn2) ineq}, for any $0<t < 0.01M\gamma_n/\log M$, with probability at least $1- t^{-1} - c_1 \exp\p{-c_2M\gamma_n} - 15\exp\p{-t}$, we have for some $c >0$,
\begin{align*}
    &M^{-1} \sum_{i \in \mathcal{I}_{-k,\beta}} \Gamma_i \frac{\exp\p{-4X_i^\top \bar \alpha}}{\exp\p{-3X_i^\top \alpha^*}}\p{X_i^\top \widehat \Delta_\alpha}^8= M^{-1} \sum_{i \in \mathcal{I}_{-k,\beta}} \Gamma_i \exp\p{-4tX_i^\top \widehat \Delta_\alpha - X_i^\top \alpha^*}\p{X_i^\top \widehat \Delta_\alpha}^8\\
    &\qquad\leq k_0^{-1}\gamma_n^{-1} M^{-1} \sum_{i \in \mathcal{I}_{-k,\beta}} \Gamma_i \exp\p{4\abs{X_i^\top \widehat \Delta_\alpha}}\p{X_i^\top \widehat \Delta_\alpha}^8\leq tck_0^{-1}\norm{\widehat \Delta_\alpha}_2^8.
\end{align*}
Together with Lemma~\ref{lemma constant concentration of w}, it follows that for any $0<t < 0.01M\gamma_n/\log M$, with probability at least $1- 2t^{-1} - c_1 \exp\p{-c_2M\gamma_n} - 15\exp\p{-t}$, for some constant $C_2>0$,
$
    P_1 \leq t C_2\norm{\widehat \Delta_\alpha}_2^4.
$

For $P_2$, by H{\"o}lder inequality, for some constant $r>0$ we have
\begin{align*}
    P_2 &\leq \p{M^{-1} \sum_{i \in \mathcal{I}_{-k,\beta}} \Gamma_i \frac{\exp\p{-4X_i^\top \bar \alpha}}{\exp\p{-2X_i^\top \alpha^*}}\p{X_i^\top \widehat \Delta_\alpha}^{8} }^{1/2}\p{M^{-1} \sum_{i \in \mathcal{I}_{-k,\beta}} \Gamma_i \p{X_i^\top \widetilde \Delta_\beta}^{4}}^{1/2}.
\end{align*}
By Lemma E.6(d) of \cite{zhang2023semi}, with probability at least $1-2\exp(-t)$, for some constant $c_3>0$, it holds that
\begin{align*}
    &M^{-1}\sum_{i \in \mathcal{I}_{-k,\beta}}\Gamma_i\p{X_i^\top \widetilde \Delta_\beta}^{4} \leq c_3\gamma_n\p{1 + (\frac{s_\beta \log d}{M\gamma_n})^{1/2} + \frac{(s_\beta \log M\log d )^{2}}{M\gamma_n}}\p{\norm{s_\beta^{-1/2}\widetilde \Delta_\beta}_1^{4} + \norm{\widetilde \Delta_\beta}_2^{4}}\\
    &\qquad\leq c_3\gamma_n\p{1 + (\frac{s_\beta \log d}{M\gamma_n})^{1/2} + \frac{(s_\beta \log M\log d )^{2}}{M\gamma_n}}\p{\norm{s_\beta^{-1/2}\widetilde \Delta_\beta}_1^{2} + \norm{\widetilde \Delta_\beta}_2^{2}}^{2}
\end{align*}

By Lemma~\ref{lemma Exp(Xn2) ineq}, with probability at least $1- 2t^{-1} - c_1 \exp\p{-c_2M\gamma_n} - 15\exp\p{-t}$, for some $c_r>0$,
\begin{align*}
    &M^{-1} \sum_{i \in \mathcal{I}_{-k,\beta}} \Gamma_i \frac{\exp\p{-4X_i^\top \bar \alpha}}{\exp\p{-2X_i^\top \alpha^*}}\p{X_i^\top \widehat \Delta_\alpha}^{8} = M^{-1} \sum_{i \in \mathcal{I}_{-k,\beta}} \Gamma_i \exp\p{-4tX_i^\top \widehat \Delta_\alpha-2X_i^\top \alpha^*}\p{X_i^\top \widehat \Delta_\alpha}^{8}\\
    &\qquad\leq k_0^{-2}\gamma_n^{-2} M^{-1} \sum_{i \in \mathcal{I}_{-k,\beta}} \Gamma_i \exp\p{4\abs{X_i^\top \widehat \Delta_\alpha}}\p{X_i^\top \widehat \Delta_\alpha}^{8}= k_0^{-2}\gamma_n^{-1}c_rt\norm{\widehat \Delta_\alpha}_2^{8}.
\end{align*}
Then for any $0<t < 0.01M\gamma_n/\log M$, with probability at least $1- 2t^{-1} - c_1 \exp\p{-c_2M\gamma_n} - 17\exp\p{-t}$, for some constant $C_r>0$, 
\begin{align*}
    P_2 \leq tC_r a_n^{1/2}\norm{\widehat \Delta_\alpha}_2^4\p{\norm{s_\beta^{-1/2}\widetilde \Delta_\beta}_1^{2} + \norm{\widetilde \Delta_\beta}_2^{2}}.
\end{align*}

For $P_3$, by identical analysis to $P_2$, we have that for any $0<t < 0.01M\gamma_n/\log M$, with probability at least $1- 2t^{-1} - c_1 \exp\p{-c_2M\gamma_n} - 17\exp\p{-t}$, for some constant $C_r>0$, 
\begin{align*}
    P_3 \leq tC_r \tilde a_n^{1/2}\norm{\widehat \Delta_\alpha}_2^4\p{\norm{s_{\tilde\beta}^{-1/2}\Delta}_1^{2} + \norm{\Delta}_2^{2}}.
\end{align*}
It follows that, with probability at least $1- 6t^{-1} - 3c_1 \exp\p{-c_2M\gamma_n} - 49\exp\p{-t}$, for some constant $K_5>0$,
\begin{align*}
    A_3 &\leq K_5t^{1/2}\norm{\widehat \Delta_\alpha}_2^2
    \p{1 + a_n^{1/4}(\norm{s_\beta^{-1/2}\widetilde \Delta_\beta}_1^{2} + \norm{\widetilde \Delta_\beta}_2^{2})^{1/2}}\p{\delta \bar l_2(\alpha^*, \widetilde \beta, \Delta)}^{1/2}\\
    &\quad + K_5t^{1/2}\norm{\widehat \Delta_\alpha}_2^2
    \tilde a_n^{1/4}(\norm{s_{\tilde\beta}^{-1/2}\Delta}_1^{2} + \norm{\Delta}_2^{2})^{1/2}
    \p{\delta \bar l_2(\alpha^*, \widetilde \beta, \Delta)}^{1/2}.
\end{align*}

Together $A_1, A_2, A_3$, we have
\begin{align*}
    \delta \bar l_2(\alpha^*, \widetilde \beta, \Delta)^{1/2}
    &\leq K_3t^{1/2}\norm{\widehat \Delta_\alpha}\\
    &\quad + K_4 \p{a_n\p{s_\beta^{-2}\norm{\widetilde \Delta_\beta}_1^4 + \norm{\widetilde \Delta_\beta}_2^4}
    + \tilde a_n\p{s_{\tilde \beta}^{-2}\norm{\Delta}_1^4 + \norm{\Delta}_2^4}}^{1/4}\\
    &\quad + K_5t^{1/2}\norm{\widehat \Delta_\alpha}_2^2\p{1 + a_n^{1/4}(\norm{s_\beta^{-1/2}\widetilde \Delta_\beta}_1^{2} + \norm{\widetilde \Delta_\beta}_2^{2})^{1/2}}\\
    &\quad + K_5t^{1/2}\norm{\widehat \Delta_\alpha}_2^2\tilde a_n^{1/4}(\norm{s_{\tilde\beta}^{-1/2}\Delta}_1^{2} + \norm{\Delta}_2^{2})^{1/2},
\end{align*}
which implies that for some constant $K_6>0$,
\begin{align*}
    K_6\delta \bar l_2(\alpha^*, \widetilde \beta, \Delta)
    &\leq t\norm{\widehat \Delta_\alpha}^2_2 + t\norm{\widehat \Delta_\alpha}^4_2\\
    &\quad + a_n^{1/2}\norm{\widehat \Delta_\alpha}^4_2\p{s_\beta^{-1}\norm{\widetilde \Delta_\beta}_1^2 + \norm{\widetilde \Delta_\beta}_2^2}\\
    &\quad + \tilde a_n^{1/2}\norm{\widehat \Delta_\alpha}^4_2\p{s_{\tilde \beta}^{-1}\norm{\Delta}_1^2 + \norm{\Delta}_2^2 }.
\end{align*}
In addition, by Lemma~\ref{lemma mar delta l2 concentration}, with probability at least $1 - c_1' \exp\p{-c_2'M\gamma_n}$,
\[
    \delta \bar l_2\p{\alpha^*, \widetilde \beta, \Delta} \geq \kappa_1' \norm{\Delta}_2^2 -  \kappa_2' {\log d}/{M\gamma_n} \norm{\Delta}_1^2.
\]
By Lemma~\ref{lemma sparsity for norm1 and norm2}, if
\[
\lambda_\beta > \max\br{2\norm{\nabla_\beta \bar l_2(\alpha^*, \beta^*)}_\infty, 2\|\nabla_\beta \bar l_2(\widehat \alpha, \widetilde \beta)\|_\infty},
\]
we have
$s_\beta^{-1}\norm{\widetilde \Delta_\beta}_1^2 \leq 16\norm{\widetilde \Delta_\beta}_2^2$ and
$s_{\tilde \beta}^{-1}\norm{\Delta}_1^2 \leq 16\norm{\Delta}_2^2.$
Then for any $0<t < 0.01M\gamma_n/\log M$, with probability at least $1- c_1't^{-1} - c_2' \exp\p{-c_3'M\gamma_n} - c_4'\exp\p{-t}$, with some $K_7, K_8, K_9>0$,
\begin{align*}
K_7\norm{\Delta}_2^2 - K_8\frac{ s_{\tilde \beta}\log d}{M\gamma_n} \norm{\Delta}_2^2 &\leq t\norm{\widehat \Delta_\alpha}_2^2 + t\norm{\widehat \Delta_\alpha}_2^4 + K_9a_n^{1/2}\norm{\widehat \Delta_\alpha}_2^4\norm{\widetilde \Delta_\beta}_2^2+ K_9\tilde a_n^{1/2}\norm{\widehat \Delta_\alpha}_2^4\norm{\Delta}_2^2.
\end{align*}
Let 
$ b_n = K_7 - K_8{ s_{\tilde \beta}\log d}/{M\gamma_n} - K_9\tilde a_n^{1/2}\norm{\widehat \Delta_\alpha}_2^4$ and 
$c_n = t + t\norm{\widehat \Delta_\alpha}_2^2 + K_9a_n^{1/2}\norm{\widehat \Delta_\alpha}_2^2\norm{\widetilde \Delta_\beta}_2^2,$
then
$
    b_n\norm{\Delta}_2^2 \leq c_n\norm{\widehat \Delta_\alpha}_2^2.
$
If $\lambda_\alpha \asymp \lambda_\beta \asymp (\log d/(n\gamma_n))^{1/2}$ and $n\gamma_n \gg \max\br{s_\alpha, s_\beta\log n , (\log n)^2} \log d$, by Lemma~\ref{lemma consistency for alpha and betatilde} and Lemma~\ref{lemma sparsity betatilde -> beta}, as $n,d \rightarrow \infty$, it follows that $s_{\tilde \beta} = O_p\p{s_\beta}$, 
\begin{align*}
    a_n^{1/2}\norm{\widehat \Delta_\alpha}_2^{2}\norm{\widetilde \Delta_\beta}_2^{2} 
 &= O_p\p{\p{1  + \frac{s_\beta \log n\log d}{(n\gamma_n)^{1/2}}}\frac{s_\alpha s_\beta (\log d)^2}{(n\gamma_n)^2}}=O_p\p{\frac{s_\alpha s_\beta (\log d)^2}{(n\gamma_n)^2} + \frac{s_\alpha s_\beta^2 \log n(\log d)^3}{(n\gamma_n)^{5/2}}},
\end{align*}
and
\begin{align*}
    \tilde a_n^{1/2}\norm{\widehat \Delta_\alpha}_2^4 = O_p\p{\frac{s_\alpha^2 (\log d)^2}{(n\gamma_n)^2} + \frac{s_\alpha^2 s_\beta \log n(\log d)^3}{(n\gamma_n)^{5/2}}}.
\end{align*}
Since $s_\alpha \ll n\gamma_n/\log d$ and $s_\beta \ll n\gamma_n/(\log n\log d)$, 
\begin{align*}
    b_n - K_7 &= o_p(1) + o_p\p{\frac{s_\alpha s_\beta (\log d)^2}{(n\gamma_n)^{3/2}}},\quad c_n = o_p(1) + O_p\p{1 + \frac{s_\alpha s_\beta \log n(\log d)^2}{(n\gamma_n)^{3/2}}}.
\end{align*}
In addition, suppose that $s_\alpha s_\beta \ll (n\gamma_n)^{3/2}/(\log n(\log d)^2)$, then 
$b_n - K_7 = o_p(1)$ and $c_n = O_p(1),$
which implies that
$ \norm{\Delta}_2 = O_p\p{\norm{\widehat \Delta_\alpha}_2}=O_p\p{(s_\alpha \log d/(n\gamma_n))^{1/2}}.$
Finally, by Lemma~\ref{lemma sparsity for norm1 and norm2} and Lemma~\ref{lemma sparsity betatilde -> beta},
$
     \norm{\Delta}_1 = O_p\p{(s_\alpha s_\beta \log d/(n\gamma_n))^{1/2}}.
$
\end{proof}

\subsection{Proof of Proposition~\ref{proposition mar nuisance consistency body}}

\begin{proof}
    Proposition~\ref{proposition mar nuisance consistency body} follows directly from Lemma~\ref{lemma consistency for alpha and betatilde} and Lemma~\ref{lemma consistency of betahat and betatilde for product sparsity}.

\end{proof}

\section{Proof of results in Section~\ref{sec: mcar}}

\subsection{Auxiliary lemmas for Theorem~\ref{theorem for the asymptotics under MCAR}}
\begin{lemma}\label{lemma gamma ratio convergence}
Suppose that $\Gamma_i \overset{i.i.d.}{\sim} Bernoulli(\gamma_n)$. It follows that
$$
\frac{\gamma_n}{n^{-1}\sum_{i=1}^n \Gamma_i}-1=O_p\left((n\gamma_n)^{-1 / 2}\p{1-\gamma_n}^{1/2}\right) .
$$
\end{lemma}
\begin{proof}
Let $\widehat \gamma = n^{-1}\sum_{i=1}^n \Gamma_i$. Then
$$
\E\left[\left\{\frac{\widehat \gamma-\gamma_n}{\gamma_n}\right\}^2\right]=\gamma_n^{-2} n^{-1} \E\left[\left(\Gamma_i-\gamma_n\right)^2\right]=n^{-1} \gamma_n^{-1}\left(1-\gamma_n\right)=O\left(\left(n \gamma_n\right)^{-1}\p{1-\gamma_n}\right) .
$$
by Chebyshev's inequality, we have
\[
{\widehat \gamma-\gamma_n}/{\gamma_n}=O_p\left((n \gamma_n)^{-1 / 2}\p{1-\gamma_n}^{1/2}\right) .
\]
By the fact that
\[
{\widehat \gamma-\gamma_n}/{\widehat \gamma}\left\{1+{\widehat \gamma-\gamma_n}/{\gamma_n}\right\}={\widehat \gamma-\gamma_n}/{\gamma_n},
\]
we have
$$
-\frac{\widehat \gamma-\gamma_n}{\widehat \gamma }=-\left\{1+\frac{\widehat \gamma-\gamma_n}{\gamma_n}\right\}^{-1} \frac{\widehat \gamma-\gamma_n}{\gamma_n}=O_p\left((n \gamma_n)^{-1 / 2}\p{1-\gamma_n}^{1/2} \right) .
$$
\end{proof}

\begin{lemma} \label{lemma the oracle asymptotics under MCAR}
Let Assumptions~\ref{assumption mcar} and \ref{assumption mcar lasso (a)} hold. Then $\sigma_n^2 \asymp \gamma_n^{-1}$ and
$
\sigma_n^{-1}n^{1/2}(\widetilde \theta - \theta) \xrightarrow{d} \mathcal{N}(0,1),
$
where $\widetilde \theta = n^{-1}\sum_{i=1}^n \br{X_i^\top \beta^*  + {\Gamma_i}/{\gamma_n}\p{Y_i - X_i^\top \beta^*}}$.
\end{lemma}

\begin{proof}
    Let $U_i = X_i^\top  \beta^* + {\Gamma_i}/{\gamma_n}\p{Y_i - X_i^\top \beta^*}$. 
    Then under Assumption~\ref{assumption mcar},
    \begin{align*}
        \E[U_i] &= \E\sbr{X_i^\top \beta^*} + \E\left[\frac{\Gamma_i}{\gamma_n}\right]\E[\p{Y_i - X_i^\top \beta^*}] = \E[Y_i] = \theta.
    \end{align*}
Observe that
\begin{align*}
    \sigma_n^2 &= \Var[U_i] = \E\left[ \p{X_i^\top \beta^*  + \frac{\Gamma_i}{\gamma_n}w_i - \theta}^2\right]=\E\left[ \p{X_i^\top \beta^* - \theta}^2 + \frac{1}{\gamma_n}w_i^2\right],
\end{align*}
where the last equation comes from Lemma~\ref{lemma least squared residual properties}. Thus, by Lemma~\ref{sub-gaussian properties} (c), we have
\begin{align}
    \sigma_n^2 \asymp \gamma_n^{-1}. \label{shrinking rate of sigma_n under MCAR}
\end{align}
In addition, by Lemma~\ref{sub-gaussian properties} (c), there exist some constants $C_1, C_2 > 0$ such that 
$
\E\left[\abs{X_i^\top \beta^*}^{2+c}\right] \leq \p{C_1\sigma}^{2+c}$ and $\E\left[\abs{w_i}^{2+c}\right] \leq \p{C_2\sigma_w}^{2+c}.
$
Then,
\begin{align*}
    \norm{U_i - \theta}_{\P, 2+c} &= \norm{X_i^\top \beta^*- \theta  + \frac{\Gamma_i}{\gamma_n}w_i }_{\P, 2+c}\leq \norm{X_i^\top \beta^*}_{\P, 2+c} + \theta + \E\left[ \frac{\Gamma_i}{\gamma_n^{2+c}}w_i^{2+c}\right]^{\frac{1}{2+c}}\\
    &=\E\left[\abs{X_i^\top \beta^*}^{2+c}\right]^{\frac{1}{2+c}} + \theta + \gamma_n^{\frac{1}{2+c}-1}\E\left[\abs{w_i}^{2+c}\right]^{\frac{1}{2+c}}\leq C_1\sigma + \theta + C_2\sigma_w\gamma_n^{\frac{1}{2+c}-1},
\end{align*}
which implies 
\begin{align*}
    n^{-\frac{c}{2}}\gamma_n^{1+\frac{c}{2}}\E\left[ \abs{U_i - \theta}^{2+c} \right] &\leq n^{-\frac{c}{2}}\gamma_n^{1+\frac{c}{2}}\p{C_1\sigma + \theta + C_2\sigma_w\gamma_n^{\frac{1}{2+c}-1}}^{2+c} = n^{-\frac{c}{2}}\gamma_n^{1+\frac{c}{2}}O\p{\gamma_n^{-1-c}}. 
\end{align*}
That is,
\begin{align}
    n^{-\frac{c}{2}}\gamma_n^{1+\frac{c}{2}}\E\left[ \abs{U_i - \theta}^{2+c} \right] = O\p{(n\gamma_n)^{-\frac{c}{2}}} = o(1). \label{shrinking rate of Lyapunov condition}
\end{align}
Then for any $\delta>0$, as $N, d \rightarrow \infty$,
\begin{align*}
    \gamma_n\E\left[\p{U_i - \theta}^{2}\mathbbm{1}_{\br{\abs{U_i - \theta}>\delta(n/\gamma_n)^{1/2}}}\right] \leq \delta^{-c}n^{-\frac{c}{2}}\gamma_n^{1+\frac{c}{2}}\E\left[\p{U_i - \theta}^{2+c}\right] = o(1).
\end{align*}
By Proposition 2.27 (Lindeberg-Feller) of \cite{van2000asymptotic}, we have
$
        \sigma_n^{-1}n^{1/2}(\widetilde \theta - \theta) \xrightarrow{d} \mathcal{N}(0,1).
$
\end{proof}

\begin{lemma} \label{lemma the mean estimator consistency under MCAR}
    Let Assumptions~\ref{assumption mcar} and \ref{assumption mcar lasso (a)} hold. 
   if $n\gamma_n \gg \p{s\vee \p{\log n}^2}\log d$ and choose some 
$\lambda_n \asymp (\log d/(n\gamma_n))^{1/2}$, then as $n, d \rightarrow \infty$,
    \begin{align*}
        \widehat \theta -\widetilde \theta &= O_p\p{(\frac{s\log d}{n^2\gamma_n^2})^{1/2}},\quad\widehat \theta -\theta = O_p \p{(n\gamma_n)^{-1/2}}.
    \end{align*}
\end{lemma}

\begin{proof}
Note that $\widetilde \theta$ can be written as
\begin{align*}
        \widetilde \theta =n^{-1}\sum_{i=1}^n \br{\p{\Gamma_i X_i + \p{1-\Gamma_i }\Bar{X}_0 }^\top  \beta^*  +  \frac{\Gamma_i}{\gamma_n}\p{Y_i - X_i^\top \beta^*}}.
\end{align*}
Let $\widehat \Delta^{(-k)} = \widehat \beta^{(-k)} - \beta^*$ and $\widehat \gamma^{(k)} = n_k^{-1}\sum_{i\in \mathcal{I}_k} \Gamma_i$. Observe that $\widehat \theta$ can be decomposed as
    \begin{align*}
        \widehat \theta = \widetilde \theta + I_1 + I_2 + I_3 + I_4,
    \end{align*}
    where
    \begin{align*}
        I_1 &= K^{-1}\sum_{k=1}^K n_k^{-1}\sum_{i\in \mathcal{I}_k} \p{\Gamma_i - \frac{\Gamma_i}{\gamma_n}} \p{X_i^\top \widehat \Delta^{(-k)} - \E\left[X_i\right]^\top \widehat \Delta^{(-k)}},\\
        I_2 &=  K^{-1}\sum_{k=1}^K \p{\frac{1}{\gamma_n}-\frac{1}{\widehat \gamma^{(k)}}} n_k^{-1}\sum_{i\in \mathcal{I}_k}  \Gamma_i \p{X_i^\top \widehat \Delta^{(-k)} - \E\left[X_i\right]^\top \widehat \Delta^{(-k)}},\\
        I_3 &= K^{-1}\sum_{k=1}^K\p{\frac{1}{\widehat \gamma^{(k)} } - \frac{1}{\gamma_n}} n_k^{-1}\sum_{i\in \mathcal{I}_k} \Gamma_iw_i,\quad
        I_4 = K^{-1}\sum_{k=1}^K n_k^{-1}\sum_{i\in \mathcal{I}_k} (1-\Gamma_i)\p{\Bar{X}_0^\top \widehat \Delta^{(-k)} - \E\left[X_i\right]^\top \widehat \Delta^{(-k)}}.
    \end{align*}

For $I_1$, let 
$$I_{1k} = n_k^{-1}\sum_{i\in \mathcal{I}_k}\p{\Gamma_i - \frac{\Gamma_i}{\gamma_n}} \p{X_i^\top \widehat \Delta^{(-k)} - \E\left[X_i\right]^\top \widehat \Delta^{(-k)}}.$$
Define $\mathcal{D}_{\mathcal{I}_{-k}} = \br{\p{\Gamma_i, \Gamma_iX_i, \Gamma_iY_i} \in \mathcal{D}_n \mid i \in \mathcal{I}_{-k}}$. Then
$$
\E[I_{1k}\mid \mathcal{D}_{\mathcal{I}_{-k}}] = n_k^{-1}\sum_{i\in \mathcal{I}_k}\E\left[\Gamma_i - \frac{\Gamma_i}{\gamma_n} \right]\E\left[X_i^\top \widehat \Delta^{(-k)} - \E\left[X_i\right]^\top \widehat \Delta^{(-k)}\mid \mathcal{D}_{\mathcal{I}_{-k}} \right] = 0,
$$
and
\begin{align*}
    &\E\left[I_{1k}^2\mid \mathcal{D}_{\mathcal{I}_{-k}}\right] = \E\left[\br{n_k^{-1}\sum_{i\in \mathcal{I}_k} \p{\Gamma_i - \frac{\Gamma_i}{\gamma_n}} \p{X_i^\top \widehat \Delta^{(-k)} - \E\left[X_i\right]^\top \widehat \Delta^{(-k)}}}^2 \mid \mathcal{D}_{\mathcal{I}_{-k}} \right]\\
    &\qquad= \E\left[n_k^{-2}\sum_{i\in \mathcal{I}_k} \p{\Gamma_i - \frac{\Gamma_i}{\gamma_n}}^2\p{X_i^\top \widehat \Delta^{(-k)} - \E\left[X_i\right]^\top \widehat \Delta^{(-k)}}^2 \mid \mathcal{D}_{\mathcal{I}_{-k}} \right]\\
    &\qquad\leq n_k^{-2}\sum_{i\in \mathcal{I}_k} \E\left[\p{\Gamma_i - \frac{\Gamma_i}{\gamma_n}}^2 \right]\E\left[ \p{X_i^\top \widehat \Delta^{(-k)}}^2 \mid \mathcal{D}_{\mathcal{I}_{-k}} \right]\\
    &\qquad=n_k^{-1}\E\left[\p{\Gamma_i - \frac{\Gamma_i}{\gamma_n}}^2 \right]\E\left[ \p{X_i^\top \widehat \Delta^{(-k)}}^2 \mid \mathcal{D}_{\mathcal{I}_{-k}} \right]=(1-\gamma_n)^2(n_k\gamma_n)^{-1}\widehat \Delta^{(-k),\top}\E\sbr{X_iX_i^\top}\widehat \Delta^{(-k)}
\end{align*}
By Lemma~\ref{sub-gaussian properties}(c) and Proposition~\ref{proposition lasso consistency body},
\begin{align}
    \Delta^{(-k),\top}\E\sbr{X_iX_i^\top}\widehat \Delta^{(-k)} \leq 2\sigma^2\norm{\widehat \Delta^{(-k)}}_2^2 = O_p\p{\frac{s\log d}{n\gamma_n}}. \label{lasso ineq bounding quadratic form expectation}
\end{align}
Then by Lemma~\ref{lemma convergence of conditional random variable},
\[
I_{1k} = O_p\p{(n\gamma_n)^{-1/2}(s\log d/(n\gamma_n))^{1/2}}.
\]
Thus,
\[
    I_1 = K^{-1}\sum_{k=1}^K I_{1k} = O_p\p{({s\log d}/{n^2\gamma_n^2})^{1/2}}.
\]

For $I_2$, let
\[
I_{2k} = n_k^{-1}\sum_{i\in \mathcal{I}_k}  \Gamma_i \p{X_i^\top \widehat \Delta^{(-k)} - \E\left[X_i\right]^\top \widehat \Delta^{(-k)}}.
\]
Then we have
\begin{align*}
    \E[I_{2k}\mid \mathcal{D}_{\mathcal{I}_{-k}}] =n_k^{-1}\sum_{i\in \mathcal{I}_k}  \E\left[\Gamma_i\right] \E\left[\p{X_i^\top \widehat \Delta^{(-k)} - \E\left[X_i\right]^\top \widehat \Delta^{(-k)}}\mid \mathcal{D}_{\mathcal{I}_{-k}}\right] = 0,
\end{align*}
and 
\begin{align*}
&\E\left[I_{2k}^2\mid \mathcal{D}_{\mathcal{I}_{-k}}\right] = n_k^{-2} \sum_{i\in \mathcal{I}_k}\E\left[\Gamma_i\right]\E\left[\p{X_i^\top \widehat \Delta^{(-k)} - \E\left[X_i\right]^\top \widehat \Delta^{(-k)}}^2\mid \mathcal{D}_{\mathcal{I}_{-k}}\right]\\
&\qquad= n_k^{-1}\gamma_n\br{\E\left[\p{X_i^\top \widehat \Delta^{(-k)}}^2 \mid \mathcal{D}_{\mathcal{I}_{-k}}\right] - \E\left[X_i^\top \widehat \Delta^{(-k)}\mid \mathcal{D}_{\mathcal{I}_{-k}}\right]^2}\\
&\qquad\leq n_k^{-1}\gamma_n\E\left[\p{X_i^\top \widehat \Delta^{(-k)}}^2 \mid \mathcal{D}_{\mathcal{I}_{-k}}\right]= n_k^{-1}\gamma_n\Delta^{(-k),\top}\E\sbr{X_iX_i^\top}\widehat \Delta^{(-k)}.
\end{align*}
Then by \eqref{lasso ineq  bounding quadratic form expectation}, and Lemma~\ref{lemma convergence of conditional random variable}, we have
\[
I_{2k} = O_p \p{n^{-1/2}\gamma_n^{1/2}(s\log d/(n\gamma_n))^{1/2}}.
\]
Together with Lemma~\ref{lemma gamma ratio convergence}, we have ${1}/{\gamma_n}-{1}/{\widehat \gamma^{(k)}} = O_p\p{n^{-1/2}\gamma_n^{-3/2}}$, then
\[
I_2 =  K^{-1}\sum_{k=1}^K \p{{1}/{\gamma_n}-{1}/{\widehat \gamma^{(k)}}} I_{2k} = O_p\p{({s\log d}/{n^3\gamma_n^3})^{1/2}}.
\]

For $I_3$, let $I_{3k} = n_k^{-1}\sum_{i\in \mathcal{I}_k} \Gamma_iw_i$. By Lemma~\ref{lemma least squared residual properties},
\begin{align*}
    \E[I_{3k}\mid \mathcal{D}_{\mathcal{I}_{-k}}] = n_k^{-1}\sum_{i\in \mathcal{I}_k}\E[\Gamma_iw_i\mid \mathcal{D}_{\mathcal{I}_{-k}}] = n_k^{-1}\sum_{i\in \mathcal{I}_k}\E[\Gamma_i] \E[w_i] = 0,
\end{align*}
and 
$
    \E[I_{3k}^2\mid \mathcal{D}_{\mathcal{I}_{-k}}] = n_k^{-2} \sum_{i\in \mathcal{I}_k} \E[\Gamma_iw_i^2] = n_k^{-1} \gamma_n \E[w_i^2].
$
By Lemma~\ref{sub-gaussian properties}(c), we have \begin{align}
    \E[w_i^2] \leq 2\sigma_w^2. \label{lasso ineq w_i^2 expectation}
\end{align} 
Thus, by Chebyshev's inequality,
$
I_{3k} = O_p \p{({\gamma_n}/{n})^{1/2}}.
$
Similarly, with Lemma~\ref{lemma gamma ratio convergence}, it follows that 
$$
I_3 = K^{-1}\sum_{k=1}^K\p{\frac{1}{\widehat \gamma^{(k)} } - \frac{1}{\gamma_n}} I_{3k} = O_p \p{(n\gamma_n)^{-1}}.
$$

For $I_4$, let $\widehat \gamma = n^{-1}\sum_{i =1}^n \Gamma_i$.
Then we have
\begin{align*}
\E\sbr{\p{\p{1-\widehat \gamma} - \p{1-\gamma_n } }^2 } &= \E\sbr{\p{\widehat \gamma - \gamma_n}^2} = \E\sbr{n^{-2}\sum_{i=1}^n \p{\Gamma_i - \gamma_n}^2} = n^{-1}\gamma_n\p{1-\gamma_n},
\end{align*}
which implies that 
\begin{align}
    \p{1-\widehat \gamma} - \p{1-\gamma_n } = O_p\p{n^{-1/2}\gamma_n^{1/2}\p{1-\gamma_n}^{1/2}}. \label{1-widehat gamma rate}
\end{align}
Similarly,
$
\p{1-\widehat \gamma^{(k)}} - \p{1-\gamma_n } = O_p\p{n^{-1/2}\gamma_n^{1/2}\p{1-\gamma_n}^{1/2}}.
$
Since Lemma~\ref{lemma gamma ratio convergence} can also be applied to $1-\Gamma_i$, we have
\begin{align*}
    \frac{1-\gamma_n}{1-\widehat \gamma^{(k)}} - 1 = O_p\p{\p{n\p{1-\gamma_n}}^{-1/2}\gamma_n^{1/2} }, \quad
    \frac{1-\gamma_n}{1-\widehat \gamma} - 1 = O_p\p{\p{n\p{1-\gamma_n}}^{-1/2}\gamma_n^{1/2} }.
\end{align*}
Then
\begin{align}
    \frac{1-\widehat \gamma^{(k)}}{1-\widehat \gamma} - 1 =& \p{1-\gamma_n}^{-1}\p{\frac{1-\gamma_n}{1-\widehat \gamma} - 1}\p{\p{1-\widehat \gamma^{(k)}} - \p{1-\gamma_n } } +\p{\frac{1-\gamma_n}{1-\widehat \gamma} - 1 } + \p{\frac{1-\widehat \gamma^{(k)}}{1-\gamma_n} - 1} \notag \\
    =& O_p\p{\gamma_n\p{n\p{1-\gamma_n}}^{-1}} + O_p\p{\gamma_n^{1/2}\p{n\p{1-\gamma_n}}^{-1/2}} \label{variance gamma(k) -1 ratio rate}.
\end{align}
Let
\begin{align*}
    I_{4k} &= \p{1-\widehat \gamma }\p{\Bar{X}_0^\top \widehat \Delta^{(-k)} - \E\left[X_i\right]^\top \widehat \Delta^{(-k)}} = n^{-1}\sum_{i=1}^n \p{1-\Gamma_i}\p{ X_i^\top \widehat \Delta^{(-k)} - \E\left[X_i\right]^\top \widehat \Delta^{(-k)}}.
\end{align*}
By Lemma~\ref{lemma conditional independence}, $(1-\Gamma_i)X_i \perp \mathcal{D}_{\mathcal{I}_{-k}}  \mid \Gamma_{1:n}$ for $i=1,\dots,n$. Then under Assumption~\ref{assumption mcar}, for $i=1,\dots,n$,
\begin{align*}
    &\E\sbr{\p{1-\Gamma_i}\p{ X_i^\top \widehat \Delta^{(-k)} - \E\left[X_i\right]^\top \widehat \Delta^{(-k)}} \mid \mathcal{D}_{\mathcal{I}_{-k}}, \Gamma_{1:n}}= (1-\Gamma_i)\E\sbr{X_i - \E\left[X_i\right] \mid \Gamma_{1:n}}^\top \widehat \Delta^{(-k)}= 0.
\end{align*}
and
\begin{align*}
    &\E\sbr{I_{4k}^2 \mid  \mathcal{D}_{\mathcal{I}_{-k}}, \Gamma_{1:n} } =n^{-2} \sum_{i=1}^n \E\sbr{ \p{\p{1-\Gamma_i}\p{X_i - \E[X_i]}^\top \widehat \Delta^{(-k)}}^2 \mid  \mathcal{D}_{\mathcal{I}_{-k}}, \Gamma_{1:n} }\\
    &\qquad\leq n^{-2}\sum_{i=1}^n \E\sbr{ \p{1-\Gamma_i} \p{X_i^\top \widehat \Delta^{(-k)}}^2 \mid  \mathcal{D}_{\mathcal{I}_{-k}}, \Gamma_{1:n} }=n^{-2}\sum_{i=1}^n\p{1-\Gamma_i}\widehat \Delta^{(-k),\top}\E\sbr{X_iX_i^\top\mid   \Gamma_{1:n}}\widehat \Delta^{(-k)}\\
    &\qquad\leq \p{n^{-1}\sum_{i=1}^n\p{1-\Gamma_i}}^2 \widehat \Delta^{(-k),\top}\E\sbr{X_iX_i^\top}\widehat \Delta^{(-k)}=\p{1-\widehat \gamma}^2\widehat \Delta^{(-k),\top}\E\sbr{X_iX_i^\top}\widehat \Delta^{(-k)},
\end{align*}
where the last inequality comes from the fact either $\sum_{i=1}^n\p{1-\Gamma_i} = 0$ or $\sum_{i=1}^n\p{1-\Gamma_i} \geq 1$ holds. By \eqref{lasso ineq  bounding quadratic form expectation}, \eqref{1-widehat gamma rate}, and Lemma~\ref{lemma convergence of conditional random variable}, we have
$
    I_{4k} = O_p\p{\p{1-\gamma_n}(s\log d/(n\gamma_n))^{1/2}}. 
$

We also have another upper bound:
\begin{align*}
    \E\sbr{I_{4k}^2 \mid  \mathcal{D}_{\mathcal{I}_{-k}}, \Gamma_{1:n} } &\leq \p{n^{-2}\sum_{i=1}^n\p{1-\Gamma_i}}\widehat \Delta^{(-k),\top}\E\sbr{X_iX_i^\top}\widehat \Delta^{(-k)}\leq n^{-1}\widehat \Delta^{(-k),\top}\E\sbr{X_iX_i^\top}\widehat \Delta^{(-k)}.
\end{align*}
By \eqref{lasso ineq  bounding quadratic form expectation} and Lemma~\ref{lemma convergence of conditional random variable},
$
    I_{4k} = O_p\p{n^{-1/2}(s\log d/(n\gamma_n))^{1/2}}.
$

Therefore, by \eqref{variance gamma(k) -1 ratio rate},
\begin{align*}
    I_4 &=K^{-1}\sum_{k=1}^K\p{\frac{1-\widehat \gamma^{(k)}}{1-\widehat \gamma} -1}I_{4k} + K^{-1}\sum_{k=1}^KI_{4k}\\
    &=O_p\p{n^{-1/2}\p{\gamma_n\p{1-\gamma_n}}^{1/2} (s\log d/(n\gamma_n))^{1/2} } + O_p\p{n^{-1/2}(s\log d/(n\gamma_n))^{1/2}}= O_p\p{(\frac{s\log d}{n^2\gamma_n})^{1/2} }.
\end{align*}

In conclusion,
\begin{align*}
    \widehat \theta^{(k)} - \widetilde \theta^{(k)} = O_p\p{(\frac{s\log d}{n^2\gamma_n^2})^{1/2}} + O_p\p{(\frac{s\log d}{n^3\gamma_n^3})^{1/2}}  + O_p\p{(\frac{s\log d}{n^2\gamma_n})^{1/2}} + O_p \p{(n\gamma_n)^{-1}},
\end{align*}
which implies
$
\widehat \theta = \widetilde \theta + O_p\p{ ({s\log d}/{n^2\gamma_n^2})^{1/2}}.
$
Together with Lemma~\ref{lemma the oracle asymptotics under MCAR}, we have
$\widehat \theta -\theta = O_p \p{(n\gamma_n)^{-1/2}}.$
\end{proof}

\begin{lemma} \label{lemma variance consistency under MCAR}
Let Assumptions~\ref{assumption mcar} and \ref{assumption mcar lasso (a)} hold. Choose
\[
\lambda_n \asymp (\log d/(n\gamma_n))^{1/2}.
\]
If $n\gamma_n \gg \p{s\vee \p{\log n}^2}\log d$, then as $n, d \rightarrow \infty$,
    \begin{align*}
        \widehat \sigma^2 = \sigma_n^2\br{1 + O_p\p{(s\log d/(n\gamma_n))^{1/2}}}.
    \end{align*}
\end{lemma}

\begin{proof}
    Let $\widehat \gamma^{(k)} = n_k^{-1}\sum_{i\in \mathcal{I}_k} \Gamma_i$,
$g^{(-k)}_i = \Gamma_iX_i^\top \widehat \beta^{(-k)}  + {\Gamma_i}/{\widehat \gamma^{(k)}}\p{Y_i - X_i^\top \widehat \beta^{(-k)}}$, and
        $g_i = \Gamma_iX_i^\top \beta^*  + {\Gamma_i}/{\gamma_n}\p{Y_i - X_i^\top \beta^*}.$
Then
    $$\widehat \sigma^2 =  n^{-1}\sum_{k=1}^K\sum_{i\in \mathcal{I}_k} \p{g^{(-k)}_i}^2 +n^{-1}\sum_{k=1}^K\sum_{i\in \mathcal{I}_k}\p{1-\Gamma_i}\widehat \beta^{(-k),\top}\Bar \Xi_0 \widehat \beta^{(-k)} - \p{\widehat \theta}^2.$$
Define 
$
    \widetilde \sigma^2 = n^{-1}\sum_{i=1}^n \p{g_i}^2 + n^{-1}\sum_{i=1}^n(1-\Gamma_i)\p{X_i^\top \beta^*}^2 - \theta^2.
$
Then we have the following decomposition:
\begin{align*}
    \widehat \sigma^2 - \widetilde \sigma^2 = B_1 + B_2 - B_3,
\end{align*}
where
\begin{align*}
    B_1 &= n^{-1}\sum_{k=1}^K\sum_{i\in \mathcal{I}_k} \p{g^{(-k)}_i}^2 - n^{-1}\sum_{i=1}^n \p{g_i}^2, \\
    B_2 &=n^{-1}\sum_{k=1}^K\sum_{i\in \mathcal{I}_k}\p{1-\Gamma_i}\widehat \beta^{(-k),\top}\Bar \Xi_0 \widehat \beta^{(-k)} -n^{-1}\sum_{i=1}^n(1-\Gamma_i)\p{X_i^\top \beta^*}^2, \quad
    B_3 = \p{\widehat \theta}^2 - \theta^2.
\end{align*}
Let $\widehat \gamma = n^{-1}\sum_{i=1}^n \Gamma_i$, $\widehat \gamma^{(-k)} = \abs{\mathcal{I}_{-k}}^{-1}\sum_{i \in \mathcal{I}_{-k}}\Gamma_i$, and $\widehat \Delta^{(-k)} = \widehat \beta^{(-k)} - \beta^*$.
Then
\begin{align*}
    \p{g^{(-k)}_i - g_i}^2 &= \br{\p{\Gamma_i-\frac{\Gamma_i}{\gamma_n}}X_i^\top  \widehat \Delta^{(-k)} + \p{\frac{\Gamma_i}{\gamma_n} - \frac{\Gamma_i}{\widehat \gamma^{(k)}}}X_i^\top  \widehat \Delta^{(-k)} + \p{\frac{\Gamma_i}{\widehat \gamma^{(k)}} - \frac{\Gamma_i}{\gamma_n}}w_i}^2\\
    &\leq 2\br{\p{\Gamma_i-\frac{\Gamma_i}{\gamma_n}}X_i^\top  \widehat \Delta^{(-k)} + \p{\frac{\Gamma_i}{\gamma_n} - \frac{\Gamma_i}{\widehat \gamma^{(k)}}}X_i^\top  \widehat \Delta^{(-k)}}^2 + 2\br{\p{\frac{\Gamma_i}{\widehat \gamma^{(k)}} - \frac{\Gamma_i}{\gamma_n}}w_i}^2\\
    &\leq 4\br{\p{\Gamma_i-\frac{\Gamma_i}{\gamma_n}}^2 + \p{\frac{\Gamma_i}{\widehat \gamma^{(k)}} - \frac{\Gamma_i}{\gamma_n}}^2}\p{X_i^\top  \widehat \Delta^{(-k)}}^2 + 2\p{\frac{\Gamma_i}{\widehat \gamma^{(k)}} - \frac{\Gamma_i}{\gamma_n}}^2 w_i^2\\
    &\leq \frac{4\Gamma_i}{\gamma_n^2}\br{1 + \p{\frac{\gamma_n}{\widehat \gamma^{(k)}} - 1}^2}\p{X_i^\top  \widehat \Delta^{(-k)}}^2 + \frac{2\Gamma_i}{\gamma_n^2}\p{\frac{\gamma_n}{\widehat \gamma^{(k)}} - 1}^2 w_i^2.
\end{align*}
For the summation, we have
\begin{align*}
    n_k^{-1}\sum_{i \in \mathcal{I}_{k}}\p{g^{(-k)}_i - g_i}^2 \leq& 4\br{1 + \p{\frac{\gamma_n}{\widehat \gamma^{(k)}} - 1}^2} {n_k^{-1}\sum_{i \in \mathcal{I}_{k}} \frac{\Gamma_i}{\gamma_n^2}\p{X_i^\top  \widehat \Delta^{(-k)}}^2}+ 2\p{\frac{\gamma_n}{\widehat \gamma^{(k)}} - 1}^2 { n_k^{-1}\sum_{i \in \mathcal{I}_{k}} \frac{\Gamma_i}{\gamma_n^2}w_i^2}.
\end{align*}
With \eqref{lasso ineq bounding quadratic form expectation} and \eqref{lasso ineq w_i^2 expectation}, it is clear that
\begin{align*}
\E\sbr{n_k^{-1}\sum_{i \in \mathcal{I}_{k}} \frac{\Gamma_i}{\gamma_n^2}\p{X_i^\top  \widehat \Delta^{(-k)}}^2 \mid \mathcal{D}_{\mathcal{I}_{-k}}} &= \p{n_k\gamma_n}^{-1}\sum_{i \in \mathcal{I}_{k}}\E\sbr{\p{X_i^\top  \widehat \Delta^{(-k)}}^2\mid \mathcal{D}_{\mathcal{I}_{-k}}} = O_p\p{\frac{s\log d}{n\gamma_n^{2}}},\\
\E\sbr{n_k^{-1}\sum_{i \in \mathcal{I}_{k}} \frac{\Gamma_i}{\gamma_n^2}w_i^2 \mid \mathcal{D}_{\mathcal{I}_{-k}} } &= \p{n_k\gamma_n}^{-1}\sum_{i \in \mathcal{I}_{k}}\E\sbr{w_i^2} = O_p\p{\gamma_n^{-1}}.
\end{align*}
By Lemma~\ref{lemma gamma ratio convergence}, we have $\p{\gamma_n/\widehat \gamma^{(k)}} - 1 = O_p\p{\p{n\gamma_n}^{-1/2}}$. Together with Lemma~\ref{lemma convergence of conditional random variable}, we have
\begin{align}
    n_k^{-1}\sum_{i \in \mathcal{I}_{k}}\p{g^{(-k)}_i - g_i}^2 = O_p\p{\frac{s\log d}{n\gamma_n^{2}}} + O_p\p{\gamma_n^{-1}\p{n\gamma_n}^{-1}} = O_p\p{\frac{s\log d}{n\gamma_n^{2}}}. \label{convergence rate of the sum of (g_-k - g)^2}
\end{align}
In addition, we have
\begin{align*}
   \E\sbr{ n^{-1}\sum_{i=1}^n g_i^2} =& n^{-1}\sum_{i=1}^n \E\sbr{\p{\Gamma_iX_i^\top \beta^*  + \frac{\Gamma_i}{\gamma_n}w_i}^2}=  \gamma_n\E\sbr{\p{X_i^\top \beta^*}^2}  + \gamma_n^{-1}\E\sbr{w_i^2}\leq 2\gamma_n\sigma^2 + 2\gamma_n^{-1}\sigma_w^2.
\end{align*}
Thus, by Markov inequality, 
\begin{align}
    n^{-1}\sum_{i=1}^n \p{ g_i}^2 = O_p\p{\gamma_n^{-1}}. \label{convergence rate of the sum of g^2}
\end{align}

For $B_1$, since $a^2-b^2 = (a-b)^2 - 2b(a-b)$, we have
\begin{align*}
    \abs{B_1} &= \abs{n^{-1}\sum_{k=1}^K\sum_{i\in \mathcal{I}_k} \br{\p{g^{(-k)}_i}^2 - \p{g_i}^2}}\leq n^{-1}\sum_{k=1}^K\sum_{i\in \mathcal{I}_k}\p{g^{(-k)}_i - g_i}^2 + 2n^{-1}\sum_{k=1}^K\sum_{i\in \mathcal{I}_k}\abs{g_i\p{g^{(-k)}_i - g_i}}\\
    &\leq n^{-1}\sum_{k=1}^K\sum_{i\in \mathcal{I}_k}\p{g^{(-k)}_i - g_i}^2 + 2\br{n^{-1}\sum_{i=1}^n\p{ g_i}^2}^{1/2}\br{n^{-1}\sum_{k=1}^K\sum_{i\in \mathcal{I}_k}\p{g^{(-k)}_i - g_i}^2}^{1/2}.
\end{align*}
By \eqref{convergence rate of the sum of (g_-k - g)^2} and \eqref{convergence rate of the sum of g^2}, we have
\begin{align*}
    B_1 = O_p\p{\frac{s\log d}{n\gamma_n^{2}}} + O_p\p{\gamma_n^{-1/2}(\frac{s\log d}{n\gamma_n^2})^{1/2}} = O_p\p{\gamma_n^{-1}(s\log d/(n\gamma_n))^{1/2}}.
\end{align*}

For $B_2$, define
$
    B_{2k} := n_k^{-1}\sum_{i\in \mathcal{I}_k}\p{1-\Gamma_i}\widehat \beta^{(-k),\top}\Bar \Xi_0 \widehat \beta^{(-k)} - n_k^{-1}\sum_{i\in \mathcal{I}_k}(1-\Gamma_i)\p{X_i^\top \beta^*}^2.
$
Then we have
\begin{align*}
    B_{2k} &= \p{\frac{1-\widehat \gamma^{(k)}}{1-\widehat \gamma}} n^{-1}\sum_{i=1}^n \p{1-\Gamma_i}\p{X_i^\top \widehat \Delta^{(-k)} + X_i^\top\beta^*}^2 - n_k^{-1}\sum_{i\in \mathcal{I}_k}(1-\Gamma_i)\p{X_i^\top \beta^*}^2\\
    &= \p{\frac{1-\widehat \gamma^{(k)}}{1-\widehat \gamma} - 1} \p{B_{2k,1} +  B_{2k,2} + B_{2k,3}} + B_{2k,1} + B_{2k,3} + B_{2k,4} + B_{2k,5},
\end{align*}
where
\begin{align*}
    B_{2k,1} &= n^{-1}\sum_{i=1}^n \p{1-\Gamma_i}\p{X_i^\top \widehat \Delta^{(-k)}}^2,\quad
    B_{2k,2} = n^{-1} \sum_{i=1}^n \p{1-\Gamma_i}\p{ X_i^\top\beta^*}^2,\\
    B_{2k,3} &= 2n^{-1} \sum_{i=1}^n \p{1-\Gamma_i}\p{X_i^\top \widehat \Delta^{(-k)}} \p{X_i^\top\beta^*},\\
    B_{2k,4} &= n^{-1} \sum_{i=1}^n \p{1-\Gamma_i}\p{ X_i^\top\beta^*}^2 - \E\sbr{\p{1-\Gamma_i}\p{ X_i^\top\beta^*}^2},\\
    B_{2k,5} &= \E\sbr{\p{1-\Gamma_i}\p{ X_i^\top\beta^*}^2} - n_k^{-1}\sum_{i\in \mathcal{I}_k}(1-\Gamma_i)\p{X_i^\top \beta^*}^2.
\end{align*}
By Lemma~\ref{lemma conditional independence}, $(1-\Gamma_i)X_i \perp \mathcal{D}_{\mathcal{I}_{-k}}  \mid \Gamma_{1:n}$ for $i=1,\dots,n$. Thus, by \eqref{lasso ineq bounding quadratic form expectation}, we have
\begin{align*}
    &\E\sbr{B_{2k,1} \mid \mathcal{D}_{\mathcal{I}_{-k}}, \Gamma_{1:n}}  = \E\sbr{n^{-1}\sum_{i=1}^n \p{1-\Gamma_i}\p{X_i^\top \widehat \Delta^{(-k)}}^2 \mid \mathcal{D}_{\mathcal{I}_{-k}}, \Gamma_{1:n}} \\
    &\qquad  =\p{1-\widehat \gamma}\widehat \Delta^{(-k),\top}\E\sbr{X_iX_i^\top}\widehat \Delta^{(-k)} = O_p\p{\p{1-\gamma_n}\frac{s\log d}{n\gamma_n}}.
\end{align*}
By Lemma~\ref{lemma convergence of conditional random variable} we have
\begin{align}
    B_{2k,1} = O_p\p{\p{1-\gamma_n}\frac{s\log d}{n\gamma_n}}. \label{variance B_2k term 1}
\end{align}
Since 
$\E\sbr{\p{1-\Gamma_i}\p{ X_i^\top\beta^*}^2} = \p{1-\gamma_n}\E\sbr{\p{ X_i^\top\beta^*}^2} \leq 2\sigma^2 \p{1-\gamma_n},$
by Markov inequality, we have
\begin{align}
    B_{2k,2} = n^{-1} \sum_{i=1}^n \p{1-\Gamma_i}\p{ X_i^\top\beta^*}^2 = O_p\p{1-\gamma_n }. \label{variance B_2k term 2}
\end{align}
Then with \eqref{variance B_2k term 1}, \eqref{variance B_2k term 2}, and by Cauchy-Schwarz inequality, 
\begin{align}
    \abs{B_{2k,3}} &\leq n^{-1} \sum_{i=1}^n \abs{\p{1-\Gamma_i}\p{X_i^\top \widehat \Delta^{(-k)}} \p{X_i^\top\beta^*} }\notag\\
    &\leq \br{n^{-1} \sum_{i=1}^n\p{1-\Gamma_i}\p{X_i^\top \widehat \Delta^{(-k)}}^2
    n^{-1} \sum_{i=1}^n\p{1-\Gamma_i}\p{X_i^\top\beta^*}^2}^{1/2}\notag\\
    &= O_p\p{\p{1-\gamma_n}(s\log d/(n\gamma_n))^{1/2}}. \label{variance B_2k term 3}
\end{align}
Moreover, by Lemma~\ref{sub-gaussian properties}(c),
\begin{align*}
    &\E\sbr{\br{n^{-1} \sum_{i=1}^n \p{1-\Gamma_i}\p{ X_i^\top\beta^*}^2 - \E\sbr{\p{1-\Gamma_i}\p{ X_i^\top\beta^*}^2}}^2}= n^{-1}\Var\sbr{\p{1-\Gamma_i}\p{ X_i^\top\beta^*}^2}\\
    &\qquad\leq n^{-1} \E\sbr{\p{1-\Gamma_i}\p{ X_i^\top\beta^*}^4} = n^{-1}\p{1-\gamma_n} \E\sbr{\p{ X_i^\top\beta^*}^4}\leq 4\sigma^4n^{-1}\p{1-\gamma_n}.
\end{align*}
By Chebyshev's inequality, we have 
\begin{align}
    B_{2k,4} = n^{-1} \sum_{i=1}^n \p{1-\Gamma_i}\p{ X_i^\top\beta^*}^2 - \E\sbr{\p{1-\Gamma_i}\p{ X_i^\top\beta^*}^2} = O_p\p{n^{-1/2}\p{1-\gamma_n}^{1/2}}. \label{variance B_2k term 4}
\end{align}
Analogously, we have
\begin{align}
    B_{2k,5} = n_k^{-1} \sum_{i\in \mathcal{I}_k} \p{1-\Gamma_i}\p{ X_i^\top\beta^*}^2 - \E\sbr{\p{1-\Gamma_i}\p{ X_i^\top\beta^*}^2} = O_p\p{n^{-1/2}\p{1-\gamma_n}^{1/2}}. \label{variance B_2k term 5}
\end{align}
Let $\widehat \gamma^{(k)} = n_k^{-1}\sum_{i \in \mathcal{I}_{k}}\Gamma_i$. By \eqref{variance gamma(k) -1 ratio rate} we have
\begin{align}
    \frac{1-\widehat \gamma^{(k)}}{1-\widehat \gamma} - 1 = O_p\p{\gamma_n\p{n\p{1-\gamma_n}}^{-1}} + O_p\p{\gamma_n^{1/2}\p{n\p{1-\gamma_n}}^{-1/2}}.
\end{align}
Together with \eqref{variance B_2k term 1}-\eqref{variance B_2k term 5}, we have
\begin{align*}
    B_{2k} =& O_p\sbr{\br{\gamma_n\p{n\p{1-\gamma_n}}^{-1} + \gamma_n^{1/2}\p{n\p{1-\gamma_n}}^{-1/2}}\p{1-\gamma_n}\frac{s\log d}{n\gamma_n}}\\
    &+ O_p\sbr{\br{\gamma_n\p{n\p{1-\gamma_n}}^{-1} + \gamma_n^{1/2}\p{n\p{1-\gamma_n}}^{-1/2}}\p{1-\gamma_n}}\\
    &+ O_p\sbr{\br{\gamma_n\p{n\p{1-\gamma_n}}^{-1} + \gamma_n^{1/2}\p{n\p{1-\gamma_n}}^{-1/2}}\p{1-\gamma_n}(s\log d/(n\gamma_n))^{1/2}}\\
    &+O_p\p{\p{1-\gamma_n}\frac{s\log d}{n\gamma_n}+\p{1-\gamma_n}(s\log d/(n\gamma_n))^{1/2}+n^{-1/2}\p{1-\gamma_n}^{1/2}+n^{-1/2}\p{1-\gamma_n}^{1/2}}\\
    =& O_p\p{\p{1-\gamma_n}(s\log d/(n\gamma_n))^{1/2}+n^{-1}\gamma_n+n^{-1/2}\p{1-\gamma_n}^{1/2}}.
\end{align*}
Therefore, 
\begin{align*}
    B_2 = K^{-1}\sum_{k=1}^K B_{2k} = O_p\p{\p{1-\gamma_n}(s\log d/(n\gamma_n))^{1/2}}.
\end{align*}

For $B_3$, by Lemma~\ref{lemma the mean estimator consistency under MCAR}, it follows that
$
    B_3 = \p{\widehat \theta}^2 - \theta^2 = \p{ \widehat \theta - \theta}^2 + 2\theta\p{ \widehat \theta - \theta} = O_p\p{\p{n\gamma_n}^{-1/2}}.
$

In conclusion, we have
\begin{align*}
    \widehat \sigma^2 - \widetilde \sigma^2 &= O_p\p{\gamma_n^{-1}(s\log d/(n\gamma_n))^{1/2}} + O_p\p{\p{1-\gamma_n}(s\log d/(n\gamma_n))^{1/2}}+O_p\p{\p{n\gamma_n}^{-1/2}}\\
    & = O_p\p{\gamma_n^{-1}(s\log d/(n\gamma_n))^{1/2}}.
\end{align*}
Note that $\widetilde \sigma^2$ can be written as
\begin{align*}
    \widetilde \sigma^2 = n^{-1}\sum_{i=1}^n\p{X_i^\top \beta^* + \frac{\Gamma_i}{\gamma_n}\p{Y_i - X_i^\top \beta^*}}^2 - \theta^2 = n^{-1}\sum_{i=1}^n\p{X_i^\top \beta^* + \frac{\Gamma_i}{\gamma_n}w_i}^2 - \theta^2.
\end{align*}
Since $\E\sbr{X_i^\top \beta^* + {\Gamma_i}/{\gamma_n}\p{Y_i - X_i^\top \beta^*}} = \theta$, we have
\begin{align*}
    \sigma_n^2 &= \Var\p{X_i^\top \beta^*  + \frac{\Gamma_i}{\gamma_n}\p{Y_i - X_i^\top \beta^*}} = \E\sbr{\p{X_i^\top \beta^* + \frac{\Gamma_i}{\gamma_n}w_i}^2} - \theta^2.
\end{align*}
Note that
\begin{align*}
    \widetilde \sigma^2 - \sigma_n^2 = n^{-1}\sum_{i=1}^n\p{X_i^\top \beta^* + \frac{\Gamma_i}{\gamma_n}w_i}^2 - \E\sbr{\p{X_i^\top \beta^* + \frac{\Gamma_i}{\gamma_n}w_i}^2}.
\end{align*}
Then
\begin{align*}
    &\E\sbr{\p{\widetilde \sigma^2 - \sigma_n^2}^2}= n^{-2}\sum_{i=1}^n \Var \sbr{\p{X_i^\top \beta^* + \frac{\Gamma_i}{\gamma_n}w_i}^2}\leq 8n^{-1}\E\sbr{\p{X_i^\top \beta^*}^4 + \p{\frac{\Gamma_i}{\gamma_n}w_i}^4},
\end{align*}
where the last inequality comes from H{\"o}lder inequality. By Lemma~\ref{sub-gaussian properties} (c), we have
$
    \E\sbr{\p{X_i^\top \beta^*}^4} \leq 4\sigma^4$ and $\E\sbr{\p{{\Gamma_i}/{\gamma_n}w_i}^4} = \gamma_n^{-3} \E\sbr{w_i^4} \leq 4\sigma_w^4\gamma_n^{-3}.
$
It follows that 
$
    \E\sbr{\p{\widetilde \sigma^2 - \sigma_n^2}^2} \leq 32\p{\sigma^4 + \sigma_w^4\gamma_n^{-3}}n^{-1}.
$
By Chebyshev's inequality, 
$
    \widetilde \sigma^2 - \sigma_n^2 = O_p\p{\gamma_n^{-1}\p{n\gamma_n}^{-1/2}}.
$
Together with $\sigma_n^2 \asymp \gamma_n^{-1}$ from Lemma~\ref{lemma the oracle asymptotics under MCAR}, we have
\begin{align*}
    \widehat \sigma^2 = \widehat \sigma^2 -\widetilde \sigma^2 + \widetilde \sigma^2 - \sigma_n^2 + \sigma_n^2 = \sigma_n^2 + O_p\p{\gamma_n^{-1}\p{n\gamma_n}^{-1/2}} + O_p\p{\gamma_n^{-1}(s\log d/(n\gamma_n))^{1/2}},
\end{align*}
which implies that $\widehat \sigma^2 = \sigma_n^2\br{1 + O_p\p{(s\log d/(n\gamma_n))^{1/2}}}.$
\end{proof}

\subsection{Proof of Theorem~\ref{theorem for the asymptotics under MCAR}}

\begin{proof} By Lemmas~\ref{lemma the mean estimator consistency under MCAR} and \ref{lemma variance consistency under MCAR},
$\widehat \theta -\theta = O_p \p{(n\gamma_n)^{-1/2}}$ and  $\widehat \sigma^2 = \sigma_n^2\br{1 + o_p(1)}$. In addition, by Lemma~\ref{lemma the oracle asymptotics under MCAR}, Lemma~\ref{lemma the mean estimator consistency under MCAR}, and Lemma~\ref{lemma variance consistency under MCAR},
$
        \widehat\sigma^{-1}n^{1/2}(\widehat \theta -\widetilde \theta) = O_p\p{\p{n\gamma_n}^{1/2}{ ({s\log d}/{n^2\gamma_n^2})^{1/2}}} = o_p(1).
$
Then by Lemma~\ref{lemma the oracle asymptotics under MCAR} and Slutsky's theorem, we have
$$\widehat\sigma^{-1}n^{1/2}(\widehat \theta - \theta) = \frac{\sigma_n}{\widehat\sigma} \cdot \sigma_n^{-1}n^{1/2}(\widetilde \theta - \theta) + \widehat\sigma^{-1}n^{1/2}(\widehat \theta -\widetilde \theta) \xrightarrow{d} \mc{N}(0,1).$$
\end{proof}

\subsection{Auxiliary lemmas for Theorem~\ref{plm Asymptotics theorem}}

\begin{lemma} \label{lemma the plm oracle asymptotics under MCAR}
Let Assumptions~\ref{assumption mcar} and \ref{assumption mcar plm (a)} hold. Then $\sigma_{plm}^2 \asymp \gamma_n^{-1}$ and as $n, d \rightarrow \infty$,
$
\sigma_{plm}^{-1}n^{1/2}(\widetilde \theta_{plm} - \theta) \xrightarrow{d} \mathcal{N}(0,1),
$
where $\widetilde \theta_{plm} = n^{-1}\sum_{i=1}^n \br{X_i^\top \beta^*_{plm} + f^*\p{Z_i} + {\Gamma_i}/{\gamma_n}\epsilon_i }$.
\end{lemma}
\begin{proof}
    Let $V_i = X_i^\top \beta^*_{plm} + f^*\p{Z_i} + {\Gamma_i}/{\gamma_n}\epsilon_i$. Then 
    \begin{align*}
        \E[V_i] &= \E\sbr{X_i^\top \beta^*_{plm} + f^*\p{Z_i} } + \E\left[\frac{\Gamma_i}{\gamma_n}\right]\E\sbr{\p{Y_i - X_i^\top \beta^*_{plm} - f^*\p{Z_i}}} = \E[Y_i] = \theta.
    \end{align*}
Observe that
\begin{align*}
    \sigma_{plm}^2 = \Var[V_i] &= \E\left[ \p{X_i^\top \beta^*_{plm} + f^*\p{Z_i} + \frac{\Gamma_i}{\gamma_n}\epsilon_i - \theta}^2\right]= \E\left[ \p{X_i^\top \beta^*_{plm}+ f^*\p{Z_i} - \theta}^2 + \frac{1}{\gamma_n}\epsilon_i^2\right],
\end{align*}
where we use Lemma~\ref{lemma oracle plm properties} to obtain the last equation. Under Assumption~\ref{assumption mcar plm (a)}, the second moments of $X_i^\top \beta^*_{plm}$, $f^*\p{Z_i}$, and $\epsilon_i$ are bounded, then by Minkowski's inequality, we have that the second moment of $X_i^\top \beta^*_{plm}+ f^*\p{Z_i} - \theta$ is bounded. In addition, the second moment of $\epsilon_i$ is assumed to be lower bounded. Then it follows that
\begin{align}
    \sigma_{plm}^2 \asymp \gamma_n^{-1}. \label{plm shrinking rate of sigma under MCAR}
\end{align}
Take $c=2$ and we have
$
\E\left[\abs{X_i^\top \beta^*_{plm}}^{2+c}\right] \leq {\sigma}^{2+c}$,  $\E\left[\abs{f^*\p{Z_i}}^{2+c}\right] \leq {\sigma_f}^{2+c}$, and $\E\left[\abs{\epsilon_i}^{2+c}\right] \leq {\sigma_\epsilon}^{2+c}.
$
Then,
\begin{align*}
    &\norm{V_i - \theta}_{\P, 2+c} = \norm{X_i^\top \beta^*_{plm} + f^*\p{Z_i} - \theta  + \frac{\Gamma_i}{\gamma_n}\epsilon_i }_{\P, 2+c}\\
    &\qquad\leq \norm{X_i^\top \beta^*_{plm}}_{\P, 2+c} + \norm{f^*\p{Z_i}}_{\P, 2+c} + \theta + \E\left[ \frac{\Gamma_i}{\gamma_n^{2+c}}\epsilon_i^{2+c}\right]^{\frac{1}{2+c}}\\
    &\qquad=\E\left[\abs{X_i^\top \beta^*_{plm}}^{2+c}\right]^{\frac{1}{2+c}} + \E\left[\abs{f^*\p{Z_i}}^{2+c}\right]^{\frac{1}{2+c}} + \theta + \gamma_n^{\frac{1}{2+c}-1}\E\left[\abs{\epsilon_i}^{2+c}\right]^{\frac{1}{2+c}}\leq \sigma + \sigma_f + \theta + \sigma_\epsilon \gamma_n^{\frac{1}{2+c}-1},
\end{align*}
which implies that
\begin{align*}
    n^{-\frac{c}{2}}\gamma_n^{1+\frac{c}{2}}\E\left[ \abs{V_i - \theta}^{2+c} \right] &\leq n^{-\frac{c}{2}}\gamma_n^{1+\frac{c}{2}}\p{\sigma + \sigma_f + \theta + \sigma_\epsilon\gamma_n^{\frac{1}{2+c}-1}}^{2+c} = n^{-\frac{c}{2}}\gamma_n^{1+\frac{c}{2}}O\p{\gamma_n^{-1-c}}. 
\end{align*}
That is,
\begin{align}
    n^{-\frac{c}{2}}\gamma_n^{1+\frac{c}{2}}\E\left[ \abs{V_i - \theta}^{2+c} \right] = O\p{(n\gamma_n)^{-\frac{c}{2}}} = o(1). \label{plm shrinking rate of Lyapunov condition}
\end{align}
Then for any $\delta>0$, as $N, d \rightarrow \infty$,
\begin{align*}
    \gamma_n\E\left[\p{V_i - \theta}^{2}\mathbbm{1}_{\br{\abs{V_i - \theta}>\delta(n/\gamma_n)^{1/2}}}\right] \leq \delta^{-c}n^{-\frac{c}{2}}\gamma_n^{1+\frac{c}{2}}\E\left[\p{V_i - \theta}^{2+c}\right] = o(1).
\end{align*}
By Proposition 2.27 (Lindeberg-Feller) of \cite{van2000asymptotic}, we have
$
        \sigma_{plm}^{-1}n^{1/2}(\widetilde \theta_{plm} - \theta) \xrightarrow{d} \mathcal{N}(0,1).
$
\end{proof}

\begin{lemma}\label{lemma plm mean estimator consistency}
    Let Assumptions~\ref{assumption mcar} and \ref{assumption mcar plm (a)} hold. If $n\gamma_n \gg 1$, then as $n \rightarrow \infty$, 
    \begin{align*}
        \widehat \theta_{plm} - \widetilde \theta_{plm} &= o_p\p{\p{n\gamma_n}^{-1/2}},\quad
        \widehat \theta_{plm} -\theta = O_p\p{\p{n\gamma_n}^{-1/2}}.
    \end{align*}
\end{lemma}

\begin{proof}
Let $\widehat \gamma = n^{-1}\sum_{i=1}^n \Gamma_i$, $\widehat \gamma^{\p{k}} = n_k^{-1}\sum_{i\in \mathcal{I}_{k}}^n \Gamma_i$, $\widehat \Delta^{(-k)}_{plm} = \widehat \beta^{(-k)}_{plm} - \beta^*_{plm}$, and $\widehat D^{(-k)} = \widehat f^{(-k)} - f^*$. Note that $\widehat \theta_{plm} = K^{-1}\sum_{k=1}^K \widehat \theta^{(-k)}_{plm}$, where
\begin{align*}
    \widehat \theta^{(k)}_{plm} &= n_{k}^{-1}\sum_{i\in\mathcal{I}_k} \br{\p{\Gamma_iX_i+ \p{1-\Gamma_i}\Bar{X}_0 }^\top \widehat \beta^{(-k)} + \widehat f^{(-k)}\p{Z_i}}+ n_{k}^{-1}\sum_{i \in \mathcal{I}_k}\frac{\Gamma_i}{\widehat \gamma^{(k)}}\p{Y_i - X_i^\top \widehat \beta^{(-k)} - \widehat f^{(-k)}\p{Z_i}},
\end{align*}
and that 
\begin{align*}
    \widetilde \theta_{plm} &= n^{-1}\sum_{i=1}^n \br{X_i^\top \beta^*_{plm} + f^*\p{Z_i} + \frac{\Gamma_i}{\gamma_n}\epsilon_i }=n^{-1}\sum_{i=1}^n \br{\p{\Gamma_iX_i + \p{1-\Gamma_i}\Bar{X}_0}^\top \beta^*_{plm} + f^*\p{Z_i} + \frac{\Gamma_i}{\gamma_n}\epsilon_i },
\end{align*}
where $\epsilon_i = Y_i - X_i^\top \beta^*_{plm} - f^*\p{Z_i}$. Then 
\begin{align*}
    &\widehat \theta_{plm} - \widetilde \theta_{plm}= T_1 + T_2 + T_3,
\end{align*}
where
\begin{align*}
    T_1 &= n^{-1}\sum_{k=1}^K\sum_{i \in \mathcal{I}_k} \p{\Gamma_i - \frac{\Gamma_i}{\widehat \gamma^{(k)} } } \br{ \p{X_i - \E\sbr{X_i } }^\top\widehat \Delta^{(-k)}_{plm} + \widehat D^{(-k)}\p{Z_i} - \E_Z\sbr{\widehat D^{(-k)}\p{Z}} },\\
    T_2 &= n^{-1}\sum_{k=1}^K\sum_{i \in \mathcal{I}_k} \p{1-\Gamma_i}\br{ \p{\Bar{X}_0 - \E\sbr{X_i } }^\top\widehat \Delta^{(-k)}_{plm} + \widehat D^{(-k)}\p{Z_i} - \E_Z\sbr{\widehat D^{(-k)}\p{Z}} },\\
    T_3 &= \p{\widehat \gamma^{-1} - \gamma_n^{-1}}n^{-1}\sum_{i=1}^n\Gamma_i\epsilon_i.
\end{align*}

For $T_1$, define
\begin{align*}
    T_{1k} = n_k^{-1}\sum_{i \in \mathcal{I}_k} \Gamma_i \br{ \p{X_i - \E\sbr{X_i } }^\top\widehat \Delta^{(-k)}_{plm} + \widehat D^{(-k)}\p{Z_i} - \E_Z\sbr{\widehat D^{(-k)}\p{Z}} }.
\end{align*}
Let $\mathcal{L}_{\mathcal{I}_{-k}} = \br{\p{Z_i, \Gamma_i, \Gamma_iX_i, \Gamma_iY_i}: i \in \mathcal{I}_{-k}}$. For any $i \in \mathcal{I}_k$, it holds that
\begin{align*}
    &\E\sbr{\p{X_i - \E\sbr{X_i } }^\top\widehat \Delta^{(-k)}_{plm} + \widehat D^{(-k)}\p{Z_i} - \E_Z\sbr{\widehat D^{(-k)}\p{Z}} \mid \mathcal{L}_{\mathcal{I}_{-k}}} \\
    &=\E\sbr{\p{X_i - \E\sbr{X_i } }^\top\widehat \Delta^{(-k)}_{plm} \mid \mathcal{L}_{\mathcal{I}_{-k}}} + \E\sbr{\widehat D^{(-k)}\p{Z_i} - \E_Z\sbr{\widehat D^{(-k)}\p{Z}} \mid \mathcal{L}_{\mathcal{I}_{-k}}}\\
    &=\E\sbr{X_i - \E\sbr{X_i } \mid \mathcal{L}_{\mathcal{I}_{-k}}}^\top\widehat \Delta^{(-k)}_{plm}=\E\sbr{X_i - \E\sbr{X_i } }^\top\widehat \Delta^{(-k)}_{plm}= 0,
\end{align*}
and 
\begin{align*}
    &\E\sbr{T_{1k}^2 \mid \mathcal{L}_{\mathcal{I}_{-k}}}= n_k^{-2}\sum_{i \in \mathcal{I}_{k} }\E\sbr{\Gamma_i\br{ \p{X_i - \E\sbr{X_i }}^\top\widehat \Delta^{(-k)}_{plm} + \widehat D^{(-k)}\p{Z_i} - \E_Z\sbr{\widehat D^{(-k)}\p{Z}}}^2 \mid \mathcal{L}_{\mathcal{I}_{-k}} }\\
    &\qquad=n_k^{-1}\E\sbr{\Gamma_i}\E\sbr{\br{ \p{X_i - \E\sbr{X_i }}^\top\widehat \Delta^{(-k)}_{plm} + \widehat D^{(-k)}\p{Z_i} - \E_Z\sbr{\widehat D^{(-k)}\p{Z}}}^2 \mid \mathcal{L}_{\mathcal{I}_{-k}} }\\
    &\qquad\leq 2\gamma_n n_k^{-1}\E\sbr{\p{ X_i^\top\widehat \Delta^{(-k)}_{plm}}^2\mid \widehat \Delta^{(-k)}_{plm} } +2 \gamma_n n_k^{-1}\E_Z\sbr{\widehat D^{(-k)}\p{Z}}.
\end{align*}
For any $i \in \mathcal{I}_{k}$, $X_i\perp \widehat \Delta^{(-k)}_{plm}$, we have
\begin{align} \label{plm bound for the quadratic linear form}
    \E\sbr{\p{ X_i^\top\widehat \Delta^{(-k)}_{plm}}^2\mid \widehat \Delta^{(-k)}_{plm} }= \Delta^{(-k),\top}_{plm}\E\sbr{ X_i X_i^\top}\widehat \Delta^{(-k)}_{plm} \leq \tilde \kappa_u \norm{\widehat \Delta^{(-k)}_{plm}}_2^2.
\end{align}
Thus,
\begin{align*}
    \E\sbr{T_{1k}^2 \mid \mathcal{L}_{\mathcal{I}_{-k}} } \leq 2\tilde \kappa_u \gamma_n n_k^{-1}\norm{\widehat \Delta^{(-k)}_{plm}}_2^2 + 2\gamma_n n_k^{-1}\E_Z\sbr{\widehat D^{(-k)}\p{Z}} = o_p\p{n^{-1}\gamma_n}.
\end{align*}
By Lemma~\ref{lemma convergence of conditional random variable}, it follows that 
$T_{1k} = o_p\p{\gamma_n \p{n\gamma_n}^{-1/2}}.$
Since by Lemma~\ref{lemma gamma ratio convergence}, 
\begin{align}\label{plm 1-1/widehat gamma rate}
    1-\frac{1}{\widehat \gamma^{(k)}} &= 1 - \frac{1}{\gamma_n} + \frac{1}{\gamma_n} - \frac{1}{\widehat \gamma^{(k)}} = O_p\p{\gamma_n^{-1}} + O_p\p{\gamma_n^{-1}\p{n\gamma_n}^{-1/2}\p{1-\gamma_n}^{1/2}} = O_p\p{\gamma_n^{-1}},
\end{align}
we have
$
    T_1 = K^{-1}\sum_{k=1}^K\p{1-{1}/{\widehat \gamma^{(k)}}}T_{1k} = o_p\p{\p{n\gamma_n}^{-1/2}}.
$

For $T_2$, we have
\begin{align*}
    T_2 = K^{-1}\sum_{k=1}^kT_{2ak} +  K^{-1}\sum_{k=1}^kT_{2bk},
\end{align*}
where
\begin{align*}
    T_{2ak} &= n_k^{-1}\sum_{i \in \mathcal{I}_k} \p{1-\Gamma_i}\p{\Bar{X}_0 - \E\sbr{X_i } }^\top\widehat \Delta^{(-k)}_{plm},\\
    T_{2bk} &= n_k^{-1}\sum_{i \in \mathcal{I}_k} \p{1-\Gamma_i}\p{\widehat D^{(-k)}\p{Z_i} - \E_Z\sbr{\widehat D^{(-k)}\p{Z}}}.
\end{align*}
Since $\Bar{X}_0 = {\sum_{i=1}^n \p{1-\Gamma_i} X_i}/{\sum_{i=1}^n \p{1-\Gamma_i}}$, $T_{2ak}$ can be written as
\begin{align*}
    T_{2ak} &= \p{1-\widehat \gamma^{(k)}}\p{\Bar{X}_0 - \E\sbr{X_i } }^\top\widehat \Delta^{(-k)}_{plm}=\frac{1-\widehat \gamma^{(k)}}{1-\widehat \gamma}\br{n^{-1}\sum_{i=1}^n\p{1-\Gamma_i}\p{X_i - \E\sbr{X_i}}^\top \widehat \Delta^{(-k)}_{plm}}.
\end{align*}
For $\mathcal{D}'_{\mathcal{I}_{-k}}  = \br{\p{\Gamma_j, \Gamma_i Z_j, \Gamma_jX_j, \Gamma_jY_j} \mid j \in \mathcal{I}_{-k}}$, by Lemma~\ref{lemma conditional independence}, for any $i=1,\dots,n$,
\begin{align}
    (1-\Gamma_i)X_i \perp \mathcal{D}'_{\mathcal{I}_{-k}} | \Gamma_{1:n}. \label{plm conditional independence between 1-Gamma_i X_i and L_I_-k}
\end{align}
 Then for any $i \in [n]$, under Assumption~\ref{assumption mcar},
 \begin{align*}
     &\E\sbr{\p{1-\Gamma_i}\p{X_i - \E\sbr{X_i}}^\top \widehat \Delta^{(-k)}_{plm} \mid \mathcal{D}'_{\mathcal{I}_{-k}}, \Gamma_{1:n}} = \p{1-\Gamma_i}\E\sbr{X_i - \E\sbr{X_i} \mid  \Gamma_{1:n}}^\top \widehat \Delta^{(-k)}_{plm}= 0.
 \end{align*}
 and
\begin{align*}
& \E\sbr{\br{n^{-1}\sum_{i=1}^n\p{1-\Gamma_i}\p{X_i - \E\sbr{X_i}}^\top \widehat \Delta^{(-k)}_{plm}}^2 \mid  \mathcal{D}'_{\mathcal{I}_{-k}}, \Gamma_{1:n} }\\
    &\qquad\leq n^{-2}\sum_{i=1}^n \E\sbr{ \p{1-\Gamma_i} \p{X_i^\top \widehat \Delta^{(-k)}_{plm}}^2 \mid  \mathcal{D}'_{\mathcal{I}_{-k}}, \Gamma_{1:n} }=\p{n^{-2}\sum_{i=1}^n\p{1-\Gamma_i}}\widehat \Delta^{(-k),\top}_{plm}\E\sbr{X_iX_i^\top}\widehat \Delta^{(-k)}_{plm}\\
    &\qquad\leq \p{n^{-1}\sum_{i=1}^n\p{1-\Gamma_i}}^2 \widehat \Delta^{(-k),\top}_{plm}\E\sbr{X_iX_i^\top}\widehat \Delta^{(-k)}_{plm}=\p{1-\widehat \gamma}^2\widehat \Delta^{(-k),\top}_{plm}\E\sbr{X_iX_i^\top}\widehat \Delta^{(-k)}_{plm}.
\end{align*}

We also have another bound:
\begin{align*}
    &\E\sbr{\br{n^{-1}\sum_{i=1}^n\p{1-\Gamma_i}\p{X_i - \E\sbr{X_i}}^\top \widehat \Delta^{(-k)}_{plm}}^2 \mid  \mathcal{D}'_{\mathcal{I}_{-k}}, \Gamma_{1:n} }\\
    &\qquad\leq \p{n^{-2}\sum_{i=1}^n\p{1-\Gamma_i}}\widehat \Delta^{(-k),\top}_{plm}\E\sbr{X_iX_i^\top}\widehat \Delta^{(-k)}_{plm}= n^{-1}\p{1-\widehat \gamma}\widehat \Delta^{(-k),\top}_{plm}\E\sbr{X_iX_i^\top}\widehat \Delta^{(-k)}_{plm}.
\end{align*}
By \eqref{1-widehat gamma rate} and \eqref{plm bound for the quadratic linear form}, 
\begin{align*}
    n^{-1}\sum_{i=1}^n\p{1-\Gamma_i}\p{X_i - \E\sbr{X_i}}^\top \widehat \Delta^{(-k)}_{plm} = o_p\p{\p{1-\gamma_n}\land \p{n^{-1/2}\p{1-\gamma_n}^{1/2}}}.
\end{align*}
By \eqref{variance gamma(k) -1 ratio rate}, we have
\begin{align*}
    \frac{1-\widehat \gamma^{(k)}}{1-\widehat \gamma} - 1 = O_p\p{\gamma_n\p{n\p{1-\gamma_n}}^{-1}+\gamma_n^{1/2}\p{n\p{1-\gamma_n}}^{-1/2}},
\end{align*}
then
\begin{align*}
    T_{2ak} &= \p{\frac{1-\widehat \gamma^{(k)}}{1-\widehat \gamma}-1}\br{n^{-1}\sum_{i=1}^n\p{1-\Gamma_i}\p{X_i - \E\sbr{X_i}}^\top \widehat \Delta^{(-k)}_{plm}}+ n^{-1}\sum_{i=1}^n\p{1-\Gamma_i}\p{X_i - \E\sbr{X_i}}^\top \widehat \Delta^{(-k)}_{plm}\\
    &=o_p\sbr{\br{\gamma_n\p{n\p{1-\gamma_n}}^{-1}}\p{1-\gamma_n}+\br{\gamma_n^{1/2}\p{n\p{1-\gamma_n}}^{-1/2}}\p{1-\gamma_n} +n^{-1/2}\p{1-\gamma_n}^{1/2}}\\
    &= o_p\p{\p{\gamma_n \vee \p{1-\gamma_n}}^{1/2}n^{-1/2}}=o_p\p{n^{-1/2}}.
\end{align*}
Consider $T_{2bk}$, we have
\begin{align*}
    \E\sbr{T_{2bk} \mid \mathcal{L}_{\mathcal{I}_{-k}} } &= n_k^{-1}\sum_{i \in \mathcal{I}_k} \E\sbr{\p{1-\Gamma_i}\p{\widehat D^{(-k)}\p{Z_i} - \E_Z\sbr{\widehat D^{(-k)}\p{Z}}}\mid\mathcal{L}_{\mathcal{I}_{-k}} }\\
    &=n_k^{-1}\sum_{i \in \mathcal{I}_k} \E\sbr{\p{1-\Gamma_i}}\E\sbr{\p{\widehat D^{(-k)}\p{Z_i} - \E_Z\sbr{\widehat D^{(-k)}\p{Z}}}\mid \mathcal{L}_{\mathcal{I}_{-k}} }=0,
\end{align*}
and 
\begin{align*}
    \E\sbr{T_{2bk}^2 \mid \mathcal{L}_{\mathcal{I}_{-k}} } &= n_k^{-2}\sum_{i \in \mathcal{I}_k}\E\sbr{\p{1-\Gamma_i}\p{\widehat D^{(-k)}\p{Z_i} - \E_Z\sbr{\widehat D^{(-k)}\p{Z}}}^2\mid \mathcal{L}_{\mathcal{I}_{-k}} }\\
    &\leq n_k^{-2}\sum_{i \in \mathcal{I}_k}\E\sbr{\p{1-\Gamma_i}}\E\sbr{\p{\widehat D^{(-k)}\p{Z_i}}^2\mid \mathcal{L}_{\mathcal{I}_{-k}} }= \p{1-\gamma_n}n_{k}^{-1}\E_Z\sbr{\p{\widehat D^{(-k)}\p{Z}}^2 }.
\end{align*}
Thus, under Assumption~\ref{assumption mcar plm (a)}(b),
$
    T_{2bk} = o_p\p{\p{1-\gamma_n}^{1/2}n^{-1/2}}.
$
Together with $T_{2ak} = o_p\p{n^{-1/2}}$, we have
$
    T_2 = K^{-1}\sum_{k=1}^kT_{2ak} +  K^{-1}\sum_{k=1}^kT_{2bk} = o_p\p{n^{-1/2}}.
$

For $T_3$, by Lemma~\ref{lemma oracle plm properties}, we have
$
    \E\sbr{n^{-1}\sum_{i=1}^n\Gamma_i\epsilon_i} = n^{-1}\sum_{i=1}^n\E\sbr{\Gamma_i\epsilon_i} = \E\sbr{\Gamma_i}\E\sbr{\epsilon_i} = 0
$
and 
\begin{align*}
\E\sbr{\p{n^{-1}\sum_{i=1}^n\Gamma_i\epsilon_i}^2} &= n^{-2}\sum_{i=1}^n \E\sbr{\Gamma_i\epsilon_i^2}=n^{-1}\E\sbr{\Gamma_i}\E\sbr{\epsilon_i^2}=n^{-1} \gamma_n \E\sbr{\epsilon_i^2}\leq C\sigma_{\epsilon}^2n^{-1}\gamma_n.
\end{align*}
Then by Chebyshev's inequality,
$
    n^{-1}\sum_{i=1}^n\Gamma_i\epsilon_i = O_p\p{n^{-1/2}\gamma_n^{1/2}}.
$
By Lemma~\ref{lemma gamma ratio convergence}, we have $\widehat \gamma^{-1} - \gamma_n^{-1} = O_p\p{\gamma_n^{-1}\p{n\gamma_n}^{-1/2}\p{1-\gamma_n}^{1/2}}$, then
$
    T_3 = \p{\widehat \gamma^{-1} - \gamma_n^{-1}}n^{-1}\sum_{i=1}^n\Gamma_i\epsilon_i = O_p\p{\p{1-\gamma_n}^{1/2}\p{n\gamma_n}^{-1}}.
$

In conclusion, we have
\begin{align*}
    \widehat \theta_{plm} - \widetilde \theta_{plm}
    &= T_1 + T_2 + T_3\\
    &= o_p\p{\p{n\gamma_n}^{-1/2}} + o_p\p{n^{-1/2}} + O_p\p{\p{1-\gamma_n}^{1/2}\p{n\gamma_n}^{-1}}\\
    &= o_p\p{\p{n\gamma_n}^{-1/2}}.
\end{align*}
In addition, by Lemma~\ref{lemma the plm oracle asymptotics under MCAR}, we have
\begin{align*}
    \widehat \theta_{plm} -\theta &= \widehat \theta_{plm}- \widetilde \theta_{plm} + \widetilde \theta_{plm} - \theta = o_p\p{\p{n\gamma_n}^{-1/2}} + O_p\p{\p{n\gamma_n}^{-1/2}} = O_p\p{\p{n\gamma_n}^{-1/2}}.
\end{align*}
\end{proof}

\begin{lemma}\label{lemma plm variance estimator consistency}
    Let Assumptions~\ref{assumption mcar} and \ref{assumption mcar plm (a)} hold. If $n\gamma_n \gg 1$, then as $n, d \rightarrow \infty$,
$
        \widehat \sigma^2_{plm} = \sigma^2_{plm}\br{1 + o_p\p{1}}.
$
\end{lemma}
\begin{proof}
    Let $\epsilon_i = Y_i -  X_i^\top \beta^*_{plm} - f^*\p{Z_i}$, 
    \begin{align*}
        g_{i, plm}^{(-k)} &= \Gamma_i\p{X_i^\top\widehat \beta^{(-k)}_{plm} + \widehat f^{(-k)}\p{Z_i} } + \frac{\Gamma_i\p{Y_i - X_i^\top\widehat \beta^{(-k)}_{plm} - \widehat f^{(-k)}\p{Z_i} }}{\sum_{i\in\mathcal{I}_{k}} \Gamma_i},\\
        g_{i, plm}^* &= \Gamma_i \p{X_i^\top \beta^*_{plm} + f^*\p{Z_i}} + \frac{\Gamma_i}{\gamma_n}\epsilon_i,\quad
         \widetilde \sigma_{plm}^2 = n^{-1}\sum_{i=1}^n \p{X_i^\top \beta^*_{plm} + f^*\p{Z_i} + \frac{\Gamma_i}{\gamma_n}\epsilon_i}^2 - \theta^2.
    \end{align*}
Then we have
    \begin{align*}
        \widetilde \sigma_{plm}^2 &=n^{-1}\sum_{i=1}^n \p{g_{i, plm}^*}^2 + n^{-1}\sum_{i=1}^n\p{1-\Gamma_i}\p{X_i^\top \beta^*_{plm}}^2+ n^{-1}\sum_{i=1}^n \p{1-\Gamma_i}\p{ f^*\p{Z_i}}^2\\
        &\qquad + 2n^{-1}\sum_{i=1}^n \p{1-\Gamma_i}\p{X_i^\top \beta^*_{plm}} f^*\p{Z_i} - \theta^2.
    \end{align*}
Note that 
\begin{align*}
    \widehat \sigma_{plm}^2 &= n^{-1}\sum_{k=1}^K\sum_{i\in\mathcal{I}_{k}} \br{\Gamma_i\p{X_i^\top\widehat \beta^{(-k)}_{plm} + \widehat f^{(-k)}\p{Z_i} } + \frac{\Gamma_i\p{Y_i - X_i^\top\widehat \beta^{(-k)}_{plm} - \widehat f^{(-k)}\p{Z_i} }}{\sum_{i\in\mathcal{I}_{k}} \Gamma_i} }^2 \notag \\
    &\qquad+ n^{-1}\sum_{k=1}^K\sum_{i\in\mathcal{I}_{k}}\p{1-\Gamma_i} \widehat \beta^{(-k), \top}_{plm} \Bar{\Xi}_0 \widehat \beta^{(-k)}_{plm} + n^{-1}\sum_{k=1}^K\sum_{i\in\mathcal{I}_{k}}\p{1-\Gamma_i} \p{\widehat f^{(-k)}\p{Z_i}}^2 \notag \\
    &\qquad+ 2n^{-1}\sum_{k=1}^K\sum_{i\in\mathcal{I}_{k}}\p{1-\Gamma_i} \p{X_i^\top\widehat \beta^{(-k)}_{plm}}\widehat f^{(-k)}\p{Z_i}  - \p{\widehat \theta_{plm}}^2.
\end{align*}
Thus, 
\begin{align*}
    \widehat \sigma_{plm}^2 - \widetilde \sigma_{plm}^2 =& R_1 + R_2 - R_3 + R_4  + 2R_5,
\end{align*}
where 
\begin{align*}
    R_1 &= n^{-1}\sum_{k=1}^K\sum_{i\in\mathcal{I}_{k}}\br{ \p{g_{i, plm}^{(-k)}}^2 - \p{g_{i, plm}^*}^2},\quad
    R_2 = n^{-1}\sum_{k=1}^K\sum_{i\in\mathcal{I}_{k}}\p{1-\Gamma_i} \br{\p{\widehat f^{(-k)}\p{Z_i}}^2 - \p{f^*\p{Z_i}}^2},\\
    R_3 &= \p{\widehat \theta_{plm}}^2 - \theta^2,\quad
R_4 = n^{-1}\sum_{k=1}^K\sum_{i\in\mathcal{I}_{k}} \p{1-\Gamma_i} \p{\widehat \beta^{(-k), \top}_{plm} \Bar{\Xi}_0 \widehat \beta^{(-k)}_{plm} - \p{X_i^\top \beta^*_{plm}}^2},\\
    R_5 &= n^{-1}\sum_{k=1}^K\sum_{i\in\mathcal{I}_{k}}\p{1-\Gamma_i} \br{\p{\bar X_0^\top\widehat \beta^{(-k)}_{plm}}\widehat f^{(-k)}\p{Z_i} - \p{X_i^\top \beta^*_{plm}} f^*\p{Z_i}}.
\end{align*}
Let $\widehat \gamma = n^{-1}\sum_{i=1}^n \Gamma_i$, $\widehat \gamma^{\p{k}} = n_k^{-1}\sum_{i\in \mathcal{I}_{k}}^n \Gamma_i$, $\widehat \Delta^{(-k)}_{plm} = \widehat \beta^{(-k)}_{plm} - \beta^*_{plm}$, and $\widehat D^{(-k)} = \widehat f^{(-k)} - f^*$.
Since
\begin{align*}
    &\p{g_{i, plm}^{(-k)} - g_{i, plm}^*}^2 = \Gamma_i \br{ \p{1-\frac{1}{\widehat \gamma^{(k)}}}\p{X_i^\top \widehat \Delta^{(-k)}_{plm} + \widehat D^{(-k)}\p{Z_i}} + \p{\frac{1}{\widehat \gamma^{(k)}} - \frac{1}{\gamma_n}}\epsilon_i }^2\\
    &\qquad\leq 3\Gamma_i \p{1-\frac{1}{\widehat \gamma^{(k)}}}^2\p{X_i^\top \widehat \Delta^{(-k)}_{plm}}^2 + 3\Gamma_i \p{1-\frac{1}{\widehat \gamma^{(k)}}}^2\p{\widehat D^{(-k)}\p{Z_i}}^2 + 3\Gamma_i \p{\frac{1}{\widehat \gamma^{(k)}} - \frac{1}{\gamma_n}}^2\epsilon_i^2,
\end{align*}
we have
\begin{align*}
    n_k^{-1}\sum_{i\in\mathcal{I}_{k}}\p{g_{i, plm}^{(-k)} - g_{i, plm}^*}^2 &\leq 3\p{1-\frac{1}{\widehat \gamma^{(k)}}}^2 n_k^{-1}\sum_{i\in\mathcal{I}_{k}}\Gamma_i\p{X_i^\top \widehat \Delta^{(-k)}_{plm}}^2\\
    &\qquad + 3\p{1-\frac{1}{\widehat \gamma^{(k)}}}^2 n_k^{-1}\sum_{i\in\mathcal{I}_{k}}\Gamma_i\p{\widehat D^{(-k)}\p{Z_i}}^2+ 3\p{\frac{1}{\widehat \gamma^{(k)}} - \frac{1}{\gamma_n}}^2 n_k^{-1}\sum_{i\in\mathcal{I}_{k}}\Gamma_i\epsilon_i^2.
\end{align*}
By \eqref{plm bound for the quadratic linear form}, we have
\begin{align*}
    \E\sbr{n_k^{-1}\sum_{i\in\mathcal{I}_{k}}\Gamma_i\p{X_i^\top \widehat \Delta^{(-k)}_{plm}}^2 \mid \mathcal{L}_{\mathcal{I}_{-k}}} &= n_k^{-1}\sum_{i\in\mathcal{I}_{k}}\E\sbr{\Gamma_i}\E\sbr{\p{X_i^\top \widehat \Delta^{(-k)}_{plm}}^2 \mid \mathcal{L}_{\mathcal{I}_{-k}}} \leq \gamma_n \sigma\norm{\widehat \Delta^{(-k)}_{plm}}_2^2.
\end{align*}
Thus, by Assumption~\ref{assumption mcar plm (a)}(b) and Lemma~\ref{lemma convergence of conditional random variable},
$
    n_k^{-1}\sum_{i\in\mathcal{I}_{k}}\Gamma_i\p{X_i^\top \widehat \Delta^{(-k)}_{plm}}^2 = o_p\p{\gamma_n}.
$
Similarly, since
\begin{align*}
\E\sbr{n_k^{-1}\sum_{i\in\mathcal{I}_{k}}\Gamma_i\p{\widehat D^{(-k)}\p{Z_i}}^2 \mid \mathcal{L}_{\mathcal{I}_{-k}} } &= n_k^{-1}\sum_{i\in\mathcal{I}_{k}}\E\sbr{\Gamma_i}\E\sbr{\p{\widehat D^{(-k)}\p{Z_i}}^2 \mid \widehat D^{(-k)} }\\
    &=\gamma_n\E_Z\sbr{\p{\widehat f^{(-k)}\p{Z} - f^*\p{Z}}^2},
\end{align*}
by Assumption~\ref{assumption mcar plm (a)}(b) and Lemma~\ref{lemma convergence of conditional random variable},
$
    n_k^{-1}\sum_{i\in\mathcal{I}_{k}}\Gamma_i\p{\widehat D^{(-k)}\p{Z_i}}^2 = o_p\p{\gamma_n}.
$
Under Assumption~\ref{assumption mcar plm (a)}(a), for some constant $C_1>0$,
$
\E\sbr{n_k^{-1}\sum_{i\in\mathcal{I}_{k}}\Gamma_i\epsilon_i^2} = n_k^{-1}\sum_{i\in\mathcal{I}_{k}}\E\sbr{\Gamma_i}\E\sbr{\epsilon_i^2}\leq C_1\sigma_\epsilon^2 \gamma_n.
$
Then by Markov inequality, 
$
    n_k^{-1}\sum_{i\in\mathcal{I}_{k}}\Gamma_i\epsilon_i^2 = O_p\p{\gamma_n}.
$
By Lemma~\ref{lemma gamma ratio convergence}, we have 
\begin{align*}
    \frac{1}{\gamma_n}-\frac{1}{\widehat \gamma^{(k)}} = O_p\p{\gamma_n^{-1}\p{n\gamma_n}^{-1/2}\p{1-\gamma_n}^{1/2}},
\end{align*}
which implies
\begin{align*}
    1-\frac{1}{\widehat \gamma^{(k)}} = 1 - \frac{1}{\gamma_n} + \frac{1}{\gamma_n}-\frac{1}{\widehat \gamma^{(k)}} = O_p\p{\gamma_n^{-1}} + O_p\p{\gamma_n^{-1}\p{n\gamma_n}^{-1/2}\p{1-\gamma_n}^{1/2}} =  O_p\p{\gamma_n^{-1}}.
\end{align*}
Combine all the rates together,
\begin{align*}
    n_k^{-1}\sum_{i\in\mathcal{I}_{k}}\p{g_{i, plm}^{(-k)} - g_{i, plm}^*}^2 &= o_p\p{\gamma_n^{-2}\gamma_n} + O_p\p{\gamma_n^{-2}\p{n\gamma_n}^{-1}\p{1-\gamma_n}\gamma_n}=o_p\p{\gamma_n^{-1}},
\end{align*}
which implies 
\begin{align}\label{plm variance estimator rate of sum (g-k - g)^2}
    n^{-1}\sum_{k=1}^K\sum_{i\in\mathcal{I}_{k}}\p{g_{i, plm}^{(-k)} - g_{i, plm}^*}^2 &= K^{-1}\sum_{k=1}^Kn_k^{-1}\sum_{i\in\mathcal{I}_{k}}\p{g_{i, plm}^{(-k)} - g_{i, plm}^*}^2 =o_p\p{\gamma_n^{-1}}.
\end{align}
Moreover, since $X_i^\top \beta^*_{plm}$, $f^*(Z_i)$, and $\epsilon_i$ have bounded second moments, there exists some constant $C_2>0$ such that
\begin{align*}
\E\sbr{n^{-1}\sum_{i=1}^n \p{g_{i, plm}^*}^2}&= n^{-1}\sum_{i=1}^n \E\sbr{\Gamma_i}\E\sbr{ \p{X_i^\top \beta^*_{plm} + f^*\p{Z_i} + \gamma_n^{-1}\epsilon_i }^2}\\
&\leq \gamma_n n^{-1}\sum_{i=1}^n \br{3\E\sbr{ \p{X_i^\top \beta^*_{plm}}^2}+ 3\E\sbr{ \p{f^*\p{Z_i}}^2} + 3\E\sbr{ \p{\gamma_n^{-1}\epsilon_i }^2}}\\
&\leq \gamma_n\p{3C_2\sigma^2 + 3C_2\sigma_f^2 + 3C_2\gamma_n^{-2}\sigma_\epsilon^2},
\end{align*}
then by Markov inequality,
$
    n^{-1}\sum_{i=1}^n \p{g_{i, plm}^*}^2 = O_p\p{\gamma_n^{-1}}.
$

For $R_1$, note that $a^2 - b^2 = \p{a-b}^2 + 2b\p{a-b}$, then 
\begin{align*}
    R_1 &= n^{-1}\sum_{k=1}^K\sum_{i\in\mathcal{I}_{k}}\p{g_{i, plm}^{(-k)} - g_{i, plm}^*}^2 + 2n^{-1}\sum_{k=1}^K\sum_{i\in\mathcal{I}_{k}}g_{i, plm}^*\p{g_{i, plm}^{(-k)} - g_{i, plm}^*}\\
    &\leq n^{-1}\sum_{k=1}^K\sum_{i\in\mathcal{I}_{k}}\p{g_{i, plm}^{(-k)} - g_{i, plm}^*}^2 + 2\br{n^{-1}\sum_{i=1}^n \p{g_{i, plm}^*}^2 }^{1/2}\br{n^{-1}\sum_{k=1}^K\sum_{i\in\mathcal{I}_{k}}\p{g_{i, plm}^{(-k)} - g_{i, plm}^*}^2}^{1/2}\\
    &=o_p\p{\gamma_n^{-1} } + o_p\p{\gamma_n^{-1/2}\gamma_n^{-1/2}}=o_p\p{\gamma_n^{-1} }.
\end{align*}

For $R_2$, we have
\begin{align*}
    R_2 &= n^{-1}\sum_{k=1}^K\sum_{i\in\mathcal{I}_{k}}\p{1-\Gamma_i} \p{\widehat f^{(-k)}\p{Z_i} - f^*\p{Z_i}}^2  + 2n^{-1}\sum_{k=1}^K\sum_{i\in\mathcal{I}_{k}}\p{1-\Gamma_i}f^*\p{Z_i}\p{\widehat f^{(-k)}\p{Z_i} - f^*\p{Z_i}}\\
    &\leq n^{-1}\sum_{k=1}^K\sum_{i\in\mathcal{I}_{k}}\p{1-\Gamma_i} \p{\widehat f^{(-k)}\p{Z_i} - f^*\p{Z_i}}^2\\
    &\qquad + 2\br{n^{-1}\sum_{i=1}^n\p{1-\Gamma_i}\p{f^*\p{Z_i}}^2n^{-1}\sum_{k=1}^K\sum_{i\in\mathcal{I}_{k}}\p{1-\Gamma_i}\p{\widehat f^{(-k)}\p{Z_i} - f^*\p{Z_i}}^2}^{1/2}.
\end{align*}
Since
\begin{align*}
    &\E\sbr{n_k^{-1}\sum_{i\in\mathcal{I}_{k}}\p{1-\Gamma_i} \p{\widehat f^{(-k)}\p{Z_i} - f^*\p{Z_i}}^2 \mid \mathcal{L}_{\mathcal{I}_{-k}}}\\
    &\qquad= n_k^{-1}\sum_{i\in\mathcal{I}_{k}}\E\sbr{1-\Gamma_i}\E\sbr{ \p{\widehat f^{(-k)}\p{Z_i} - f^*\p{Z_i}}^2 \mid \widehat f^{(-k)}}= (1-\gamma_n)\E_Z\sbr{\p{\widehat f^{(-k)}\p{Z} - f^*\p{Z}}^2},
\end{align*}
under Assumption~\ref{assumption mcar plm (a)}(b), by Lemma~\ref{lemma convergence of conditional random variable}, it follows that
\begin{align*}
&n^{-1}\sum_{k=1}^K\sum_{i\in\mathcal{I}_{k}}\p{1-\Gamma_i} \p{\widehat f^{(-k)}\p{Z_i} - f^*\p{Z_i}}^2 \\
&\qquad= K^{-1}\sum_{k=1}^Kn_k^{-1}\sum_{i\in\mathcal{I}_{k}}\p{1-\Gamma_i} \p{\widehat f^{(-k)}\p{Z_i} - f^*\p{Z_i}}^2 = o_p\p{1-\gamma_n}.
\end{align*}
Since $f^*\p{Z_i}$ has a bounded second moment, we have
$$
    \E\sbr{n^{-1}\sum_{i=1}^n\p{1-\Gamma_i}\p{f^*\p{Z_i}}^2} = n^{-1}\sum_{i=1}^n\E\sbr{1-\Gamma_i}\E\sbr{\p{f^*\p{Z_i}}^2} \leq C_2\sigma_f^2 \p{1-\gamma_n},
$$
then by Markov inequality, 
$
    n^{-1}\sum_{i=1}^n\p{1-\Gamma_i}\p{f^*\p{Z_i}}^2 = O_p\p{1-\gamma_n}.
$
Thus,
\begin{align*}
    R_2 = o_p\p{1-\gamma_n} + o_p\p{\p{1-\gamma_n}^{1/2}\p{1-\gamma_n}^{1/2}} = o_p\p{1-\gamma_n}.
\end{align*}

For $R_3$, by Lemma~\ref{lemma plm mean estimator consistency} we have
\begin{align*}
    R_3 &= \p{\widehat \theta_{plm}}^2 - \theta^2=\p{\widehat \theta_{plm} - \theta}^2 + 2\theta\p{\widehat \theta_{plm} - \theta}=O_p\p{\p{n\gamma_n}^{-1/2} },
\end{align*}

For $R_4$, define
\begin{align*}
    R_{4k} := n_k^{-1}\sum_{i\in \mathcal{I}_k}\p{1-\Gamma_i}\widehat \beta^{(-k),\top}_{plm}\Bar \Xi_0 \widehat \beta^{(-k)}_{plm} - n_k^{-1}\sum_{i\in \mathcal{I}_k}(1-\Gamma_i)\p{X_i^\top \beta^*_{plm}}^2.
\end{align*}
Then we have
\begin{align*}
    R_{4k} &= \p{\frac{1-\widehat \gamma^{(k)}}{1-\widehat \gamma}} n^{-1}\sum_{i=1}^n \p{1-\Gamma_i}\p{X_i^\top \widehat \Delta^{(-k)}_{plm} + X_i^\top\beta^*_{plm}}^2 - n_k^{-1}\sum_{i\in \mathcal{I}_k}(1-\Gamma_i)\p{X_i^\top \beta^*_{plm}}^2\\
    &= \p{\frac{1-\widehat \gamma^{(k)}}{1-\widehat \gamma} - 1} \p{R_{4k,1} +  R_{4k,2} + R_{4k,3}} + R_{4k,1} + R_{4k,3} + R_{4k,4} + R_{4k,5},
\end{align*}
where
\begin{align*}
    R_{4k,1} &= n^{-1}\sum_{i=1}^n \p{1-\Gamma_i}\p{X_i^\top \widehat \Delta^{(-k)}_{plm}}^2,\quad
    R_{4k,2} = n^{-1} \sum_{i=1}^n \p{1-\Gamma_i}\p{ X_i^\top\beta^*_{plm} }^2,\\
    R_{4k,3} &= 2n^{-1} \sum_{i=1}^n \p{1-\Gamma_i}\p{X_i^\top \widehat \Delta^{(-k)}_{plm}} \p{X_i^\top\beta^*_{plm}},\\
    R_{4k,4} &= n^{-1} \sum_{i=1}^n \p{1-\Gamma_i}\p{ X_i^\top\beta^*_{plm}}^2 - \E\sbr{\p{1-\Gamma_i}\p{ X_i^\top\beta^*_{plm}}^2},\\
    R_{4k,5} &= \E\sbr{\p{1-\Gamma_i}\p{ X_i^\top\beta^*_{plm}}^2} - n_k^{-1}\sum_{i\in \mathcal{I}_k}(1-\Gamma_i)\p{X_i^\top \beta^*_{plm}}^2.
\end{align*}
With \eqref{plm bound for the quadratic linear form} and \eqref{plm conditional independence between 1-Gamma_i X_i and L_I_-k}, we have
\begin{align*}
    &\E\sbr{R_{4k,1} \mid \mathcal{D}'_{\mathcal{I}_{-k}}, \Gamma_{1:n}}  = \E\sbr{n^{-1}\sum_{i=1}^n \p{1-\Gamma_i}\p{X_i^\top \widehat \Delta^{(-k)}_{plm}}^2 \mid \mathcal{D}'_{\mathcal{I}_{-k}}, \Gamma_{1:n}} \\
    &\qquad =n^{-1}\sum_{i=1}^n \p{1-\Gamma_i}\widehat \Delta^{(-k),\top}_{plm}\E\sbr{X_iX_i^\top}\widehat \Delta^{(-k)}_{plm} \leq \tilde \kappa_u \p{1-\widehat \gamma}\norm{\widehat \Delta^{(-k)}_{plm}}_2^2.
\end{align*}
By \eqref{1-widehat gamma rate}, $1 - \widehat \gamma = O_p\p{1-\gamma_n}$, then under Assumption~\ref{assumption mcar plm (a)}(b) and by Lemma~\ref{lemma convergence of conditional random variable}, we have
\begin{align}\label{plm variance B_2k term 1}
    R_{4k,1} = n^{-1}\sum_{i=1}^n \p{1-\Gamma_i}\p{X_i^\top \widehat \Delta^{(-k)}_{plm}}^2 = o_p\p{1-\gamma_n}.
\end{align}
Since 
$\E\sbr{\p{1-\Gamma_i}\p{ X_i^\top\beta^*_{plm}}^2} = \p{1-\gamma_n}\E\sbr{\p{ X_i^\top\beta^*_{plm}}^2} \leq C_2\sigma^2 \p{1-\gamma_n},$
by Markov inequality,
\begin{align}
    R_{4k,2} = n^{-1} \sum_{i=1}^n \p{1-\Gamma_i}\p{ X_i^\top\beta^*_{plm}}^2 = O_p\p{1-\gamma_n }. \label{plm variance B_2k term 2}
\end{align}
Then with \eqref{plm variance B_2k term 1}, \eqref{plm variance B_2k term 2}, and by Cauchy-Schwarz inequality, 
\begin{align}
    \abs{R_{4k,3}} &\leq n^{-1} \sum_{i=1}^n \abs{\p{1-\Gamma_i}\p{X_i^\top \widehat \Delta^{(-k)}_{plm}} \p{X_i^\top\beta^*_{plm}} }\notag\\
    &\leq \br{n^{-1} \sum_{i=1}^n\p{1-\Gamma_i}\p{X_i^\top \widehat \Delta^{(-k)}_{plm}}^2}^{1/2}\br{n^{-1} \sum_{i=1}^n\p{1-\Gamma_i}\p{X_i^\top\beta^*_{plm}}^2}^{1/2}= o_p\p{1-\gamma_n}. \label{plm variance B_2k term 3}
\end{align}
Moreover, under Assumption~\ref{assumption mcar plm (a)}(a) the fourth moment of $X_i^\top \beta^*$ is bounded, then
\begin{align*}
    &\E\sbr{\br{n^{-1} \sum_{i=1}^n \p{1-\Gamma_i}\p{ X_i^\top\beta^*_{plm}}^2 - \E\sbr{\p{1-\Gamma_i}\p{ X_i^\top\beta^*_{plm}}^2}}^2}\\
    &\qquad\leq n^{-1} \E\sbr{\p{1-\Gamma_i}\p{ X_i^\top\beta^*_{plm}}^4} = n^{-1}\p{1-\gamma_n} \E\sbr{\p{ X_i^\top\beta^*_{plm}}^4}\leq \sigma^4n^{-1}\p{1-\gamma_n}.
\end{align*}
By Chebyshev's inequality, we have 
\begin{align}
    R_{4k,4} &= n^{-1} \sum_{i=1}^n \p{1-\Gamma_i}\p{ X_i^\top\beta^*_{plm}}^2 - \E\sbr{\p{1-\Gamma_i}\p{ X_i^\top\beta^*_{plm}}^2} = O_p\p{n^{-1/2}\p{1-\gamma_n}^{1/2}}. \label{plm variance B_2k term 4}
\end{align}
Similarly, we have
\begin{align}
    R_{4k,5} &= n_k^{-1} \sum_{i\in \mathcal{I}_k} \p{1-\Gamma_i}\p{ X_i^\top\beta^*_{plm}}^2 - \E\sbr{\p{1-\Gamma_i}\p{ X_i^\top\beta^*_{plm}}^2} = O_p\p{n^{-1/2}\p{1-\gamma_n}^{1/2}}. \label{plm variance B_2k term 5}
\end{align}
By \eqref{variance gamma(k) -1 ratio rate} we have
\begin{align*}
    \frac{1-\widehat \gamma^{(k)}}{1-\widehat \gamma} - 1 = O_p\p{\gamma_n\p{n\p{1-\gamma_n}}^{-1}} + O_p\p{\gamma_n^{1/2}\p{n\p{1-\gamma_n}}^{-1/2}}.
\end{align*}
Together with \eqref{plm variance B_2k term 1}-\eqref{plm variance B_2k term 5}, we have
\begin{align*}
    R_{4k} =& o_p\sbr{\br{\gamma_n\p{n\p{1-\gamma_n}}^{-1} + \gamma_n^{1/2}\p{n\p{1-\gamma_n}}^{-1/2}}\p{1-\gamma_n}}\\
    &+ O_p\sbr{\br{\gamma_n\p{n\p{1-\gamma_n}}^{-1} + \gamma_n^{1/2}\p{n\p{1-\gamma_n}}^{-1/2}}\p{1-\gamma_n}}\\
    &+ o_p\sbr{\br{\gamma_n\p{n\p{1-\gamma_n}}^{-1} + \gamma_n^{1/2}\p{n\p{1-\gamma_n}}^{-1/2}}\p{1-\gamma_n}}\\
    &+o_p\p{1-\gamma_n} + o_p\p{1-\gamma_n}+O_p\p{n^{-1/2}\p{1-\gamma_n}^{1/2}} + O_p\p{n^{-1/2}\p{1-\gamma_n}^{1/2}}\\
    =& o_p\p{1-\gamma_n} + O_p\p{n^{-1}\gamma_n} + O_p\p{n^{-1/2}\p{1-\gamma_n}^{1/2}}.
\end{align*}
Therefore, 
\begin{align*}
    R_4 = K^{-1}\sum_{k=1}^K R_{4k} = o_p\p{1-\gamma_n} + O_p\p{n^{-1}\gamma_n} + O_p\p{n^{-1/2}\p{1-\gamma_n}^{1/2}}.
\end{align*}

For $R_5$, note that $ab - cd = (a-c)(b-d) + d(a-c) + c(b-d)$. Then
\begin{align*}
    R_5 =& n^{-1}\sum_{k=1}^K\sum_{i\in\mathcal{I}_{k}}\p{1-\Gamma_i} \br{\p{\bar X_0^\top\widehat \beta^{(-k)}_{plm}}\widehat f^{(-k)}\p{Z_i} - \p{\bar X_0^\top \beta^*_{plm}} f^*\p{Z_i}}\\
    =& K^{-1}\sum_{k=1}^K\p{1-\widehat \gamma^{(k)} }\br{\bar X_0^\top \widehat \Delta_{plm}^{(-k)} \widehat D^{(-k)}\p{Z_i} + \bar X_0^\top\widehat \Delta_{plm}^{(-k)} f^*\p{Z_i} + \bar X_0^\top \beta^* \widehat D^{(-k)}\p{Z_i}}.
\end{align*}
Observe that
\begin{align*}
    &\bar X_0^\top \widehat \Delta_{plm}^{(-k)} \widehat D^{(-k)}\p{Z_i} + \bar X_0^\top\widehat \Delta_{plm}^{(-k)} f^*\p{Z_i} + \bar X_0^\top \beta^* \widehat D^{(-k)}\p{Z_i}\\
    &\qquad= \p{1-\widehat \gamma}n^{-1}\sum_{i=1}^n\p{1-\Gamma_i}\p{X_i^\top \widehat \Delta_{plm}^{(-k)} \widehat D^{(-k)}\p{Z_i} + X_i^\top\widehat \Delta_{plm}^{(-k)} f^*\p{Z_i} + X_i^\top \beta^* \widehat D^{(-k)}\p{Z_i}}.
\end{align*}
By Minkowski's inequality and H{\"o}lder inequality,
\begin{align*}
    &\norm{\p{1-\Gamma_i}\p{X_i^\top \widehat \Delta_{plm}^{(-k)} \widehat D^{(-k)}\p{Z_i} + X_i^\top\widehat \Delta_{plm}^{(-k)} f^*\p{Z_i} + X_i^\top \beta^* \widehat D^{(-k)}\p{Z_i}}}_{\P, 1} \\
    &\quad\leq \norm{\p{1-\Gamma_i}X_i^\top \widehat \Delta_{plm}^{(-k)} \widehat D^{(-k)}\p{Z_i} }_{\P, 1}  + \norm{\p{1-\Gamma_i}X_i^\top\widehat \Delta_{plm}^{(-k)} f^*\p{Z_i}}_{\P, 1} + \norm{ \p{1-\Gamma_i}X_i^\top \beta^* \widehat D^{(-k)}\p{Z_i}}_{\P, 1} \\
    &\quad\leq \norm{\p{1-\Gamma_i}X_i^\top \widehat \Delta_{plm}^{(-k)}}_{\P, 2} \norm{D^{(-k)}\p{Z_i}}_{\P, 2} + \norm{\p{1-\Gamma_i}X_i^\top\widehat \Delta_{plm}^{(-k)}}_{\P, 2}\norm{f^*\p{Z_i}}_{\P, 2} \\
    &\quad\quad + \norm{X_i^\top \beta^*}_{\P, 2}\norm{\widehat D^{(-k)}\p{Z_i}}_{\P, 2}\\
    &\quad\overset{(i)}{\leq} C\sigma\norm{\p{1-\Gamma_i}\norm{ \widehat \Delta_{plm}^{(-k)}}_2}_{\P, 2}\p{ \norm{D^{(-k)}\p{Z_i}}_{\P, 2} + \norm{f^*\p{Z_i}}_{\P, 2}} + C\sigma \norm{\widehat D^{(-k)}\p{Z_i}}_{\P, 2}\overset{(ii)}{=}o_p\p{1},
\end{align*}
where in $(i)$ we use $\E\sbr{\p{1-\Gamma_i}\p{X_i^\top \widehat \Delta_{plm}^{(-k)}}^2 \mid \Gamma_i} = \p{1-\Gamma_i}\widehat \Delta_{plm}^{(-k),\top}\E\sbr{X_iX_i^\top}\widehat \Delta_{plm}^{(-k)} \leq C\sigma^2\norm{\widehat \Delta_{plm}^{(-k)}}_2^2$ by \eqref{plm conditional independence between 1-Gamma_i X_i and L_I_-k}, and in $(ii)$ we use Assumption~\ref{assumption mcar plm (a)}(b). By Markov inequality, it follows that
$
    R_5 = o_p\p{1}.
$

In conclusion, we have
\begin{align*}
     &\widehat \sigma_{plm}^2 - \widetilde \sigma_{plm}^2 = R_1 - R_2 - R_3 + R_4  + R_5\\
     &\qquad= o_p\p{\gamma_n^{-1}+\p{1-\gamma_n}} + O_p\p{\p{n\gamma_n}^{-1/2} } + o_p\p{1-\gamma_n} + O_p\p{n^{-1}\gamma_n+n^{-1/2}\p{1-\gamma_n}^{1/2}} + o_p\p{1}\\
     &\qquad= o_p\p{\gamma_n^{-1}}.
\end{align*}

Since $\sigma_{plm}^2 = \E\sbr{\p{X_i^\top \beta^*_{plm} + f^*\p{Z_i} + {\Gamma_i}/{\gamma_n}\epsilon_i}^2} - \theta^2$, we have
\begin{align*}
        &\E\sbr{\widetilde \sigma_{plm}^2 -  \sigma_{plm}^2} = n^{-1}\sum_{i=1}^n \E\sbr{ \p{X_i^\top \beta^*_{plm} + f^*\p{Z_i} + \frac{\Gamma_i}{\gamma_n}\epsilon_i}^2} - E\sbr{\p{X_i^\top \beta^*_{plm} + f^*\p{Z_i} + \frac{\Gamma_i}{\gamma_n}\epsilon_i}^2}=0,
    \end{align*}
and 
\begin{align*}
    \E\sbr{\p{\widetilde \sigma_{plm}^2 -  \sigma_{plm}^2}^2} &= n^{-2}\sum_{i=1}^n\Var\sbr{\p{X_i^\top \beta^*_{plm} + f^*\p{Z_i} + \frac{\Gamma_i}{\gamma_n}\epsilon_i}^2}\\
    &\leq n^{-1}\E\sbr{\p{X_i^\top \beta^*_{plm} + f^*\p{Z_i} + \frac{\Gamma_i}{\gamma_n}\epsilon_i}^4}\\
    &\leq 27n^{-1}\E\sbr{\p{X_i^\top \beta^*_{plm}}^4} + 27n^{-1}\E\sbr{\p{f^*\p{Z_i}}^4} + 27n^{-1}\E\sbr{\p{\frac{\Gamma_i}{\gamma_n}\epsilon_i}^4},
\end{align*}
where the last inequality comes from H{\"o}lder inequality. Note that $\E\sbr{\p{X_i^\top \beta^*_{plm}}^4} \leq \sigma^4$, $\E\sbr{\p{f^*\p{Z_i}}^4} \leq \sigma_f^4$, and 
$
    \E\sbr{\p{{\Gamma_i}/{\gamma_n}\epsilon_i}^4} = \gamma_n^{-4}\E\sbr{\Gamma_i}\E\sbr{\epsilon_i^4}\leq \gamma_n^{-3} \sigma_{\epsilon}^4.
$
Thus,
$
    \E\sbr{\p{\widetilde \sigma_{plm}^2 -  \sigma_{plm}^2}^2} = O_p\p{\gamma_n^{-2}\p{n\gamma_n}^{-1}}.
$
By Chebyshev's inequality, we have
$
    \widetilde \sigma_{plm}^2 -  \sigma_{plm}^2 = O_p\p{\gamma_n^{-1}\p{n\gamma_n}^{-1/2}}.
$
By Lemma~\ref{lemma the plm oracle asymptotics under MCAR}, we have $\sigma_{plm}^2 \asymp \gamma_n^{-1}$, then it follows that
\begin{align*}
    \widehat \sigma_{plm}^2 -  \sigma_{plm}^2 &= \widehat \sigma_{plm}^2 - \widetilde \sigma_{plm}^2 + \widetilde \sigma_{plm}^2 -  \sigma_{plm}^2= o_p\p{\gamma_n^{-1}} + O_p\p{\gamma_n^{-1}\p{n\gamma_n}^{-1/2}}=\sigma_{plm}^2 o_p\p{1}.
\end{align*}
\end{proof}

\subsection{Proof of Theorem~\ref{plm Asymptotics theorem}}
\begin{proof}
By Lemmas~\ref{lemma plm mean estimator consistency} and \ref{lemma plm variance estimator consistency},
     $\widehat \theta_{plm} -\theta = O_p \p{(n\gamma_n)^{-1/2}}$ and 
    $\widehat \sigma_{plm}^2 = \sigma_{plm}^2\br{1 + o_p(1)}$. In addition, by Lemmas~\ref{lemma the plm oracle asymptotics under MCAR}, \ref{lemma plm mean estimator consistency}, and \ref{lemma plm variance estimator consistency}, we have $\widehat\sigma_{plm}^2 \asymp \gamma_n^{-1}$ and 
    \begin{align*}
        \widehat\sigma_{plm}^{-1}n^{1/2}(\widehat \theta_{plm} -\widetilde \theta_{plm}) = o_p\p{\gamma_n^{1/2}n^{1/2}\p{n\gamma_n}^{-1/2}} = o_p(1).
    \end{align*}
Then by Lemma~\ref{lemma the plm oracle asymptotics under MCAR} and Slutsky's theorem, 
$$\widehat\sigma_{plm}^{-1}n^{1/2}(\widehat \theta_{plm} - \theta) = \frac{\sigma_{plm}}{\widehat\sigma_{plm}} \cdot \sigma_{plm}^{-1}n^{1/2}(\widetilde \theta_{plm} - \theta) + \widehat\sigma_{plm}^{-1}n^{1/2}(\widehat \theta_{plm} -\widetilde \theta_{plm}) \xrightarrow{d} \mc{N}(0,1).$$
\end{proof}

\section{Proof of results in Section~\ref{sec: mar}}
For simplicity, throughout this section, we ignore the superscript $(-k)$ and subscripts $PS$, $OR$ in the proof when the context is clear.

\subsection{Auxiliary lemmas for generalizability}

\begin{lemma}\label{lemma mar theta tilde normality}
   Let Assumptions~\ref{assumption mar}  and \ref{assumption mar nuisance (a)} hold. Assume either $\gamma_n\p{X} = g\p{X^\top \alpha^*_{PS}}$ or $\mu\p{X} = X^\top \beta^*_{OR}$ holds. If $n\gamma_n \gg 1$, then as $n, d \rightarrow \infty$, $\sigma_g^2 \asymp \gamma_n^{-1}$ and
    $
    \sigma^{-1}n^{1/2}\p{\widetilde \theta_g - \theta_g} \rightarrow \mathcal{N}\p{0,1},
    $
    where $\widetilde \theta_g = n^{-1}\sum_{i=1}^n \br{X_i^\top \beta^*_{OR} + { \Gamma_i}/{g\p{X_i^\top\alpha^*_{PS}}}\p{Y_i - X_i^\top\beta^*_{OR}}}$.
\end{lemma}
\begin{proof}
    Let $w_i = Y_i - X_i^\top \beta^*$. Define $V_i = X_i^\top \beta^* + {\Gamma_i}/{g\p{X_i^\top\alpha^*}}w_i = X_i^\top \beta^* + \Gamma_i\p{1+ \exp\p{-X_i^\top \alpha^*}}w_i.$
When either $\gamma_n\p{X} = g\p{X^\top \alpha^*}$ or $\mu\p{X} = X^\top \beta^*$,
    \begin{align*}
        \E\sbr{V_i} - \theta_g &= \E\sbr{\p{1 - \frac{\Gamma_i}{g\p{X_i^\top \alpha^*}}}\p{X_i^\top \beta^* - Y_i}}=\E\sbr{\E\sbr{\p{1 - \frac{\Gamma_i}{g\p{X_i^\top \alpha^*}}}\p{X_i^\top \beta^* - Y_i} \mid X_i}}\\
        &=\E\sbr{\p{1 - \frac{\E\sbr{\Gamma_i \mid X_i}}{g\p{X_i^\top \alpha^*}}}\p{X_i^\top \beta^* - \E\sbr{Y_i \mid X_i}}}= 0.
    \end{align*}
    By Lyapunov's central limit theorem, it suffices to prove for some $\delta>0$ and $C>0$,
    \begin{align}
        \lim_{n\rightarrow \infty}n^{-\delta/2}\sigma^{-\p{2+\delta}} \E\sbr{\abs{V_i - \theta_g}^{2+\delta}} = 0.
    \end{align}
    If $\E\sbr{Y_i\mid X_i} = X_i^\top \beta^*$, i.e., $\E\sbr{w_i\mid X_i} = 0$, then
    \begin{align*}
        \sigma^2_g &=  \E\sbr{\p{V_i - \theta_g}^2}=\E\sbr{\p{X_i^\top\beta^* - \theta_g}^2} + \E\sbr{\p{\frac{\Gamma_i}{g\p{X_i^\top \alpha^*}}w_i}^2}+ 2\E\sbr{\p{X_i^\top \beta^* - \theta_g}\frac{\Gamma_i}{g\p{X_i^\top \alpha^*}}w_i}\\
        &=\E\sbr{\p{X_i^\top\beta^* - \theta_g}^2} + \E\sbr{\p{\frac{\Gamma_i}{g\p{X_i^\top \alpha^*}}w_i}^2}+ 2\E\sbr{\E\sbr{\p{X_i^\top \beta^* - \theta_g}\frac{\Gamma_i}{g\p{X_i^\top \alpha^*}}\mid X_i}\E\sbr{w_i\mid X_i}}\\
        &=\E\sbr{\p{X_i^\top\beta^* - \theta_g}^2} + \E\sbr{\frac{\Gamma_i}{\p{g\p{X_i^\top \alpha^*}}^2}w_i^2} \overset{(i)}{\asymp} \gamma_n^{-1},
    \end{align*}
    where (i) uses $g\p{X_i^\top \alpha^*} \asymp \gamma_n$ by Assumption~\ref{assumption mar nuisance (a)}(a). By Lemma~\ref{sub-gaussian properties}(c) and Minkowski's inequality, 
$
       \norm{\p{X_i^\top\beta^* - \theta_g}^2}_{\P, 2} \leq \norm{\p{X_i^\top\beta^*}^2}_{\P, 2} + \abs{\theta_g} \leq C_1\sigma + \abs{\theta_g}$ and $
        \gamma_n\delta_w\leq \E\sbr{\Gamma_iw_i^2}= \gamma_n\E\sbr{w_i^2\mid \Gamma_i=1} \leq\gamma_n\sigma_w^2.
$
If $\P\p{\Gamma_i = 1\mid X_i} = g\p{X_i^\top \alpha^*}$, then
\begin{align*}
    \sigma^2_g &=  \E\sbr{\p{V_i - \theta_g}^2} =\Var\p{Y_i} + \E\sbr{\p{\frac{\Gamma_i}{g\p{X_i^\top \alpha^*}}-1}^2w_i^2}+ \E\sbr{\p{Y_i - \theta_g}\p{\frac{\Gamma_i}{g\p{X_i^\top \alpha^*}}-1}w_i}\\
    &=\Var\p{Y_i} + \E\sbr{\p{\frac{\Gamma_i}{g\p{X_i^\top \alpha^*}}-1}^2w_i^2}+ 2\E\sbr{\E\sbr{\frac{\Gamma_i}{g\p{X_i^\top \alpha^*}}-1 \mid X_i}\E\sbr{\p{Y_i - \theta_g}w_i \mid X_i}}\\
    &=\Var\p{Y_i} + \E\sbr{\frac{\Gamma_i}{\p{g\p{X_i^\top \alpha^*}}^2}w_i^2} -2\E\sbr{\frac{\Gamma_i}{g\p{X_i^\top \alpha^*}}w_i^2} + \E\sbr{w_i^2} \asymp \gamma_n^{-1}.
\end{align*}
By Lemma~\ref{sub-gaussian properties}(c), we have
$\E\sbr{\Gamma_iw_i^2}\asymp \gamma_n$, $\E\sbr{w_i^2} \leq \sigma_w^2$, and by Minkowski's inequality, 
$
    \Var\p{Y_i}^{1/2} \leq \norm{X_i^\top \beta^* + w_i}_{\P,2} \leq \norm{X_i^\top \beta^*}_{\P,2} + \norm{w_i}_{\P,2} \leq C_2{\sigma+\sigma_w}.
$
Thus, when either $\gamma_n\p{X} = g\p{X^\top \alpha^*}$ or $\mu\p{X} = X^\top \beta^*$, it holds that
\begin{align}
    \sigma^2_g \asymp \gamma_n^{-1}. \label{mar sigma2 asymp gamma-1}
\end{align}
On the other hand, by Minkowski's inequality, 
$
    \norm{V_i - \theta_g}_{\P, 2+\delta} \leq \norm{X_i^\top \beta^*}_{\P, 2+\delta}   + \norm{{\Gamma_i}/{g\p{X_i^\top \alpha^*}}w_i}_{\P, 2+\delta} + \abs{\theta_g}.
$
Choose $\delta=2$. By Lemma~\ref{sub-gaussian properties}(c), we have $\norm{X_i^\top \beta^*}_{\P, 2+\delta} \leq C_3\sigma$ and
\begin{align*}
    \norm{\frac{\Gamma_i}{g\p{X_i^\top \alpha^*}}w_i}_{\P, 2+\delta} &= \E\sbr{\frac{\Gamma_i}{\p{g\p{X_i^\top \alpha^*}}^{4}}w_i^{4}}^{1/4}\leq k_0^{-1}\gamma_n^{-1}\E\sbr{\Gamma_iw_i^{4}}^{1/4}\leq k_0^{-1}\sigma_w\gamma_n^{-3/4}.
\end{align*}
Thus, 
\begin{align*}
    &\E\sbr{\abs{V_i - \theta_g}^{2+\delta}} = \norm{V_i - \theta}_{\P, 2+\delta}^{{2+\delta}}\leq \p{\norm{X_i^\top \beta^*}_{\P, 2+\delta}   + \norm{\frac{\Gamma_i}{g\p{X_i^\top \alpha^*}}w_i}_{\P, 2+\delta} + \abs{\theta_g}}^{{2+\delta}}\\
    &\qquad\leq \p{\abs{\theta_g} + C_3\sigma + k_0^{-1}\sigma_w\gamma_n^{-3/4}}^4=\gamma_n^{-3}\p{\p{\abs{\theta_g} + C_3\sigma}\gamma_n^{3/4} + k_0^{-1}\sigma_w}^4\leq \p{\p{\abs{\theta_g} + C_3\sigma} + k_0^{-1}\sigma_w}^4\gamma_n^{-3}.
\end{align*}
Thus, there exists some $C>0$ such that for $\delta = 2$,
\begin{align*}
        \lim_{n\rightarrow \infty} n^{-\delta/2}\sigma^{-\p{2+\delta}} \E\sbr{\abs{V_i - \theta_g}^{2+\delta}} \leq \lim_{n\rightarrow \infty} n^{-1}C\gamma_n^{2}\gamma_n^{-3} = \lim_{n\rightarrow \infty} C\p{n\gamma_n}^{-1} = 0.
    \end{align*}
\end{proof}

\begin{lemma}\label{lemma infty norm Ui}
Let Assumption~\ref{assumption mar nuisance (a)} hold. If $n\gamma_n \gg \log n \log d$, then as $n, d \rightarrow \infty$, it holds that 
\begin{align*}
    \norm{n_k^{-1}\sum_{i\in \mathcal{I}_k}\br{\p{1-\frac{\Gamma_i}{g\p{X_i^\top \alpha^*}}}X_i}}_\infty = O_p\p{(\frac{\log d}{n\gamma_n})^{1/2}}.
\end{align*}  
\end{lemma}
\begin{proof}
Let 
$
    U_i = \p{1-{\Gamma_i}/{g\p{X_i^\top \alpha^*}}}X_i.
$
    By the construction of $\alpha^*$, we have
$
        \E\sbr{\p{1-{\Gamma_i}/{g\p{X_i^\top \alpha^*}}}X_i} = \boldsymbol{0},
$
which implies $\E\sbr{U_i} = \boldsymbol{0}$. By Lemma~\ref{sub-gaussian properties}(c), for any $1 \leq j \leq d$, we have
$
        \abs{U_i^\top e_j} \leq k_0^{-1}\gamma_n^{-1}\abs{X_i^\top e_j} \leq  k_0^{-1}\gamma_n^{-1}\abs{X_i^\top e_j},
$
    which implies
$
        \sup_{i\in \mathcal{I}_k}\sup_{1\leq j \leq d} \norm{U_i^\top e_j}_{\psi_2} \leq k_0^{-1}\gamma_n^{-1}\sigma.
$
    In addition, by \eqref{exp(-2X'alpha = gamma-2)},
    \begin{align*}
        &\E\sbr{\p{U_i^\top e_j}^2} =(1-\gamma_n)\E\sbr{\p{X_i^\top e_j}^2 \mid \Gamma_i=0} + \gamma_n\E\sbr{\p{1-\frac{1}{g\p{X_i^\top \alpha^*}}}^2\p{X_i^\top e_j}^2 \mid \Gamma_i=1}\\
    &\qquad\leq \E\sbr{\p{X_i^\top e_j}^2 \mid \Gamma_i=0} + k_0^{-2}\gamma_n^{-1}\E\sbr{\p{X_i^\top e_j}^2 \mid \Gamma_i=1}\leq 2\p{1+k_0^{-2}\gamma_n^{-1}}\sigma^2,
    \end{align*}
    which implies 
$
        \sup_{1\leq j \leq d} {1}/{n_k}\sum_{i\in \mathcal{I}_k}\E\sbr{\p{U_i^\top e_j}^2} \leq 2\p{1+k_0^{-2}\gamma_n^{-1}}\sigma^2.
$
Then by Lemma~\ref{Lemma Tail Bounds for Maximums}, when $n\gamma_n \gg \log n \log d$, we have
\begin{align*}
    \norm{\frac{1}{n_k}\sum_{i\in \mathcal{I}_k}U_i}_\infty = O_p\p{(\frac{\log d}{n_k\gamma_n})^{1/2} + \frac{\log d(\log n_k)^{1/2}}{n_k\gamma_n}} = O_p\p{(\frac{\log d}{n\gamma_n})^{1/2}}.
\end{align*}

\end{proof}

\begin{lemma}\label{lemma infty norm Ji}
Let Assumption~\ref{assumption mar nuisance (a)} hold. If $n\gamma_n \gg (\log n)^2 \log d$, then as $n, d \rightarrow \infty$, it holds that 
\begin{align*}
    \norm{n_k^{-1}\sum_{i\in \mathcal{I}_k}\Gamma_i\exp\p{-X_i^\top \alpha^*}w_iX_i}_\infty = O_p\p{(\frac{\log d}{n\gamma_n})^{1/2}}.
\end{align*}    
\end{lemma}
\begin{proof}
Let 
$
    J_i = \Gamma_i\exp\p{-X_i^\top \alpha^*}w_iX_i.
$
    By the construction of $\beta^*$,
$
    \E\sbr{\Gamma_i\exp\p{-X_i^\top \alpha^*}w_iX_i} = \boldsymbol{0},
$
which implies $\E\sbr{J_i} = \boldsymbol{0}$. By \eqref{exp(-X'alpha = gamma-1)}, for any $1 \leq j \leq d$, we have
$
        \abs{J_i^\top e_j} \leq k_0^{-1}\gamma_n^{-1}w_iX_i^\top e_j.
$
By Lemma~\ref{sub-gaussian properties}(d),
    \begin{align*}
        \sup_{i\in \mathcal{I}_k}\sup_{1\leq j \leq d} \norm{J_i^\top e_j}_{\psi_1} \leq k_0^{-1}\gamma_n^{-1}\norm{w_i}_{\psi_2}\norm{X_i^\top e_j}_{\psi_2} \leq k_0^{-1}\gamma_n^{-1} \sigma_w\sigma.
    \end{align*}
    In addition, by \eqref{exp(-2X'alpha = gamma-2)} and Lemma~\ref{sub-gaussian properties}(c), for some $C>0$,
    \begin{align*}
        \E\sbr{\p{J_i^\top e_j}^2} &= \E\sbr{\Gamma_i\exp\p{-2X_i^\top \alpha^*}w_i^2(X_i^\top e_j)^2}\leq \E\sbr{\Gamma_i\exp\p{-2X_i^\top \alpha^*}w_i^2\p{X_i^\top e_j}^2}\\
    &\leq k_0^{-2}\gamma_n^{-2}\E\sbr{\Gamma_iw_i^2\p{X_i^\top e_j}^2}=k_0^{-2}\gamma_n^{-1}\E\sbr{w_i^2\p{X_i^\top e_j}^2\mid \Gamma_i=1}\\
&\leq k_0^{-2}\gamma_n^{-1}\E\sbr{w_i^4\mid \Gamma_i=1}^{1/2}\E\sbr{\p{X_i^\top e_j}^4\mid \Gamma_i=1}^{1/2}\leq C\sigma_w^2 \sigma^2 k_0^{-2}\gamma_n^{-1}.
    \end{align*}
    which implies
\[
        \sup_{1\leq j \leq d} {1}/{n_k}\sum_{i\in \mathcal{I}_k}\E\sbr{\p{J_i^\top e_j}^2} \leq C\sigma_w^2 \sigma^2 k_0^{-2}\gamma_n^{-1}.
\]
By Lemma~\ref{Lemma Tail Bounds for Maximums}, when $n\gamma_n \gg (\log n)^2 \log d$,
\begin{align*}
    \norm{\frac{1}{n_k}\sum_{i\in \mathcal{I}_k}J_i}_\infty = O_p\p{(\frac{\log d}{n_k\gamma_n})^{1/2} + \frac{\log n_k\log d}{n_k\gamma_n}} = O_p\p{(\frac{\log d}{n\gamma_n})^{1/2}}.
\end{align*}
\end{proof}

\begin{lemma}\label{lemma gradient infty norm 1-Gamma X}
Let Assumption~\ref{assumption mar nuisance (a)} hold. If $n\gamma_n \gg \log n \log d$, then as $n, d \rightarrow \infty$, it holds that 
    \begin{align*}
        \norm{n_k^{-1}\sum_{i\in \mathcal{I}_k}\p{1-\Gamma_i}\p{X_i - \E\sbr{X_i\mid \Gamma_i=0}} }_\infty = O_p\p{(\frac{\log d}{n})^{1/2}}.
    \end{align*}
\end{lemma}
\begin{proof}
First we have 
   \begin{align*}
       \E\sbr{\p{1-\Gamma_i}(X_i - \E\sbr{X_i\mid \Gamma_i=0})} = (1-\gamma_n)\E\sbr{X_i - \E\sbr{X_i\mid \Gamma_i=0} \mid \Gamma_i=0} = \boldsymbol{0},
   \end{align*}
For each $1 \leq j \leq d$, by Lemma~\ref{sub-gaussian properties}(c), $\E\sbr{X_i^\top e_j\mid \Gamma_i=0}) \leq C_1\sigma$ for some $C_1>0$. Then by Lemma~\ref{sub-gaussian properties}(a), for some $C_2>0$,
\begin{align*}
    \norm{\p{1-\Gamma_i}(X_i - \E\sbr{X_i\mid \Gamma_i=0})^\top e_j}_{\psi_2} \leq \norm{X_i^\top e_j}_{\psi_2} + \norm{\E\sbr{X_i^\top e_j\mid \Gamma_i=0}}_{\psi_2} \leq C_2\sigma,
\end{align*}
which implies
\begin{align*}
    \max_{i\in\mc{I}_k}\max_{1\leq j \leq d} \norm{\p{1-\Gamma_i}(X_i - \E\sbr{X_i\mid \Gamma_i=0})^\top e_j}_{\psi_2} \leq C_2\sigma.
\end{align*}
In addition, by Lemma~\ref{sub-gaussian properties}(c),
\begin{align*}
    &\E\sbr{\p{\p{1-\Gamma_i}(X_i - \E\sbr{X_i\mid \Gamma_i=0})^\top e_j}^2}= (1-\gamma_n)\E\sbr{\p{(X_i - \E\sbr{X_i\mid \Gamma_i=0})^\top e_j}^2 \mid \Gamma_i=0}\\
    &\qquad\leq \E\sbr{(X_i^\top e_j)^2 \mid \Gamma_i=0}\leq 2\sigma^2.
\end{align*}
Then we have
\begin{align*}
    \max_{1\leq j \leq d}n_k^{-1}\sum_{i \in \mc{I}_k}\E\sbr{\p{\p{1-\Gamma_i}(X_i - \E\sbr{X_i\mid \Gamma_i=0})^\top e_j}^2} \leq 2\sigma^2.
\end{align*}
By Lemma~\ref{Lemma Tail Bounds for Maximums}, when $n \gg \log n \log d$, as $n, d \rightarrow \infty$,
\begin{align*}
    \norm{n_k^{-1}\sum_{i\in \mathcal{I}_k}\p{1-\Gamma_i}\p{X_i - \E\sbr{X_i\mid \Gamma_i=0}} }_\infty &= O_p\p{(\frac{\log d}{n_k})^{1/2} + \frac{\log d(\log n_k)^{1/2}}{n_k}}= O_p\p{(\frac{\log d}{n})^{1/2}}.
\end{align*}
\end{proof}

\begin{lemma}\label{lemma mar theta hat - theta tilde consistency}
    Let Assumption~\ref{assumption mar nuisance (a)} hold. Choose $\lambda_\alpha \asymp \lambda_\beta \asymp (\log d/(n\gamma_n))^{1/2}$. If $n\gamma_n \gg \max\{s_\alpha, s_\beta\log n$, $(\log n)^2\} \log d$ and  $s_\alpha s_\beta \ll (n\gamma_n)^{3/2}/(\log n(\log d)^2)$, then as $n, d\rightarrow \infty$, 
    \begin{align*}
        \widehat \theta_g - \widetilde \theta_g &= O_p\p{\frac{\p{s_\alpha + (s_\alpha s_\beta)^{1/2}} \log d}{n\gamma_n}},
    \end{align*}
In addition, let Assumption~\ref{assumption mar} hold. Assume that $\mu(X) = X^\top \beta^*_{OR}$ holds. Then as $n, d\rightarrow \infty$, 
\begin{align*}
\widehat \theta_g - \widetilde \theta_g&=O_p\p{\frac{(s_\alpha \log d)^{1/2}}{n\gamma_n} +\frac{(s_\alpha s_\beta)^{1/2} \log d}{n\gamma_n}}.
\end{align*}
\end{lemma}
\begin{proof}
    Let
\begin{align*}
    \Psi_i\p{\alpha, \beta}
    &= \Gamma_i X_i^\top \beta + { \Gamma_i}/{g\p{X_i^\top\alpha}}\p{Y_i - X_i^\top\beta},\\
    \widetilde \theta_g^{\p{k}}
    &= n_k^{-1}\sum_{i \in \mathcal{I}_{k}} \p{\p{1-\Gamma_i}\bar X_0^\top \beta^* + \Psi_i\p{\alpha^*, \beta^*}}.
\end{align*}
    Then $\widehat \theta_g - \widetilde \theta_g = K^{-1} \sum_{k=1}^K \p{\widehat \theta_g^{\p{k}} - \widetilde \theta_g^{\p{k}}}$. Note that
    \begin{align*}
         \widehat \theta_g^{(k)} = n_k^{-1}\sum_{i\in \mathcal{I}_k} \p{(1-\Gamma_i)\Bar{X}_0^\top \widehat \beta^{(-k)} + \Psi_i\p{\widehat \alpha^{(-k)}, \widehat \beta^{(-k)}} }.
    \end{align*}
Let $\widehat \gamma_{k} = n_k^{-1}\sum_{i \in \mc{I}_k} \Gamma_i$, $\widehat \gamma = n^{-1}\sum_{i=1}^n \Gamma_i$, $w_i = Y_i - X_i^\top \beta^*$, $\widehat \Delta_\alpha^{(-k)} = \widehat \alpha^{(-k)} - \alpha^*$, and $\widehat \Delta_\beta^{(-k)} = \widehat \beta^{(-k)} - \beta^*$, $\widetilde \Delta_\beta^{(-k)} = \widetilde \beta^{(-k)} - \beta^*$, and $\Delta^{(-k)} = \widehat \beta^{(-k)} - \widetilde \beta^{(-k)}$. Then for each $1\leq k\leq K$, we have
    \begin{align*}
        \widehat \theta_g^{\p{k}} - \widetilde \theta_g^{\p{k}} &= n_k^{-1}\sum_{i\in \mathcal{I}_k} \br{\Psi_i\p{\widehat \alpha^{(-k)}, \widehat \beta^{(-k)}} - \Psi_i\p{\alpha^*, \beta^*}} +  n_k^{-1}\sum_{i\in \mathcal{I}_k} (1-\Gamma_i)\Bar{X}_0^\top \widehat \Delta_{\beta}^{(-k)}\\
        &= r_1 - r_2 + r_3 - r_4 + r_5,
    \end{align*}
where
\begin{align*}
    r_1 &= n_k^{-1}\sum_{i\in \mathcal{I}_k}\p{1-\frac{\Gamma_i}{g\p{X_i^\top\alpha^*}}}X_i^\top  \widehat \Delta_{\beta}^{(-k)},\quad
    r_2 = n_k^{-1}\sum_{i\in \mathcal{I}_k}\p{\frac{ \Gamma_i}{g\p{X_i^\top\widehat \alpha^{(-k)}}} -\frac{\Gamma_i}{g\p{X_i^\top\alpha^*}} }X_i^\top  \widehat \Delta_{\beta}^{(-k)},\\
    r_3 &= n_k^{-1}\sum_{i\in \mathcal{I}_k}\p{\frac{ \Gamma_i}{g\p{X_i^\top\widehat \alpha^{(-k)}}} -\frac{\Gamma_i}{g\p{X_i^\top\alpha^*}} }w_i,\quad
    r_4 = n_k^{-1}\sum_{i\in \mathcal{I}_k}\p{1-\Gamma_i}\p{X_i - \E\sbr{X_i\mid \Gamma_i=0}}^\top\widehat \Delta_{\beta}^{(-k)},\\
    r_5 &= \p{1-\widehat\gamma_k}\p{1-\widehat \gamma}^{-1}n^{-1}\sum_{i=1}^n \p{1-\Gamma_i}\p{X_i -\E\sbr{X_i\mid \Gamma_i=0}}^\top\widehat \Delta_{\beta}^{(-k)}.
\end{align*}

For $r_1$, we have $r_1 = r_{11} + r_{12}$, where
\begin{align*}
    r_{11} &= n_k^{-1}\sum_{i\in \mathcal{I}_k}\p{1-\frac{\Gamma_i}{g\p{X_i^\top\alpha^*}}}X_i^\top   \Delta^{(-k)},\quad
    r_{12} = n_k^{-1}\sum_{i\in \mathcal{I}_k}\p{1-\frac{\Gamma_i}{g\p{X_i^\top\alpha^*}}}X_i^\top  \widetilde \Delta_{\beta}^{(-k)}.
\end{align*}
By Lemma~\ref{lemma infty norm Ui}, we have
\begin{align*}
    r_{11} \leq \norm{n_k^{-1}\sum_{i\in \mathcal{I}_k}\p{1-\frac{\Gamma_i}{g\p{X_i^\top\alpha^*}}}X_i}_\infty   \norm{\Delta^{(-k)}}_1 = O_p\p{\norm{\Delta^{(-k)}}_1(\frac{\log d}{n\gamma_n})^{1/2}}.
\end{align*}

 Since $\widetilde \Delta_\beta^{(-k)} \perp \mc{D}_{\mc{I}_k}$, by first order optimal of $\alpha^*$, for any $i \in \mc{I}_k$, we have 
\begin{align*}
    \E\sbr{\p{1-\frac{\Gamma_i}{g\p{X_i^\top \alpha^*}}}X_i^\top\widetilde \Delta_\beta^{(-k)} \mid \widetilde \beta^{(-k)}}= \E\sbr{\p{1-\Gamma_i - \Gamma_i\exp\p{-X_i^\top \alpha^*}}X_i}^\top\widetilde \Delta_\beta^{(-k)} = 0.
\end{align*}
Then we have
\begin{align*}
&\E\sbr{r_{12}^2 \mid \widetilde \beta^{(-k)}} = n_k^{-1}\E\sbr{\br{\p{1-\frac{\Gamma_i}{g\p{X_i^\top \alpha^*}}}X_i^\top\widetilde \Delta_\beta^{(-k)}}^2 \mid \widetilde \beta^{(-k)}}\\
&\qquad= n_k^{-1}(1-\gamma_n)\E\sbr{\p{X_i^\top\widetilde \Delta_\beta^{(-k)}}^2 \mid \Gamma_i=0, \widetilde \beta^{(-k)}}\\
&\qquad\qquad+ n_k^{-1}\gamma_n\E\sbr{\p{1-\frac{1}{g\p{X_i^\top \alpha^*}}}^2\p{X_i^\top\widetilde \Delta_\beta^{(-k)}}^2 \mid \Gamma_i=1, \widetilde \beta^{(-k)}}.
\end{align*}
By Assumption~\ref{assumption mar nuisance (a)}(a), we have $1-(g\p{X_i^\top \alpha^*})^{-1} \leq k_0^{-1}\gamma_n^{-1}$. By Lemma~\ref{sub-gaussian properties}(c),
\[
    \E\sbr{\p{X_i^\top\widetilde \Delta_\beta^{(-k)}}^2 \mid \Gamma_i,\widetilde \beta^{(-k)}} \leq 2\sigma^2\norm{\widetilde \Delta_\beta^{(-k)}}_2^2.
\]
Thus, by Lemma~\ref{lemma convergence of conditional random variable},
\[
    r_{12} = O_p\p{(n\gamma_n)^{-1/2}\norm{\widetilde \Delta_\beta^{(-k)}}_2}.
\]
It follows that
\begin{align*}
    r_1 = O_p\p{\norm{\Delta^{(-k)}}_1(\frac{\log d}{n\gamma_n})^{1/2} + (n\gamma_n)^{-1/2}\norm{\widetilde \Delta_\beta^{(-k)}}_2}.
\end{align*}

For $r_2$, by H{\"o}lder inequality, we have
\begin{align*}
    r_2 \leq \br{n_k^{-1}\sum_{i\in \mathcal{I}_k}\Gamma_i\p{\frac{ 1}{g\p{X_i^\top\widehat \alpha^{(-k)}}} -\frac{1}{g\p{X_i^\top\alpha^*}} }^2}^{1/2}\br{n_k^{-1}\sum_{i\in \mathcal{I}_k}\Gamma_i\p{X_i^\top  \widehat \Delta_{\beta}^{(-k)}}^2}^{1/2}
\end{align*}

By Taylor's theorem and \eqref{exp(-X'alpha = gamma-1)}, for some $t \in (0,1)$,
\begin{align*}
    \p{\frac{1}{g\p{X_i^\top \widehat \alpha^{(-k)}}} - \frac{1}{g\p{X_i^\top \alpha^*}}}^2 &= \p{\exp\p{-tX_i^\top \widehat \Delta_\alpha^{(-k)} - X_i^\top\alpha^*}X_i^\top \widehat \Delta_\alpha^{(-k)}}^2\notag \\
    &\leq k_0^{-2}\gamma_n^{-2}\exp\p{-2\abs{X_i^\top \widehat \Delta_\alpha^{(-k)}}}\p{X_i^\top \widehat \Delta_\alpha^{(-k)}}^2. 
\end{align*}
Then by Lemma~\ref{lemma Exp(Xn2) ineq}, 
\begin{align*}
    n_k^{-1}\sum_{i\in \mathcal{I}_k}\Gamma_i\p{\frac{ 1}{g\p{X_i^\top\widehat \alpha^{(-k)}}} -\frac{1}{g\p{X_i^\top\alpha^*}} }^2 = O_p\p{\gamma_n^{-1}\norm{\widehat \Delta_\alpha^{(-k)}}_2^2}.
\end{align*}
By Lemma~\ref{lemma conditional independence}, we have $\mc{D}_{\mc{I}_k} \perp \widehat \Delta_{\beta}^{(-k)} \mid \Gamma_{\mc{I}_k}$, then for any $i \in \mc{I}_k$, by Lemma~\ref{sub-gaussian properties}(c), 
\begin{align*}
    \E\sbr{\p{X_i^\top  \widehat \Delta_{\beta}^{(-k)}}^2 \mid \Gamma_{\mc{I}_k}, \widehat \Delta_{\beta}^{(-k)}} = \E\sbr{\p{X_i^\top  \widehat \Delta_{\beta}^{(-k)}}^2 \mid \Gamma_i, \widehat \Delta_{\beta}^{(-k)}} \leq 2\sigma^2\norm{\widehat \Delta_\beta^{(-k)}}_2^2.
\end{align*}
Since $\widehat \gamma_k = O_p(\gamma_n)$ due to Lemma~\ref{lemma concentrate gamma}, then
$
    n_k^{-1}\sum_{i\in \mathcal{I}_k}\Gamma_i\p{X_i^\top  \widehat \Delta_{\beta}^{(-k)}}^2 = O_p\p{\gamma_n\norm{\widehat \Delta_\beta^{(-k)}}_2^2}.
$
Thus,
$
    r_2 = O_p\p{\norm{\widehat \Delta_\alpha^{(-k)}}_2\norm{\widehat \Delta_\beta^{(-k)}}_2}.
$

For $r_3$, by Taylor's theorem, there exists $\widetilde \alpha$ between $\widehat \alpha^{(-k)}$ and $\alpha^*$ such that
\begin{align*}
    r_3 &= n_k^{-1}\sum_{i\in \mathcal{I}_k}\Gamma_i\br{-\exp\p{-X_i^\top \alpha^*}\p{X_i^\top\widehat \Delta_\alpha^{(-k)}} + \exp\p{-X_i^\top\widetilde \alpha}\p{X_i^\top\widehat \Delta_\alpha^{(-k)}}^2}w_i=-r_{31} + r_{32},
\end{align*}
where
\begin{align*}
    r_{31} &=n_k^{-1}\sum_{i\in \mathcal{I}_k}\Gamma_i\exp\p{-X_i^\top \alpha^*}w_i\p{X_i^\top\widehat \Delta_\alpha^{(-k)}},\quad
    r_{32} = n_k^{-1}\sum_{i\in \mathcal{I}_k}\Gamma_i\exp\p{-X_i^\top\widetilde \alpha}w_i\p{X_i^\top\widehat \Delta_\alpha^{(-k)}}^2.
\end{align*}
By Lemma~\ref{lemma infty norm Ji}, we have
\begin{align*}
r_{31}&\leq \norm{n_k^{-1}\sum_{i\in \mathcal{I}_k}\Gamma_i\exp\p{-X_i^\top \alpha^*}w_iX_i}_\infty \norm{\widehat \Delta_\alpha^{(-k)}}_1 =O_p\p{\norm{\widehat \Delta_\alpha^{(-k)}}_1(\frac{\log d}{n \gamma_n})^{1/2}}.
\end{align*}
In addition, by \eqref{exp(-X'alpha = gamma-1)},
$
\exp\p{-X_i^\top\widetilde \alpha} \leq  \exp\p{\abs{X_i^\top\widehat \Delta_\alpha^{(-k)}} - X_i^\top\alpha^*} \leq k_0^{-1}\gamma_n^{-1}\exp\p{X_i^\top\abs{\widehat \Delta_\alpha^{(-k)}}}.
$
By H{\"o}lder inequality, we have
\begin{align*}
   r_{32}&\leq k_0^{-1}\gamma_n^{-1}n_k^{-1}\sum_{i\in \mathcal{I}_k}\Gamma_i\exp\p{\abs{X_i^\top\widehat \Delta_\alpha^{(-k)}}}w_i\p{X_i^\top\widehat \Delta_\alpha^{(-k)}}^2\\
    &\leq k_0^{-1}\gamma_n^{-1}\br{n_k^{-1}\sum_{i\in \mathcal{I}_k}\Gamma_iw_i^2}^{1/2}\br{n_k^{-1}\sum_{i\in \mathcal{I}_k}\Gamma_i\exp\p{2\abs{X_i^\top\widehat \Delta_\alpha^{(-k)}}}\p{X_i^\top\widehat \Delta_\alpha^{(-k)}}^4}^{1/2}.
\end{align*}
Since $\E\sbr{\Gamma_iw_i^2} = \gamma_n \E\sbr{w_i^2\mid\Gamma_i=1} \leq \gamma_n\sigma_w^2$, then
$
    n_k^{-1}\sum_{i\in \mathcal{I}_k}\Gamma_iw_i^2 = O_p\p{\gamma_n}.
$
By Lemma~\ref{lemma Exp(Xn2) ineq}, we have
\begin{align*}
    n_k^{-1}\sum_{i\in \mathcal{I}_k}\Gamma_i\exp\p{2\abs{X_i^\top\widehat \Delta_\alpha^{(-k)}}}\p{X_i^\top\widehat \Delta_\alpha^{(-k)}}^4 = O_p\p{\gamma_n\norm{\widehat \Delta_\alpha^{(-k)}}_2^4}.
\end{align*}
Thus, we have
\begin{align*}
    n_k^{-1}\sum_{i\in \mathcal{I}_k}\Gamma_i\exp\p{-X_i^\top\widetilde \alpha}w_i\p{X_i^\top\widehat \Delta_\alpha^{(-k)}}^2 = O_p\p{\norm{\widehat \Delta_\alpha^{(-k)}}_2^2}.
\end{align*}
It follows that
$
    r_3 = O_p\p{\norm{\widehat \Delta_\alpha^{(-k)}}_1({\log d}/{n \gamma_n})^{1/2} + \norm{\widehat \Delta_\alpha^{(-k)}}_2^2}.
$

For $r_4$, we have $r_{4} = r_{41} + r_{42}$, where
\begin{align*}
    r_{41} &= n_k^{-1}\sum_{i\in \mathcal{I}_k}\p{1-\Gamma_i}\p{X_i - \E\sbr{X_i\mid \Gamma_i=0}}^\top\Delta^{(-k)},\\
    r_{42} &= n_k^{-1}\sum_{i\in \mathcal{I}_k}\p{1-\Gamma_i}\p{X_i - \E\sbr{X_i\mid \Gamma_i=0}}^\top\widetilde \Delta_{\beta}^{(-k)}.
\end{align*}
By Lemma~\ref{lemma gradient infty norm 1-Gamma X}, 
\begin{align*}
    r_{41} \leq \norm{n_k^{-1}\sum_{i\in \mathcal{I}_k}\p{1-\Gamma_i}\p{X_i - \E\sbr{X_i\mid \Gamma_i=0}}}_\infty \norm{\Delta^{(-k)}}_1 = O_p\p{\norm{\Delta^{(-k)}}_1(\frac{\log d}{n})^{1/2}}.
\end{align*}
Since $\widetilde \Delta_{\beta}^{(-k)} \perp (1-\Gamma_i)X_i $ for $i \in \mc{D}_{\mc{I}_k}$, then
\begin{align*}
    \E\sbr{r_{42} \mid \widetilde \Delta_{\beta}^{(-k)}} = n_k^{-1}\sum_{i\in \mathcal{I}_k}\E\sbr{\p{1-\Gamma_i}\p{X_i - \E\sbr{X_i\mid \Gamma_i=0}}^\top\widetilde \Delta_{\beta}^{(-k)} \mid \widetilde \Delta_{\beta}^{(-k)}} = 0,
\end{align*}
and for $i \in \mc{D}_{\mc{I}_k}$, by Lemma~\ref{sub-gaussian properties}(c),
\begin{align*}
    \E\sbr{r_{42}^2\mid \widetilde \Delta_{\beta}^{(-k)}} &= n_k^{-1}\E\sbr{\p{1-\Gamma_i}\p{\p{X_i - \E\sbr{X_i\mid \Gamma_i=0}}^\top\widetilde \Delta_{\beta}^{(-k)}}^2 \mid \widetilde \Delta_{\beta}^{(-k)}}\\
    &=n_k^{-1}(1-\gamma_n)\E\sbr{\p{\p{X_i - \E\sbr{X_i\mid \Gamma_i=0}}^\top\widetilde \Delta_{\beta}^{(-k)}}^2 \mid \Gamma_i=0, \widetilde \Delta_{\beta}^{(-k)}}\\
    &\leq n_k^{-1}\E\sbr{\p{X_i^\top\widetilde \Delta_{\beta}^{(-k)}}^2\mid \Gamma_i=0, \widetilde \Delta_{\beta}^{(-k)}}\leq 2\sigma^2 n_k^{-1}\norm{\widetilde \Delta_{\beta}^{(-k)}}_2^2.
\end{align*}
Thus, $r_{42} = O_p\p{n^{-1/2}\norm{\widetilde \Delta_{\beta}^{(-k)}}_2}$. It follows that
\begin{align*}
    r_{4} = O_p\p{\norm{\Delta^{(-k)}}_1(\frac{\log d}{n})^{1/2} + n^{-1/2}\norm{\widetilde \Delta_{\beta}^{(-k)}}_2}.
\end{align*}

For $r_5$, we have $r_5 = \p{1-\widehat\gamma_k}\p{1-\widehat \gamma}^{-1}(r_{51} + r_{52})$, where
\begin{align*}
    r_{51} &= n^{-1}\sum_{i=1}^n \p{1-\Gamma_i}\p{X_i -\E\sbr{X_i\mid \Gamma_i=0}}^\top\Delta^{(-k)},\\
    r_{52} &= n^{-1}\sum_{i=1}^n \p{1-\Gamma_i}\p{X_i -\E\sbr{X_i\mid \Gamma_i=0}}^\top\widetilde \Delta_{\beta}^{(-k)}.
\end{align*}
Similar to $r_{41}$, we have
$
    r_{51} = O_p\p{\norm{\Delta^{(-k)}}_1({\log d}/{n})^{1/2}}.
$
By Lemma~\ref{lemma conditional independence}, we have $\widetilde \Delta_{\beta}^{(-k)} \perp (1-\Gamma_i)X_i \mid \Gamma_{1:n}$ for $i \in [n]$. Since
\begin{align*}
    &\E\sbr{\p{1-\Gamma_i}(X_i -\E\sbr{X_i\mid \Gamma_i=0})^\top\widetilde \Delta_{\beta}^{(-k)} \mid \Gamma_{1:n}, \widetilde \Delta_{\beta}^{(-k)}} \\
    &\qquad= \p{1-\Gamma_i}\E\sbr{\p{X_i -\E\sbr{X_i\mid \Gamma_i=0}}^\top\widetilde \Delta_{\beta}^{(-k)} \mid \Gamma_i=0, \widetilde \Delta_{\beta}^{(-k)}} = 0,
\end{align*}
and for $i \neq j$,
\begin{align*}
    \p{1-\Gamma_i}(X_i -\E\sbr{X_i\mid \Gamma_i=0})  \perp \p{1-\Gamma_j}(X_j -\E\sbr{X_j\mid \Gamma_i=0}) \mid (\Gamma_i,\Gamma_j, \widetilde \Delta_{\beta}^{(-k)}),
\end{align*}
then by Lemma~\ref{sub-gaussian properties}(c), 
\begin{align*}
    \E\sbr{r_{52}^2\mid \Gamma_{1:n}, \widetilde \Delta_{\beta}^{(-k)}} &= n^{-2}\sum_{i=1}^n \E\sbr{\p{1-\Gamma_i}\p{(X_i -\E\sbr{X_i\mid \Gamma_i=0})^\top\widetilde \Delta_{\beta}^{(-k)}}^2 \mid \Gamma_{1:n}, \widetilde \Delta_{\beta}^{(-k)}}\\
    &\leq n^{-2}\sum_{i=1}^n \E\sbr{\p{1-\Gamma_i}\p{X_i^\top\widetilde \Delta_{\beta}^{(-k)}}^2 \mid \Gamma_{i}=0, \widetilde \Delta_{\beta}^{(-k)}}\\
    &\leq n^{-2}\sum_{i=1}^n\p{1-\Gamma_i}2\sigma^2\norm{ \widetilde \Delta_{\beta}^{(-k)}}_2^2= 2\sigma^2n^{-1}(1-\widehat\gamma)\norm{ \widetilde \Delta_{\beta}^{(-k)}}_2^2,
\end{align*}
which implies 
$
    r_{52} = O_p\p{n^{-1/2}\norm{ \widetilde \Delta_{\beta}^{(-k)}}_2}.
$
Thus,
$
    r_5 = O_p\p{\norm{\Delta^{(-k)}}_1({\log d}/{n})^{1/2} + n^{-1/2}\norm{ \widetilde \Delta_{\beta}^{(-k)}}_2}.
$

In conclusion, we have
\begin{align*}
    \widehat \theta_g^{\p{k}} - \widetilde \theta_g^{\p{k}} &= O_p\p{\norm{\Delta^{(-k)}}_1(\frac{\log d}{n\gamma_n})^{1/2} + (n\gamma_n)^{-1/2}\norm{\widetilde \Delta_\beta^{(-k)}}_2+\norm{\widehat \Delta_\alpha^{(-k)}}_2\norm{\widehat \Delta_\beta^{(-k)}}_2} \\
    &\qquad +O_p\p{\norm{\widehat \Delta_\alpha^{(-k)}}_1(\frac{\log d}{n \gamma_n})^{1/2} + \norm{\widehat \Delta_\alpha^{(-k)}}_2^2+\norm{\Delta^{(-k)}}_1(\frac{\log d}{n})^{1/2} + n^{-1/2}\norm{\widetilde \Delta_{\beta}^{(-k)}}_2}\\
    &\qquad + O_p\p{\norm{\Delta^{(-k)}}_1(\frac{\log d}{n})^{1/2} + n^{-1/2}\norm{ \widetilde \Delta_{\beta}^{(-k)}}_2}.
\end{align*}
By Proposition~\ref{proposition mar nuisance consistency body} and Lemma~\ref{lemma consistency of betahat and betatilde for product sparsity}, we have
\begin{align*}
    \widehat \theta_g^{\p{k}} - \widetilde \theta_g^{\p{k}} &=O_p\p{\norm{\widehat \Delta_\alpha^{(-k)}}_2^2 + \norm{\widehat \Delta_\alpha^{(-k)}}_2\norm{\widehat \Delta_\beta^{(-k)}}_2+\p{\norm{\Delta^{(-k)}}_1 +  \norm{\widehat \Delta_\alpha^{(-k)}}_1}(\frac{\log d}{n \gamma_n})^{1/2}}\\ &=O_p\p{  \frac{\p{s_\alpha + (s_\alpha s_\beta)^{1/2}} \log d}{n\gamma_n}}.
\end{align*}

Let $\E\sbr{Y_i \mid X_i} = X_i^\top \beta^*$. For $r_3$, by Lemma~\ref{lemma conditional independence}, $(\Gamma_iX_i, \Gamma_iY_i)\perp \widehat \alpha^{(-k)} \mid \Gamma_i$ for $i \in \mathcal{I}_k$. Then
\begin{align*}
    &\E\sbr{\p{\frac{ \Gamma_i}{g\p{X_i^\top\widehat \alpha^{(-k)}}} -\frac{\Gamma_i}{g\p{X_i^\top\alpha^*}} }w_i \mid \widehat \alpha^{(-k)}, \Gamma_i, X_i}=\p{\frac{ 1}{g\p{X_i^\top\widehat \alpha^{(-k)}}} -\frac{1}{g\p{X_i^\top\alpha^*}} }\E\sbr{\Gamma_iw_i\mid \Gamma_i, X_i}\\
    &\qquad\overset{(i)}{=} \Gamma_i\p{\frac{ 1}{g\p{X_i^\top\widehat \alpha^{(-k)}}} -\frac{1}{g\p{X_i^\top\alpha^*}} }\E\sbr{w_i\mid X_i} = 0,
\end{align*}
where we use Assumption~\ref{assumption mar} in $(i)$. Then by mean value theorem and \eqref{exp(-X'alpha = gamma-1)}, for some $t \in (0,1)$,
\begin{align*}
    &\E\sbr{r_3^2 \mid \Gamma_{1:n},\widehat \alpha^{(-k)}}= n_k^{-2}\sum_{i\in \mathcal{I}_k}\E\sbr{\p{\frac{ \Gamma_i}{g\p{X_i^\top\widehat \alpha^{(-k)}}} -\frac{\Gamma_i}{g\p{X_i^\top\alpha^*}} }^2w_i^2 \mid \Gamma_i=1, \widehat \alpha^{(-k)}}\\
    &=n_k^{-2}\sum_{i\in \mathcal{I}_k}\Gamma_i\E\sbr{\p{\exp(-2t\widehat \Delta_\alpha^{(-k)} - 2X_i^\top \alpha^*)}\p{X_i^\top \widehat \Delta_\alpha^{(-k)}}^2w_i^2 \mid \Gamma_i=1, \widehat \alpha^{(-k)}}\\
    &\leq k_0^{-2}\gamma_n^{-2}n_k^{-2}\sum_{i\in \mathcal{I}_k}\Gamma_i\E\sbr{\exp\p{2\abs{\widehat \Delta_\alpha^{(-k)} }}\p{X_i^\top \widehat \Delta_\alpha^{(-k)}}^2w_i^2\mid \Gamma_i=1, \widehat \alpha^{(-k)}}\\
    &\leq  k_0^{-2}\gamma_n^{-2}n_k^{-2}\sum_{i\in \mathcal{I}_k}\Gamma_i\E\sbr{\exp\p{4\abs{\widehat \Delta_\alpha^{(-k)} }}\p{X_i^\top \widehat \Delta_\alpha^{(-k)}}^4\mid \Gamma_i=1, \widehat \alpha^{(-k)}}^{1/2}\E\sbr{w_i^4\mid \Gamma_i=1}^{1/2}.
\end{align*}
By Lemma~\ref{lemma concentrate gamma} and Lemma~\ref{lemma Exp(Xn2) ineq}, we have
$
    r_3 = O_p\p{(n\gamma_n)^{-1/2}\norm{\widehat \Delta_\alpha^{(-k)}}_2}.
$
It follows that 
\begin{align*}
    \widehat \theta_g^{\p{k}} - \widetilde \theta_g^{\p{k}} &=O_p\p{\p{n\gamma_n}^{-1/2}\norm{\widehat \Delta_\alpha^{(-k)}}_2 +  \norm{\widehat \Delta_\alpha^{(-k)}}_2\norm{\widehat \Delta_\beta^{(-k)}}_2+\norm{\Delta^{(-k)}}_1(\frac{\log d}{n \gamma_n})^{1/2}}\\ &=O_p\p{ \frac{(s_\alpha \log d)^{1/2}}{n\gamma_n} +\frac{(s_\alpha s_\beta)^{1/2} \log d}{n\gamma_n}}.
\end{align*}
\end{proof}

\begin{lemma}\label{lemma mar i.i.d. quadratic form}
    Suppose $\left(\mathbf{X}_i\right)_{i=1}^m$ are i.i.d. sub-Gaussian random vectors. Then, for any (possibly random) $\boldsymbol{\Delta} \in \mathbb{R}^d$, as $m, d \rightarrow \infty$,
$$
\sup _{\boldsymbol{\Delta} \in \mathbb{R}^d /\{\mathbf{0}\}} \frac{m^{-1} \sum_{i=1}^m\left(\mathbf{X}_i^{\top} \boldsymbol{\Delta}\right)^2}{\|\boldsymbol{\Delta}\|_1^2 m^{-1} \log d+\|\boldsymbol{\Delta}\|_2^2}=O_p(1) .
$$
\end{lemma}
Lemma~\ref{lemma mar i.i.d. quadratic form} follows from Lemma C.7 of \cite{zhang2021dynamic}.

\begin{lemma}\label{lemma psi diff square}
Let Assumption~\ref{assumption mar nuisance (a)} hold. Choose $\lambda_\alpha \asymp (\log d/ (n\gamma_n))^{1/2}$. If $n\gamma_n \gg \max\br{s_{\alpha},\log n}\log d$, then
    \begin{align*}
        n_k^{-1}\sum_{i\in \mathcal{I}_k}\p{\Psi_i\p{\widehat \alpha^{(-k)}, \widehat \beta^{(-k)}} - \Psi_i\p{\alpha^*, \beta^*}}^2 = O_p\p{\gamma_n^{-1}\br{\norm{\widehat \Delta_\alpha^{(-k)}}_2^2 + \norm{\widehat \Delta_\beta^{(-k)}}_2^2}}.
    \end{align*}
\end{lemma}

\begin{proof}
From Lemma~\ref{lemma mar theta hat - theta tilde consistency},
\begin{align*}
    \Psi_i\p{\widehat \alpha^{(-k)}, \widehat \beta^{(-k)}} &= \Psi_i\p{\alpha^*, \beta^*} +\Gamma_i\p{1-\frac{1}{g\p{X_i^\top\alpha^*}}}X_i^\top  \widehat \Delta_{\beta}^{(-k)}\\
    &\qquad +\p{\frac{ \Gamma_i}{g\p{X_i^\top\widehat \alpha^{(-k)}}} -\frac{\Gamma_i}{g\p{X_i^\top\alpha^*}} }X_i^\top  \widehat \Delta_{\beta}^{(-k)} +\p{\frac{ \Gamma_i}{g\p{X_i^\top\widehat \alpha^{(-k)}}} -\frac{\Gamma_i}{g\p{X_i^\top\alpha^*}} }w_i.
\end{align*}
Then by Cauchy-Schwarz inequality,
\begin{align*}
    n_k^{-1}\sum_{i\in \mathcal{I}_k}\p{\Psi_i\p{\widehat \alpha^{(-k)}, \widehat \beta^{(-k)}} - \Psi_i\p{\alpha^*, \beta^*}}^2 &\leq 3(R_1 + R_2 + R_3),
\end{align*}
where
\begin{align*}
    R_1 &= n_k^{-1}\sum_{i\in \mathcal{I}_k}\Gamma_i\p{1-\frac{1}{g\p{X_i^\top\alpha^*}}}^2\p{X_i^\top  \widehat \Delta_{\beta}^{(-k)}}^2,\\
    R_2 &= n_k^{-1}\sum_{i\in \mathcal{I}_k}\p{\frac{ \Gamma_i}{g\p{X_i^\top\widehat \alpha^{(-k)}}} -\frac{\Gamma_i}{g\p{X_i^\top\alpha^*}} }^2\p{X_i^\top  \widehat \Delta_{\beta}^{(-k)}}^2,\\
    R_3 &= n_k^{-1}\sum_{i\in \mathcal{I}_k}\p{\frac{ \Gamma_i}{g\p{X_i^\top\widehat \alpha^{(-k)}}} -\frac{\Gamma_i}{g\p{X_i^\top\alpha^*}} }^2w_i^2
\end{align*}
By mean value theorem and \eqref{exp(-X'alpha = gamma-1)}, for some $t \in (0,1)$,
\begin{align*}
    \p{\frac{1}{g\p{X_i^\top\widehat \alpha^{(-k)}}} -\frac{1}{g\p{X_i^\top\alpha^*}}}^2 &= \p{\exp\p{-tX_i^\top \widehat\Delta_\alpha^{(-k)} - X_i^\top \alpha^*}\p{X_i^\top \widehat\Delta_\alpha^{(-k)}}}^2\\
    &\leq k_0^{-2}\gamma_n^{-2}\exp\p{2\abs{X_i^\top \widehat\Delta_\alpha^{(-k)}}}\p{X_i^\top \widehat\Delta_\alpha^{(-k)}}^2.
\end{align*}
Since $\p{1-{1}/{g\p{X_i^\top\alpha^*}}} \leq k_0^{-1}\gamma_n^{-1}$, by H{\"o}lder inequality,
\begin{align*}
   R_1 &\leq k_0^{-2}\gamma_n^{-2}\p{n_k^{-1}\sum_{i\in \mathcal{I}_k} \Gamma_i}^{1/2}R_a^{1/2},\quad
   R_2 \leq k_0^{-2}\gamma_n^{-2}R_a^{1/2}R_b^{1/2},\quad
   R_3 \leq k_0^{-2}\gamma_n^{-2}R_{b}^{1/2}R_{x}^{1/2},
\end{align*}
where
\begin{align*}
    R_a &= n_k^{-1}\sum_{i\in \mathcal{I}_k}\Gamma_i\p{X_i^\top  \widehat \Delta_{\beta}^{(-k)}}^4,\;\;
    R_b = n_k^{-1}\sum_{i\in \mathcal{I}_k}\Gamma_i\exp\p{4\abs{X_i^\top \widehat\Delta_\alpha^{(-k)}}}\p{X_i^\top \widehat\Delta_\alpha^{(-k)}}^4,\;\;
    R_c = n_k^{-1}\sum_{i\in \mathcal{I}_k}\Gamma_iw_i^4.
\end{align*}
By Lemma~\ref{lemma conditional independence}, $\Gamma_i X_i\perp \widehat \Delta_\beta^{(-k)} \mid \Gamma_i$ for $i \in \mc{I}_k$. By Lemma~\ref{sub-gaussian properties}(c), for some $C_1>0$,
\begin{align*}
    \E\sbr{\Gamma_i\p{X_i^\top\widehat \Delta_\beta^{(-k)}}^4 \mid \Gamma_i,\widetilde \beta^{(-k)}} \leq \Gamma_iC_1\sigma^2\norm{\widehat \Delta_\beta^{(-k)}}_2^4.
\end{align*}
Since
\begin{align*}
    \E\sbr{R_a \mid \Gamma_{1:n},\widetilde \beta^{(-k)}} = n_k^{-1}\sum_{i\in \mathcal{I}_k}\E\sbr{\Gamma_i\p{X_i^\top  \widehat \Delta_{\beta}^{(-k)}}^4\mid \Gamma_i,\widetilde \beta^{(-k)}} \leq n_k^{-1}\sum_{i\in \mathcal{I}_k}\Gamma_iC_1\sigma^2\norm{\widehat \Delta_\beta^{(-k)}}_2^4,
\end{align*}
by Lemma~\ref{lemma concentrate gamma}, we have $n_k^{-1}\sum_{i\in \mathcal{I}_k}\Gamma_i = O_p\p{\gamma_n}$ and 
$
    R_a = O_p\p{\gamma_n\norm{\widehat \Delta_\beta^{(-k)}}_2^2}.
$

By Lemma~\ref{lemma Exp(Xn2) ineq}, if $\lambda_\alpha \asymp (\log d/ (n\gamma_n))^{1/2}$ and $n\gamma_n \gg \max\br{s_{\alpha},\log n}\log d$, we have
\[
    R_b = O_p\p{\gamma_n\norm{\widehat \Delta_\alpha^{(-k)}}_2^4}.
\]
Note that
\[
    \E\sbr{R_c} = \E\sbr{\Gamma_iw_i^4} = \gamma_n\E\sbr{w_i^4 \mid \Gamma_i=1} \leq \gamma_n\sigma_w^4,
\]
we have
\[
    R_c = O_p\p{\gamma_n}.
\]
In conclusion,
\begin{align*}
    n_k^{-1}\sum_{i\in \mathcal{I}_k}\p{\Psi_i\p{\widehat \alpha^{(-k)}, \widehat \beta^{(-k)}} - \Psi_i\p{\alpha^*, \beta^*}}^2 = O_p\p{\gamma_n^{-1}\br{\norm{\widehat \Delta_\alpha^{(-k)}}_2^2 + \norm{\widehat \Delta_\beta^{(-k)}}_2^2}}.
\end{align*}
\end{proof}

\begin{lemma}\label{lemma mar sigma hat = sigma (1+o(1))}
    Let Assumptions~\ref{assumption mar}  and \ref{assumption mar nuisance (a)} hold. Assume either $\gamma_n\p{X} = g\p{X^\top \alpha^*_{PS}}$ or $\mu\p{X} = X^\top \beta^*_{OR}$ holds. Choose $\lambda_\alpha \asymp \lambda_\beta \asymp (\log d/(n\gamma_n))^{1/2}$. If $n\gamma_n \gg \max \{s_\alpha, s_\beta\log n$, $(\log n)^2\} \log d$ and  $s_\alpha s_\beta \ll (n\gamma_n)^{3/2}/(\log n(\log d)^2)$, then as $n, d\rightarrow \infty$, 
    \begin{align*}
        \widehat \sigma^2_g = \sigma^2_g\br{1 + O_p\p{(\frac{(s_\alpha + s_\beta)\log d}{n\gamma_n})^{1/2}}}.
    \end{align*}
\end{lemma}
\begin{proof}
Let $\widetilde \sigma^2_g =n^{-1}\sum_{i =1}^n\p{1-\Gamma_i} \beta^{*,\top}\bar \Xi_0  \beta^{*} + n^{-1}\sum_{i=1}^n\Psi\p{\alpha^*,\beta^*,W_i}^2 - \theta_g^2$. Then
\begin{align*}
    \widehat \sigma^2_g - \widetilde \sigma^2_g = K^{-1}\sum_{k=1}^K\p{s_{k1} - s_{k2} + s_{k3}},
\end{align*}
where
\begin{align*}
    s_{k1} &= n_k^{-1}\sum_{i \in \mathcal{I}_{k}} \br{\Psi_i\p{\widehat \alpha^{(-k)},\widehat \beta^{(-k)}}^2 - \Psi_i\p{\alpha^*,\beta^*}^2},\\
    s_{k2} &= n_k^{-1}\sum_{i \in \mathcal{I}_{k}}\p{1-\Gamma_i}\p{\widehat \beta^{(-k),\top}\bar \Xi_0 \widehat \beta^{(-k)} -  \beta^{*,\top}\bar \Xi_0  \beta^{*}},\quad
    s_{k3} = \widehat \theta^2_g - \theta^2_g.
\end{align*}

    For $s_{k1}$, note that $a^2 - b^2 = \p{a-b}^2 + 2b\p{a-b} $. Thus,
\begin{align*}
    s_{k1} &= n_k^{-1}\sum_{i\in \mathcal{I}_k}\p{\Psi_i\p{\widehat \alpha^{(-k)}, \widehat \beta^{(-k)}} - \Psi_i\p{\alpha^*, \beta^*}}^2+ 2n_k^{-1}\sum_{i\in \mathcal{I}_k}\p{\Psi_i\p{\widehat \alpha^{(-k)}, \widehat \beta^{(-k)}} - \Psi_i\p{\alpha^*, \beta^*}}\Psi_i\p{\alpha^*, \beta^*}\\
    &\leq n_k^{-1}\sum_{i\in \mathcal{I}_k}\p{\Psi_i\p{\widehat \alpha^{(-k)}, \widehat \beta^{(-k)}} - \Psi_i\p{\alpha^*, \beta^*}}^2\\
    &\qquad + 2\br{n_k^{-1}\sum_{i\in \mathcal{I}_k}\p{\Psi_i\p{\widehat \alpha^{(-k)}, \widehat \beta^{(-k)}} - \Psi_i\p{\alpha^*, \beta^*}}^2}^{1/2}\br{n_k^{-1}\sum_{i\in \mathcal{I}_k}\p{\Psi_i\p{\alpha^*, \beta^*}}^2}^{1/2}.
\end{align*}
By \eqref{exp(-X'alpha = gamma-1)} and Lemma~\ref{sub-gaussian properties}(c),
\begin{align*}
    \E\sbr{\p{\Psi_i\p{\alpha^*, \beta^*}}^2} &= \E\sbr{\p{\Gamma_i X_i^\top \beta^* + \frac{ \Gamma_i}{g\p{X_i^\top\alpha^*}}w_i}^2}=2\E\sbr{\Gamma_i \p{X_i^\top \beta^*}^2} + 2\E\sbr{\frac{ \Gamma_i}{\p{g\p{X_i^\top\alpha^*}}^2}w_i^2}\\
    &\leq 2\E\sbr{\p{X_i^\top \beta^*}^2} + 2(1+k_0^{-1}\gamma_n^{-1})^2\E\sbr{\Gamma_iw_i^2}\\
    &=2\E\sbr{\p{X_i^\top \beta^*}^2} + 2(1+k_0^{-1}\gamma_n^{-1})^2\gamma_n\E\sbr{w_i^2\mid \Gamma_i=1} C\gamma_n^{-1}.
\end{align*}
Thus,
$
    n_k^{-1}\sum_{i\in \mathcal{I}_k}\p{\Psi_i\p{\alpha^*, \beta^*}}^2 = O_p\p{\gamma_n^{-1}}.
$
By Lemma~\ref{lemma psi diff square}, we have
\begin{align}
    s_{k1} = O_p\p{\gamma_n^{-1}(\norm{\widehat \Delta_\alpha^{(-k)}}_2^2 + \norm{\widehat \Delta_\beta^{(-k)}}_2^2)^{1/2}}. \label{sk1}
\end{align}

For $s_{k2}$,
\begin{align*}
    s_{k2} &= n_k^{-1}\sum_{i \in \mathcal{I}_{k}}\p{1-\Gamma_i} \widehat \Delta_\beta^{(-k),\top} \bar \Xi_0 \widehat \Delta_\beta^{(-k)}+ 2n_k^{-1}\sum_{i \in \mathcal{I}_{k}}\p{1-\Gamma_i}\beta^{*,\top} \bar \Xi_0 \widehat \Delta_\beta^{(-k)}\\
    &=\p{1-\widehat \gamma_k}\p{1-\widehat \gamma}^{-1}\br{s_{k21} + 2s_{k22}},
\end{align*}
where
\begin{align*}
    s_{k21} &= n^{-1}\sum_{i=1}^n \p{1-\Gamma_i}\p{X_i^\top \widehat \Delta_\beta^{(-k)}}^2, \quad
    s_{k22} = n^{-1}\sum_{i=1}^n\p{1-\Gamma_i}\p{X_i^\top \widehat \Delta_\beta^{(-k)}}\p{X_i^\top \beta^*}.
\end{align*}
By Lemma~\ref{lemma mar i.i.d. quadratic form}, given $\Gamma_{1:n}$,
\begin{align*}
    n^{-1}\sum_{i=1}^n \p{1-\Gamma_i}\p{X_i^\top \Delta^{(-k)}}^2 &= \p{1-\widehat \gamma}\p{n\p{1-\widehat \gamma}}^{-1}\sum_{i=1}^n \p{1-\Gamma_i}\p{X_i^\top\Delta^{(-k)}}^2\\
    &=O_p\p{\p{1-\widehat \gamma}\br{\norm{\Delta^{(-k)}}_1^2\frac{\log d}{n\p{1-\widehat \gamma}} + \norm{\Delta^{(-k)}}_2^2}}.
\end{align*}
Since $\widetilde \beta^{(-k)} \perp (1-\Gamma_i)X_i \mid \Gamma_{i}$ for any $i \in [n]$ due to Lemma~\ref{lemma conditional independence}, we have 
\begin{align*}
    \E\sbr{\p{1-\Gamma_i}\p{X_i^\top \widetilde \Delta_\beta^{(-k)}}^2\mid \Gamma_{1:n},\widetilde \beta^{(-k)}} = \E\sbr{\p{1-\Gamma_i}\p{X_i^\top \widetilde \Delta_\beta^{(-k)}}^2\mid \Gamma_i,\widetilde \beta^{(-k)}} \leq 2\sigma^2\norm{\widehat \Delta_\beta^{(-k)}}_2^2.
\end{align*}
It follows that 
\begin{align*}
    s_{k21} &\leq  2n^{-1}\sum_{i=1}^n \p{1-\Gamma_i}\p{X_i^\top \Delta}^2 + 2n^{-1}\sum_{i=1}^n \p{1-\Gamma_i}\p{X_i^\top \widetilde \Delta_\beta^{(-k)}}^2 \\
    &= O_p\p{\norm{\Delta^{(-k)}}_1^2\frac{\log d}{n} + \norm{\Delta^{(-k)}}_2^2 + \norm{\widehat \Delta_\beta^{(-k)}}_2^2}.
\end{align*}
Since
\[
    \E\sbr{\p{1-\Gamma_i}\p{X_i^\top \beta^*}^2}
    = \p{1-\gamma_n}\E\sbr{\p{X_i^\top \beta^*}^2 \mid \Gamma_i=0}
    \leq \p{1-\gamma_n}C\sigma^2,
\]
we have
\[
    n^{-1}\sum_{i=1}^n\p{1-\Gamma_i}\p{X_i^\top \beta^*}^2 = O_p\p{1-\gamma_n}.
\]
By Cauchy-Schwarz inequality, 
\begin{align*}
    s_{k22}&=n^{-1}\sum_{i=1}^n\p{1-\Gamma_i}\p{X_i^\top \widehat\Delta_\beta^{(-k)}}\p{X_i^\top \beta^*} \\
    &\leq \br{n^{-1}\sum_{i=1}^n\p{1-\Gamma_i}\p{X_i^\top \widehat\Delta_\beta^{(-k)}}^2}^{1/2}\br{n^{-1}\sum_{i=1}^n\p{1-\Gamma_i}\p{X_i^\top \beta^*}^2 }^{1/2}\\
    &= O_p\p{\norm{\Delta^{(-k)}}_1(\frac{\log d}{n})^{1/2}+ \norm{\Delta^{(-k)}}_2+  \norm{\widehat \Delta_\beta^{(-k)}}_2}.
\end{align*}
Thus, with $\widehat \gamma_k-\gamma_n = o_p(1)$ and $\widehat \gamma - \gamma_n = o_p(1)$,
\begin{align}
    s_{k2} &= O_p\p{\norm{\Delta^{(-k)}}_1(\frac{\log d}{n})^{1/2}+ \norm{\Delta^{(-k)}}_2 + \norm{\widehat \Delta_\beta^{(-k)}}_2}.\label{sk2}
\end{align}

For $s_{k3}$, use the fact $a^2 - b^2 = \p{a-b}^2 + 2b\p{a-b} $ again and we have
$
    s_{k3} = \widehat \theta^2_g - \theta^2_g = \p{\widehat \theta_g - \theta_g}^2 + 2\theta\p{\widehat \theta_g - \theta_g}.
$
By Lemma~\ref{lemma mar theta tilde normality} and Lemma~\ref{lemma mar theta hat - theta tilde consistency}, we have
\begin{align}
    s_{k3} = O_p\p{\widehat \theta_g - \theta_g} =  O_p\p{\frac{\p{s_\alpha + (s_\alpha s_\beta)^{1/2}} \log d}{n\gamma_n} + \p{n\gamma_n}^{-1/2}}. \label{sk3}
\end{align}
Together with \eqref{sk1} and \eqref{sk2}, by Proposition~\ref{proposition mar nuisance consistency body}, we have
\begin{align}
    \widehat \sigma^2_g - \widetilde \sigma^2_g &= O_p\p{\gamma_n^{-1}(\norm{\widehat \Delta_\alpha^{(-k)}}_2^2+ \norm{\widehat \Delta_\beta^{(-k)}}_2^2)^{1/2}+\norm{\Delta^{(-k)}}_1(\frac{\log d}{n})^{1/2}+ \norm{\Delta^{(-k)}}_2 + \norm{\widehat \Delta_\beta^{(-k)}}_2} \notag \\
    &\qquad + O_p\p{\frac{\p{s_\alpha + (s_\alpha s_\beta)^{1/2}}\log d}{n\gamma_n} + \p{n\gamma_n}^{-1/2}}\notag\\
&=O_p\p{\gamma_n^{-1}(\frac{(s_\alpha + s_\beta)\log d}{n\gamma_n})^{1/2}}. \label{mar sigma hat - sigma tilde}
\end{align}

Note that
$
    \widetilde \sigma_g^2 =n^{-1}\sum_{i=1}^n \p{X_i^\top \beta^* + {\Gamma_i}/{g\p{X_i^\top \alpha^*}}w_i}^2 - \theta^2.
$
Then
\begin{align*}
    \widetilde \sigma_g^2 - \sigma_g^2 = n^{-1}\sum_{i=1}^n \br{\p{X_i^\top \beta^* + \frac{\Gamma_i}{g\p{X_i^\top \alpha^*}}w_i}^2 - \E\sbr{\p{X_i^\top \beta^* + \frac{\Gamma_i}{g\p{X_i^\top \alpha^*}}w_i}^2}},
\end{align*}
and $\E\sbr{\widetilde \sigma_g^2} = \sigma_g^2$. Thus, by Assumption~\ref{assumption mar nuisance (a)}, for some constant $C>0$,
\begin{align*}
    \E\sbr{\p{\widetilde \sigma^2_g - \sigma^2_g}^2} &\leq n^{-1}\E\sbr{\p{X_i^\top \beta^* + \frac{\Gamma_i}{g\p{X_i^\top \alpha^*}}w_i}^4}\leq 8n^{-1}\E\sbr{\p{X_i^\top \beta^*}^4} + 8n^{-1}\E\sbr{\frac{\Gamma_i}{\p{g\p{X_i^\top \alpha^*}}^4}w_i^4}\\
    &\leq 8n^{-1}\E\sbr{\p{X_i^\text{•}op \beta^*}^4} + 8(1 + k_0^{-1}\gamma_n^{-1})^4n^{-1}\E\sbr{\Gamma_iw_i^4}\\
    &\leq 8n^{-1}\E\sbr{\p{X_i^\top \beta^*}^4} + 8(1 + k_0^{-1}\gamma_n^{-1})^4n^{-1}\gamma_n\E\sbr{w_i^4\mid \Gamma_i=1}\leq Cn^{-1}\gamma_n^{-3},
\end{align*}
which implies $\widetilde \sigma_g^2 - \sigma_g^2 = O_p\p{\gamma_n^{-1}\p{n\gamma_n}^{-1/2}}$. Together with \eqref{mar sigma hat - sigma tilde},
\begin{align*}
    \widehat \sigma^2_g - \sigma^2_g &=  \widehat \sigma^2_g -\widetilde \sigma^2_g + \widetilde \sigma^2_g -\sigma^2_g=O_p\p{\gamma_n^{-1}\p{(\frac{(s_\alpha + s_\beta)\log d}{n\gamma_n})^{1/2} + \p{n\gamma_n}^{-1/2}}}.
\end{align*}
By Lemma~\ref{lemma mar theta tilde normality}, $\sigma_g^2 \asymp \gamma_n^{-1}$, then
$
    \widehat \sigma^2_g = \sigma^2_g\br{1 + O_p\p{({(s_\alpha + s_\beta)\log d}/{n\gamma_n})^{1/2}}}.
$
\end{proof}

\begin{lemma}\label{lemma key results for generalizability}
Let Assumptions~\ref{assumption mar}  and \ref{assumption mar nuisance (a)} hold. Assume either $\gamma_n\p{X} = g\p{X^\top \alpha^*_{PS}}$ or $\mu\p{X} = X^\top \beta^*_{OR}$ holds. Choose $\lambda_\alpha \asymp \lambda_\beta \asymp (\log d/(n\gamma_n))^{1/2}$. If $n\gamma_n \gg  (\log n)^2\log d$, $s_\alpha = o((n\gamma_n)^{1/2}/\log d)$ and $s_\alpha s_\beta = o(n\gamma_n/(\log n\p{\log d}^2))$, then as $n, d\rightarrow \infty$, $\widehat \theta_g -\theta_g = O_p \p{(n\gamma_n)^{-1/2}}$, $\widehat \sigma^2_g = \sigma_g^2\br{1 + o_p(1)}$, and $\widehat\sigma_g^{-1}n^{1/2}(\widehat \theta_g - \theta_g)\xrightarrow{d} \mc{N}(0,1)$.

In addition, Assume that $\mu\p{X} = X^\top \beta^*_{OR}$ holds. Choose $\lambda_\alpha \asymp \lambda_\beta \asymp (\log d/(n\gamma_n))^{1/2}$. If $n\gamma_n \gg  (\log n)^2\log d$ and $s_\alpha s_\beta = o(n\gamma_n/(\log n\p{\log d}^2))$, then as $n, d\rightarrow \infty$, $\widehat \theta_g -\theta_g = O_p \p{(n\gamma_n)^{-1/2}}$, $\widehat \sigma^2_g = \sigma_g^2\br{1 + o_p(1)}$, and $\widehat\sigma_g^{-1}n^{1/2}(\widehat \theta_g - \theta_g)\xrightarrow{d} \mc{N}(0,1)$.
\end{lemma}

\begin{proof}
    When $n\gamma_n \gg  (\log n)^2\log d$, $s_\alpha = o((n\gamma_n)^{1/2}/\log d)$ and $s_\alpha s_\beta = o(n\gamma_n/(\log n\p{\log d}^2))$, we have $n\gamma_n \gg \max \{s_\alpha$, $ s_\beta\log n, (\log n)^2\} \log d$ and  $s_\alpha s_\beta \ll (n\gamma_n)^{3/2}/(\log n(\log d)^2)$.
    
    First, by Lemma~\ref{lemma mar theta tilde normality} and Lemma~\ref{lemma mar theta hat - theta tilde consistency}, we have
    \begin{align*}
        \widehat \theta_g - \theta_g = O_p\p{\frac{\p{s_\alpha + (s_\alpha s_\beta)^{1/2}}\log d}{n\gamma_n} + \p{n\gamma_n}^{-1/2}} = O_p(\p{n\gamma_n}^{-1/2}).
    \end{align*}

    Second, by Lemma~\ref{lemma mar sigma hat = sigma (1+o(1))}, 
    \begin{align*}
        \widehat \sigma^2_g &= \sigma^2_g\br{1 + O_p\p{(\frac{(s_\alpha + s_\beta)\log d}{n\gamma_n})^{1/2}}}= \sigma^2_g\br{1 + o_p(1)}.
    \end{align*}

    Third, by Lemma~\ref{lemma mar theta hat - theta tilde consistency}, we have
    \begin{align*}
        \widehat \theta_g - \widetilde\theta_g = O_p\p{\frac{\p{s_\alpha + (s_\alpha s_\beta)^{1/2}} \log d}{n\gamma_n}} = o_p(\p{n\gamma_n}^{-1/2}).
    \end{align*}
    By Lemma~\ref{lemma mar theta tilde normality}, since $\sigma_g^2 \asymp \gamma_n^{-1}$, by Slutsky's theorem,
    \begin{align*}
        \widehat\sigma^{-1}_gn^{1/2}(\widehat \theta_g - \theta_g) = \widehat\sigma_g^{-1}\sigma_g\p{\sigma_g^{-1}n^{1/2}(\widetilde \theta_g - \theta_g) +  \sigma_g^{-1}n^{1/2}\p{\widehat \theta_g - \widetilde\theta_g}} \rightarrow \mc{N}(0,1).
    \end{align*}

In addition, if $\mu\p{X} = X^\top \beta^*_{OR}$, we only need $s_\alpha s_\beta = o(n\gamma_n/(\log n(\log d)^2))$, which implies $s_\alpha = o(n\gamma_n/(\log n(\log d)^2))$. Under these conditions, by Lemma~\ref{lemma mar theta hat - theta tilde consistency},
\begin{align*}
    \widehat \theta_g - \widetilde\theta_g =   O_p\p{\frac{(s_\alpha \log d)^{1/2}}{n\gamma_n} +\frac{(s_\alpha s_\beta)^{1/2} \log d}{n\gamma_n}} = o_p((n\gamma_n)^{-1/2}).
\end{align*}
The rest of proof follows similarly.
\end{proof}

\subsection{Auxiliary lemmas for transportability}
\begin{lemma}\label{lemma mar transportability theta t normal}
    Let Assumptions~\ref{assumption mar} and \ref{assumption mar nuisance (a)} hold. Assume either $\gamma_n\p{X} = g\p{X^\top \alpha^*_{PS}}$ or $\mu\p{X} = X^\top \beta^*_{OR}$ holds. If $n\gamma_n \gg 1$, then as $n, d \rightarrow \infty$, $\sigma_t^2 \asymp \gamma_n^{-1}$ and
$
    \sigma_t^{-1}n^{1/2}\p{\widetilde\theta_t - \theta_t} \rightarrow \mathcal{N}\p{0,1},
$
    where $\widetilde \theta_t = n^{-1}\sum_{i=1}^n \br{{1-\Gamma_i}/{1-\gamma_n}\p{X_i^\top \beta^*_{OR}-\theta_t} + {\Gamma_i\p{1-g\p{X_i^\top \alpha^*_{PS}}}}/{\p{1-\gamma_n}g\p{X_i^\top \alpha^*_{PS}}}\p{Y_i - X_i^\top\beta^*_{OR}}} + \theta_t$
\end{lemma}
\begin{proof}
Let $w_i = Y_i - X_i^\top \beta^*$.
\begin{align*}
    V_{t,i} = \frac{1-\Gamma_i}{1-\gamma_n}\p{X_i^\top \beta^*-\theta_t} + \frac{\Gamma_i\p{1-g\p{X_i^\top \alpha^*}}}{\p{1-\gamma_n}g\p{X_i^\top \alpha^*}}\p{Y_i - X_i^\top\beta^*} + \theta_t.
\end{align*}
When either $\gamma_n\p{X} = g\p{X^\top \alpha^*_{PS}}$ or $\mu\p{X} = X^\top \beta^*_{OR}$ holds, under Assumption~\ref{assumption mar} we have
    \begin{align*}
        \E\sbr{\widetilde\theta_t - \theta_t} &= \E\sbr{\frac{1-\Gamma_i}{1-\gamma_n}X_i^\top \beta^*} + \E\sbr{\frac{\Gamma_i\p{1-g\p{X_i^\top \alpha^*}}}{\p{1-\gamma_n}g\p{X_i^\top \alpha^*}}\p{Y_i - X_i^\top\beta^*}} - \E\sbr{\frac{1-\Gamma_i}{1-\gamma_n}}\theta_t\\
         &=\E\sbr{\frac{1-\Gamma_i}{1-\gamma_n}\p{X_i^\top \beta^* - Y_i}} + \E\sbr{\frac{\Gamma_i\p{1-g\p{X_i^\top \alpha^*}}}{\p{1-\gamma_n}g\p{X_i^\top \alpha^*}}\p{Y_i - X_i^\top\beta^*}}\\
        &=\E\sbr{\p{\frac{1-\Gamma_i}{1-\gamma_n}-\frac{\Gamma_i\p{1-g\p{X_i^\top \alpha^*}}}{\p{1-\gamma_n}g\p{X_i^\top \alpha^*}}}\p{X_i^\top \beta^* - Y_i}}\\
    &=\E\sbr{\E\sbr{\p{\frac{g\p{X_i^\top \alpha^*} - \Gamma_i}{\p{1-\gamma_n}g\p{X_i^\top \alpha^*}}}\p{X_i^\top \beta^* - Y_i} \mid X_i}}\\
    &=\E\sbr{\p{\frac{\E\sbr{g\p{X_i^\top \alpha^*} - \Gamma_i\mid X_i}}{\p{1-\gamma_n}g\p{X_i^\top \alpha^*}}}\E\sbr{X_i^\top \beta^* - Y_i\mid X_i}}=0.
    \end{align*}
    By Lyapunov's central limit theorem, it suffices to prove for some $\delta>0$ and $C>0$,
    \begin{align}
        \lim_{n\rightarrow \infty}n^{-\delta/2}\sigma^{-\p{2+\delta}} \E\sbr{\abs{V_{t,i} - \theta_t}^{2+\delta}} = 0.
    \end{align}
    If $\E\sbr{Y_i\mid X_i} = X_i^\top \beta^*$, then $\E\sbr{w_i\mid X_i} = 0$ and
    \begin{align*}
    \sigma_t^2 &= \E\sbr{\br{\frac{1-\Gamma_i}{1-\gamma_n}\p{X_i^\top \beta^*-\theta_t} + \frac{\Gamma_i\p{1-g\p{X_i^\top \alpha^*}}}{\p{1-\gamma_n}g\p{X_i^\top \alpha^*}}w_i}^2}\\
    &=\E\sbr{\p{\frac{1-\Gamma_i}{1-\gamma_n}\p{X_i^\top \beta^*-\theta_t}}^2} + \E\sbr{\Gamma_i\p{\frac{\p{1-g\p{X_i^\top \alpha^*}}}{\p{1-\gamma_n}g\p{X_i^\top \alpha^*}}}^2w_i^2}.
    \end{align*}
    since 
    \begin{align*}
    &\E\sbr{\p{\frac{1-\Gamma_i}{1-\gamma_n}\p{X_i^\top \beta^*-\theta_t}}\frac{\Gamma_i\p{1-g\p{X_i^\top \alpha^*}}}{\p{1-\gamma_n}g\p{X_i^\top \alpha^*}}w_i}\\
    &\qquad=\E\sbr{\E\sbr{\p{\frac{1-\Gamma_i}{1-\gamma_n}\p{X_i^\top \beta^*-\theta_t}}\frac{\Gamma_i\p{1-g\p{X_i^\top \alpha^*}}}{\p{1-\gamma_n}g\p{X_i^\top \alpha^*}}\mid X_i}\E\sbr{w_i \mid X_i}}=0.
    \end{align*}

    Under Assumption~\ref{assumption mar}, we have
    \begin{align*}
        \theta_t = \E\sbr{Y_i \mid \Gamma_i=0} &= \E\sbr{\E\sbr{Y_i \mid \Gamma_i=0, X_i}\mid \Gamma_i=0}=\E\sbr{\E\sbr{Y_i \mid  X_i}\mid \Gamma_i=0}=\E\sbr{X_i^\top \beta^* \mid \Gamma_i=0}
    \end{align*}
Under Assumption~\ref{assumption mar nuisance (a)},
    \begin{align*}
        \E\sbr{\p{\frac{1-\Gamma_i}{1-\gamma_n}\p{X_i^\top \beta^*-\theta_t}}^2} &= \p{1-\gamma_n}^{-1}\E\sbr{\p{X_i^\top \beta^*-\theta_t}^2 \mid \Gamma_i=0}\\
        &\leq \E\sbr{\p{X_i^\top \beta^*}^2 \mid \Gamma_i=0}\leq C\sigma^2.
    \end{align*}
In addition, we have
\begin{align*}
    \E\sbr{\Gamma_i\p{\frac{\p{1-g\p{X_i^\top \alpha^*}}}{\p{1-\gamma_n}g\p{X_i^\top \alpha^*}}}^2w_i^2} &= \gamma_n\E\sbr{\p{\frac{\p{1-g\p{X_i^\top \alpha^*}}}{\p{1-\gamma_n}g\p{X_i^\top \alpha^*}}}^2w_i^2 \mid \Gamma_i=1}\\
    &\asymp \gamma_n\frac{\p{1+\gamma_n^{-1}}^2}{\p{1-\gamma_n}^2 }\E\sbr{w_i^2\mid \Gamma_i=1}\asymp \gamma_n^{-1}.
\end{align*}
It follows that when $\E\sbr{w_i \mid X_i} = 0$,
$
    \sigma_t^2 \asymp \gamma_n^{-1}.
$

On the other hand, if $\E\sbr{\Gamma_i \mid X_i} = g\p{X_i^\top \alpha^*}$, we have
\begin{align*}
    \sigma_t^2 &= \E\sbr{\p{V_{t,i} - \theta_t}^2} \\
    &= \E\sbr{\br{\frac{1-\Gamma_i}{1-\gamma_n}\p{X_i^\top \beta^*-\theta_t} + \frac{\Gamma_i\p{1-g\p{X_i^\top \alpha^*}}}{\p{1-\gamma_n}g\p{X_i^\top \alpha^*}}w_i}^2}\\
    &=\E\sbr{\br{\frac{1-\Gamma_i}{1-\gamma_n}\p{Y_i - \theta_t} + \p{\frac{\Gamma_i\p{1-g\p{X_i^\top \alpha^*}}}{\p{1-\gamma_n}g\p{X_i^\top \alpha^*}  } -\frac{1-\Gamma_i}{1-\gamma_n}}w_i}^2 }\\
    &=\E\sbr{\br{\p{\frac{1-\Gamma_i}{1-\gamma_n}\p{Y_i - \theta_t}} + \frac{\Gamma_i - g\p{X_i^\top \alpha^*}}{\p{1-\gamma_n}g\p{X_i^\top \alpha^*}  } w_i}^2 }\\
    &=\E\sbr{\p{\frac{1-\Gamma_i}{1-\gamma_n}\p{Y_i - \theta_t}}^2} + \E\sbr{\p{\frac{\Gamma_i - g\p{X_i^\top \alpha^*}}{\p{1-\gamma_n}g\p{X_i^\top \alpha^*}  }}^2 w_i^2}\\
    &\quad + 2\E\sbr{\p{\frac{1-\Gamma_i}{1-\gamma_n}\p{Y_i - \theta_t}}\p{\frac{\Gamma_i - g\p{X_i^\top \alpha^*}}{\p{1-\gamma_n}g\p{X_i^\top \alpha^*}  } w_i}}.
\end{align*}
Note that 
\begin{align*}
    &\E\sbr{\p{\frac{1-\Gamma_i}{1-\gamma_n}\p{Y_i - \theta_t}}\p{\frac{\Gamma_i - g\p{X_i^\top \alpha^*}}{\p{1-\gamma_n}g\p{X_i^\top \alpha^*}  } w_i}} \\
    &= -\p{1-\gamma_n}^{-1}\E\sbr{\p{Y_i - \theta_t}w_i \mid \Gamma_i = 0}\\
    &\leq c_0^{-1}\p{\E\sbr{w_i^2\mid \Gamma_i = 0} + \E\sbr{w_i^2 \mid \Gamma_i = 0}^{1/2}\E\sbr{\p{X_i^\top\beta^*}^2\mid \Gamma_i = 0}^{1/2} + \abs{\theta_t}\E\sbr{\abs{w_i}\mid \Gamma_i = 0}}\\
    &\leq c_0^{-1}C\p{c_0^{-1}\sigma_w^2 + c_0^{-1}\sigma_w\sigma + c_0^{-1}\sigma_w}.
\end{align*}
Similarly, we have
\begin{align*}
    &\E\sbr{\p{\frac{1-\Gamma_i}{1-\gamma_n}\p{Y_i - \theta_t}}^2} = \E\sbr{\p{Y_i - \theta_t}^2 \mid \Gamma_i=0}\leq \E\sbr{Y_i^2\mid \Gamma_i=0}\\
    &\qquad\leq 2\p{\E\sbr{w_i^2 \mid \Gamma_i = 0} + \E\sbr{\p{X_i^\top \beta}^2 \mid \Gamma_i = 0}}\leq 2\p{c_0^{-1}C_1\sigma_w^2 + C_2\sigma^2}.
\end{align*}
In addition, 
\begin{align*}
    &\E\sbr{\p{\frac{\Gamma_i - g\p{X_i^\top \alpha^*}}{\p{1-\gamma_n}g\p{X_i^\top \alpha^*}  }}^2 w_i^2} = \gamma_n\E\sbr{\p{\frac{1 - g\p{X_i^\top \alpha^*}}{\p{1-\gamma_n}g\p{X_i^\top \alpha^*}  }}^2 w_i^2 \mid \Gamma_i=1}+ \p{1-\gamma_n}^{-1}\E\sbr{ w_i^2 \mid \Gamma_i = 0}\\
    &\qquad\asymp \gamma_n\p{\frac{1-\gamma_n^{-1}}{\p{1-\gamma_n}\gamma_n}}^2\E\sbr{w_i^2 \mid \Gamma_i=1}+ \p{1-\gamma_n}^{-1}\E\sbr{ w_i^2 \mid \Gamma_i = 0}\asymp \gamma_n^{-1}.
\end{align*}
It follows that when $\E\sbr{\Gamma_i \mid X_i} = g\p{X_i^\top \alpha^*}$,
$
    \sigma_t^2 \asymp \gamma_n^{-1}.
$

Thus, when either $\gamma_n\p{X} = g\p{X^\top \alpha^*_{PS}}$ or $\mu\p{X} = X^\top \beta^*_{OR}$ holds, we have
\begin{align}
    \sigma_t^2 \asymp \gamma_n^{-1}. \label{sigma t asymp gamma-1 transportability}
\end{align}
By Minkowski's inequality, we have
\begin{align*}
    \norm{V_{t,i}-\theta_t}_{\P, 2+\delta}
    &\leq \norm{{1-\Gamma_i}/{1-\gamma_n}X_i^\top \beta^*}_{\P, 2+\delta}\\
    &\quad + \norm{{\Gamma_i\p{1-g\p{X_i^\top \alpha^*}}}/{\p{1-\gamma_n}g\p{X_i^\top \alpha^*}}w_i}_{\P, 2+\delta}
    + c_0^{-1}\abs{\theta_t}.
\end{align*}
Choose $\delta = 2$. Then under Assumption~\ref{assumption mar nuisance (a)},
\[
    \E\sbr{\abs{{1-\Gamma_i}/{1-\gamma_n}X_i^\top \beta^*}^4} \leq c_0^{-4}\E\sbr{\p{X_i^\top \beta^*}^4} \leq c_0^{-4}C_3^{4}\sigma^4.
\]
In addition,
\begin{align*}
    \E\sbr{\abs{\frac{\Gamma_i\p{1-g\p{X_i^\top \alpha^*}}}{\p{1-\gamma_n}g\p{X_i^\top \alpha^*}}w_i}^4} &= \gamma_n\p{1-\gamma_n}^{-4} \E\sbr{\abs{\frac{\p{1-g\p{X_i^\top \alpha^*}}}{g\p{X_i^\top \alpha^*}}w_i}^4 \mid \Gamma_i = 1}\\
    &\leq c_0^{-4}\gamma_n^{-3}\E\sbr{w_i^4 \mid \Gamma_i = 1}\leq c_0^{-4}\gamma_n^{-3}\sigma_w^4.
\end{align*}
It follows that
\begin{align*}
    \norm{V_{t,i}-\theta_t}_{\P, 4} &\leq c_0^{-1}C_3\sigma + c_0^{-1}\sigma_w\gamma_n^{-3/4} + c_0^{-1}\abs{\theta_t}\leq \gamma_n^{-3/4}\p{c_0^{-1}C_3\sigma + c_0^{-1}\sigma_w + c_0^{-1}\abs{\theta_t}}.
\end{align*}
Thus, if $n\gamma_n \gg 1$, as $n, d \rightarrow \infty$,
\begin{align*}
    \lim_{n\rightarrow \infty}n^{-\delta/2}\sigma^{-\p{2+\delta}} \E\sbr{\abs{V_{t,i} - \theta_t}^{2+\delta}} \leq \lim_{n\rightarrow \infty} Cn^{-1}\gamma_n^{2}\gamma_n^{-3} = \lim_{n\rightarrow \infty} C\p{n\gamma_n}^{-1} = 0.
    \end{align*}
\end{proof}

\begin{lemma}\label{lemma mar transportability theta hat - theta tilde consistency}
    Let Assumption~\ref{assumption mar nuisance (a)} hold. Choose $\lambda_\alpha \asymp \lambda_\beta \asymp (\log d/(n\gamma_n))^{1/2}$. If $n\gamma_n \gg \max\{s_\alpha, s_\beta\log n$, $(\log n)^2\} \log d$ and  $s_\alpha s_\beta \ll (n\gamma_n)^{3/2}/(\log n(\log d)^2)$, then as $n, d\rightarrow \infty$, 
    \begin{align*}
        \widehat \theta_t - \widetilde \theta_t &= O_p\p{\frac{\p{s_\alpha + (s_\alpha s_\beta)^{1/2}} \log d}{n\gamma_n}}.
    \end{align*}
In addition, let Assumption~\ref{assumption mar} hold. Assume that $\mu(X) = X^\top \beta^*_{OR}$ holds. Then as $n, d\rightarrow \infty$, 
\begin{align*}
    \widehat \theta_t - \widetilde \theta_t &= O_p\p{\frac{(s_\alpha \log d)^{1/2}}{n\gamma_n} +\frac{(s_\alpha s_\beta)^{1/2} \log d}{n\gamma_n} }.
\end{align*}
\end{lemma}
\begin{proof}
    Let $\widehat \Delta_\alpha^{(-k)} = \widehat \alpha^{(-k)} - \alpha^*$, $\widehat \Delta_\beta^{(-k)} = \widehat \beta^{(-k)} - \beta^*$, $\widetilde \Delta_\beta^{(-k)} = \widetilde \beta^{(-k)} - \beta^*$, $\Delta^{(-k)} = \widehat \beta^{(-k)} - \widetilde \beta^{(-k)}$, $\widehat \gamma_k = n_k^{-1}\sum_{i\in \mathcal{I}_k} \Gamma_i$, $\widehat \gamma = n^{-1}\sum_{i=1}^n\Gamma_i$, and
\begin{align*}
    \Phi_i\p{\alpha, \beta} = \Gamma_i\frac{1-g\p{X_i^\top \alpha}}{g\p{X_i^\top \alpha}}\p{Y_i - X_i^\top \beta} = \Psi_i\p{\alpha, \beta} - \Gamma_i Y_i.
\end{align*}
Then
\begin{align*}
    \widehat \theta_t^{(k)} &= \p{1-\widehat \gamma_k}^{-1}n_k^{-1}\sum_{i\in \mathcal{I}_k} \br{\p{1-\Gamma_i}\bar X_0^\top \widehat \beta^{(-k)} + \Phi_i\p{\widehat \alpha^{(-k)}, \widehat \beta^{(-k)}}},\\
    \widetilde \theta_t^{(k)} &= \p{1-\gamma_n}^{-1}n_k^{-1}\sum_{i\in \mathcal{I}_k} \br{\p{1-\Gamma_i}\bar X_0^\top  \beta^* + \Phi_i\p{\alpha^*, \beta^*}} + \frac{\widehat \gamma_k - \gamma_n}{1-\gamma_n}\theta_t,
\end{align*}
where $\widetilde \theta_t = K^{-1}\sum_{k=1}^K \widetilde \theta_t^{(k)}$. Note that $\Phi_i\p{\widehat \alpha^{(-k)}, \widehat \beta^{(-k)}} - \Phi_i\p{\alpha^*, \beta^*} = \Psi_i\p{\widehat \alpha^{(-k)}, \widehat \beta^{(-k)}} - \Psi_i\p{\alpha^*, \beta^*}$ and 
\begin{align*}
    \p{1-\widehat \gamma_{k}}\widehat\theta_t^{(k)} - \p{1-\gamma_n}\p{\widetilde\theta_t^{(k)} - \frac{\widehat \gamma_k - \gamma_n}{1-\gamma_n}\theta_t} = \widehat \theta_g^{(k)} - \widetilde \theta_g^{(k)}.
\end{align*}
By results in Lemma~\ref{lemma mar theta hat - theta tilde consistency} we have
\begin{align*}
    \p{1-\widehat \gamma_{k}}\widehat\theta_t^{(k)} - \p{1-\gamma_n}\p{\widetilde\theta_t^{(k)} - \frac{\widehat \gamma_k - \gamma_n}{1-\gamma_n}\theta_t} = O_p\p{  \frac{\p{s_\alpha + (s_\alpha s_\beta)^{1/2}} \log d}{n\gamma_n}}.
\end{align*}

On the other hand, let  
\begin{align*}
    \bar\theta_t^{(k)} = n_k^{-1}\sum_{i \in \mathcal{I}_k}\br{\frac{1-\Gamma_i}{1-\gamma_n}X_i^\top \beta^* + \Phi_i\p{\alpha^*, \beta^*}} + \frac{\widehat \gamma_k - \gamma_n}{1-\gamma_n}\theta_t.
\end{align*}
Then
\begin{align*}
    \widetilde \theta_t^{(k)} - \bar\theta_t^{(k)} = q_1 - q_2 + q_3,
\end{align*}
where
\begin{align*}
    q_1 &= \frac{1-\widehat \gamma_k}{1-\gamma_n}\p{1-\widehat \gamma}^{-1} n^{-1}\sum_{i=1}^n\br{\p{1-\Gamma_i}X_i - \E\sbr{\p{1-\Gamma_i}X_i}}^\top \beta^*,\\
    q_2 &= \frac{1}{1-\gamma_n} n_k^{-1}\sum_{i\in \mathcal{I}_k} \br{\p{1-\Gamma_i}X_i - \E\sbr{\p{1-\Gamma_i}X_i}}^\top \beta^*,\\
    q_3 &= \p{\frac{1-\widehat \gamma_k}{1-\gamma_n}\p{1-\widehat \gamma}^{-1} - \frac{1}{1-\gamma_n}} \E\sbr{\p{1-\Gamma_i}X_i}^\top \beta^*.
\end{align*}

For $q_1$, let 
\begin{align*}
    q_{1a} = n^{-1}\sum_{i=1}^n\br{\p{1-\Gamma_i}X_i - \E\sbr{\p{1-\Gamma_i}X_i}}^\top \beta^*.
\end{align*}
Then $\E\sbr{q_{1a}}=0$. By Assumption~\ref{assumption mar nuisance (a)},
\begin{align*}
    \E\sbr{q_{1a}^2} &= n^{-1}\E\sbr{\br{\p{\p{1-\Gamma_i}X_i - \E\sbr{\p{1-\Gamma_i}X_i}}^\top \beta^*}^2}\leq n^{-1} \E\sbr{\p{1-\Gamma_i}\p{X_i^\top \beta^*}^2}\\
    &= \p{1-\gamma_n}n^{-1} \E\sbr{\p{X_i^\top \beta^*}^2 \mid \Gamma_i = 0}=O_p\p{n^{-1}}.
\end{align*}
Thus, 
$
    q_{1a} = O_p\p{n^{-1/2}}.
$
By Lemma~\ref{lemma gamma ratio convergence},
\begin{align}
    \frac{1-\gamma_n}{1-\widehat \gamma_{k}} - 1 = O_p\p{(\frac{\gamma_n}{n})^{1/2}} \quad \text{and} \quad \frac{1}{1-\widehat \gamma_k} = O_p\p{1}. \label{q4b gamma}
\end{align}
Thus,
$
    q_1 = O_p\p{n^{-1/2}}.
$
Similarly, we have
$
    q_2 = O_p\p{n^{-1/2}}.
$
By Lemma~\ref{lemma concentrate gamma} and Lemma~\ref{lemma gamma ratio convergence}, we have
\begin{align*}
\frac{1-\widehat \gamma_k}{1-\widehat \gamma} - 1 
    &= \frac{1-\widehat \gamma_k - \p{1-\gamma_n}}{1-\gamma_n}\p{\frac{1-\gamma_n}{1-\widehat \gamma}-1} + \p{\frac{1-\gamma_n}{1-\widehat \gamma}-1} + \frac{1-\widehat \gamma_k - \p{1-\gamma_n}}{1-\gamma_n} \notag \\
    &= O_p\p{\p{n\p{1-\gamma_n}}^{-1/2}}O_p\p{\p{n\p{1-\gamma_n}}^{-1/2}\gamma_n^{1/2}} + O_p\p{n\p{1-\gamma_n}}^{-1/2}\notag\\
    &= O_p\p{\gamma_n^{1/2}\p{n\p{1-\gamma_n}}^{-1}} + O_p\p{n\p{1-\gamma_n}}^{-1/2}.
\end{align*}
Since $1-\gamma_n\geq c_0$, we have
$
    q_3 = O_p\p{n^{-1/2}}.
$
Thus,
$
    \widetilde \theta_t^{(k)} - \bar\theta_t^{(k)} = O_p\p{n^{-1/2}}.
$

Note that by Lemma~\ref{lemma mar transportability theta t normal}, $\bar \theta_t^{(k)} - \theta_t = O_p\p{\p{n\gamma_n}^{-1/2}}$. Then by \eqref{q4b gamma}, we have
\begin{align*}
   \widehat\theta_t^{(k)} - \widetilde\theta_t^{(k)} &= \p{\frac{1-\gamma_n}{1-\widehat \gamma_{k}} - 1}\p{\widetilde\theta_t^{(k)} - \bar \theta_t^{(k)}} + \p{\frac{1-\gamma_n}{1-\widehat \gamma_{k}}-1}\p{\bar \theta_t^{(k)} - \theta_t}  + \frac{1}{1-\widehat \gamma_k}O_p\p{\p{s_\alpha + s_\beta}\frac{\log d}{n}}\\
   &=O_p\p{\gamma_n^{1/2}n^{-1}} + O_p\p{n^{-1}} +  O_p\p{  \frac{\p{s_\alpha + (s_\alpha s_\beta)^{1/2}} \log d}{n\gamma_n}}= O_p\p{\frac{\p{s_\alpha + (s_\alpha s_\beta)^{1/2}} \log d}{n\gamma_n}}.
\end{align*}
It follows that 
$
    \widehat \theta_t - \widetilde \theta_t = O_p\p{  {\p{s_\alpha + (s_\alpha s_\beta)^{1/2}} \log d}/{n\gamma_n}}.
$

In addition, when $\mu(X_i) = X_i^\top \beta^*$, by Lemma~\ref{lemma mar theta hat - theta tilde consistency}, we have 
\begin{align*}
    \p{1-\widehat \gamma_{k}}\widehat\theta_t^{(k)} - \p{1-\gamma_n}\p{\widetilde\theta_t^{(k)} - \frac{\widehat \gamma_k - \gamma_n}{1-\gamma_n}\theta_t} = O_p\p{\frac{(s_\alpha \log d)^{1/2}}{n\gamma_n} +\frac{(s_\alpha s_\beta)^{1/2} \log d}{n\gamma_n}}.
\end{align*}
By identical analysis, we have
$
     \widehat \theta_t - \widetilde \theta_t = O_p\p{{(s_\alpha \log d)^{1/2}}/{n\gamma_n} +{(s_\alpha s_\beta)^{1/2} \log d}/{n\gamma_n}}.
$
\end{proof}

\begin{lemma}\label{lemma mar transportability variance consistency}
    Let Assumptions~\ref{assumption mar} and \ref{assumption mar nuisance (a)} hold. Assume either $\gamma_n\p{X} = g\p{X^\top \alpha^*_{PS}}$ or $\mu\p{X} = X^\top \beta^*_{OR}$ holds. Choose $\lambda_\alpha \asymp \lambda_\beta \asymp (\log d/(n\gamma_n))^{1/2}$. If $n\gamma_n \gg \max\{s_\alpha, s_\beta\log n$, $(\log n)^2\} \log d$ and  $s_\alpha s_\beta \ll (n\gamma_n)^{3/2}/(\log n(\log d)^2)$, then as $n, d\rightarrow \infty$, 
    \begin{align*}
        \widehat \sigma_t^2 = \sigma_t^2\br{1+O_p\p{(\frac{\p{s_\alpha + s_\beta}\log d}{n\gamma_n})^{1/2} + \gamma_n(\frac{s_\alpha s_\beta \log d }{n\gamma_n})^{1/2}}}.
    \end{align*}
\end{lemma}

\begin{proof}
    Note that 
    \begin{align*}
        \widehat \sigma_t^2 - \widetilde \sigma_t^2 = K^{-1}\sum_{k=1}^K\p{p_{k1a}-p_{k1b} + p_{k2a}-p_{k2b} - 2\p{p_{k3a}-p_{k3b}} + p_{k4a} - p_{k4b}},
    \end{align*}
where
\begin{align*}
    &p_{k1a} = \p{1-\widehat \gamma_k}^{-1} \widehat \beta^{(-k),\top}\bar \Xi_0 \widehat \beta^{(-k)}, \quad
    p_{k1b} = \frac{1-\widehat \gamma_k}{\p{1-\gamma_n}^{2}} \beta^{*,\top}\bar \Xi_0 \beta^{*},\\
    &p_{k2a} = \p{1-\widehat \gamma_k}^{-2}n_k^{-1}\sum_{i \in \mathcal{I}_k} \Phi_i\p{\widehat \alpha^{(-k)},\widehat \beta^{(-k)}}^2, \quad
    p_{k2b} =\p{1-\gamma_n}^{-2}n_k^{-1}\sum_{i \in \mathcal{I}_k} \Phi_i\p{\alpha^*, \beta^*}^2,\\
        &p_{k3a} = \p{1-\widehat \gamma_k}^{-1}\p{\bar X_0^\top \widehat \beta^{(-k)}}\widehat \theta_t, \quad
        p_{k3b} = \p{1-\gamma_n}^{-2}\p{1-\widehat \gamma_k}\p{\bar X_0^\top \beta^*}\theta_t\\
        &p_{k4a} =\p{1-\widehat \gamma_k}^{-1}\widehat \theta_t^2,\quad
        p_{k4b} = \p{1-\gamma_n}^{-2}\p{1-\widehat \gamma_k}\theta_t^2.
\end{align*}
Recall $s_{k21}, s_{k22}$ in the proof of Lemma~\ref{lemma mar sigma hat = sigma (1+o(1))},
\begin{align*}
    &\p{1-\widehat \gamma_k}^2p_{k1a} - \p{1-\gamma_n}^2p_{k1b} = \p{1-\widehat \gamma_k}\p{\widehat \Delta_\beta^{(-k),\top}\bar \Xi_0 \widehat \Delta_\beta^{(-k)} + 2\beta^{*,\top}\Xi_0 \widehat \Delta_\beta^{(-k)}}\\
    &\qquad\leq s_{k21} + 2s_{k22}=O_p\p{\norm{\Delta^{(-k)}}_1(\frac{\log d}{n})^{1/2}+ \norm{\Delta^{(-k)}}_2+  \norm{\widehat \Delta_\beta^{(-k)}}_2}.
\end{align*}
Under Assumption~\ref{assumption mar nuisance (a)},
    \begin{align*}
        &\E\sbr{\Phi_i\p{\alpha^*, \beta^*}^2} = \E\sbr{\frac{\Gamma_i}{\p{1-\gamma_n}^2}\p{\frac{1-g\p{X_i^\top \alpha^*}}{g\p{X_i^\top \alpha^*}}}^2w_i^2}\leq c_0^{-2}k_0^{-2}\gamma_n^{-2}\E\sbr{\Gamma_iw_i^2}\\
        &\qquad=c_0^{-2}k_0^{-2}\gamma_n^{-1}\E\sbr{w_i^2 \mid \Gamma_i=1}\leq c_0^{-2}k_0^{-2}\sigma_w^{-1}\gamma_n^{-1},
    \end{align*}
which implies 
\begin{align}
    p_{k2b} = O_p\p{\p{1-\gamma_n}^{-2}\gamma_n^{-1}
} = O_p\p{\gamma_n^{-1}}. \label{pk2b}
\end{align}
Since $a^2 - b^2 = (a-b)^2 + 2(a-b)b$ and 
$
    \Phi_i\p{\widehat \alpha^{(-k)}, \widehat \beta^{(-k)}} - \Phi_i\p{\alpha^*, \beta^*} = \Psi_i\p{\widehat \alpha^{(-k)}, \widehat \beta^{(-k)}} - \Psi_i\p{\alpha^*, \beta^*},
$
by Cauchy-Schwarz inequality, we have
\begin{align*}
    &\p{1-\widehat \gamma_k}^2p_{k2a} - \p{1-\gamma_n}^2p_{k2b} \\
    &\quad= n_k^{-1}\sum_{i \in \mathcal{I}_k}\p{\Psi_i\p{\widehat \alpha^{(-k)}, \widehat \beta^{(-k)}} - \Psi_i\p{\alpha^*, \beta^*}}^2+ 2n_k^{-1}\sum_{i \in \mathcal{I}_k}\p{\Psi_i\p{\widehat \alpha^{(-k)}, \widehat \beta^{(-k)}} - \Psi_i\p{\alpha^*, \beta^*}}\Phi_i\p{\alpha^*, \beta^*}\\
    &\quad\leq n_k^{-1}\sum_{i \in \mathcal{I}_k}\p{\Psi_i\p{\widehat \alpha^{(-k)}, \widehat \beta^{(-k)}} - \Psi_i\p{\alpha^*, \beta^*}}^2 \\
    &\quad\quad + 2\br{n_k^{-1}\sum_{i \in \mathcal{I}_k}\p{\Psi_i\p{\widehat \alpha^{(-k)}, \widehat \beta^{(-k)}} - \Psi_i\p{\alpha^*, \beta^*}}^2}^{1/2}\p{p_{k2b}}^{1/2}\\
    &\quad\overset{(i)}{=}O_p\p{\gamma_n^{-1}\br{\norm{\widehat \Delta_\alpha^{(-k)}}_2^2 + \norm{\widehat \Delta_\beta^{(-k)}}_2^2}} + O_p\p{\gamma_n^{-1}\br{\norm{\widehat \Delta_\alpha^{(-k)}}_2 + \norm{\widehat \Delta_\beta^{(-k)}}_2}},
\end{align*}
    where we use Lemma~\ref{lemma psi diff square} in $(i)$. 
In addition,
\begin{align*}
    &\p{1-\widehat \gamma_k}^2p_{k3a} - \p{1- \gamma_n}^2p_{k3b} = \p{1-\widehat \gamma_k}\br{\bar X_0^\top \widehat \Delta_\beta^{(-k)}\p{\widehat \theta_t - \theta_t} + \bar X_0^\top \beta^* \p{\widehat \theta_t - \theta_t} + \bar X_0^\top\widehat \Delta_\beta^{(-k)}\theta_t}\\
    &\qquad=\frac{1-\widehat \gamma_k}{1-\widehat \gamma}\p{\widehat \theta_t - \theta_t + \theta_t} n^{-1}\sum_{i=1}^n\p{1-\Gamma_i}X_i^\top\Delta^{(-k)}+\frac{1-\widehat \gamma_k}{1-\widehat \gamma}\p{\widehat \theta_t - \theta_t}n^{-1}\sum_{i=1}^n\p{1-\Gamma_i}X_i^\top\beta^*\\
    &\qquad\qquad +\frac{1-\widehat \gamma_k}{1-\widehat \gamma}\p{\widehat \theta_t - \theta_t + \theta_t} n^{-1}\sum_{i=1}^n\p{1-\Gamma_i}X_i^\top\widetilde \Delta_\beta^{(-k)}.
\end{align*}
By Lemma~\ref{sub-gaussian properties}(c), for some $C>0$,
$
    \max_{1\leq i \leq n}\max_{1\leq k \leq d} \norm{\p{\p{1-\Gamma_i}X_i - \E\sbr{\p{1-\Gamma_i}X_i}}^\top e_k}_{\psi_2} \leq C\sigma,
$
and
$
    \max_{1\leq k \leq d}n^{-1}\sum_{i=1}^n\E\sbr{\br{\p{\p{1-\Gamma_i}X_i - \E\sbr{\p{1-\Gamma_i}X_i}}^\top e_k}^2} \leq C\sigma_2.
$
Since $n\gamma_n \gg (\log n)^2\log d$, by Lemma~\ref{Lemma Tail Bounds for Maximums}, 
\begin{align*}
    \norm{n^{-1}\sum_{i=1}^n\p{1-\Gamma_i}X_i - \E\sbr{\p{1-\Gamma_i}X_i} }_\infty = O_p\p{(\frac{\log d}{n})^{1/2}}.
\end{align*}
Thus,
\begin{align*}
    &n^{-1}\sum_{i=1}^n\p{1-\Gamma_i}X_i^\top\Delta^{(-k)}\\
    &\qquad \leq \norm{n^{-1}\sum_{i=1}^n\p{1-\Gamma_i}X_i - \E\sbr{\p{1-\Gamma_i}X_i} }_\infty\norm{\Delta^{(-k)}}_1\\
    &\qquad\quad + \norm{\E\sbr{\p{1-\Gamma_i}X_i} }_\infty\norm{\Delta^{(-k)}}_1\\
    &\qquad = O_p\p{\norm{\Delta^{(-k)}}_1}.
\end{align*}
By Lemma~\ref{sub-gaussian properties} (c),
$
    \E\sbr{\p{1-\Gamma_i}X_i^\top \beta^*} \leq \E\sbr{\abs{X_i^\top \beta^*}} \leq C\sigma,
$
we have
$
    n^{-1}\sum_{i=1}^n\p{1-\Gamma_i}X_i^\top\beta^* = O_p\p{1}.
$
Since $\widetilde \Delta_\beta^{(-k)} \perp (1-\Gamma_i)X_i \mid \Gamma_{1:n}$, then by Lemma~\ref{sub-gaussian properties}(c),
\begin{align*}
    \E\sbr{\p{1-\Gamma_i}X_i^\top\widetilde \Delta_\beta^{(-k)} \mid \Gamma_{1:n}, \widetilde\beta^{(-k)}} = \p{1-\Gamma_i}\E\sbr{X_i\mid \Gamma_i=0}^\top\widetilde \Delta_\beta^{(-k)} \leq 2\sigma\norm{\widetilde \Delta_\beta^{(-k)}}_2.
\end{align*}
By Lemma~\ref{lemma convergence of conditional random variable}, 
$
    n^{-1}\sum_{i=1}^n\p{1-\Gamma_i}X_i^\top\widetilde \Delta_\beta^{(-k)} = O_p\p{\norm{\widetilde \Delta_\beta^{(-k)}}_2}.
$
By Lemma~\ref{lemma mar transportability theta t normal} and Lemma~\ref{lemma mar transportability theta hat - theta tilde consistency},
\begin{align*}
    \widehat \theta_t - \theta_t &= O_p\p{\frac{\p{s_\alpha + (s_\alpha s_\beta)^{1/2}} \log d}{n\gamma_n} + \p{n\gamma_n}^{-1/2}}.
\end{align*}
Moreover, by Lemma~\ref{lemma concentrate gamma} and Lemma~\ref{lemma gamma ratio convergence},
\begin{align*}
    \frac{1-\widehat \gamma_k}{1-\widehat \gamma} = \frac{1-\widehat \gamma_k}{1-\gamma_n}\p{\frac{1-\gamma_n}{1-\widehat \gamma} - 1} + \frac{1-\widehat \gamma_k}{1-\gamma_n} = O_p\p{1},
\end{align*}
which implies
\begin{align}
    p_{k3b} = \p{1-\gamma_n}^{-2}\p{1-\widehat \gamma_k}\p{\bar X_0^\top \beta^*}\theta_t = O_p\p{1}. \label{pk3b}
\end{align}
Thus, 
\begin{align*}
    &\p{1-\widehat \gamma_k}^2p_{k3a} - \p{1- \gamma_n}^2p_{k3b} = O_p\p{\norm{\Delta^{(-k)}}_1+\norm{\widetilde \Delta_\beta^{(-k)}}_2 + \frac{\p{s_\alpha + (s_\alpha s_\beta)^{1/2}} \log d}{n\gamma_n} + \p{n\gamma_n}^{-1/2}}.
\end{align*}
It is clear that 
\begin{align}
    p_{k4b} = \frac{1-\widehat \gamma_k}{\p{1-\gamma_n}^2}\theta_t^2 = O_p\p{1}, \label{pk4b}
\end{align}
and
\begin{align*}
    &\p{1-\widehat \gamma_k}^2p_{k4a} - \p{1- \gamma_n}^2p_{k4b} = \p{1-\widehat \gamma_k}\br{\p{\widehat \theta_t - \theta_t}^2 + 2\p{\widehat \theta_t - \theta_t}\theta_t}\\
    &\qquad=O_p\p{\frac{\p{s_\alpha + (s_\alpha s_\beta)^{1/2}} \log d}{n\gamma_n} + \p{n\gamma_n}^{-1/2}}.
\end{align*}
In conclusion, by Proposition~\ref{proposition mar nuisance consistency body}, Lemma~\ref{lemma consistency of betahat and betatilde for product sparsity}, and Lemma~\ref{lemma consistency for alpha and betatilde}, we have
\begin{align*}
    a_1 = \p{1-\widehat \gamma_k}^2p_{k1a} - \p{1-\gamma_n}^2p_{k1b} &= O_p\p{ (\frac{\p{s_\alpha + s_\beta}\log d}{n\gamma_n})^{1/2}},\\
    a_2 =\p{1-\widehat \gamma_k}^2p_{k2a} - \p{1-\gamma_n}^2p_{k2b} &= O_p\p{\gamma_n^{-1} (\frac{\p{s_\alpha + s_\beta}\log d}{n\gamma_n})^{1/2}},\\
    a_3 =\p{1-\widehat \gamma_k}^2p_{k3a} - \p{1-\gamma_n}^2p_{k3b} &= O_p\p{(\frac{s_\alpha s_\beta \log d }{n\gamma_n})^{1/2}},\\
    a_4 =\p{1-\widehat \gamma_k}^2p_{k4a} - \p{1-\gamma_n}^2p_{k4b} &= O_p\p{\frac{\p{s_\alpha + (s_\alpha s_\beta)^{1/2}} \log d}{n\gamma_n} + \p{n\gamma_n}^{-1/2}}.
\end{align*}
Note that
\begin{align*}
     p_{k1b} = \frac{1-\widehat \gamma_k}{\p{1-\gamma_n}^{2}}\beta^{*,\top}\bar \Xi_0 \beta^{*} = \frac{\p{1-\widehat \gamma_k}\p{1-\widehat \gamma}}{\p{1-\gamma_n}^{2}}n^{-1}\sum_{i=1}^n\p{1-\Gamma_i}\p{X_i^\top \beta^*}^2.
\end{align*}
By Lemma~\ref{lemma concentrate gamma}, $1-\widehat \gamma_k = O_p\p{1-\gamma_n}$, $1-\widehat \gamma = O_p\p{1-\gamma_n}$. By Lemma~\ref{sub-gaussian properties}(c),
\begin{align*}
    \E\sbr{\p{1-\Gamma_i}\p{X_i^\top \beta^*}^2} = \p{1-\gamma_n}\E\sbr{\p{X_i^\top \beta^*}^2 \mid \Gamma_i=0} = O_p\p{1-\gamma_n},
\end{align*}
which implies
\begin{align}
    p_{k1b} = O_p\p{1-\gamma_n}. \label{pk1b}
\end{align}
In addition, by Lemma~\ref{lemma gamma ratio convergence}, ${1-\gamma_n}/{1-\widehat \gamma_{k}} - 1 = O_p\p{\gamma_n^{1/2}\p{n\p{1-\gamma_n}}^{-1/2}}$ and
\begin{align}
    \p{\frac{1-\gamma_n}{1-\widehat \gamma_{k}}}^2 - 1 = \p{\frac{1-\gamma_n}{1-\widehat \gamma_{k}} - 1}^2 + 2\p{\frac{1-\gamma_n}{1-\widehat \gamma_{k}}-1} = O_p\p{\gamma_n^{1/2}\p{n\p{1-\gamma_n}}^{-1/2}}. \label{1-gamma hat square ratio}
\end{align}
By \eqref{pk2b}-\eqref{1-gamma hat square ratio}, we have
\begin{align*}
    &p_{k1a}-p_{k1b} + p_{k2a}-p_{k2b} - 2\p{p_{k3a}-p_{k3b}} + p_{k4a} - p_{k4b}\\
    &\qquad=\p{\p{\frac{1-\gamma_n}{1-\widehat \gamma_k}}^2 - 1}\p{p_{k1b} + p_{k2b} - 2p_{k3b} + p_{k4b}}  + \frac{1}{\p{1-\widehat \gamma_k}^2}\p{a_1 + a_2 -2a_3 + a_4}\\
    &\qquad=O_p\p{ \gamma_n^{-1} (\frac{\p{s_\alpha + s_\beta}\log d}{n\gamma_n})^{1/2} + (\frac{s_\alpha s_\beta \log d }{n\gamma_n})^{1/2}}.
\end{align*}
Then it follows that
\begin{align*}
    \widehat \sigma_t^2 - \widetilde \sigma_t^2 = O_p\p{\gamma_n^{-1} (\frac{\p{s_\alpha + s_\beta}\log d}{n\gamma_n})^{1/2} + (\frac{s_\alpha s_\beta \log d }{n\gamma_n})^{1/2}}.
\end{align*}

On the other hand, since
\begin{align*}
    \sigma_t^2 = \E\sbr{\br{\frac{1-\Gamma_i}{1-\gamma_n}\p{X_i^\top \beta^*-\theta_t} + \frac{\Gamma_i\p{1-g\p{X_i^\top \alpha^*}}}{\p{1-\gamma_n}g\p{X_i^\top \alpha^*}}\p{Y_i - X_i^\top\beta^*}}^2},
\end{align*}
we have $\E\sbr{\widetilde \sigma_t^2 - \sigma_t^2} = 0$.
Thus,
\begin{align*}
    &\E\sbr{\p{\widetilde \sigma_t^2 - \sigma_t^2}^2} \leq n^{-1}\E\sbr{\br{\frac{1-\Gamma_i}{1-\gamma_n}\p{X_i^\top \beta^*-\theta_t} + \frac{\Gamma_i\p{1-g\p{X_i^\top \alpha^*}}}{\p{1-\gamma_n}g\p{X_i^\top \alpha^*}}w_i}^4}\\
    &\qquad=n^{-1}\p{1-\gamma_n}^{-4}\br{\E\sbr{\p{1-\Gamma_i}\p{X_i^\top \beta^*-\theta_t}^4}-\E\sbr{\p{\frac{\Gamma_i\p{1-g\p{X_i^\top \alpha^*}}}{g\p{X_i^\top \alpha^*}}}^4w_i^4}}\\
    &\qquad\leq Cn^{-1}\p{1-\gamma_n}^{-4}\br{\p{1-\gamma_n}\E\sbr{\p{X_i^\top \beta^*}^4 \mid \Gamma_i=0} + \gamma_n^{-4}\E\sbr{\Gamma_iw_i^4}}\\
    &\qquad\leq Cn^{-1}\p{1-\gamma_n}^{-4}\br{\p{1-\gamma_n}\E\sbr{\p{X_i^\top \beta^*}^4 \mid \Gamma_i=0} + \gamma_n^{-3}\E\sbr{w_i^4 \mid \Gamma_i = 1}}=O_p\p{n^{-1}\gamma_n^{-3}},
\end{align*}
which implies
$
    \widetilde \sigma_t^2 - \sigma_t^2 = O_p\p{\gamma_n^{-1}\p{n\gamma_n}^{-1/2}}. 
$
By Lemma~\ref{lemma mar transportability theta t normal}, $\sigma_t^2 \asymp \gamma_n^{-1}$. Then
\begin{align*}
    \widehat \sigma_t^2  &= \sigma_t^2 + O_p\p{\gamma_n^{-1}\p{n\gamma_n}^{-1/2}} + O_p\p{\gamma_n^{-1} (\frac{\p{s_\alpha + s_\beta}\log d}{n\gamma_n})^{1/2}+ (\frac{s_\alpha s_\beta \log d }{n\gamma_n})^{1/2}}\\
    &=\sigma_t^2\br{1+O_p\p{(\frac{\p{s_\alpha + s_\beta}\log d}{n\gamma_n})^{1/2} + \gamma_n(\frac{s_\alpha s_\beta \log d }{n\gamma_n})^{1/2}}}.
\end{align*}
\end{proof}

\begin{lemma}\label{lemma key results for transportability}
Let Assumptions~\ref{assumption mar}  and \ref{assumption mar nuisance (a)} hold. Assume either $\gamma_n\p{X} = g\p{X^\top \alpha^*_{PS}}$ or $\mu\p{X} = X^\top \beta^*_{OR}$ holds. Choose $\lambda_\alpha \asymp \lambda_\beta \asymp (\log d/(n\gamma_n))^{1/2}$. If $n\gamma_n \gg  (\log n)^2\log d$, $s_\alpha = o((n\gamma_n)^{1/2}/\log d)$ and
\[
s_\alpha s_\beta = o\{n\gamma_n/(\log n\p{\log d}^2)\},
\]
then as $n, d\rightarrow \infty$, $\widehat \theta_t -\theta_t = O_p \p{(n\gamma_n)^{-1/2}}$, $\widehat \sigma^2_t = \sigma_t^2\br{1 + o_p(1)}$, and $\widehat\sigma_t^{-1}n^{1/2}(\widehat \theta_t - \theta_t)\xrightarrow{d} \mc{N}(0,1)$.

In addition, Assume that $\mu\p{X} = X^\top \beta^*_{OR}$ holds. Choose $\lambda_\alpha \asymp \lambda_\beta \asymp (\log d/(n\gamma_n))^{1/2}$. If $n\gamma_n \gg  (\log n)^2\log d$ and $s_\alpha s_\beta = o\{n\gamma_n/(\log n\p{\log d}^2)\}$, then as $n, d\rightarrow \infty$, $\widehat \theta_t -\theta_t = O_p \p{(n\gamma_n)^{-1/2}}$, $\widehat \sigma^2_t = \sigma_t^2\br{1 + o_p(1)}$, and $\widehat\sigma_t^{-1}n^{1/2}(\widehat \theta_t - \theta_t)\xrightarrow{d} \mc{N}(0,1)$.
\end{lemma}

\begin{proof}
    When $n\gamma_n \gg  (\log n)^2\log d$, $s_\alpha = o((n\gamma_n)^{1/2}/\log d)$ and $s_\alpha s_\beta = o(n\gamma_n/(\log n\p{\log d}^2))$, we have $n\gamma_n \gg \max \{s_\alpha$, $ s_\beta\log n, (\log n)^2\} \log d$ and  $s_\alpha s_\beta \ll (n\gamma_n)^{3/2}/(\log n(\log d)^2)$.
    
    First, by Lemma~\ref{lemma mar transportability theta t normal} and Lemma~\ref{lemma mar transportability theta hat - theta tilde consistency}, we have
    \begin{align*}
        \widehat \theta_t - \theta_t = O_p\p{\frac{\p{s_\alpha + (s_\alpha s_\beta)^{1/2}}\log d}{n\gamma_n} + \p{n\gamma_n}^{-1/2}} = O_p(\p{n\gamma_n}^{-1/2}).
    \end{align*}

    Second, by Lemma~\ref{lemma mar sigma hat = sigma (1+o(1))}, 
    \begin{align*}
        \widehat \sigma_t^2 &= \sigma_t^2\br{1+O_p\p{(\frac{\p{s_\alpha + s_\beta}\log d}{n\gamma_n})^{1/2} + \gamma_n(\frac{s_\alpha s_\beta \log d }{n\gamma_n})^{1/2}}}= \sigma^2_g\br{1 + o_p(1)}.
    \end{align*}

    Third, by Lemma~\ref{lemma mar transportability theta hat - theta tilde consistency}, we have
    \begin{align*}
        \widehat \theta_g - \widetilde\theta_g = O_p\p{\frac{\p{s_\alpha + (s_\alpha s_\beta)^{1/2}} \log d}{n\gamma_n}} = o_p(\p{n\gamma_n}^{-1/2}).
    \end{align*}
    By Lemma~\ref{lemma mar transportability theta t normal}, since $\sigma_t^2 \asymp \gamma_n^{-1}$, by Slutsky's theorem,
    \begin{align*}
        \widehat\sigma^{-1}_tn^{1/2}(\widehat \theta_t - \theta_t) = \widehat\sigma_t^{-1}\sigma_t\p{\sigma_g^{-1}n^{1/2}(\widetilde \theta_t - \theta_t) +  \sigma_t^{-1}n^{1/2}\p{\widehat \theta_t - \widetilde\theta_t}} \rightarrow \mc{N}(0,1).
    \end{align*}

In addition, if $\mu\p{X} = X^\top \beta^*_{OR}$, we only need $s_\alpha s_\beta = o(n\gamma_n/(\log n(\log d)^2))$, which implies $s_\alpha = o(n\gamma_n/(\log n(\log d)^2))$. Under these conditions, by Lemma~\ref{lemma mar transportability theta hat - theta tilde consistency},
\begin{align*}
    \widehat \theta_t - \widetilde\theta_t =   O_p\p{\frac{(s_\alpha \log d)^{1/2}}{n\gamma_n} +\frac{(s_\alpha s_\beta)^{1/2} \log d}{n\gamma_n}} = o_p((n\gamma_n)^{-1/2}).
\end{align*}
The rest of proof follows similarly.
\end{proof}

\subsection{Proof of Theorem~\ref{theorem mar generalizability Asymptotics body}}
\begin{proof}
    Theorem~\ref{theorem mar generalizability Asymptotics body} follows directly from Lemma~\ref{lemma key results for generalizability} and Lemma~\ref{lemma key results for transportability}.
\end{proof}

\section{Proof of results in Section~\ref{sec: causal'}}\label{sec: proof causal}

For each $a \in \{0,1\}$, define 
\begin{align*}
    \widetilde \beta^{(-k)}_{a,g} &= \argmin_{\beta \in \R^d} \br{ M^{-1}\sum_{i \in \mathcal{I}_{-k,\beta}} \Gamma_{a,i} \exp\p{-X_i^\top \alpha_{a,g}^*}\p{Y_i - X_i^\top \beta}^2 + \lambda_\beta \norm{\beta}_1 },\\
    \widetilde \beta^{(-k)}_{a,t} &= \argmin_{\beta \in \R^d} \br{ M^{-1}\sum_{i \in \mathcal{I}_{-k,\beta}} \Gamma_{a,i} \exp\p{-X_i^\top \widehat \alpha_{a,t}^*}\p{Y_i - X_i^\top \beta}^2 + \lambda_\beta \norm{\beta}_1 }.
\end{align*}
In addition, for each $a \in \{0,1\}$, we define the sparsity levels of $\alpha_{a,g}^*$, $\beta_{a,g}^*$ and $\alpha_{a,t}^*$, $\beta_{a,t}^*$ as $s_{\alpha_{a,g}}, s_{\beta_{a,g}}$ and $s_{\alpha_{a,t}}, s_{\beta_{a,t}}$, respectively.

\subsection{Auxiliary lemmas for generalizability}
\begin{lemma}\label{lemma causal generalizability}
    Let Assumption~\ref{assumption causal} hold and for each \(a \in \{0,1\}\), Assumption~\ref{assumption mar nuisance (a)} hold with \((\alpha^*_{PS}, \beta^*_{OR}, w_{OR})\) replaced by \((\alpha^*_{a,g}, \beta^*_{a,g}, w_{a,g})\). Then 

    (a) For each $a\in \{0,1\}$, $\Gamma_{a,i} \perp Y_i\p{a} \mid X_i$;

    (b) For each $a \in \{0,1\}$, there exists some constant $k_0, c_0 \in (0,1)$ such that 
    \begin{align*}
        k_0(1-\gamma_{a,g})/\gamma_{a,g} \leq (1-g(X^\top \alpha^*_{a,g}))/g(X^\top \alpha^*_{a,g}) \leq k_0^{-1}(1-\gamma_{a,g})/\gamma_{a,g},
    \end{align*}
    where $\gamma_{a,g} = \P\p{\Gamma_{a,i}=1} \asymp \gamma_n$ and $1-\gamma_{a,g} \geq c_0$. In addition, for some $q>1$ and $\nu>0$, $\E[\P(\Gamma_{a,i}=1\mid X_i)^q] \leq \nu\gamma_{a,g}^q$.
    
    (c) For each $a \in \{0,1\}$, there exist some constants $\lambda_l,\sigma >0$, $X_i$ is sub-Gaussian given $\Gamma_{a,i}=\Gamma_a$ such that $\norm{X_i^\top \beta^*_{a,g}}_{\psi_2} \leq \sigma$ and
    \begin{align*}
        \norm{X_i}_{g, \Gamma_a, \psi_2} := \sup_{v \in \R^d, \norm{v}_2 = 1}\inf \br{t > 0: \E\sbr{\exp\br{\frac{\p{X_i^\top v}^2}{t^2}} \mid \Gamma_{a,i}=\Gamma_a}\leq 2}\leq \sigma.
    \end{align*}
    In addition, $\lambda_{\min}\p{\E\sbr{X_iX_i^\top \mid \Gamma_{a,i}=1}} \geq \lambda_l$.

    (d) For each $a \in \{0,1\}$, let $w_{a,g,i} = Y_i\p{a} - X_i^\top \beta^*_{a,g}$. There exist some constants $\sigma_w, \delta_w>0$, $w_{a,g,i}$ is sub-Gaussian with $\norm{w_{a,g,i}}_{\psi_2} \leq \sigma_w$ and 
    \begin{align*}
        \E\sbr{w_{a,g,i}^8 \mid \Gamma_{a,i}=1} \leq \sigma_w^8 \quad \text{and} \quad \E\sbr{w_{a,g,i}^2 \mid \Gamma_{a,i}=1} \geq \delta_w.
    \end{align*}
\end{lemma}
\begin{proof}
    (a) follows directly from Assumption~\ref{assumption causal}.

    (b) Under Assumption~\ref{assumption causal}, 
    \begin{align*}
        \gamma_{a,g} &= \P\p{\Gamma_{a,i}=1}=\E\sbr{\Gamma_i\E\sbr{\mathbbm{1}_{\br{A_i = a}} \mid \Gamma_i}} \asymp \E\sbr{\Gamma_i} = \gamma_n,\\
        1-\gamma_{a,g} &= \P\p{\Gamma_{a,i}=0}=\E\sbr{\p{1-\Gamma_i}\E\sbr{\mathbbm{1}_{\br{A_i = a}} \mid \Gamma_i}} \asymp \E\sbr{\p{1-\Gamma_i}} = 1-\gamma_n,
    \end{align*}
    which implies (b).

    (c) Let $\Gamma_a=1$. For any $t>0$ and $v \in \R^d, \norm{v}_2 = 1$, under Assumption~\ref{assumption causal},
    \begin{align*}
       &\E\sbr{\exp\br{\frac{\p{X_i^\top v}^2}{t^2}} \mid \Gamma_{a,i}=1} = \gamma_{a,g}^{-1}\E\sbr{\Gamma_{i}\mathbbm{1}_{\br{A_i = a}}\exp\br{\frac{\p{X_i^\top v}^2}{t^2}}}\\
       &\qquad=\gamma_{a,g}^{-1}\E\sbr{\E\sbr{\mathbbm{1}_{\br{A_i = a}} \mid \Gamma_i, X_i}\Gamma_{i}\exp\br{\frac{\p{X_i^\top v}^2}{t^2}}}\asymp \gamma_{n}^{-1}\E\sbr{\Gamma_{i}\exp\br{\frac{\p{X_i^\top v}^2}{t^2}}}\\
       &\qquad=\E\sbr{\exp\br{\frac{\p{X_i^\top v}^2}{t^2}} \mid \Gamma_i = 1}.
    \end{align*}
Thus,
$
    \norm{X_i}_{g, 1, \psi_2} \asymp \norm{X_i}_{1, \psi_2}.
$

    Let $\Gamma_a = 0$. Note that under Assumption~\ref{assumption causal},
    \begin{align*}
        \P\p{\Gamma_{a,i}=0} &= \P\sbr{\br{\Gamma_i=0}\cup \br{\Gamma_i=1, A_i \neq a}}=\P\p{\Gamma_i=0} + \P\p{A_i \neq a \mid \Gamma_i=1}\P\p{\Gamma_i=1}\\
        &\geq 1-\gamma_n + \eta_0 \gamma_n \geq \eta_0.
    \end{align*}
    In addition, for any $t>0$ and $v \in \R^d, \norm{v}_2 = 1$,
    \begin{align*}
        \E\sbr{\exp\br{\frac{\p{X_i^\top v}^2}{t^2}} \mid \Gamma_{a,i}=0}
        &= \p{1-\gamma_{a,g}}^{-1}\E\sbr{\p{1-\Gamma_{i}\mathbbm{1}_{\br{A_i = a}}}\exp\br{\frac{\p{X_i^\top v}^2}{t^2}}}\\
        &\leq \eta^{-1}\E\sbr{\exp\br{\frac{\p{X_i^\top v}^2}{t^2}}}.
    \end{align*}
    Thus,
$
    \norm{X_i}_{g, 0, \psi_2} \leq (\log \p{1/\eta_0})^{1/2}\norm{X_i}_{\psi_2}.
$

(d) Note that $\gamma_{a,g} = \P\p{\Gamma_{a,i}=1}=\E\sbr{\Gamma_i\E\sbr{\mathbbm{1}_{\br{A_i = a}} \mid \Gamma_i}}$. Under Assumption~\ref{assumption causal},
$
    \eta_0 \gamma_n \leq \gamma_{a,g} \leq \p{1-\eta_0} \gamma_n.
$
Then we have
\begin{align*}
    \E\sbr{w_{a,g,i}^8 \mid \Gamma_{a,i}=1} &= \gamma_{a,g}^{-1}\E\sbr{\Gamma_i\mathbbm{1}_{\br{A_i = a}}w_{a,g,i}^8}\leq \eta_0^{-1} \gamma_n^{-1} \E\sbr{\Gamma_iw_{a,g,i}^8}=\eta_0^{-1}\E\sbr{w_{a,g,i}^8 \mid \Gamma_i=1} \leq \eta_0^{-1}\sigma_w^8,
\end{align*}
and
\begin{align*}
    \E\sbr{w_{a,g,i}^2 \mid \Gamma_{a,i}=1}&= \gamma_{a,g}^{-1}\E\sbr{\Gamma_i\mathbbm{1}_{\br{A_i = a}}w_{a,g,i}^2}\geq \p{1-\eta_0}^{-1} \gamma_n^{-1} \E\sbr{\Gamma_iw_{a,g,i}^2}\\
    &=\p{1-\eta_0}^{-1}\E\sbr{w_{a,g,i}^2 \mid \Gamma_i=1} \geq \eta^{-1}\delta_w.
\end{align*}
\end{proof}

\begin{lemma}\label{lemma causal generalizability tau tilde normal}
Let Assumption~\ref{assumption causal} hold, and for each \(a \in \{0,1\}\), Assumption~\ref{assumption mar nuisance (a)} hold with \((\alpha^*_{PS}, \beta^*_{OR}, w_{OR})\) replaced by \((\alpha^*_{a,g}, \beta^*_{a,g}, w_{a,g})\). Let either $\P\p{\Gamma_{a,i} = 1 \mid X_i} = g\p{X_i^\top \alpha^*_{a,g}}$ or $\E\sbr{Y_i(a) \mid X_i} = X_i^\top \beta^*_{a,g}$ holds for each $a\in\{0,1\}$. If $n\gamma_n \gg 1$, then as $n, d \rightarrow \infty$, $\Sigma_{g}^2 \asymp \gamma_n^{-1}$ and 
$
\Sigma_{g}^{-1}n^{1/2}\p{\widetilde \tau_{g} - \tau_{g}} \rightarrow \mathcal{N}\p{0,1},
$
    where $\widetilde \tau_{g} = \widetilde \tau_{1,g} - \widetilde \tau_{0,g}$ and $\widetilde \tau_{a,g} = n^{-1}\sum_{i=1}^n X_i^\top \beta^*_{a,g} +  {\Gamma_{a,i}}/{g\p{X_i^\top \alpha^*_{a,g}}}\p{Y_i - X_i^\top \beta^*_{a,g}}$.
\end{lemma}

\begin{proof}
    Let $ Q_{g,i} = Q_{1,g,i} - Q_{0,g,i}$,
where
\begin{align*}
    Q_{a,g,i} = X_i^\top \beta^*_{a,g} +  \frac{\Gamma_{a,i}}{g\p{X_i^\top \alpha^*_{a,g}}}\p{Y_i - X_i^\top \beta^*_{a,g}}.
\end{align*}
Note that $\Gamma_{a,i} = \Gamma_i\mathbbm{1}_{\br{A_i=a}}$, then under Assumption~\ref{assumption causal}, we have $\Gamma_{a,i} \perp Y_i(a) \mid X_i$ and $\Gamma_{a,i}Y_i = \Gamma_{a,i}Y_i\p{a}$.
    It follows that
    \begin{align*}
       \E\sbr{Q_{a,g,i}} -  \E\sbr{Y_i\p{a}} &= \E\sbr{X_i^\top \beta^*_{a,g} +  \frac{\Gamma_{a,i}}{g\p{X_i^\top \alpha^*_{a,g}}}\p{Y_i\p{a} - X_i^\top \beta^*_{a,g}}} - \E\sbr{Y_i\p{a}}\\
       &=\E\sbr{\p{\frac{\Gamma_{a,i}}{g\p{X_i^\top \alpha^*_{a,g}}} - 1}  \p{Y_i\p{a} - X_i^\top \beta^*_{a,g}}}\\
       &=\E\sbr{\frac{\E\sbr{\Gamma_{a,i} - g\p{X_i^\top \alpha^*_{a,g}} \mid X_i}}{g\p{X_i^\top \alpha^*_{a,g}}}  \E\sbr{Y_i\p{a} - X_i^\top \beta^*_{a,g} \mid X_i}}.
    \end{align*}
When either $\P\sbr{\Gamma_{a,i} = 1 \mid X_i} = g\p{X_i^\top \alpha^*_{a,g}}$ or $\E\sbr{Y_i(a) \mid X_i} = X_i^\top \beta^*_{a,g}$ holds for each $a\in\{0,1\}$, we have
$
    \E\sbr{Q_{a,g,i}} -  \E\sbr{Y_i\p{a}} = 0.
$
Thus,
$
       \E\sbr{\widetilde \tau_{g}} =\E\sbr{Q_{g,i}} = \tau_g.
$
    By Lyapunov's central limit theorem, it suffices to prove for some $\delta>0$,
    \begin{align}
        \lim_{n\rightarrow \infty}n^{-\delta/2}\Sigma_g^{-\p{2+\delta}} \E\sbr{\abs{Q_{g,i} - \theta}^{2+\delta}} = 0.
    \end{align}

    Note that 
    \begin{align*}
        \Sigma_g^2 &= \E\sbr{\p{X_i^\top \beta^*_{1,g}- X_i^\top \beta^*_{0,g} - \tau_g}^2} + \E\sbr{\frac{\Gamma_iA_i}{\p{g\p{X_i^\top \alpha^*_{1,g}}}^2}\p{Y_i(1) - X_i^\top \beta^*_{1,g}}^2}\\
        &\qquad + \E\sbr{\frac{\Gamma_i\p{1-A_i}}{\p{g\p{X_i^\top \alpha^*_{0,g}}}^2}\p{Y_i(0) - X_i^\top \beta^*_{0,g}}^2}\\
        &\qquad + 2\E\sbr{\p{X_i^\top \beta^*_{1,g}- X_i^\top \beta^*_{0,g} - \tau_g}\frac{\Gamma_iA_i}{g\p{X_i^\top \alpha^*_{1,g}}}\p{Y_i(1) - X_i^\top \beta^*_{1,g}}}\\
        &\qquad - 2\E\sbr{\p{X_i^\top \beta^*_{1,g}- X_i^\top \beta^*_{0,g} - \tau_g}\frac{\Gamma_i\p{1-A_i}}{g\p{X_i^\top \alpha^*_{0,g}}}\p{Y_i(0) - X_i^\top \beta^*_{0,g}}}.
    \end{align*}
    By Minkowski's inequality,
    \begin{align*}
        \norm{X_i^\top \beta^*_{1,g}- X_i^\top \beta^*_{0,g} - \tau_g}_{\P,2} &\leq \norm{X_i^\top \beta^*_{1,g}}_{\P,2} + \norm{X_i^\top \beta^*_{0,g}}_{\P,2} + \abs{\tau_g}\leq C_1 \sigma +  \abs{\tau_g},
    \end{align*}
    which implies
$
        \E\sbr{\p{X_i^\top \beta^*_{1,g}- X_i^\top \beta^*_{0,g} - \tau_g}^2} \leq \p{C_1 \sigma +  \abs{\tau_g}}^2.
$
    By Lemma~\ref{lemma causal generalizability},
    \begin{align*}
        &\E\sbr{\frac{\Gamma_iA_i}{\p{g\p{X_i^\top \alpha^*_{1,g}}}^2}\p{Y_i(1) - X_i^\top \beta^*_{1,g}}^2} \\
        &\qquad= \E\sbr{\E\sbr{A_i \mid \Gamma_i, X_i}\E\sbr{\frac{\Gamma_i}{\p{g\p{X_i^\top \alpha^*_{1,g}}}^2}\p{Y_i(1) - X_i^\top \beta^*_{1,g}}^2 \mid \Gamma_i, X_i}}\\
        &\qquad\asymp \E\sbr{\frac{\Gamma_i}{\p{g\p{X_i^\top \alpha^*_{1,g}}}^2}\p{Y_i(1) - X_i^\top \beta^*_{1,g}}^2}= \gamma_n\E\sbr{\frac{1}{\p{g\p{X_i^\top \alpha^*_{1,g}}}^2}\p{Y_i(1) - X_i^\top \beta^*_{1,g}}^2 \mid \Gamma_i=1}\\
        &\qquad\asymp \gamma_n^{-1}\E\sbr{\p{Y_i(1) - X_i^\top \beta^*_{1,g}}^2 \mid \Gamma_i=1}\asymp \gamma_n^{-1}.
    \end{align*}
    Similarly, we have
$
        \E\sbr{{\Gamma_i\p{1-A_i}}/{\p{g\p{X_i^\top \alpha^*_{0,g}}}^2}\p{Y_i(0) - X_i^\top \beta^*_{0,g}}^2} \asymp \gamma_n^{-1}.
$
Moreover, by H{\"o}lder inequality,
\begin{align*}
    &\abs{\E\sbr{\p{X_i^\top \beta^*_{1,g}- X_i^\top \beta^*_{0,g} - \tau_g}\frac{\Gamma_iA_i}{g\p{X_i^\top \alpha^*_{1,g}}}\p{Y_i(1) - X_i^\top \beta^*_{1,g}}}}\\
    &\qquad=\abs{\gamma_n\E\sbr{\frac{A_i}{g\p{X_i^\top \alpha^*_{1,g}}}\p{X_i^\top \beta^*_{1,g}- X_i^\top \beta^*_{0,g} - \tau_g}\p{Y_i(1) - X_i^\top \beta^*_{1,g}}\mid \Gamma_i=1} }\\
    &\qquad\leq \abs{k_0^{-1}\E\sbr{\p{X_i^\top \beta^*_{1,g}- X_i^\top \beta^*_{0,g} - \tau_g}\p{Y_i(1) - X_i^\top \beta^*_{1,g}}\mid \Gamma_i=1} }\\
    &\qquad\leq k_0^{-1}\E\sbr{\p{X_i^\top \beta^*_{1,g}- X_i^\top \beta^*_{0,g} - \tau_g}^2\mid \Gamma_i=1}^{1/2}\E\sbr{\p{Y_i(1) - X_i^\top \beta^*_{1,g}}^2\mid \Gamma_i=1}^{1/2}\\
    &\qquad\leq k_0^{-1}\p{C_1\sigma + \abs{\tau_g}}\sigma_w.
\end{align*}
Similarly, we have
\begin{align*}
    \abs{\E\sbr{\p{X_i^\top \beta^*_{1,g}- X_i^\top \beta^*_{0,g} - \tau_g}\frac{\Gamma_i\p{1-A_i}}{g\p{X_i^\top \alpha^*_{0,g}}}\p{Y_i(0) - X_i^\top \beta^*_{0,g}}}} \leq k_0^{-1}\p{C_1\sigma + \abs{\tau_g}}\sigma_w.
\end{align*}
Thus,
$
    \Sigma_g^2 \asymp \gamma_n^{-1}.
$

On the other hand, choose $\delta = 2$. Then
\begin{align*}
    \norm{Q_{g,i} - \tau_g}_{\P,4} &\leq \norm{ X_i^\top \beta^*_{1,g}- X_i^\top \beta^*_{0,g} - \tau_g}_{\P,4} + \norm{\frac{\Gamma_iA_i}{g\p{X_i^\top \alpha^*_{1,g}}}\p{Y_i(1) - X_i^\top \beta^*_{1,g}}}_{\P,4} \\
    &\qquad + \norm{\frac{\Gamma_{i}\p{1-A_i}}{g\p{X_i^\top \alpha^*_{0,g}}}\p{Y_i(0) - X_i^\top \beta^*_{0,g}}}_{\P,4}\\
    &\leq \norm{ X_i^\top \beta^*_{1,g}- X_i^\top \beta^*_{0,g} - \tau_g}_{\P,4} + k_0^{-1}\gamma_n^{-1}\norm{\Gamma_i\p{Y_i(1) - X_i^\top \beta^*_{1,g}}}_{\P,4}\\
    &\qquad + k_0^{-1}\gamma_n^{-1}\norm{\Gamma_i\p{Y_i(0) - X_i^\top \beta^*_{0,g}}}_{\P,4}.
\end{align*}
In addition, we have
\begin{align*}
    \norm{ X_i^\top \beta^*_{1,g}- X_i^\top \beta^*_{0,g} - \tau_g}_{\P,4} &\leq \norm{ X_i^\top \beta^*_{1,g}}_{\P,4} + \norm{X_i^\top \beta^*_{0,g}}_{\P,4} + \abs{\tau_g}\leq C_2\sigma + \abs{\tau_g},
\end{align*}
and
\begin{align*}
    \norm{\Gamma_i\p{Y_i(1) - X_i^\top \beta^*_{1,g}}}_{\P,4} &= \E\sbr{\Gamma_i\p{Y_i(1) - X_i^\top \beta^*_{1,g}}^4}^{1/4}\\
    &=\gamma_n^{1/4}\E\sbr{\p{Y_i(1) - X_i^\top \beta^*_{1,g}}^4 \mid \Gamma_i=1}^{1/4}\leq \sigma_w\gamma_n^{1/4}.
\end{align*}
Similarly, we have
$
    \norm{\Gamma_i\p{Y_i(0) - X_i^\top \beta^*_{0,g}}}_{\P,4} \leq \sigma_w\gamma_n^{1/4}.
$
Thus,
\begin{align*}
    \E\sbr{\abs{Q_{g,i} - \tau_g}^4} = \norm{Q_{g,i} - \tau_g}_{\P,4}^4 &\leq \p{C_2\sigma + \abs{\tau_g} + 2k_0^{-1}\sigma_w\gamma_n^{-3/4}}^4\\
    &\leq \gamma_n^{-3}\p{C_2\sigma + \abs{\tau_g} + 2k_0^{-1}\sigma_w}^4=: C_3\gamma_n^{-3}.
\end{align*}
It follows that for $\delta=2$, if $n\gamma_n \gg 1$, as $n, d \rightarrow \infty$,
\begin{align*}
    \lim_{n\rightarrow \infty}n^{-1}\Sigma_g^{-4} \E\sbr{\abs{Q_{g,i} - \theta}^{4}} \leq \lim_{n\rightarrow \infty}C\p{n\gamma_n}^{-1} = 0.
\end{align*}
\end{proof}
\begin{lemma}\label{lemma causal generalizability nuisance consistency}
    Let Assumption~\ref{assumption causal} hold, and for each \(a \in \{0,1\}\), Assumption~\ref{assumption mar nuisance (a)} holds with \((\alpha^*_{PS}, \beta^*_{OR}, w_{OR})\) replaced by \((\alpha^*_{a,g}, \beta^*_{a,g}, w_{a,g})\). Then as $n, d \rightarrow \infty$, it holds that

(a) Choose $\lambda_\alpha \asymp (\log d/(n\gamma_n))^{1/2}$. If $n\gamma_n \gg \max\br{s_{\alpha_{a,g}}, \log n} \log d$, then
    $$\norm{\widehat \alpha_{a,g}^{(-k)} - \alpha^*_{a,g}}_1 = O_p\p{s_{\alpha_{a,g}}(\log d/(n\gamma_n))^{1/2}} \text{ and } \norm{\widehat \alpha_{a,g}^{(-k)} - \alpha^*_{a,g}}_2 = O_p\p{(\frac{s_{\alpha_{a,g}} \log d}{n\gamma_n})^{1/2}}.$$

(b) Choose $\lambda_\alpha \asymp (\log d/(n\gamma_n))^{1/2}$. If $n\gamma_n \gg \max\br{s_{\beta_{a,g}}, (\log n)^2} \log d$, then
$$\norm{\widetilde \beta_{a,g}^{(-k)} - \beta^*_{a,g}}_1 = O_p\p{s_{\beta_{a,g}}(\log d/(n\gamma_n))^{1/2}} \text{ and } \norm{\widetilde \beta_{a,g}^{(-k)} - \beta^*_{a,g}}_2 = O_p\p{(\frac{s_{\beta_{a,g}} \log d}{n\gamma_n})^{1/2}}.$$
 
    (c) Choose $\lambda_\alpha \asymp \lambda_\beta \asymp (\log d/(n\gamma_n))^{1/2}$. If
\[
n\gamma_n \gg \max\br{s_{\alpha_{a,g}}, s_{\beta_{a,g}}\log n , (\log n)^2} \log d
\]
and $s_{\alpha_{a,g}} s_{\beta_{a,g}} \ll (n\gamma_n)^{3/2}/(\log n(\log d)^2)$, then
    \begin{align*}
     \norm{\widehat \beta_{a,g}^{(-k)} - \widetilde \beta_{a,g}^{(-k)}}_1 = O_p\p{(\frac{s_{\alpha_{a,g}} s_{\beta_{a,g}} \log d}{n\gamma_n})^{1/2}}\text{ and } \norm{\widehat \beta_{a,g}^{(-k)} - \widetilde \beta_{a,g}^{(-k)}}_2 = O_p\p{(\frac{s_{\alpha_{a,g}} \log d}{n\gamma_n})^{1/2}}.
    \end{align*}
\end{lemma}
\begin{proof}
     Lemma~\ref{lemma causal generalizability nuisance consistency} follows by repeating the arguments used in the proofs of Lemmas~\ref{lemma consistency for alpha and betatilde} and \ref{lemma consistency of betahat and betatilde for product sparsity}.
\end{proof}

\begin{lemma}\label{lemma causal generalizability hat tau - tilde tau consistency}
    Let Assumption~\ref{assumption causal} hold, and suppose that, for each \(a \in \{0,1\}\), Assumption~\ref{assumption mar nuisance (a)} holds with \((\alpha^*_{PS}, \beta^*_{OR}, w_{OR})\) replaced by \((\alpha^*_{a,g}, \beta^*_{a,g}, w_{a,g})\). Choose $\lambda_\alpha \asymp \lambda_\beta \asymp (\log d/(n\gamma_n))^{1/2}$. If
    \[
        n\gamma_n \gg \sum_{a=0,1}\max\br{s_{\alpha_{a,g}}, s_{\beta_{a,g}}\log n , (\log n)^2} \log d,
        \qquad
        \sum_{a=0,1}s_{\alpha_{a,g}} s_{\beta_{a,g}} \ll (n\gamma_n)^{3/2}/\{\log n(\log d)^2\},
    \]
    then as \(n, d\rightarrow \infty\),
    \begin{align*}
        \widehat \tau_g - \widetilde \tau_g &= O_p\p{\sum_{a=0,1}\frac{\p{s_{\alpha_{a,g}} + (s_{\alpha_{a,g}} s_{\beta_{a,g}})^{1/2}} \log d}{n\gamma_n}}.
    \end{align*}
In addition, assume that $\E\sbr{Y_i(a) \mid X_i} = X_i^\top \beta^*_{a,g}$ for all $a \in \{0,1\}$. Then as $n, d\rightarrow \infty$, 
\begin{align*}
\widehat \tau_g - \widetilde \tau_g&=O_p\p{\sum_{a=0,1}\br{\frac{(s_{\alpha_{a,g}} \log d)^{1/2}}{n\gamma_n} +\frac{(s_{\alpha_{a,g}} s_{\beta_{a,g}})^{1/2} \log d}{n\gamma_n}}}.
\end{align*}
\end{lemma}
\begin{proof}
    Lemma~\ref{lemma causal generalizability hat tau - tilde tau consistency} follows by repeating the proof of Lemma~\ref{lemma mar theta hat - theta tilde consistency}.
\end{proof}

\begin{lemma}\label{lemma causal generalizability hat sigma = sigma(1+o(1))}
    Let Assumption~\ref{assumption causal} hold, and for each \(a \in \{0,1\}\), Assumption~\ref{assumption mar nuisance (a)} hold with \((\alpha^*_{PS}, \beta^*_{OR}, w_{OR})\) replaced by \((\alpha^*_{a,g}, \beta^*_{a,g}, w_{a,g})\). Let either $\P\p{\Gamma_{a,i} = 1 \mid X_i} = g\p{X_i^\top \alpha^*_{a,g}}$ or $\E\sbr{Y_i(a) \mid X_i} = X_i^\top \beta^*_{a,g}$ hold for each $a\in\{0,1\}$. Choose $\lambda_\alpha \asymp \lambda_\beta \asymp (\log d/(n\gamma_n))^{1/2}$. If
\[
n\gamma_n \gg \sum_{a=0,1}\max\br{s_{\alpha_{a,g}}, s_{\beta_{a,g}}\log n , (\log n)^2} \log d
\]
and $\sum_{a=0,1}s_{\alpha_{a,g}} s_{\beta_{a,g}} \ll (n\gamma_n)^{3/2}/(\log n(\log d)^2)$, then as $n, d\rightarrow \infty$,
    \begin{align*}
    \widehat \Sigma_g^2 = \Sigma_g^2\br{1 + O_p\p{\sum_{a=0,1}(\frac{ \p{s_{\alpha_{a,g}}+s_{\beta_{a,g}}} \log d}{n\gamma_n})^{1/2}}}.
\end{align*}
\end{lemma}
\begin{proof}
    Let $\widehat \Delta_{\alpha_{a,g}}^{(-k)} = \widehat \alpha^{(-k)}_{a,g} - \alpha^*_{a,g}$, $\widehat \Delta_{\beta_{a,g}}^{(-k)} = \widehat \beta^{(-k)}_{a,g} - \beta^*_{a,g}$, $\Delta_{a,g}^{(-k)} = \widehat \beta^{(-k)}_{a,g} - \beta^*_{a,g}$, $\widetilde \Delta_{\beta_{a,g}}^{(-k)} = \widetilde \beta^{(-k)}_{a,g} - \beta^*_{a,g}$, $\widehat \gamma_k = n_k^{-1}\sum_{i\in \mathcal{I}_k} \Gamma_i$, $\widehat \gamma = n^{-1}\sum_{i=1}^n\Gamma_i$, and 
    \begin{align*}
        \widetilde \Sigma_g^2 &= n^{-1}\sum_{i=1}^n \p{1-\Gamma_i}\p{X_i^\top \beta^*_{1,g} - X_i^\top \beta^*_{0,g}}^2\\
        &\qquad + n^{-1}\sum_{i=1}^n \Gamma_iA_i\p{X_i^\top \beta^*_{1,g} - X_i^\top \beta^*_{0,g} + \frac{Y_i - X_i^\top \beta^*_{1,g}}{g\p{X_i^\top \alpha^*_{1,g}}}}^2\\
        &\qquad + n^{-1}\sum_{i=1}^n \Gamma_i\p{1 - A_i}\p{X_i^\top \beta^*_{1,g} - X_i^\top \beta^*_{0,g} - \frac{Y_i - X_i^\top \beta^*_{0,g}}{g\p{X_i^\top \alpha^*_{0,g}}}}^2 - \tau_g^2.
    \end{align*}
Then 
\begin{align*}
    \widehat \Sigma_g^2 - \widetilde \Sigma_g^2 = K^{-1}\sum_{k=1}^K\p{\ell_{k1a} - \ell_{k1b} + \ell_{k2a} - \ell_{k2b} +\ell_{k3a}- \ell_{k3b} -\ell_{k4}},
\end{align*}
where
\begin{align*}
\ell_{k1a} &= \p{1-\widehat \gamma_k}\p{\widehat \beta^{(-k), \top}_{1,g}\bar 
    \Xi_0 \widehat \beta^{(-k)}_{1,g} + \widehat \beta^{(-k), \top}_{0,g}\bar 
    \Xi_0 \widehat \beta^{(-k)}_{0,g} - 2\widehat \beta^{(-k), \top}_{1,g}\bar 
    \Xi_0 \widehat \beta^{(-k)}_{0,g}},\\
    \ell_{k1b} &= \p{1-\widehat \gamma_k}\p{\beta^{*, \top}_{1,g}\bar 
    \Xi_0 \beta^{*}_{1,g} +  \beta^{*, \top}_{0,g}\bar 
    \Xi_0 \beta^{*}_{0,g} - 2 \beta^{*, \top}_{1,g}\bar 
    \Xi_0 \beta^{*}_{0,g}},\\
    \ell_{k2a} &= n_k^{-1}\sum_{i \in \mathcal{I}_k}\Gamma_iA_i\p{X_i^\top \widehat \beta^{(-k)}_{1,g} - X_i^\top \widehat \beta^{(-k)}_{0,g} + \frac{Y_i - X_i^\top \widehat \beta^{(-k)}_{1,g} }{g\p{X_i^\top \widehat \alpha^{(-k)}_{1,g}}}}^2,\\
    \ell_{k2b}&=n_k^{-1}\sum_{i \in \mathcal{I}_k}\Gamma_iA_i\p{X_i^\top \beta^{*}_{1,g} - X_i^\top \beta^{*}_{0,g} + \frac{Y_i - X_i^\top \beta^{*}_{1,g} }{g\p{X_i^\top \alpha^{*}_{1,g}}}}^2,\\
    \ell_{k3a}&=n_k^{-1}\sum_{i \in \mathcal{I}_k}\Gamma_i\p{1-A_i}\p{X_i^\top \widehat \beta^{(-k)}_{1,g} - X_i^\top \widehat \beta^{(-k)}_{0,g} - \frac{Y_i - X_i^\top \widehat \beta^{(-k)}_{0,g} }{g\p{X_i^\top \widehat \alpha^{(-k)}_{0,g}}}}^2,\\
    \ell_{k3b}&=n_k^{-1}\sum_{i \in \mathcal{I}_k}\Gamma_i\p{1-A_i}\p{X_i^\top \beta^{*}_{1,g} - X_i^\top \beta^{*}_{0,g} - \frac{Y_i - X_i^\top \beta^{*}_{0,g} }{g\p{X_i^\top\alpha^{*}_{0,g}}}}^2,\quad
    \ell_{k4} = \widehat \tau_{g}^2 - \tau_{g}^2.
\end{align*}

For $\ell_{k1a} - \ell_{k1b}$, we have
\begin{align*}
    \ell_{k1a}- \ell_{k1b} &= \p{1-\widehat \gamma_k}\p{\widehat \beta^{(-k), \top}_{1,g}\bar 
    \Xi_0 \widehat \beta^{(-k)}_{1,g} - \beta^{*, \top}_{1,g}\bar 
    \Xi_0 \beta^{*}_{1,g} }+\p{1-\widehat \gamma_k}\p{ \widehat \beta^{(-k), \top}_{0,g}\bar 
    \Xi_0 \widehat \beta^{(-k)}_{0,g} -  \beta^{*, \top}_{0,g}\bar 
    \Xi_0 \beta^{*}_{0,g}}\\
    &\qquad-2\p{1-\widehat \gamma_k}\p{ \widehat \beta^{(-k), \top}_{1,g}\bar 
    \Xi_0 \widehat \beta^{(-k)}_{0,g} - \beta^{*, \top}_{1,g}\bar 
    \Xi_0 \beta^{*}_{0,g}}.
\end{align*}
Recall $s_{k2}$ in Lemma~\ref{lemma mar sigma hat = sigma (1+o(1))}. Similarly, we have
\begin{align*}
    &\p{1-\widehat \gamma_k}\p{\widehat \beta^{(-k), \top}_{1,g}\bar 
    \Xi_0 \widehat \beta^{(-k)}_{1,g} - \beta^{*, \top}_{1,g}\bar 
    \Xi_0 \beta^{*}_{1,g} } = O_p\p{\norm{\Delta_{1,g}^{(-k)}}_1(\frac{\log d}{n})^{1/2}+ \norm{\Delta_{1,g}^{(-k)}}_2 + \norm{\widehat \Delta_{\beta_{1,g}}^{(-k)}}_2},
\end{align*}
and
\begin{align*}
    &\p{1-\widehat \gamma_k}\p{ \widehat \beta^{(-k), \top}_{0,g}\bar 
    \Xi_0 \widehat \beta^{(-k)}_{0,g} -  \beta^{*, \top}_{0,g}\bar 
    \Xi_0 \beta^{*}_{0,g}}= O_p\p{\norm{\Delta_{0,g}^{(-k)}}_1(\frac{\log d}{n})^{1/2}+ \norm{\Delta_{0,g}^{(-k)}}_2 + \norm{\widehat \Delta_{\beta_{0,g}}^{(-k)}}_2}.
\end{align*}
Note that
\begin{align*}
    \p{1-\widehat \gamma_k}\p{\widehat \beta^{(-k), \top}_{1,g}\bar 
    \Xi_0 \widehat \beta^{(-k)}_{0,g} - \beta^{*, \top}_{1,g}\bar 
    \Xi_0 \beta^{*}_{0,g}} &= \frac{1-\widehat \gamma_k}{1-\widehat \gamma} \p{\ell_{k11} + \ell_{k12} + \ell_{k13}} ,
\end{align*}
where
\begin{align*}
    \ell_{k11} &= n^{-1}\sum_{i=1}^n \p{1-\Gamma_i}\p{X_i^\top \widehat \Delta_{\beta_{1,g}}^{(-k)}}^2,\quad
    \ell_{k12} = n^{-1}\sum_{i=1}^n \p{1-\Gamma_i}\p{X_i^\top \widehat \Delta_{\beta_{0,g}}^{(-k)}}\p{X_i^\top \beta^*_{1,g}},\\
    \ell_{k13} &= n^{-1}\sum_{i=1}^n \p{1-\Gamma_i}\p{X_i^\top \widehat \Delta_{\beta_{1,g}}^{(-k)}}\p{X_i^\top \beta^*_{0,g}}.
\end{align*}
Recall $s_{k21}$ and $s_{k22}$ of Lemma~\ref{lemma mar sigma hat = sigma (1+o(1))}. Similarly, we have
\begin{align*}
    \ell_{k11} &= O_p\p{\norm{\Delta_{1,g}^{(-k)}}_1^2\frac{\log d}{n} + \norm{\Delta_{1,g}^{(-k)}}_2^2 + \norm{\widehat \Delta_{\beta_{1,g}}^{(-k)}}_2^2},\\
    \ell_{k12} &= O_p\p{\norm{\Delta_{0,g}^{(-k)}}_1(\frac{\log d}{n})^{1/2}+ \norm{\Delta_{0,g}^{(-k)}}_2+ \norm{\widehat \Delta_{\beta_{0,g}}^{(-k)}}_2},\\
    \ell_{k13}&= O_p\p{\norm{\Delta_{1,g}^{(-k)}}_1(\frac{\log d}{n})^{1/2}+ \norm{\Delta_{1,g}^{(-k)}}_2+ \norm{\widehat \Delta_{\beta_{1,g}}^{(-k)}}_2}.
\end{align*}
By Lemma~\ref{lemma concentrate gamma}, $\widehat \gamma_k - \gamma_n = o_p(1)$ and  $\widehat \gamma - \gamma_n = o_p(1)$, then
\begin{align*}
    &\p{1-\widehat \gamma_k}\p{\widehat \beta^{(-k), \top}_{1,g}\bar 
    \Xi_0 \widehat \beta^{(-k)}_{0,g} - \beta^{*, \top}_{1,g}\bar 
    \Xi_0 \beta^{*}_{0,g}}= O_p\p{\sum_{a=0,1}\br{\norm{\Delta_{a,g}^{(-k)}}_1(\frac{\log d}{n})^{1/2}+ \norm{\Delta_{a,g}^{(-k)}}_2+ \norm{\widehat \Delta_{\beta_{a,g}}^{(-k)}}_2}}.
\end{align*}
Thus,
\begin{align*}
    \ell_{k1a}- \ell_{k1b} = O_p\p{\sum_{a=0,1}\br{\norm{\Delta_{a,g}^{(-k)}}_1(\frac{\log d}{n})^{1/2}+ \norm{\Delta_{a,g}^{(-k)}}_2+ \norm{\widehat \Delta_{\beta_{a,g}}^{(-k)}}_2}}.
\end{align*}

For $\ell_{k2a}- \ell_{k2b}$, let 
\begin{align*}
    \varphi_{a,i}\p{\alpha, \beta_1, \beta_0} = \Gamma_{a,i}\p{X_i^\top \beta_1 - X_i^\top \beta_0 + \frac{Y_i - X_i^\top \beta_1 }{g\p{X_i^\top \alpha}}}.
\end{align*}
Since $a^2 - b^2 = \p{a - b}^2 + 2\p{a-b}b$, it follows
\begin{align*}
    \ell_{k2a}- \ell_{k2b}&= n_k^{-1}\sum_{i \in \mathcal{I}_k}\p{\varphi_{1,i}\p{\widehat \alpha^{(-k)}_{1,g}, \widehat \beta^{(-k)}_{1,g}, \widehat \beta^{(-k)}_{0,g}} - \varphi_{1,i}\p{\alpha^{*}_{1,g}, \beta^{*}_{1,g}, \beta^{*}_{0,g}}}^2\\
    &\qquad+2 n_k^{-1}\sum_{i \in \mathcal{I}_k}\p{\varphi_{1,i}\p{\widehat \alpha^{(-k)}_{1,g}, \widehat \beta^{(-k)}_{1,g}, \widehat \beta^{(-k)}_{0,g}} - \varphi_{1,i}\p{\alpha^{*}_{1,g}, \beta^{*}_{1,g}, \beta^{*}_{0,g}}}\p{\varphi_{1,i}\p{\alpha^{*}_{1,g}, \beta^{*}_{1,g}, \beta^{*}_{0,g}}}\\
    &\leq n_k^{-1}\sum_{i \in \mathcal{I}_k}\p{\varphi_{1,i}\p{\widehat \alpha^{(-k)}_{1,g}, \widehat \beta^{(-k)}_{1,g}, \widehat \beta^{(-k)}_{0,g}} - \varphi_{1,i}\p{\alpha^{*}_{1,g}, \beta^{*}_{1,g}, \beta^{*}_{0,g}}}^2\\
    &\qquad+ 2 \p{n_k^{-1}\sum_{i \in \mathcal{I}_k}\p{\varphi_{1,i}\p{\widehat \alpha^{(-k)}_{1,g}, \widehat \beta^{(-k)}_{1,g}, \widehat \beta^{(-k)}_{0,g}} - \varphi_{1,i}\p{\alpha^{*}_{1,g}, \beta^{*}_{1,g}, \beta^{*}_{0,g}}}^2}^{1/2}\p{\ell_{k2b}}^{1/2}.
\end{align*}
Let $w_{1,g,i} = Y_i\p{1} - X_i^\top \beta^*_{1,g}$. Then
\begin{align*}
    &\p{\varphi_{1,i}\p{\widehat \alpha^{(-k)}_{1,g}, \widehat \beta^{(-k)}_{1,g}, \widehat \beta^{(-k)}_{0,g}} - \varphi_{1,i}\p{\alpha^{*}_{1,g}, \beta^{*}_{1,g}, \beta^{*}_{0,g}}}^2\\
    &\quad= \Gamma_iA_i\left\{X_i^\top \widehat \Delta_{\beta_{1,g}}^{(-k)} - X_i^\top \widehat \Delta_{\beta_{0,g}}^{(-k)} - \frac{1}{g\p{X_i^\top \widehat \alpha^{(-k)}_{1,g}}}X_i^\top\widehat \Delta_{\beta_{1,g}}^{(-k)} + \p{\frac{1}{g\p{X_i^\top \widehat \alpha^{(-k)}_{1,g}}} - \frac{1}{g\p{X_i^\top \alpha^{*}_{1,g}}}}w_{1,g,i}\right\}^2\\
    &\quad\leq 4\Gamma_iA_i\p{1- \frac{1}{g\p{X_i^\top \alpha^*_{1,g}}}}^2\p{X_i^\top \widehat \Delta_{\beta_{1,g}}^{(-k)}}^2 + 4\Gamma_iA_i\p{\frac{1}{g\p{X_i^\top \widehat \alpha^{(-k)}_{1,g}}} - \frac{1}{g\p{X_i^\top \alpha^{*}_{1,g}}}}^2\p{X_i^\top \widehat \Delta_{\beta_{1,g}}^{(-k)}}^2 \\
    &\quad\quad+4\Gamma_iA_i\p{X_i^\top \widehat \Delta_{\beta_{0,g}}^{(-k)}}^2 + 4\Gamma_iA_i\p{\frac{1}{g\p{X_i^\top \widehat \alpha^{(-k)}_{1,g}}} - \frac{1}{g\p{X_i^\top \alpha^{*}_{1,g}}}}^2w_{1,g,i}^2.
\end{align*}
Recall $R_2$ and $R_3$ of Lemma~\ref{lemma psi diff square}. Similarly, we have
\begin{align*}
    n_k^{-1}\sum_{i \in \mathcal{I}_k}\Gamma_iA_i\p{\frac{1}{g\p{X_i^\top \widehat \alpha^{(-k)}_{1,g}}} - \frac{1}{g\p{X_i^\top \alpha^{*}_{1,g}}}}^2\p{X_i^\top \widehat \Delta_{\beta_{1,g}}^{(-k)}}^2 = O_p\p{\gamma_n^{-1}\norm{\widehat \Delta_{\alpha_{1,g}}^{(-k)}}_2^2\norm{\widehat \Delta_{\beta_{1,g}}^{(-k)}}_2^2},
\end{align*}
and
\begin{align*}
    n_k^{-1}\sum_{i \in \mathcal{I}_k}\Gamma_iA_i\p{\frac{1}{g\p{X_i^\top \widehat \alpha^{(-k)}_{1,g}}} - \frac{1}{g\p{X_i^\top \alpha^{*}_{1,g}}}}^2w_{1,g,i}^2 = O_p\p{\gamma_n^{-1}\norm{\widehat \Delta_{\alpha_{1,g}}^{(-k)}}_2^2}.
\end{align*}
By Lemma~\ref{sub-gaussian properties}(c), for some $C>0$,
\begin{align*}
    &n_k^{-1}\sum_{i \in \mathcal{I}_k}\E\sbr{\Gamma_iA_i\p{1- \frac{1}{g\p{X_i^\top \alpha^*_{1,g}}}}^2\p{X_i^\top \widehat \Delta_{\beta_{1,g}}^{(-k)}}^2\mid \Gamma_{1:n}, \widehat \beta^{(-k)}_{1,g}}\\
    &\leq k_0^{-2}\gamma_n^{-2}n_k^{-1}\sum_{i \in \mathcal{I}_k} \E\sbr{\Gamma_i\p{X_i^\top \widehat \Delta_{\beta_{1,g}}^{(-k)}}^2\mid \Gamma_i, \widehat \beta^{(-k)}_{1,g}}\leq C\gamma_n^{-2}n_k^{-1}\sum_{i \in \mathcal{I}_k}\Gamma_i\norm{\widehat \Delta_{\beta_{1,g}}^{(-k)}}_2^2.
\end{align*}
Since $n_k^{-1}\sum_{i \in \mathcal{I}_k}\Gamma_i = O_p(\gamma_n)$, then
\begin{align*}
    n_k^{-1}\sum_{i \in \mathcal{I}_k} \Gamma_iA_i\p{1- \frac{1}{g\p{X_i^\top \alpha^*_{1,g}}}}^2\p{X_i^\top \widehat \Delta_{\beta_{1,g}}^{(-k)}}^2 = O_p\p{\gamma_n^{-1}\norm{\widehat \Delta_{\beta_{1,g}}^{(-k)}}_2^2}.
\end{align*}
Similarly,
$
    n_k^{-1}\sum_{i \in \mathcal{I}_k}\Gamma_iA_i\p{X_i^\top \widehat \Delta_{\beta_{0,g}}^{(-k)}}^2 = O_p\p{\gamma_n\norm{\widehat \Delta_{\beta_{0,g}}^{(-k)}}_2^2}.
$
Thus,
\begin{align*}
    &n_k^{-1}\sum_{i \in \mathcal{I}_k}\p{\varphi_{a,i}\p{\widehat \alpha^{(-k)}_{1,g}, \widehat \beta^{(-k)}_{1,g}, \widehat \beta^{(-k)}_{0,g}} - \varphi_{a,i}\p{\alpha^{*}_{1,g}, \beta^{*}_{1,g}, \beta^{*}_{0,g}}}^2\\
    &\qquad=O_p\p{\gamma_n^{-1}\p{\norm{\widehat \Delta_{\alpha_{1,g}}^{(-k)}}_2^2 + \norm{\widehat \Delta_{\beta_{1,g}}^{(-k)}}_2^2 + \gamma_n^2\norm{\widehat \Delta_{\beta_{0,g}}^{(-k)}}_2^2}}.
\end{align*}
In addition, by Lemma~\ref{lemma causal generalizability} and Lemma~\ref{sub-gaussian properties}(c),
\begin{align*}
    &\E\sbr{\Gamma_iA_i\p{X_i^\top \beta^{*}_{1,g} - X_i^\top \beta^{*}_{0,g} + \frac{Y_i - X_i^\top \beta^{*}_{1,g} }{g\p{X_i^\top \alpha^{*}_{1,g}}}}^2} \\
    &\qquad\leq 3\E\sbr{\p{X_i^\top \beta^{*}_{1,g}}^2} + 3\E\sbr{\p{X_i^\top \beta^{*}_{0,g}}^2} + 3\E\sbr{\Gamma_i\frac{w_{1,g,i}^2 }{\p{g\p{X_i^\top \alpha^{*}_{1,g}}}^2}}\\
    &\qquad\leq 3C\sigma^2 + 3\gamma_n\E\sbr{\frac{w_{1,g,i}^2 }{\p{g\p{X_i^\top \alpha^{*}_{1,g}}}^2} \mid \Gamma_i=1}\\
    &\qquad\leq 3C\sigma^2 + 3k_0^{-2}\gamma_n^{-1}\E\sbr{w_{1,g,i}^2 \mid \Gamma_i=1}\leq 3C\sigma^2 + 3\sigma_w^2k_0^{-2}\gamma_n^{-1},
\end{align*}
which implies
\begin{align*}
    \ell_{k2b} &= n_k^{-1}\sum_{i \in \mathcal{I}_k}\Gamma_iA_i\p{X_i^\top \beta^{*}_{1,g} - X_i^\top \beta^{*}_{0,g} + \frac{Y_i - X_i^\top \beta^{*}_{1,g} }{g\p{X_i^\top \alpha^{*}_{1,g}}}}^2=O_p\p{\gamma_n^{-1}}.
\end{align*}
Therefore,
\begin{align*}
    \ell_{k2a} - \ell_{k2b} &= O_p\p{\gamma_n^{-1}\p{\norm{\widehat \Delta_{\alpha_{1,g}}^{(-k)}}_2^2 + \norm{\widehat \Delta_{\beta_{1,g}}^{(-k)}}_2^2+ \gamma_n^2\norm{\widehat \Delta_{\beta_{0,g}}^{(-k)}}_2^2}}\\
    &\qquad+O_p\p{\gamma_n^{-1/2}(\norm{\widehat \Delta_{\alpha_{1,g}}^{(-k)}}_2^2 + \norm{\widehat \Delta_{\beta_{1,g}}^{(-k)}}_2^2+ \gamma_n^2\norm{\widehat \Delta_{\beta_{0,g}}^{(-k)}}_2^2)^{1/2}}O_p\p{\gamma_n^{-1/2}}\\
    &=O_p\p{\gamma_n^{-1}(\norm{\widehat \Delta_{\alpha_{1,g}}^{(-k)}}_2^2 + \norm{\widehat \Delta_{\beta_{1,g}}^{(-k)}}_2^2+ \gamma_n^2\norm{\widehat \Delta_{\beta_{0,g}}^{(-k)}}_2^2)^{1/2}}.
\end{align*}

Similarly, we have
\begin{align*}
    \ell_{k3a} - \ell_{k3b} = O_p\p{\gamma_n^{-1}(\norm{\widehat \Delta_{\alpha_{0,g}}^{(-k)}}_2^2 + \norm{\widehat \Delta_{\beta_{0,g}}^{(-k)}}_2^2 + \gamma_n^2\norm{\widehat \Delta_{\beta_{1,g}}^{(-k)}}_2^2)^{1/2}}.
\end{align*}

For $\ell_{k4}$, by Lemma~\ref{lemma causal generalizability tau tilde normal} and Lemma~\ref{lemma causal generalizability hat tau - tilde tau consistency},
\begin{align*}
    \ell_{k4} &= \widehat \tau_{g}^2 - \tau_{g}^2=\p{\widehat \tau_{g}- \tau_{g}}^2 + 2\p{\widehat \tau_{g} - \tau_{g}}\tau_{g}\\
    &=O_p\p{\sum_{a=0,1} \br{\frac{\p{s_{\alpha_{a,g}} + (s_{\alpha_{a,g}} s_{\beta_{a,g}})^{1/2}} \log d}{n\gamma_n}} + \p{n\gamma_n}^{-1/2}}.
\end{align*}
By Lemma~\ref{lemma causal generalizability nuisance consistency}, we have
\begin{align*}
    &\ell_{k1a} - \ell_{k1b} + \ell_{k2a} - \ell_{k2b} +\ell_{k3a}- \ell_{k3b} -\ell_{k4}\\
    &=O_p\p{\sum_{a=0,1}\br{\norm{\Delta_{a,g}^{(-k)}}_1(\frac{\log d}{n})^{1/2}+ \norm{\Delta_{a,g}^{(-k)}}_2+ \norm{\widehat \Delta_{\beta_{a,g}}^{(-k)}}_2}}\\
    &\qquad +O_p\p{\gamma_n^{-1}\sum_{a=0,1}\br{\norm{\widehat \Delta_{\alpha_{a,g}}^{(-k)}}_2 + \norm{\widehat \Delta_{\beta_{a,g}}^{(-k)}}_2}+\sum_{a=0,1} \br{\frac{\p{s_{\alpha_{a,g}} + (s_{\alpha_{a,g}} s_{\beta_{a,g}})^{1/2}} \log d}{n\gamma_n}} + \p{n\gamma_n}^{-1/2}}\\
    &=O_p\p{\gamma_n^{-1}\sum_{a=0,1}(\frac{ \p{s_{\alpha_{a,g}}+s_{\beta_{a,g}}} \log d}{n\gamma_n})^{1/2}},
\end{align*}
which implies
\begin{align}
     \widehat \Sigma_g^2 - \widetilde \Sigma_g^2 = O_p\p{\gamma_n^{-1}\sum_{a=0,1}(\frac{ \p{s_{\alpha_{a,g}}+s_{\beta_{a,g}}} \log d}{n\gamma_n})^{1/2}}. \label{hat sigma g - tilde sigma g rate}
\end{align}

On the other hand, note that
\begin{align*}
    \Sigma_g^2 &= \E\sbr{\br{X_i^\top \beta^*_{1,g}- X_i^\top \beta^*_{0,g} +  \frac{\Gamma_iA_i}{g\p{X_i^\top \alpha^*_{1,g}}}\p{Y_i - X_i^\top \beta^*_{1,g}}  -  \frac{\Gamma_{i}\p{1-A_i}}{g\p{X_i^\top \alpha^*_{0,g}}}\p{Y_i - X_i^\top \beta^*_{0,g}}}^2} - \tau_g^2\\
    &=\E\sbr{\p{1-\Gamma_i}\p{X_i^\top \beta^*_{1,g} - X_i^\top \beta^*_{0,g}}^2} + \E\sbr{\Gamma_iA_i\p{X_i^\top \beta^*_{1,g} - X_i^\top \beta^*_{0,g} + \frac{Y_i - X_i^\top \beta^*_{1,g}}{g\p{X_i^\top \alpha^*_{1,g}}}}^2}\\
    &\qquad + \E\sbr{\Gamma_i\p{1 - A_i}\p{X_i^\top \beta^*_{1,g} - X_i^\top \beta^*_{0,g} - \frac{Y_i - X_i^\top \beta^*_{0,g}}{g\p{X_i^\top \alpha^*_{0,g}}}}^2} -\tau_g^2.
\end{align*}
We have $\E\sbr{\widetilde \Sigma_g^2 - \Sigma_g^2} = 0$ and
\begin{align*}
    \E\sbr{\p{\widetilde \Sigma_g^2 - \Sigma_g^2}^2} &\leq n^{-1} \E\sbr{\p{1-\Gamma_i}\p{X_i^\top \beta^*_{1,g} - X_i^\top \beta^*_{0,g}}^4}\\
    &\qquad+ n^{-1} \E\sbr{\Gamma_iA_i\p{X_i^\top \beta^*_{1,g} - X_i^\top \beta^*_{0,g} + \frac{Y_i - X_i^\top \beta^*_{1,g}}{g\p{X_i^\top \alpha^*_{1,g}}}}^4}\\
        &\qquad + n^{-1} \E\sbr{\Gamma_i\p{1 - A_i}\p{X_i^\top \beta^*_{1,g} - X_i^\top \beta^*_{0,g} - \frac{Y_i - X_i^\top \beta^*_{0,g}}{g\p{X_i^\top \alpha^*_{0,g}}}}^4}\\
        &\leq n^{-1} \E\sbr{\p{X_i^\top \beta^*_{1,g} - X_i^\top \beta^*_{0,g}}^4}\\
        &\qquad+ n^{-1} \E\sbr{\Gamma_i\p{X_i^\top \beta^*_{1,g} - X_i^\top \beta^*_{0,g} + \frac{Y_i\p{1} - X_i^\top \beta^*_{1,g}}{g\p{X_i^\top \alpha^*_{1,g}}}}^4}\\
        &\qquad + n^{-1} \E\sbr{\Gamma_i\p{X_i^\top \beta^*_{1,g} - X_i^\top \beta^*_{0,g} - \frac{Y_i\p{0} - X_i^\top \beta^*_{0,g}}{g\p{X_i^\top \alpha^*_{0,g}}}}^4}.
\end{align*}
By Minkowski's inequality,
$
    \norm{X_i^\top \beta^*_{1,g} - X_i^\top \beta^*_{0,g}}_{\P,4} \leq \norm{X_i^\top \beta^*_{1,g}}_{\P,4} + \norm{X_i^\top \beta^*_{0,g}}_{\P,4}\leq C\sigma.
$
By Lemma~\ref{lemma causal generalizability},
\begin{align*}
    &\norm{\Gamma_i\p{X_i^\top \beta^*_{1,g} - X_i^\top \beta^*_{0,g} + \frac{Y_i\p{1} - X_i^\top \beta^*_{1,g}}{g\p{X_i^\top \alpha^*_{1,g}}}}}_{\P,4} \\
    &\qquad\leq \norm{X_i^\top \beta^*_{1,g} - X_i^\top \beta^*_{0,g}}_{\P,4} + \norm{\Gamma_i\frac{w_{1,g,i}}{g\p{X_i^\top \alpha^*_{1,g}}}}_{\P,4}\leq C\sigma + k_0^{-1}\gamma_n^{-1}\norm{\Gamma_iw_{1,g,i}}_{\P,4}.
\end{align*}
In addition,
$
    \E\sbr{\Gamma_iw_{1,g,i}^4} = \gamma_n\E\sbr{w_{1,g,i}^4 \mid \Gamma_i=1}\leq \gamma_n\sigma_w^4.
$
It follows that 
\begin{align*}
    \E\sbr{\Gamma_i\p{X_i^\top \beta^*_{1,g} - X_i^\top \beta^*_{0,g} + \frac{Y_i\p{1} - X_i^\top \beta^*_{1,g}}{g\p{X_i^\top \alpha^*_{1,g}}}}^4} &\leq \p{C\sigma + k_0^{-1}\sigma_w\gamma_n^{-3/4}}^4\leq \p{C\sigma + k_0^{-1}\sigma_w}^4\gamma_n^{-3}.
\end{align*}
Similarly,
\begin{align*}
    \E\sbr{\Gamma_i\p{X_i^\top \beta^*_{1,g} - X_i^\top \beta^*_{0,g} - \frac{Y_i\p{0} - X_i^\top \beta^*_{0,g}}{g\p{X_i^\top \alpha^*_{0,g}}}}^4} &\leq \p{C\sigma + k_0^{-1}\sigma_w}^4\gamma_n^{-3}.
\end{align*}
Thus,
$
    \widetilde \Sigma_g^2 - \Sigma_g^2 = O_p\p{\gamma_n^{-1}\p{n\gamma_n}^{-1/2}}.
$
Together with \eqref{hat sigma g - tilde sigma g rate},
\begin{align*}
    \widehat \Sigma_g^2-\Sigma_g^2 &= \widehat \Sigma_g^2 - \widetilde \Sigma_g^2 + \widetilde \Sigma_g^2-\Sigma_g^2= O_p\p{\gamma_n^{-1}\sum_{a=0,1}(\frac{ \p{s_{\alpha_{a,g}}+s_{\beta_{a,g}}} \log d}{n\gamma_n})^{1/2}} + O_p\p{\gamma_n^{-1}\p{n\gamma_n}^{-1/2}}\\
   &=O_p\p{\gamma_n^{-1}\sum_{a=0,1}(\frac{ \p{s_{\alpha_{a,g}}+s_{\beta_{a,g}}} \log d}{n\gamma_n})^{1/2}}.
\end{align*}
By Lemma~\ref{lemma causal generalizability tau tilde normal}, $\Sigma_g^2 \asymp \gamma_n^{-1}$, it holds that
$
    \widehat \Sigma_g^2 = \Sigma_g^2\br{1 + O_p\p{\sum_{a=0,1}({ \p{s_{\alpha_{a,g}}+s_{\beta_{a,g}}} \log d}/{n\gamma_n})^{1/2}}}.
$
\end{proof}

\subsection{Proof of Theorem~\ref{theorem causal generalizability Asymptotics body}}

\begin{proof}
Theorem~\ref{theorem causal generalizability Asymptotics body} follows from Lemmas~\ref{lemma causal generalizability hat tau - tilde tau consistency}, \ref{lemma causal generalizability tau tilde normal}, and \ref{lemma causal generalizability hat sigma = sigma(1+o(1))} by repeating the proof of Lemma~\ref{lemma key results for generalizability}.
\end{proof}

\subsection{Auxiliary lemmas for transportability}

\begin{lemma}\label{lemma causal transportability}
   Let Assumption~\ref{assumption causal} hold, and for each $a\in\{0,1\}$, Assumption~\ref{assumption mar nuisance (a)} hold with $(\alpha^*_{PS}, \beta^*_{OR}, w_{OR})$ replaced by $(\alpha^*_{a,t}, \beta^*_{a,t}, w_{a,t})$. Then 
    
    (a) For each $a \in \{0,1\}$, there exist some constants $\lambda_l,\sigma >0$, $X_i$ is sub-Gaussian given $\Gamma_{i}=\Gamma$ such that $\norm{\Omega_{a,i}X_i^\top \beta^*_{a,t}}_{\psi_2} \leq \sigma$ and
    \begin{align*}
        \norm{\Omega_{a,i}X_i}_{\Gamma, \psi_2} := \sup_{v \in \R^d, \norm{v}_2 = 1}\inf \br{t > 0: \E\sbr{\exp\br{\frac{\Omega_{a,i}\p{X_i^\top v}^2}{t^2}} \mid \Gamma_{i}=\Gamma}\leq 2}\leq \sigma.
    \end{align*}
    In addition, $\lambda_{\min}\p{\E\sbr{\Omega_{a,i}X_iX_i^\top \mid \Gamma_i=1}} \geq \lambda_l$.

    (b) For each $a \in \{0,1\}$, let $\tilde w_{a,t,i} = \Omega_{a,i}\p{Y_i\p{a} - X_i^\top \beta^*_{a,t}}$. There exist some constants $\sigma_w, \delta_w>0$, $\tilde w_{a,t,i}$ is sub-Gaussian with $\norm{\tilde w_{a,t,i}}_{\psi_2} \leq \sigma_w$ and 
    \begin{align*}
        \E\sbr{\tilde w_{a,t,i}^8 \mid \Gamma_i=1} \leq \sigma_w^8 \quad \text{and} \quad \E\sbr{\tilde w_{a,t,i}^2 \mid \Gamma_i=1} \geq \delta_w.
    \end{align*}
\end{lemma}

\begin{proof}
    (a) Note that $\abs{\Omega_{a,i}X_i} \leq \abs{X_i}$ and $\abs{\Omega_{a,i}X_i^\top \beta^*_{a,t}} \leq \abs{X_i^\top \beta^*_{a,t}}$, then $\norm{\Omega_{a,i}X_i}_{\Gamma, \psi_2} \leq \norm{X_i}_{\Gamma, \psi_2}$ and $\norm{\Omega_{a,i}X_i^\top \beta^*_{a,t}}_{\psi} \leq \sigma$. In addition, we have
    \begin{align*}
        \E\sbr{\Omega_{a,i}X_iX_i^\top \mid \Gamma_i=1} = \E\sbr{\E\sbr{\Omega_{a,i} \mid \Gamma_i=1, X_i}X_iX_i^\top}  \succcurlyeq \eta_0\E\sbr{X_iX_i^\top}.
    \end{align*}

    (b) Under Assumption~\ref{assumption causal}, for some constant $C_1,C_2>0$,
    \begin{align*}
        \E\sbr{\tilde w_{a,t,i}^8 \mid \Gamma_i=1} &= \E\sbr{A_i\p{Y_i\p{a} - X_i^\top \beta^*_{a,t}}^8\mid \Gamma_i=1}\leq \E\sbr{w_{1,t,i}^8\mid \Gamma_i=1}\leq \sigma_w^8.
    \end{align*}
    In addition, by Assumption~\ref{assumption causal},
    \begin{align*}
        \E\sbr{\tilde w_{a,t,i}^8 \mid \Gamma_i=1} &= \E\sbr{A_iw_{1,t,i}^8\mid \Gamma_i=1}=\E\sbr{\E\sbr{A_i \mid \Gamma_i=1, X_i}\E\sbr{w_{1,t,i}^8\mid X_i,\Gamma_i=1}\mid \Gamma_i=1}\\
        &\geq \eta_0 \E\sbr{w_{1,t,i}^8\mid \Gamma_i=1} \geq \eta_0\delta_w.
    \end{align*}
\end{proof}

\begin{lemma}\label{lemma causal transportability tau tilde normal}
   Let Assumption~\ref{assumption causal} hold, and for each $a\in\{0,1\}$, Assumption~\ref{assumption mar nuisance (a)} hold with $(\alpha^*_{PS}, \beta^*_{OR}, w_{OR})$ replaced by $(\alpha^*_{a,t}, \beta^*_{a,t}, w_{a,t})$. Let either $\P\p{\Gamma_{a,i}=1 \mid X_i, \Omega_{a,i}=1} = g\p{X_i^\top \alpha^*_{a,t}}$ or $\E\sbr{Y_i(a) \mid X_i} = X_i^\top \beta^*_{a,t}$ holds for each $a \in \{0,1\}$. If $n\gamma_n \gg 1$, then as $n, d \rightarrow \infty$, $\Sigma_{t}^2 \asymp \gamma_n^{-1}$ and
$
\Sigma_{t}^{-1}n^{1/2}\p{\widetilde \tau_{t} - \tau_{t}} \rightarrow \mathcal{N}\p{0,1},
$
where $\widetilde \tau_{t} = \widetilde \tau_{1,t} - \widetilde \tau_{0,t}$ and
\begin{align*}
    \widetilde \tau_{a,t} = n^{-1}\sum_{i=1}^n \br{\frac{1-\Gamma_i}{1-\gamma_n} \p{X_i^\top \beta^*_{a,t} - \tau_{a,t}} +  \frac{\Gamma_i\mathbbm{1}_{\br{A_i=a}}}{1-\gamma_n}\frac{1-g\p{X_i^\top \alpha^*_{a,t}}}{g\p{X_i^\top \alpha^*_{a,t}}}\p{Y_i - X_i^\top \beta^*_{a,t}}} + \tau_{a,t}.
\end{align*}
\end{lemma}
\begin{proof}
Let $Q_{t,i} = Q_{1,t,i} - Q_{0,t,i} + \tau_t$ where
\begin{align*}
    Q_{a,t,i} &= \frac{1-\Gamma_i}{1-\gamma_n} \p{X_i^\top \beta^*_{a,t} - \tau_{a,t}} +  \frac{\Gamma_i\mathbbm{1}_{\br{A_i=a}}}{1-\gamma_n}\frac{1-g\p{X_i^\top \alpha^*_{a,t}}}{g\p{X_i^\top \alpha^*_{a,t}}}\p{Y_i - X_i^\top \beta^*_{a,t}}.
\end{align*}
    Under Assumption~\ref{assumption causal},
$
        \mathbbm{1}_{\br{A_i=a}}Y_i = \mathbbm{1}_{\br{A_i=a}}Y_i\p{a}.
$
    Note that $\Omega_{a,i} = 1-\Gamma_i + \Gamma_i\mathbbm{1}_{\br{A_i=a}} = \mathbbm{1}_{\br{\Gamma_i=0} \cup \br{\Gamma_i=1, A_i=a}}$. Then
    \begin{align*}
        \P\p{\Gamma_i=1,A_i=a \mid X_i, \Omega_{a,i}=1}&=\frac{\P\p{\Gamma_i=1,A_i=a, \Omega_{a,i}=1} \mid X_i}{\P\p{\Omega_{a,i}=1 \mid X_i}}=\frac{\P\p{\Gamma_i=1,A_i=a} \mid X_i}{\P\p{\Omega_{a,i}=1 \mid X_i}}.
    \end{align*}
    It follows that when either $\P\sbr{\Gamma_{a,i}=1 \mid X_i, \Omega_{a,i}=1} = g\p{X_i^\top \alpha^*_{a,t}}$ or $\E\sbr{Y_i(a) \mid X_i} = X_i^\top \beta^*_{a,t}$ holds for each $a \in \{0,1\}$,
    \begin{align*}
        \E\sbr{\widetilde \tau_{a,t}} &= \E\sbr{\frac{1-\Gamma_i}{1-\gamma_n} \p{X_i^\top \beta^*_{a,t} - \tau_{a,t}} +  \frac{\Gamma_i\mathbbm{1}_{\br{A_i=a}}}{1-\gamma_n}\frac{1-g\p{X_i^\top \alpha^*_{a,t}}}{g\p{X_i^\top \alpha^*_{a,t}}}\p{Y_i(a) - X_i^\top \beta^*_{a,t}}}\\
        &=\E\sbr{\p{\frac{1-\Gamma_i}{1-\gamma_n} -  \frac{\Gamma_i\mathbbm{1}_{\br{A_i=a}}}{1-\gamma_n}\frac{1-g\p{X_i^\top \alpha^*_{a,t}}}{g\p{X_i^\top \alpha^*_{a,t}}}}\p{Y_i(a) - X_i^\top \beta^*_{a,t}}}  + \E\sbr{\frac{1-\Gamma_i}{1-\gamma_n}\p{Y_i(a) - \tau_{a,t}}}\\
&=\E\sbr{\frac{\E\sbr{\Omega_{a,i}g\p{X_i^\top \alpha^*_{a,t}} - \Gamma_i\mathbbm{1}_{\br{A_i=a}} \mid X_i}}{\p{1-\gamma_n}g\p{X_i^\top \alpha^*_{a,t}}}\E\sbr{Y_i(a) - X_i^\top \beta^*_{a,t} \mid X_i}}\\
        &=\E\sbr{\frac{g\p{X_i^\top \alpha^*_{a,t}}\P\p{\Omega_{a,i}=1 \mid X_i} - \P\p{\Gamma_i=1, A_i=a \mid X_i}}{\p{1-\gamma_n}g\p{X_i^\top \alpha^*_{a,t}}}\E\sbr{Y_i(a) - X_i^\top \beta^*_{a,t} \mid X_i}}=0.
    \end{align*}
Thus, $\E\sbr{\widetilde \tau_t} =\E\sbr{Q_{t,i}} = \tau_t$.
    By Lyapunov's central limit theorem, it suffices to prove for some $\delta>0$ and $C>0$,
    \begin{align}
        \lim_{n\rightarrow \infty}n^{-\delta/2}\Sigma_t^{-\p{2+\delta}} \E\sbr{\abs{Q_{t,i} - \tau_t}^{2+\delta}} = 0.
    \end{align}
    
Note that under Assumption~\ref{assumption mar nuisance (a)} with $(\alpha^*_{PS}, \beta^*_{OR}, w_i)$ replaced by $(\alpha^*_{a,t}, \beta^*_{a,t}, w_{a,t,i})$,
\begin{align*}
    \Var\p{Q_{t,i}} &= \Sigma_t^2\\
    &= \E\sbr{\frac{1-\Gamma_i}{\p{1-\gamma_n}^2} \p{X_i^\top \beta^*_{1,t} - X_i^\top \beta^*_{0,t} - \tau_{t}}^2 }  \\
    &\qquad + \E\sbr{\frac{\Gamma_i A_i}{\p{1-\gamma_n}^2}\p{\frac{1-g\p{X_i^\top \alpha^*_{1,t}}}{g\p{X_i^\top \alpha^*_{1,t}}}}^2\p{Y_i(1) - X_i^\top \beta^*_{1,t}}^2}\\
    &\qquad + \E\sbr{\frac{\Gamma_i\p{1-A_i}}{\p{1-\gamma_n}^2}\p{\frac{1-g\p{X_i^\top \alpha^*_{0,t}}}{g\p{X_i^\top \alpha^*_{0,t}}}}^2\p{Y_i(0) - X_i^\top \beta^*_{0,t}}^2}.
\end{align*}
Since
\begin{align*}
    \E\sbr{\frac{1-\Gamma_i}{\p{1-\gamma_n}^2} \p{X_i^\top \beta^*_{1,t} - X_i^\top \beta^*_{0,t} - \tau_{t}}^2 } \leq c_0^{-2}\E\sbr{\p{X_i^\top \beta^*_{1,t} - X_i^\top \beta^*_{0,t} - \tau_{t}}^2 },
\end{align*}
by Minkowski's inequality,
\begin{align*}
    \norm{X_i^\top \beta^*_{1,t} - X_i^\top \beta^*_{0,t} - \tau_{t}}_{\P,2} &\leq \norm{X_i^\top \beta^*_{1,t}}_{\P,2} + \norm{X_i^\top \beta^*_{0,t}}_{\P,2} + \abs{\tau_{t}} \leq C_1\sigma + \abs{\tau_{t}}.
\end{align*}
Thus,
\begin{align*}
     \E\sbr{\frac{1-\Gamma_i}{\p{1-\gamma_n}^2} \p{X_i^\top \beta^*_{1,t} - X_i^\top \beta^*_{0,t} - \tau_{t}}^2 } \leq c_0^{-2}\p{C\sigma + \abs{\tau_{t}}}^2.
\end{align*}
In addition, 
\begin{align*}
    &\E\sbr{\frac{\Gamma_i A_i}{\p{1-\gamma_n}^2}\p{\frac{1-g\p{X_i^\top \alpha^*_{1,t}}}{g\p{X_i^\top \alpha^*_{1,t}}}}^2\p{Y_i(1) - X_i^\top \beta^*_{1,t}}^2}\\
    &=\E\sbr{\E\sbr{A_i \mid \Gamma_i, X_i}\E\sbr{\frac{\Gamma_i }{\p{1-\gamma_n}^2}\p{\frac{1-g\p{X_i^\top \alpha^*_{1,t}}}{g\p{X_i^\top \alpha^*_{1,t}}}}^2\p{Y_i(1) - X_i^\top \beta^*_{1,t}}^2 \mid \Gamma_i, X_i}}\\
    &\asymp \E\sbr{\Gamma_i\p{\frac{1-g\p{X_i^\top \alpha^*_{1,t}}}{g\p{X_i^\top \alpha^*_{1,t}}}}^2\p{Y_i(1) - X_i^\top \beta^*_{1,t}}^2}\\
    &=\gamma_n\E\sbr{\p{\frac{1-g\p{X_i^\top \alpha^*_{1,t}}}{g\p{X_i^\top \alpha^*_{1,t}}}}^2\p{Y_i(1) - X_i^\top \beta^*_{1,t}}^2 \mid \Gamma_i=1}\\
    &\asymp \gamma_n^{-1}\E\sbr{\p{Y_i(1) - X_i^\top \beta^*_{1,t}}^2 \mid \Gamma_i=1}\asymp \gamma_n^{-1}.
\end{align*}
Similarly, we have
\begin{align*}
    \E\sbr{\frac{\Gamma_i\p{1-A_i}}{\p{1-\gamma_n}^2}\p{\frac{1-g\p{X_i^\top \alpha^*_{0,t}}}{g\p{X_i^\top \alpha^*_{0,t}}}}^2\p{Y_i(0) - X_i^\top \beta^*_{0,t}}^2} \asymp \gamma_n^{-1}.
\end{align*}

Choose $\delta = 2$. Then
\begin{align*}
    \E\p{\abs{Q_{t,i} - \tau_t}^4} &= \E\sbr{\frac{1-\Gamma_i}{\p{1-\gamma_n}^4} \p{X_i^\top \beta^*_{1,t} - X_i^\top \beta^*_{0,t} - \tau_{t}}^4 }  \\
    &\qquad + \E\sbr{\frac{\Gamma_i A_i}{\p{1-\gamma_n}^4}\p{\frac{1-g\p{X_i^\top \alpha^*_{1,t}}}{g\p{X_i^\top \alpha^*_{1,t}}}}^4\p{Y_i(1) - X_i^\top \beta^*_{1,t}}^4}\\
    &\qquad + \E\sbr{\frac{\Gamma_i\p{1-A_i}}{\p{1-\gamma_n}^4}\p{\frac{1-g\p{X_i^\top \alpha^*_{0,t}}}{g\p{X_i^\top \alpha^*_{0,t}}}}^4\p{Y_i(0) - X_i^\top \beta^*_{0,t}}^4}\\
    &\leq c_0^{-4} \E\sbr{\p{X_i^\top \beta^*_{1,t} - X_i^\top \beta^*_{0,t} - \tau_{t}}^4}\\
    &\qquad + c_0^{-4}\E\sbr{\Gamma_i A_i\p{\frac{1-g\p{X_i^\top \alpha^*_{1,t}}}{g\p{X_i^\top \alpha^*_{1,t}}}}^4\p{Y_i(1) - X_i^\top \beta^*_{1,t}}^4}\\
    &\qquad +  c_0^{-4}\E\sbr{\Gamma_i\p{1-A_i}\p{\frac{1-g\p{X_i^\top \alpha^*_{0,t}}}{g\p{X_i^\top \alpha^*_{0,t}}}}^4\p{Y_i(0) - X_i^\top \beta^*_{0,t}}^4}.
\end{align*}
By Minkowski's inequality,
\begin{align*}
    \norm{X_i^\top \beta^*_{1,t} - X_i^\top \beta^*_{0,t} - \tau_{t}}_{\P, 4} &\leq \norm{X_i^\top \beta^*_{1,t}}_{\P, 4}+\norm{X_i^\top \beta^*_{0,t}}_{\P, 4}+\abs{\tau_{t}}\leq C_2\sigma + \abs{\tau_{t}},
\end{align*}
which implies
$
    \E\sbr{\p{X_i^\top \beta^*_{1,t} - X_i^\top \beta^*_{0,t} - \tau_{t}}^4} \leq \p{C_2\sigma + \abs{\tau_{t}}}^4.
$
Note that
\begin{align*}
    &\E\sbr{\Gamma_i A_i\p{\frac{1-g\p{X_i^\top \alpha^*_{1,t}}}{g\p{X_i^\top \alpha^*_{1,t}}}}^4\p{Y_i(1) - X_i^\top \beta^*_{1,t}}^4}\leq \E\sbr{\Gamma_i\p{\frac{1}{g\p{X_i^\top \alpha^*_{1,t}}}}^4\p{Y_i(1) - X_i^\top \beta^*_{1,t}}^4}\\
    &\qquad=\gamma_n\E\sbr{\p{\frac{1}{g\p{X_i^\top \alpha^*_{1,t}}}}^4\p{Y_i(1) - X_i^\top \beta^*_{1,t}}^4 \mid \Gamma_i=1}.
\end{align*}
Observe that
$
    \p{{1}/{g\p{X_i^\top \alpha^*_{1,t}}}}^4 \leq k_0^{-4}\gamma_n^{-4}$ and $\E\sbr{\p{Y_i(1) - X_i^\top \beta^*_{1,t}}^4 \mid \Gamma_i=1} \leq \sigma_w^4,
$
which implies
\begin{align*}
    \E\sbr{\Gamma_i A_i\p{\frac{1-g\p{X_i^\top \alpha^*_{1,t}}}{g\p{X_i^\top \alpha^*_{1,t}}}}^4\p{Y_i(1) - X_i^\top \beta^*_{1,t}}^4} \leq k_0^{-4}\sigma_w^4 \gamma_n^{-3}.
\end{align*}
Similarly,
\begin{align*}
    \E\sbr{\Gamma_i\p{1-A_i}\p{\frac{1-g\p{X_i^\top \alpha^*_{0,t}}}{g\p{X_i^\top \alpha^*_{0,t}}}}^4\p{Y_i(0) - X_i^\top \beta^*_{0,t}}^4} \leq k_0^{-4}\sigma_w^4 \gamma_n^{-3}.
\end{align*}
Thus,
\begin{align*}
    \E\p{\abs{Q_{t,i} - \tau_t}^4} &\leq c_0^{-4}\br{\p{C_2\sigma + \abs{\tau_{t}}}^4 + 2k_0^{-4}\sigma_w^4 \gamma_n^{-3}}\\
    &\leq c_0^{-4}\br{\p{C_2\sigma + \abs{\tau_{t}}}^4 + 2k_0^{-4}\sigma_w^4}\gamma_n^{-3} =: C_3\gamma_n^{-3}.
\end{align*}
It follows that for $\delta=2$, if $n\gamma_n \gg 1$, as $n, d \rightarrow 0$,
\begin{align*}
    \lim_{n\rightarrow \infty}n^{-1}\Sigma_t^{-4} \E\sbr{\abs{Q_{t,i} - \tau_t}^{4}} \leq \lim_{n\rightarrow \infty}n^{-1}C\gamma_n^{-1} = 0.
\end{align*}
\end{proof}

\begin{lemma}\label{lemma causal transportability nuisance consistency}
    Let Assumption~\ref{assumption causal} hold, and for each \(a \in \{0,1\}\), Assumption~\ref{assumption mar nuisance (a)} hold with $(\alpha^*_{PS}, \beta^*_{OR}, w_{OR})$ replaced by $(\alpha^*_{a,t}, \beta^*_{a,t}, w_{a,t})$. Then as $n, d \rightarrow \infty$, it holds that

(a) Choose $\lambda_\alpha \asymp (\log d/(n\gamma_n))^{1/2}$. If $n\gamma_n \gg \max\br{s_{\alpha_{a,t}}, \log n} \log d$, then
    $$\norm{\widehat \alpha_{a,t}^{(-k)} - \alpha^*_{a,t}}_1 = O_p\p{s_{\alpha_{a,t}}(\log d/(n\gamma_n))^{1/2}} \text{ and } \norm{\widehat \alpha_{a,t}^{(-k)} - \alpha^*_{a,t}}_2 = O_p\p{(\frac{s_{\alpha_{a,t}} \log d}{n\gamma_n})^{1/2}}.$$

(b) Choose $\lambda_\alpha \asymp (\log d/(n\gamma_n))^{1/2}$. If $n\gamma_n \gg \max\br{s_{\beta_{a,t}}, (\log n)^2} \log d$, then
    $$\norm{\widetilde \beta_{a,t}^{(-k)} - \beta^*_{a,t}}_1 = O_p\p{s_{\beta_{a,t}}(\log d/(n\gamma_n))^{1/2}} \text{ and } \norm{\widetilde \beta_{a,t}^{(-k)} - \beta^*_{a,t}}_2 = O_p\p{(\frac{s_{\beta_{a,t}} \log d}{n\gamma_n})^{1/2}}.$$

    (c) Choose $\lambda_\alpha \asymp \lambda_\beta \asymp (\log d/(n\gamma_n))^{1/2}$. If $n\gamma_n \gg \max\br{s_{\alpha_{a,t}}, s_{\beta_{a,t}}\log n , (\log n)^2} \log d$ and  $s_{\alpha_{a,t}} s_{\beta_{a,t}} \ll (n\gamma_n)^{3/2}/(\log n(\log d)^2)$, then 
    \begin{align*}
     \norm{\widehat \beta_{a,t}^{(-k)} - \widetilde \beta_{a,t}^{(-k)}}_1 = O_p\p{(\frac{s_{\alpha_{a,t}} s_{\beta_{a,t}} \log d}{n\gamma_n})^{1/2}} \text{ and } \norm{\widehat \beta_{a,t}^{(-k)} - \widetilde \beta_{a,t}^{(-k)}}_2 = O_p\p{(\frac{s_{\alpha_{a,t}} \log d}{n\gamma_n})^{1/2}}.
    \end{align*}
\end{lemma}
\begin{proof}
    Note that $(1-\Gamma_i) = (1-\Gamma_i)\Omega_{a,i}$ and $\Gamma_{a,i} = \Gamma_i\Omega_{a,i}$. Let $(\widetilde X_{a,i}, \widetilde Y_{a,i}) = (\Omega_{a,i} X_i, \Omega_{a,i}Y_i)$, then Lemma~\ref{lemma causal generalizability nuisance consistency} follows by repeating the proof of Lemma~\ref{lemma consistency for alpha and betatilde} and Lemma~\ref{lemma consistency of betahat and betatilde for product sparsity}.
\end{proof}

\begin{lemma}\label{lemma causal transportability hat tau - tilde tau consistency}
Let Assumption~\ref{assumption causal} hold, and suppose that, for each \(a \in \{0,1\}\), Assumption~\ref{assumption mar nuisance (a)} holds with $(\alpha^*_{PS}, \beta^*_{OR}, w_{OR})$ replaced by $(\alpha^*_{a,t}, \beta^*_{a,t}, w_{a,t})$. Choose $\lambda_\alpha \asymp \lambda_\beta \asymp (\log d/(n\gamma_n))^{1/2}$. If
\[
    n\gamma_n \gg \sum_{a=0,1}\max\br{s_{\alpha_{a,t}}, s_{\beta_{a,t}}\log n , (\log n)^2} \log d,
    \qquad
    \sum_{a=0,1}s_{\alpha_{a,t}} s_{\beta_{a,t}} \ll (n\gamma_n)^{3/2}/\{\log n(\log d)^2\},
\]
then as \(n, d\rightarrow \infty\),
    \begin{align*}
        \widehat \tau_t - \widetilde \tau_t &= O_p\p{\sum_{a=0,1}\frac{\p{s_{\alpha_{a,t}} + (s_{\alpha_{a,t}} s_{\beta_{a,t}})^{1/2}} \log d}{n\gamma_n}}.
    \end{align*}
In addition, assume that $\E\sbr{Y_i(a) \mid X_i} = X_i^\top \beta^*_{a,t}$ for all $a \in \{0,1\}$. Then as $n, d\rightarrow \infty$, 
\begin{align*}
\widehat \tau_t - \widetilde \tau_t&=O_p\p{\sum_{a=0,1}\br{\frac{(s_{\alpha_{a,t}} \log d)^{1/2}}{n\gamma_n} +\frac{(s_{\alpha_{a,t}} s_{\beta_{a,t}})^{1/2} \log d}{n\gamma_n}}}.
\end{align*}
\end{lemma}

\begin{proof}
    Note that $(1-\Gamma_i) = (1-\Gamma_i)\Omega_{a,i}$ and $\Gamma_{a,i} = \Gamma_i\Omega_{a,i}$. Let $(\widetilde X_{a,i}, \widetilde Y_{a,i}) = (\Omega_{a,i} X_i, \Omega_{a,i}Y_i)$, then Lemma~\ref{lemma causal transportability hat tau - tilde tau consistency} follows by repeating the proof of Lemma~\ref{lemma mar transportability theta hat - theta tilde consistency}.
\end{proof}

\begin{lemma}\label{lemma causal transportability hat sigma = sigma(1+o(1))}
Let Assumption~\ref{assumption causal} hold, and suppose that, for each \(a \in \{0,1\}\), Assumption~\ref{assumption mar nuisance (a)} holds with $(\alpha^*_{PS}, \beta^*_{OR}, w_{OR})$ replaced by $(\alpha^*_{a,t}, \beta^*_{a,t}, w_{a,t})\). Let either \(\P\p{\Gamma_{a,i}=1 \mid X_i, \Omega_{a,i}=1} = g\p{X_i^\top \alpha^*_{a,t}}\) or \(\E\sbr{Y_i(a) \mid X_i} = X_i^\top \beta^*_{a,t}\) hold for each \(a \in \br{0,1}\). Choose $\lambda_\alpha \asymp \lambda_\beta \asymp (\log d/(n\gamma_n))^{1/2}$. If
\[
    n\gamma_n \gg \sum_{a=0,1}\max\br{s_{\alpha_{a,t}}, s_{\beta_{a,t}}\log n , (\log n)^2} \log d,
    \qquad
    \sum_{a=0,1}s_{\alpha_{a,t}} s_{\beta_{a,t}} \ll (n\gamma_n)^{3/2}/\{\log n(\log d)^2\},
\]
then as \(n, d\rightarrow \infty\),
    \begin{align*}
        \widehat \Sigma_t^2 = \Sigma_t^2\br{1 + O_p\p{\sum_{a=0,1}\br{({\p{s_{\alpha_{a,t}} + s_{\beta_{a,t}}}\log d}/{n\gamma_n})^{1/2} +\gamma_n({s_{\alpha_{a,t}} s_{\beta_{a,t}} \log d}/{n\gamma_n})^{1/2}}}}.
    \end{align*}
\end{lemma}
\begin{proof}
    Let $\widehat \Delta_{\alpha_{a,t}}^{(-k)} = \widehat \alpha^{(-k)}_{a,t} - \alpha^*_{a,t}$, $\widehat \Delta_{\beta_{a,t}}^{(-k)} = \widehat \beta^{(-k)}_{a,t} - \beta^*_{a,t}$, $\Delta_{a,g}^{(-k)} = \widehat \beta^{(-k)}_{a,t} - \beta^*_{a,t}$, $\widetilde \Delta_{\beta_{a,t}}^{(-k)} = \widetilde \beta^{(-k)}_{a,t} - \beta^*_{a,t}$, $\widehat \gamma_k = n_k^{-1}\sum_{i\in \mathcal{I}_k} \Gamma_i$, $\widehat \gamma = n^{-1}\sum_{i=1}^n\Gamma_i$, 
    \begin{align*}
        \widetilde \Sigma_t^2 &= n^{-1}\sum_{i=1}^n {1-\Gamma_i}/{\p{1-\gamma_n}^2}\p{X_i^\top \beta^*_{1,t} - X_i^\top \beta^*_{0,t} - \tau_{t}}^2 \\
        &\qquad + n^{-1}\sum_{i=1}^n{\Gamma_i A_i}/{\p{1-\gamma_n}^2}\p{{1-g\p{X_i^\top \alpha^*_{1,t}}}/{g\p{X_i^\top \alpha^*_{1,t}}}}^2\p{Y_i - X_i^\top \beta^*_{1,t}}^2\\
        &\qquad + n^{-1}\sum_{i=1}^n{\Gamma_i\p{1-A_i}}/{\p{1-\gamma_n}^2}\p{{1-g\p{X_i^\top \alpha^*_{0,t}}}/{g\p{X_i^\top \alpha^*_{0,t}}}}^2\p{Y_i - X_i^\top \beta^*_{0,t}}^2.
    \end{align*}
Then
\begin{align*}
    \widehat \Sigma_t^2 - \widetilde \Sigma_t^2 = K^{-1}\sum_{k=1}^K \p{S_{k1a} - S_{k1b} + S_{k2a} - S_{k2b} + S_{k3a} - S_{k3b}+ S_{k4a} - S_{k4b}},
\end{align*}
where
\begin{align*}
    S_{k1a} &= \bigl\{\widehat \beta_{1,t}^{(-k),\top} \Xi_0\widehat \beta_{1,t}^{(-k)}
    +\widehat \beta_{0,t}^{(-k),\top} \Xi_0\widehat \beta_{0,t}^{(-k)}
    - 2\widehat \beta_{1,t}^{(-k),\top} \Xi_0\widehat \beta_{0,t}^{(-k)}\\
    &\quad -2 \bar X_0^\top\widehat \beta_{1,t}^{(-k)}\widehat \tau_t
    + 2 \bar X_0^\top\widehat \beta_{0,t}^{(-k)}\widehat \tau_t\bigr\}
    \p{n_k^{-1}\sum_{i\in \mathcal{I}_k} \p{1-\Gamma_i}}^{-1},\\
    S_{k1b} &= \bigl\{\beta_{1,t}^{*,\top} \Xi_0\beta_{1,t}^{*}
    +\beta_{0,t}^{*,\top} \Xi_0\beta_{0,t}^{*}
    - 2\beta_{1,t}^{*,\top} \Xi_0\beta_{0,t}^{*}\\
    &\quad - 2\bar X_0^\top\beta_{1,t}^{*} \tau_t
    + 2 \bar X_0^\top\beta_{0,t}^{*} \tau_t\bigr\}\p{1-\gamma_n}^{-2},\\
    S_{k2a} &= {\widehat \tau_t^2}/{n_k^{-1}\sum_{i\in \mathcal{I}_k} \p{1-\Gamma_i}},\quad
    S_{k2b} = {n_k^{-1}\sum_{i\in \mathcal{I}_k} \p{1-\Gamma_i} \tau_t^2}/{\p{1-\gamma_n}^2},\\
    S_{k3a} &= n_k^{-1}\sum_{i\in \mathcal{I}_k}
    \Gamma_iA_i\p{\frac{1-g\p{X_i^\top \widehat \alpha_{1,t}^{(-k)}}}{g\p{X_i^\top \widehat \alpha_{1,t}^{(-k)}}}}^2\\
    &\quad \times \p{Y_i - X_i^\top \widehat \beta_{1,t}^{(-k)}}^2
    \p{n_k^{-1}\sum_{i\in \mathcal{I}_k} \p{1-\Gamma_i}}^{-2},\\
    S_{k3b} &= n_k^{-1}\sum_{i\in \mathcal{I}_k}
    \frac{\Gamma_iA_i}{\p{1-\gamma_n}^2}
    \p{\frac{1-g\p{X_i^\top \alpha^*_{1,t}}}{g\p{X_i^\top \alpha^*_{1,t}}}}^2
    \p{Y_i - X_i^\top \beta^*_{1,t}}^2,\\
    S_{k4a} &= n_k^{-1}\sum_{i\in \mathcal{I}_k}
    \Gamma_i\p{1-A_i}\p{\frac{1-g\p{X_i^\top \widehat \alpha_{0,t}^{(-k)}}}{g\p{X_i^\top \widehat \alpha_{0,t}^{(-k)}}}}^2\\
    &\quad \times \p{Y_i - X_i^\top \widehat \beta_{0,t}^{(-k)}}^2
    \p{n_k^{-1}\sum_{i\in \mathcal{I}_k} \p{1-\Gamma_i}}^{-2},\\
    S_{k4b} &= n_k^{-1}\sum_{i\in \mathcal{I}_k}
    \frac{\Gamma_i\p{1-A_i}}{\p{1-\gamma_n}^2}
    \p{\frac{1-g\p{X_i^\top \alpha^*_{0,t}}}{g\p{X_i^\top \alpha^*_{0,t}}}}^2
    \p{Y_i - X_i^\top \beta^*_{0,t}}^2.
\end{align*}

For $S_{k1a}$ and $S_{k1b}$, 
\begin{align*}
    &\p{1-\widehat \gamma_k}^2S_{k1a} - \p{1-\gamma_n}^2S_{k1b} \\
    &\qquad=  \p{1-\widehat \gamma_k}\p{\widehat \beta_{1,t}^{(-k),\top} \Xi_0\widehat \beta_{1,t}^{(-k)} - 
 \beta_{1,t}^{*,\top} \Xi_0\beta_{1,t}^{*}} + \p{1-\widehat \gamma_k}\p{\widehat \beta_{0,t}^{(-k),\top} \Xi_0\widehat \beta_{0,t}^{(-k)} - \beta_{0,t}^{*,\top} \Xi_0\beta_{0,t}^{*}}\\
 &\qquad\qquad -2\p{1-\widehat \gamma_k}\p{\widehat \beta_{1,t}^{(-k),\top} \Xi_0\widehat \beta_{0,t}^{(-k)} - \beta_{1,t}^{*,\top} \Xi_0\beta_{0,t}^{*}} -2\p{1-\widehat \gamma_k}\p{\bar X_0^\top\widehat \beta_{1,t}^{(-k)}\widehat \tau_t - \bar X_0^\top\beta_{1,t}^{*} \tau_t}\\
 &\qquad\qquad + 2\p{1-\widehat \gamma_k}\p{\bar X_0^\top\widehat \beta_{0,t}^{(-k)}\widehat \tau_t - \bar X_0^\top\beta_{0,t}^{*} \tau_t}.
\end{align*}
Similar to $\ell_{k1a} - \ell_{k1b}$ of Lemma~\ref{lemma causal generalizability hat sigma = sigma(1+o(1))}, we have
\begin{align*}
    &\p{1-\widehat \gamma_k}\p{\widehat \beta_{1,t}^{(-k),\top} \Xi_0\widehat \beta_{1,t}^{(-k)} - 
 \beta_{1,t}^{*,\top} \Xi_0\beta_{1,t}^{*}} = O_p\p{\norm{\Delta_{1,t}^{(-k)}}_1({\log d}/{n})^{1/2}+ \norm{\Delta_{1,t}^{(-k)}}_2 + \norm{\widehat \Delta_{\beta_{1,t}}^{(-k)}}_2},\\
     &\p{1-\widehat \gamma_k}\p{\widehat \beta_{0,t}^{(-k),\top} \Xi_0\widehat \beta_{0,t}^{(-k)} - \beta_{0,t}^{*,\top} \Xi_0\beta_{0,t}^{*}}=O_p\p{\norm{\Delta_{0,t}^{(-k)}}_1({\log d}/{n})^{1/2}+ \norm{\Delta_{0,t}^{(-k)}}_2 + \norm{\widehat \Delta_{\beta_{0,t}}^{(-k)}}_2},\\
    &\p{1-\widehat \gamma_k}\p{\widehat \beta^{(-k), \top}_{1,t}\bar 
    \Xi_0 \widehat \beta^{(-k)}_{0,t} - \beta^{*, \top}_{1,t}\bar 
    \Xi_0 \beta^{*}_{0,t}} = O_p\p{\sum_{a=0}^1\br{\norm{\Delta_{a,t}^{(-k)}}_1({\log d}/{n})^{1/2}+ \norm{\Delta_{a,t}^{(-k)}}_2+ \norm{\widehat \Delta_{\beta_{a,t}}^{(-k)}}_2}}.
\end{align*}
Recall $p_{k3a}$ and $p_{k3b}$ of Lemma~\ref{lemma mar transportability variance consistency}. Similarly, by Lemma~\ref{lemma causal transportability tau tilde normal} and Lemma~\ref{lemma causal transportability hat tau - tilde tau consistency},
\begin{align}
    \widehat \tau_t - \tau_t = O_p\p{\sum_{a=0,1}{\p{s_{\alpha_{a,t}} + (s_{\alpha_{a,t}} s_{\beta_{a,t}})^{1/2}} \log d}/{n\gamma_n} + \p{n\gamma_n}^{-1/2}}, \label{hat tau t - tau}
\end{align}
then
\begin{align*}
    &\p{1-\widehat \gamma_k}\p{\bar X_0^\top\widehat \beta_{1,t}^{(-k)}\widehat \tau_t - \bar X_0^\top\beta_{1,t}^{*} \tau_t} \\
    &\qquad= O_p\p{\norm{\Delta^{(-k)}_{1,t}}_1+\norm{\widetilde \Delta_{\beta_{1,t}}^{(-k)}}_2 + \sum_{a=0,1}{\p{s_{\alpha_{a,t}} + (s_{\alpha_{a,t}} s_{\beta_{a,t}})^{1/2}} \log d}/{n\gamma_n} + \p{n\gamma_n}^{-1/2}}.
\end{align*}
and
\begin{align*}
    &\p{1-\widehat \gamma_k}\p{\bar X_0^\top\widehat \beta_{0,t}^{(-k)}\widehat \tau_t - \bar X_0^\top\beta_{0,t}^{*} \tau_t}\\
    &\qquad=O_p\p{\norm{\Delta^{(-k)}_{0,t}}_1+\norm{\widetilde \Delta_{\beta_{0,t}}^{(-k)}}_2 + \sum_{a=0,1}{\p{s_{\alpha_{a,t}} + (s_{\alpha_{a,t}} s_{\beta_{a,t}})^{1/2}} \log d}/{n\gamma_n} + \p{n\gamma_n}^{-1/2}}.
\end{align*}
Thus, 
\begin{align*}
    &\p{1-\widehat \gamma_k}^2S_{k1a} - \p{1-\gamma_n}^2S_{k1b} \\
    &\qquad= O_p\p{\sum_{a=0,1}\br{\norm{\Delta_{a,t}^{(-k)}}_1({\log d}/{n})^{1/2}+ \norm{\Delta_{a,t}^{(-k)}}_2+ \norm{\widehat \Delta_{\beta_{a,t}}^{(-k)}}_2}}\\
    &\qquad\qquad + O_p\p{\sum_{a=0,1}\br{\norm{\Delta^{(-k)}_{a,t}}_1+\norm{\widetilde \Delta_{\beta_{a,t}}^{(-k)}}_2} + \sum_{a=0,1}{\p{s_{\alpha_{a,t}} + (s_{\alpha_{a,t}} s_{\beta_{a,t}})^{1/2}} \log d}/{n\gamma_n} + \p{n\gamma_n}^{-1/2}}.
\end{align*}
Note that 
\begin{align*}
    \E\sbr{\abs{S_{k1b}}} &\leq c_0^{-2}\E\sbr{\p{1-\Gamma_i}\p{X_i^\top \beta_{1,t}}^2}+c_0^{-2}\E\sbr{\p{1-\Gamma_i}\p{X_i^\top \beta_{0,t}}^2}\\
    &\qquad + 2c_0^{-2}\E\sbr{\p{1-\Gamma_i}\abs{\p{X_i^\top \beta_{1,t}}\p{X_i^\top \beta_{0,t}}}} + 2c_0^{-2}\E\sbr{\p{1-\Gamma_i}\abs{X_i^\top \beta_{1,t}}} \abs{\tau_t} \\
    &\qquad + 2c_0^{-2}\E\sbr{\p{1-\Gamma_i}\abs{X_i^\top \beta_{0,t}}} \abs{\tau_t}\\
    &\leq c_0^{-2}\E\sbr{\p{X_i^\top \beta_{1,t}}^2}+c_0^{-2}\E\sbr{\p{X_i^\top \beta_{0,t}}^2}\\
    &\qquad + 2c_0^{-2}\E\sbr{\p{X_i^\top \beta_{1,t}}^2}^{1/2}\E\sbr{\p{X_i^\top \beta_{0,t}}^2}^{1/2} + 2c_0^{-2}\E\sbr{\abs{X_i^\top \beta_{1,t}}} \abs{\tau_t} \\
    &\qquad + 2c_0^{-2}\E\sbr{\abs{X_i^\top \beta_{0,t}}} \abs{\tau_t}\\
    &\leq C_1c_0^{-2}\sigma^2 + C_2c_0^{-2}\sigma \abs{\tau_t}.
\end{align*}
Thus, $S_{k1b} = O_p\p{1}$. It follows that by \eqref{1-gamma hat square ratio} and Lemma~\ref{lemma causal transportability nuisance consistency},
\begin{align*}
    S_{k1a} - S_{k1b} &= \p{\p{{1-\gamma_n}/{1-\widehat \gamma_k}}^2 - 1}S_{k1b} + {1}/{\p{1-\widehat \gamma_k}^2}O_p\p{\sum_{a=0,1}({\p{s_{\alpha_{a,t}} + s_{\beta_{a,t}}}\log d}/{n\gamma_n})^{1/2}}\\
&=O_p\p{\sum_{a=0,1}\p{({(s_{\alpha_{a,t}} + s_{\beta_{a,t}})\log d}/{n\gamma_n})^{1/2} + ({s_{\alpha_{a,t}} s_{\beta_{a,t}} \log d}/{n\gamma_n})^{1/2}}}.
\end{align*}

For $S_{k2a}$ and $S_{k2b}$, we have
\begin{align*}
    \p{1-\widehat \gamma_k}^2S_{k2a} - \p{1-\gamma_n}^2S_{k2b} &= \p{1-\widehat \gamma_k}\p{\widehat \tau_t^2 - \tau_t^2}=\p{1-\widehat \gamma_k}\br{\p{\widehat \tau_t - \tau_t}^2 + 2\p{\widehat \tau_t - \tau_t}\tau}.
\end{align*}
By Lemma~\ref{lemma concentrate gamma} and \eqref{hat tau t - tau}, $1-\widehat \gamma_k = O_p(1-\gamma_n)$ and 
\begin{align*}
    \p{1-\widehat \gamma_k}^2S_{k2a} - \p{1-\gamma_n}^2S_{k2b} = O_p\p{\sum_{a=0,1}{\p{s_{\alpha_{a,t}} + (s_{\alpha_{a,t}} s_{\beta_{a,t}})^{1/2}} \log d}/{n\gamma_n} + \p{n\gamma_n}^{-1/2}}.
\end{align*}
By Lemma~\ref{lemma concentrate gamma}, $S_{k2b}=O_p\p{1}$. Thus, by \eqref{1-gamma hat square ratio},
\begin{align*}
    S_{k2a} - S_{k2b} &= \p{\p{{1-\gamma_n}/{1-\widehat \gamma_k}}^2 - 1}S_{k2b}+ {1}/{\p{1-\widehat \gamma_k}^2}O_p\p{\sum_{a=0,1}\p{s_{\alpha_{a,t}} + s_{\beta_{a,t}}}{\log d}/{n\gamma_n} + \p{n\gamma_n}^{-1/2}}\\
    &=O_p\p{\sum_{a=0,1}{\p{s_{\alpha_{a,t}} + (s_{\alpha_{a,t}} s_{\beta_{a,t}})^{1/2}} \log d}/{n\gamma_n} + \p{n\gamma_n}^{-1/2}}.
\end{align*}

For $S_{k3a}$ and $S_{k3b}$, let 
\begin{align*}
\Phi_{a,i}\p{\alpha, \beta} = \Gamma_{a,i}{1-g\p{X_i^\top \alpha}}/{g\p{X_i^\top \alpha}}\p{Y_i - X_i^\top \beta}.
\end{align*}
Recall $\p{1-\widehat \gamma_k}^2p_{k2a} - \p{1-\gamma_n}^2p_{k2b}$ of Lemma~\ref{lemma mar transportability variance consistency}. We have $S_{k3b} = O_p\p{\gamma_n^{-1}}$ and
\begin{align*}
    \p{1-\widehat \gamma_k}^2S_{k3a} - \p{1-\gamma_n}^2S_{k3b} &= n_k^{-1}\sum_{i \in \mathcal{I}_k} \br{\Phi_{a,i}\p{\widehat \alpha_{1,t}^{(-k)}, \widehat \beta_{1,t}^{(-k)}}^2 - \Phi_{a,i}\p{\alpha^*_{1,t}, \beta^*_{1,t}}^2}\\
    &=O_p\p{\gamma_n^{-1}\br{\norm{\widehat \Delta_{\alpha_{1,t}}^{(-k)}}_2 + \norm{\widehat \Delta_{\beta_{1,t}}^{(-k)}}_2}}.
\end{align*}
Then by Lemma~\ref{lemma causal transportability nuisance consistency},
\begin{align*}
   S_{k3a} - S_{k3b} &= \p{\p{{1-\gamma_n}/{1-\widehat \gamma_k}}^2 - 1}S_{k3b}+ {1}/{\p{1-\widehat \gamma_k}^2}O_p\p{\gamma_n^{-1}\p{\sum_{a=0,1}({\p{s_{\alpha_{a,t}} + s_{\beta_{a,t}}}\log d}/{n\gamma_n})^{1/2}}}\\
&=O_p\p{\gamma_n^{-1}\p{\sum_{a=0,1}({\p{s_{\alpha_{a,t}} + s_{\beta_{a,t}}}\log d}/{n\gamma_n})^{1/2}}}.
\end{align*}

Similarly, for $S_{k4a}$ and $S_{k4b}$, we have
\begin{align*}
    S_{k4a} - S_{k4b} = O_p\p{\gamma_n^{-1}\p{\sum_{a=0,1}({\p{s_{\alpha_{a,t}} + s_{\beta_{a,t}}}\log d}/{n\gamma_n})^{1/2}}}.
\end{align*}
Then it follows
\begin{align*}
    \widehat \Sigma_t^2 - \widetilde \Sigma_t^2 = O_p\p{\gamma_n^{-1}\p{\sum_{a=0,1}\br{({\p{s_{\alpha_{a,t}} + s_{\beta_{a,t}}}\log d}/{n\gamma_n})^{1/2} +\gamma_n({s_{\alpha_{a,t}} s_{\beta_{a,t}} \log d}/{n\gamma_n})^{1/2}}}}.
\end{align*}

On the other hand, since
\begin{align*}
       \Sigma_{t}^2
       &= \E\sbr{\frac{1-\Gamma_i}{\p{1-\gamma_n}^2}
       \p{X_i^\top \beta^*_{1,t} - X_i^\top \beta^*_{0,t} - \tau_{t}}^2 }\\
    &\quad + \E\sbr{\frac{\Gamma_i A_i}{\p{1-\gamma_n}^2}
    \p{\frac{1-g\p{X_i^\top \alpha^*_{1,t}}}{g\p{X_i^\top \alpha^*_{1,t}}}}^2
    \p{Y_i - X_i^\top \beta^*_{1,t}}^2}\\
    &\quad + \E\sbr{\frac{\Gamma_i\p{1-A_i}}{\p{1-\gamma_n}^2}
    \p{\frac{1-g\p{X_i^\top \alpha^*_{0,t}}}{g\p{X_i^\top \alpha^*_{0,t}}}}^2
    \p{Y_i - X_i^\top \beta^*_{0,t}}^2},
\end{align*}
we have $\E\sbr{\widetilde \Sigma_t^2 - \Sigma_t^2} = 0$. Then
\begin{align*}
    \E\sbr{\p{\widetilde \Sigma_t^2 - \Sigma_t^2}^2} &\leq n^{-1}\E\sbr{{1-\Gamma_i}/{\p{1-\gamma_n}^4} \p{X_i^\top \beta^*_{1,t} - X_i^\top \beta^*_{0,t} - \tau_{t}}^4 }  \\
    &\qquad + n^{-1}\E\sbr{{\Gamma_i A_i}/{\p{1-\gamma_n}^4}\p{{1-g\p{X_i^\top \alpha^*_{1,t}}}/{g\p{X_i^\top \alpha^*_{1,t}}}}^4\p{Y_i - X_i^\top \beta^*_{1,t}}^4}\\
    &\qquad + n^{-1}\E\sbr{{\Gamma_i\p{1-A_i}}/{\p{1-\gamma_n}^4}\p{{1-g\p{X_i^\top \alpha^*_{0,t}}}/{g\p{X_i^\top \alpha^*_{0,t}}}}^4\p{Y_i - X_i^\top \beta^*_{0,t}}^4}\\
    &\leq c_0^{-4} n^{-1}\E\sbr{ \p{X_i^\top \beta^*_{1,t} - X_i^\top \beta^*_{0,t} - \tau_{t}}^4 }  \\
    &\qquad + c_0^{-4}k_0^{-4}\gamma_n^{-4}n^{-1}\E\sbr{\Gamma_i\p{Y_i\p{1} - X_i^\top \beta^*_{1,t}}^4}\\
    &\qquad + c_0^{-4}k_0^{-4}\gamma_n^{-4}n^{-1}\E\sbr{\Gamma_i\p{Y_i\p{0} - X_i^\top \beta^*_{0,t}}^4}.
\end{align*}
Since
\begin{align*}
    &\norm{X_i^\top \beta^*_{1,t} - X_i^\top \beta^*_{0,t} - \tau_{t}}_{\P,4} \leq \norm{X_i^\top \beta^*_{1,t}}_{\P,4} + \norm{X_i^\top \beta^*_{0,t}}_{\P,4} + \abs{\tau_{t}} \leq C_3\sigma + \abs{\tau_{t}},\\
    &\E\sbr{\Gamma_i\p{Y_i\p{1} - X_i^\top \beta^*_{1,t}}^4} = \gamma_n\E\sbr{\p{Y_i\p{1} - X_i^\top \beta^*_{1,t}}^4 \mid \Gamma_i=1} \leq \gamma_n\sigma_w^4,\\
    &\E\sbr{\Gamma_i\p{Y_i\p{0} - X_i^\top \beta^*_{0,t}}^4} = \gamma_n\E\sbr{\p{Y_i\p{0} - X_i^\top \beta^*_{0,t}}^4 \mid \Gamma_i=1} \leq \gamma_n\sigma_w^4,
\end{align*}
we have
\begin{align*}
     \E\sbr{\p{\widetilde \Sigma_t^2 - \Sigma_t^2}^2} \leq c_0^{-4}\p{\p{C_3\sigma + \abs{\tau_{t}}}^4 + k_0^{-4}\sigma_w^4} n^{-1}\gamma_n^{-3}.
\end{align*}
It follows that
$
    \widetilde \Sigma_t^2 - \Sigma_t^2 = O_p\p{\gamma_n^{-1}\p{n\gamma_n^{-1/2}}}.
$

In conclusion, by Lemma~\ref{lemma causal transportability tau tilde normal}, $\Sigma_t^2 \asymp \gamma_n^{-1}$ and 
\begin{align*}
    \widehat \Sigma_t^2 &= \widehat \Sigma_t^2 -\widetilde \Sigma_t^2 + \widetilde \Sigma_t^2 - \Sigma_t^2 + \Sigma_t^2 \\
    &= \Sigma_t^2\br{1 + O_p\p{\sum_{a=0,1}\br{(\frac{\p{s_{\alpha_{a,t}} + s_{\beta_{a,t}}}\log d}{n\gamma_n})^{1/2} +\gamma_n(\frac{s_{\alpha_{a,t}} s_{\beta_{a,t}} \log d}{n\gamma_n})^{1/2}}}}.
\end{align*}
\end{proof}

\subsection{Proof of Theorem~\ref{theorem causal transportability Asymptotics body}}

\begin{proof}
    Theorem~\ref{theorem causal transportability Asymptotics body} follows from Lemmas~\ref{lemma causal transportability tau tilde normal}, \ref{lemma causal transportability hat tau - tilde tau consistency}, and \ref{lemma causal transportability hat sigma = sigma(1+o(1))} by repeating the proof of Lemma~\ref{lemma key results for transportability}.
\end{proof}
\end{document}